\documentclass{article}
\usepackage[utf8]{inputenc}
\usepackage{amsmath}
\usepackage{amsthm}											
\usepackage{amsfonts}
\usepackage{amssymb}
\usepackage{braket}
\usepackage{bm}
\usepackage{caption,cite}
\usepackage{sidecap}
\usepackage{subcaption}
\usepackage{comment}
\usepackage{esint}
\usepackage[a4paper, total={6.5in, 10in}]{geometry}
\usepackage{tikz}
\usetikzlibrary{matrix}
\usepackage{graphicx}
\usepackage{mathtools}
\usepackage{pdflscape}
\usepackage{BOONDOX-uprscr}
\usepackage{xcolor}
\graphicspath{ {images/} }
\usepackage[section]{placeins}
\DeclareMathOperator{\diag}{diag}

\newcommand{\norm}[1]{\left\lVert#1\right\rVert}

\newcommand{\vertiii}[1]{\left\lvert\kern-0.25ex\left\lvert\kern-0.25ex\left\lvert #1 \right\rvert\kern-0.25ex\right\rvert\kern-0.25ex\right\rvert}

\newtheorem{definition}{Definition}		
\newtheorem{assumption}{Assumption}
\newtheorem{remark}{Remark}
\newtheorem{theorem}{Theorem}

\newtheorem{lemma}{Lemma}
\newtheorem{proposition}{Proposition}
\newtheorem{corollary}{Corollary}
\newtheorem{example}{Example}
\numberwithin{lemma}{section}
\numberwithin{example}{section}
\numberwithin{proposition}{section}
\numberwithin{equation}{section}
\numberwithin{theorem}{section}
\numberwithin{corollary}{section}
\numberwithin{remark}{section}
\numberwithin{definition}{section}
\numberwithin{assumption}{section}
\numberwithin{figure}{section}
\newtheorem*{definition*}{Definition}		
\newtheorem*{assumption*}{Assumption}
\newtheorem*{remark*}{Remark}
\newtheorem*{theorem*}{Theorem}
\newtheorem*{lemma*}{Lemma}
\newtheorem*{proposition*}{Proposition}
\newtheorem*{corollary*}{Corollary}
\newtheorem*{example*}{Example}
\newtheorem*{conjecture*}{Conjecture}

\newcommand{\eps}{\varepsilon}

\newcommand{\aver}[1]{\langle #1 \rangle}
\newcommand{\absval}[1]{\left| #1 \right|}

\newcommand{\cG}{\mathcal{G}}

\newcommand{\cHm}{\mathcal{H}_1}

\newcommand{\hperp}{h_{12}}
\newcommand{\hathperp}{\hat{h}_{12}}
\newcommand{\hathperprad}{\hat{h}_{12, {\rm rad}}}
\newcommand{\hathperpang}{\hat{h}_{12, {\rm ang}}}

\newcommand{\Hbm}{H^{{\rm eff}}}

\newcommand{\Hperp}{H_{12}}

\newcommand{\genH}{A} 
\newcommand{\gencH}{\cH_0} 

\newcommand{\hcs}{\mathfrak{h}}

\newcommand{\cC}{\mathcal{C}}
\newcommand{\cD}{\mathcal{D}}
\newcommand{\rotop}{\mathcal{U}}

\newcommand{\phf}{g} 
\newcommand{\phph}{\chi} 
\newcommand{\dpm}{\varphi} 
\newcommand{\dpmH}[1]{\dpm^{(#1)}} 
\newcommand{\tf}{f} 
\newcommand{\tff}[1]{\mathfrak{f}^{(#1)}} 
\newcommand{\mswf}{\phi}

\newcommand{\eshift}{E} 

\newcommand{\res}{\rho}
\newcommand{\resgen}{\varrho}

\newcommand{\sdp}{\eta}
\newcommand{\sdpt}{\eta}
\newcommand{\srp}{\delta}
\newcommand{\swp}{\nu}

\newcommand{\ellC}{\mu}

\newcommand{\cRall}{\cQ}

\newcommand{\radNNN}{\zeta}
\newcommand{\radpNN}{\xi}
\newcommand{\angs}{\lambda}
\newcommand{\angp}{\mu}

\newcommand{\HBM}{H^{(1)}}
\newcommand{\HBMNNN}{H^{({\rm NNN})}}
\newcommand{\HBMgNN}{H^{(\nabla, {\rm NN})}}
\newcommand{\HBMt}{H^{(2)}}
\newcommand{\HBMta}{H^{(2,{\rm all})}}
\newcommand{\HBMvar}[1]{H^{(#1)}}

\newcommand{\tffta}{\tff{2, {\rm all}}}

\newcommand{\Hfull}{H^{{\rm eff}}}
\newcommand{\Hoff}{H^{{\rm off}}}

\newcommand{\HBMsd}{S}
\newcommand{\HBMtwist}{S_{\rm twist}}
\newcommand{\wfntwist}{F_{\rm twist}}

\newcommand{\HBMptwist}{S'_{\rm twist}}
\newcommand{\wfnptwist}{F'_{\rm twist}}

\newcommand{\dtwist}{\mathfrak{D}}

\newcommand{\tj}{\mathfrak{t}}

\newcommand{\gradmax}{J}
\newcommand{\gradmaxtwo}{J_2}

\newcommand{\fw}{\mathfrak{w}}
\newcommand{\ls}{\mathfrak{d}}
\newcommand{\co}{o}

\newcommand{\angvar}[1]{\mathscr{#1}}

\newcommand{\svec}{s}

\newcommand{\cB}{\mathcal{B}}
\newcommand{\cA}{\mathcal{A}}
\newcommand{\cH}{\mathcal{H}}
\newcommand{\cS}{\mathcal{S}}
\newcommand{\cT}{\mathcal{T}}
\newcommand{\cU}{\mathcal{U}}

\newcommand{\sm}{\mathscr{m}}

\newcommand{\afn}{\mathscr{z}}
\newcommand{\hoppingparam}{\gamma}
\newcommand{\lparam}{\lambda_0}

\newcommand{\cM}{\mathcal{M}}

\newcommand{\cP}{\mathcal{P}}

\newcommand{\cQ}{\mathcal{Q}}
\newcommand{\cR}{\mathcal{R}}

\newcommand{\rot}{\mathfrak{R}}

\newcommand{\eV}{\text{ eV}}
\newcommand{\meV}{\text{ meV}}

\newcommand{\hoppingT}{\mathfrak{T}}
\newcommand{\Tu}{\hoppingT_{2,2}}

\newcommand{\microtime}{T}
\newcommand{\macrotime}{t}

\def\XXint#1#2#3{{\setbox0=\hbox{$#1{#2#3}{\int}$ }
\vcenter{\hbox{$#2#3$ }}\kern-.6\wd0}}

\renewcommand{\vec}[1]{#1}

\usepackage{hyperref}
\usepackage[capitalise,nameinlink]{cleveref}

\usepackage{booktabs}
\usepackage{multirow}

\newcommand{\dee}{\ensuremath{\textrm{d}}}

\newcommand{\inty}[4]{\ensuremath{ \int_{#1}^{#2} \! #3 \, \dee#4 }}

\newcommand{\sq}[1]{{\color{black} #1}}
\newcommand{\srq}[1]{{\color{black} #1}}

\title{Higher-order continuum models for twisted bilayer graphene}

\author{Solomon Quinn, Tianyu Kong, Mitchell Luskin, Alexander B. Watson}

\begin{document}



\maketitle

\begin{abstract}
    The first-order continuum PDE model proposed by Bistritzer and MacDonald in \cite{bistritzer2011moire} accurately describes the single-particle electronic properties of twisted bilayer graphene (TBG) at small twist angles. In this paper, we 
    obtain higher-order corrections to the Bistritzer-MacDonald model via a systematic multiple-scales expansion.
    We prove that the solution of 
    the resulting higher-order PDE model
    accurately approximates the corresponding tight-binding wave function under a natural choice of parameters and given initial conditions that are spectrally localized to the monolayer Dirac points. 
    Numerical simulations of tight-binding and continuum dynamics demonstrate the validity of the higher-order continuum model. 
    Symmetries of the higher-order models are also discussed. This work extends the analysis from \cite{watson2023bistritzer}, which rigorously established the validity of the (first-order) BM model.
\end{abstract}

\section{Introduction}
\subsection{Overview}
Twisted bilayer graphene (TBG) is obtained by stacking two sheets of graphene on top of each other with a relative twist \cite{twistronics}. At incommensurate twist angles, TBG is not periodic and thus does not admit a Brillouin zone or periodic branches of spectrum \cite{genkubo17,dos17,momentumspace17}. Instead, the 
atoms form a structure which is approximately periodic with respect to the so-called moir\'e lattice, whose unit-cell area is inversely proportional to the square of the twist angle.
Thus, at small twist angles, it becomes intractable to simulate electron dynamics directly from first principles, making it essential to derive accurate effective models that reduce the complexity of the problem.

A first-order continuum model for TBG was proposed in \cite{castro2007,bistritzer2011moire}. The Bistritzer-MacDonald (BM) model \cite{bistritzer2011moire} is an exactly periodic, 
first-order PDE model that accurately captures single-particle electronic properties of TBG. Due to the periodicity of the BM Hamiltonian, one can apply Floquet-Bloch theory to parametrize its spectrum by two-dimensional surfaces over a Brillouin zone. The BM model predicts a sequence of twist angles (called ``magic angles’’) at which two of these surfaces are approximately flat (known as ``flat bands''). 
This prediction was later validated by experiments which showed that at the largest magic angle $\approx 1^\circ$, TBG admits superconducting and Mott insulating phases at low temperature \cite{cao2018correlated, cao2018unconventional}. Indeed, although the physical mechanism that produces them is not yet fully understood, it is believed that these phases occur because near to the magic angle the electron-electron energy scale becomes large compared to the single-particle energy scale defined by the flat band width.

The validity of the BM model was rigorously proven in \cite{watson2023bistritzer,kong2023modeling}, showing that for small twist angles and weakly interacting layers, the solution to the BM Schr\"odinger equation under appropriate initial conditions serves as an accurate approximation of the corresponding tight-binding dynamics. A formal derivation of a continuum model for twisted bilayer graphene (TBG) from Density Functional Theory has been given in \cite{cancesbm2023}. The goal of the present work is to improve on 
the results in \cite{watson2023bistritzer,kong2023modeling} by rigorously justifying 
a modified BM model with higher order error compared to the first-order BM model. 
In addition to proving a theorem to this effect, we also demonstrate
numerically that solutions to the higher-order model are indeed more accurate than, and have qualitative differences from, their first-order analogues. For example, the higher-order and tight-binding wave-functions exhibit a spiral pattern which is absent from the first-order BM model; see Figure \ref{fig:wave_packet_spiral}. Note that higher-order continuum models have been proposed in the physics literature, for example \cite{vafek2023continuum}, but the precise order of convergence and range of validity of these models has not been rigorously established.

\begin{figure}[h!]
    \centering
    \includegraphics[width=.65\textwidth]{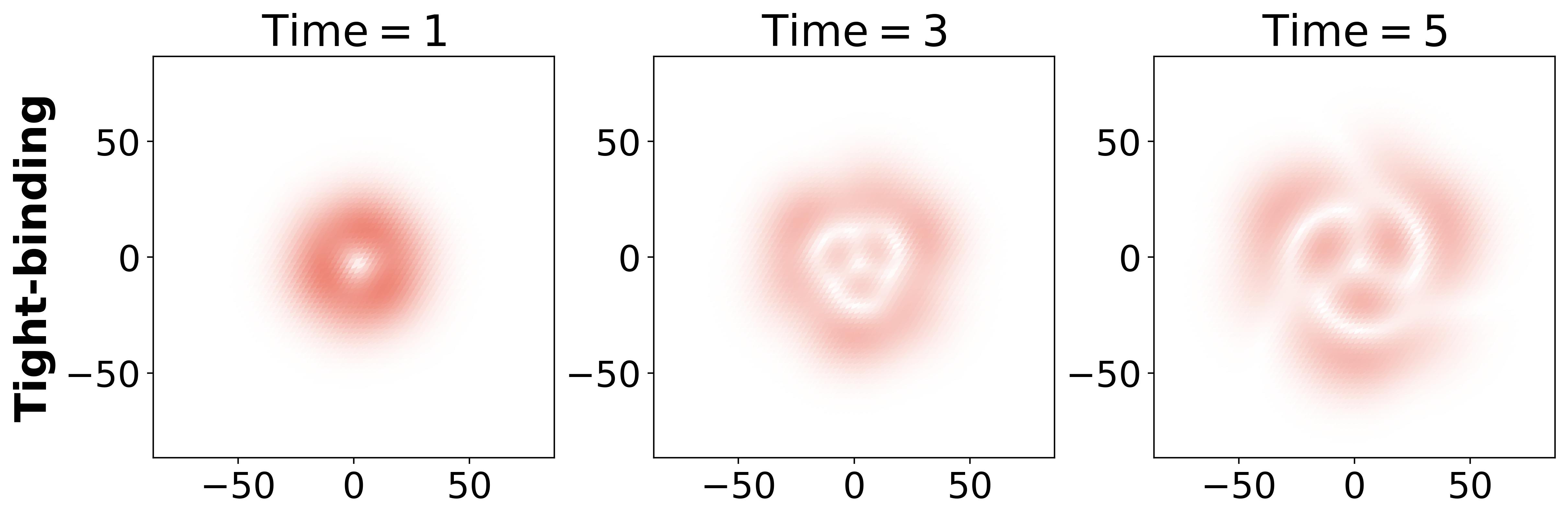}
    \includegraphics[width=.65\textwidth]{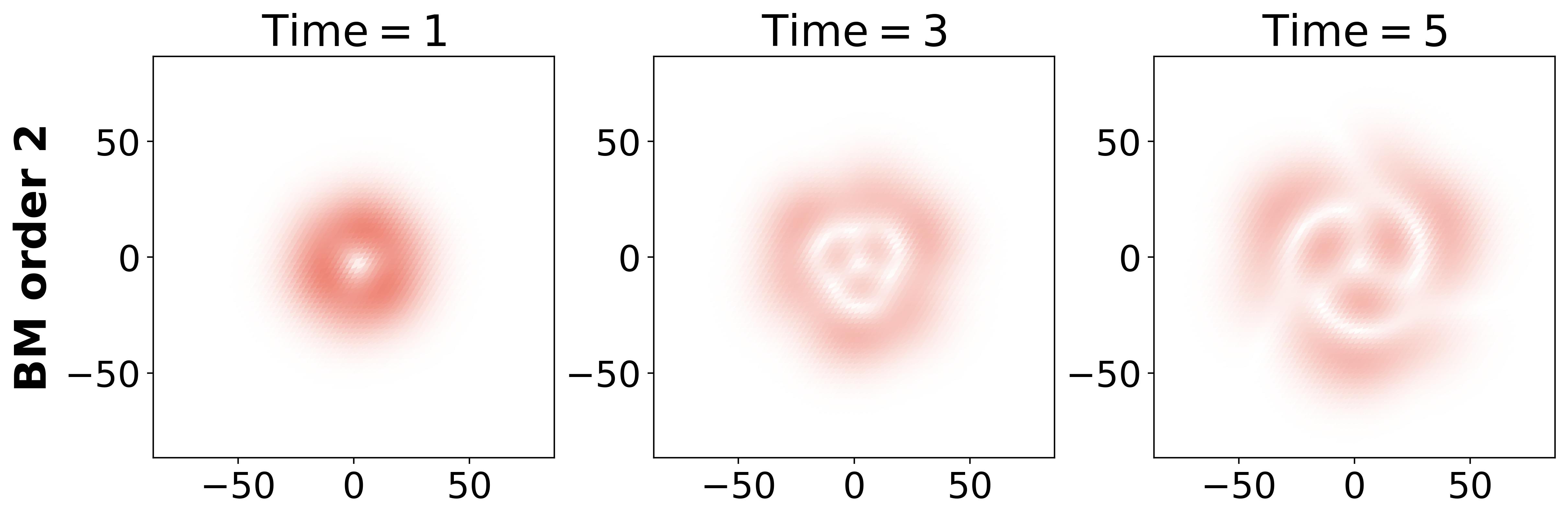}
    \includegraphics[width=.65\textwidth]{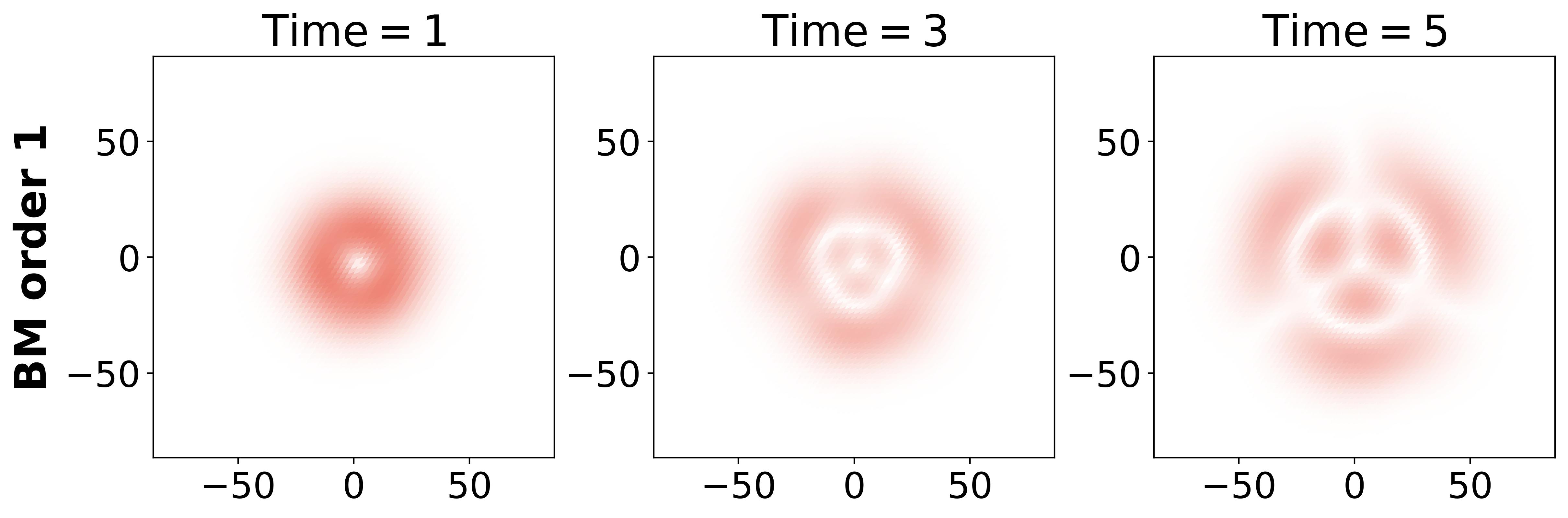}
    \caption{The modulus of the wave-function for the tight-binding model and the first and second order continuum models. The initial data is concentrated on a flat band. The group velocity is approximately zero. The tight-binding model display a spiral pattern, which is also present in the second order BM model. The axes have units \r{A}, and one unit of time is $\hbar \cdot \mathrm{eV}^{-1} \approx 6.6 \times 10^{-16} \text{ s}$.  }
    \label{fig:wave_packet_spiral}
\end{figure}

\subsection{Main results}\label{subsec:main}
This paper is concerned with approximating the solution to the initial value problem
\begin{align}\label{eq:tb}
    i \partial_\microtime \psi = H \psi, \qquad \psi \vert_{\microtime = 0} = \psi_0,
\end{align}
where $H$ is the discrete Hamiltonian for TBG defined by \eqref{eq:H}-\eqref{eq:intra}-\eqref{eq:Hperp}.
In this model, each layer of graphene is an infinite, hexagonal lattice of carbon atoms. We assume that the two layers are a small twist angle and arbitrary translation away from perfect alignment.
Electrons are localized to the carbon atom positions of each layer (tight-binding approximation) and interactions between electrons are ignored (except through their mean field in the modeling of the tight-binding parameters by Density Functional Theory). 
The probability that electrons move between lattice sites is modeled by 
intralayer and interlayer hopping functions depending only on the displacement between the lattice sites (a reduction known as the ``two-center approximation'' in the physics literature) and respectively satisfying Assumptions \ref{assumption:h} and \ref{assumption:hperp2} below. These assumptions require that the hopping functions decay rapidly at infinity. Moreover, Assumption \ref{assumption:h} imposes $2\pi/3$-rotation symmetry on the intralayer hopping function, while Assumption \ref{assumption:hperp2} roughly states that the interaction between the two layers is weak. For ease of exposition, we also assume that the interlayer hopping function has the separable form \eqref{eq:sep}, though our results would easily extend to arbitrary smooth functions with the appropriate decay properties; see Remark \ref{remark:linearity}.

With $f_0$ a sufficiently regular $\mathbb{C}^4$-valued function, the (dimensionless) BM model \cite{bistritzer2011moire} is the following initial value problem
\begin{align}\label{eq:BM}
    (i\partial_t - \HBM) f = 0, \qquad f (r,0) = f_0 (r),
\end{align}
where the self-adjoint operator $\HBM : H^1 (\mathbb{R}^2; \mathbb{C}^4) \to L^2 (\mathbb{R}^2; \mathbb{C}^4)$ is defined by
\begin{align}\label{eq:HBM_1}
\begin{split}
    \HBM &:= \begin{pmatrix}
        L & \hoppingT (r)\\
        \hoppingT^\dagger (r) & L
    \end{pmatrix}, \qquad
    L := \begin{pmatrix}
        0 & \alpha (D_{r_1} - i D_{r_2})\\
        \bar{\alpha} (D_{r_1} + i D_{r_2}) & 0
    \end{pmatrix},\\
    \hoppingT(r) &:= \frac{1}{|\Gamma|} \Bigg(
    \lambda_0 e^{-i \svec_1 \cdot r} \begin{pmatrix}
        1 & 1\\
        1 & 1
    \end{pmatrix} + 
    \lambda_2 e^{-i b_2 \cdot \ls} e^{-i \svec_2 \cdot r} \begin{pmatrix}
        1 & e^{-i2\pi/3}\\
        e^{i2\pi/3} & 1
    \end{pmatrix}\\
    & \hspace{5cm} 
    +\lambda_4 e^{ib_1 \cdot \ls}e^{-i \svec_3 \cdot r} \begin{pmatrix}
        1 & e^{i2\pi/3}\\
        e^{-i2\pi/3} & 1
    \end{pmatrix}\Bigg),
\end{split}
\end{align}
and the vectors $\svec_j$ are related by $2\pi/3$ rotations and given by
$$\svec_1 := \frac{4\pi \beta}{3a} (0,-1), \qquad 
    \svec_2 := \frac{4\pi \beta}{3a} \left(
        \frac{\sqrt{3}}{2}, \frac{1}{2}
    \right), \qquad
    \svec_3 := \frac{4\pi \beta}{3a} \left(-\frac{\sqrt{3}}{2}, \frac{1}{2}\right).$$
Here, $\alpha \ne 0$ is the Fermi velocity (see Section \ref{subsec:existence}), $|\Gamma|$ is the area of the fundamental unit cell defined in \eqref{eq:lattice}, and the constants $\lambda_0, \lambda_2, \lambda_4 \in \mathbb{C}$ are defined in \eqref{eq:rad_ang}. In the special case of a $2\pi/3$-rotation symmetric interlayer hopping function (see, e.g. \cite{bistritzer2011moire, watson2023bistritzer}), $\lambda_0 = \lambda_2 = \lambda_4$. 
The 
monolayer lattice constant $a > 0$ is defined above \eqref{eq:lattice_defn}, and the 
monolayer reciprocal lattice vectors $b_1$ and $b_2$ are given by \eqref{eq:reciprocal_lattice}. The model allows for a relative shift of the layers by $\ls \in \Gamma$; see Section \ref{subsec:bilayer}. Finally, the parameter $\beta \ne 0$ \srq{controls the twist angle of the discrete model \eqref{eq:tb}; see \eqref{eq:beta_def} for its precise definition.} 

The entries of $f$ model electron densities on each layer and graphene sublattice, with the first two entries corresponding to graphene sublattices $A$ and $B$ on layer $1$, and the next two entries corresponding to the sublattices on layer $2$.
Therefore, the operator $L$ models the single graphene sheet, while $\hoppingT (r)$ captures the coupling between layers.

We see that $\HBM$ is a first-order partial differential operator with smooth 
coefficients. 
Moreover, up to multiplication by a diagonal matrix whose nonzero entries are in $\{1, e^{i2\pi/3}, e^{-i2\pi/3}\}$, $\hoppingT(r)$ is periodic with respect to the moir\'e lattice defined in \eqref{eq:moire}; see Proposition \ref{prop:continuum_symmetries} (\ref{symm:translation}).


Starting from the discrete initial value problem \eqref{eq:tb} describing TBG, \cite{watson2023bistritzer} rigorously established the validity of the BM model \eqref{eq:BM} in a particular asymptotic regime. 
This asymptotic regime is defined by a small parameter $0 < \eps < 1$, which is proportional to three dimensionless quantities: 
\begin{enumerate}
    \item The spectral width of the wave-packet initial condition $\psi_0$ from \eqref{eq:tb}; see Theorems \ref{thm:main2} and \ref{thm:main}.
    \item The ratio of the interlayer hopping energy to the monolayer band width; see Remark \ref{remark:interlayer_hopping_strength}.
    \item The relative twist angle of the two layers (the constant of proportionality is $\beta$).
\end{enumerate}
With $\cH$ the natural $\ell^2$ Hilbert space defined by \eqref{eq:cH}, \cite[Theorem 3.1]{watson2023bistritzer} then states that the solution to the first-order Schr\"odinger equation \eqref{eq:BM} can be used to obtain a discrete function $\phph_{{\rm BM}}$ satisfying
$$\norm{\phph_{{\rm BM}} (\microtime)-\psi (\microtime)}_\cH \le C_{t_0} \eps^{1+\sdpt_-} \microtime, \qquad \qquad 0 \le \microtime \le t_0/\eps, \quad 0 < \eps < 1$$
for any fixed $t_0 > 0$,
where the parameter $0 < \sdpt_- < 1$ depends on the interlayer hopping function and can often be taken arbitrarily close to $1$; see Remark 3.5 there. We emphasize that although the discrete initial value problem depends on $\eps$, the emergent effective model \eqref{eq:BM} does not. For a detailed discussion on the experimental relevance of 
taking $\eps \to 0$, see \cite[Section 1.1]{watson2023bistritzer}.

In this paper, we identify the next-order corrections to the BM model in the same asymptotic regime. These terms consist of the periodic point-wise multiplication operator $\HBMNNN$, and the first- and second-order differential operators $\HBMgNN$ and $\HBMt$ all defined in Section \ref{subsec:multiscale}. 
Our effective model is
\begin{align}\label{eq:eff}
    (i \partial_t - \Hfull) f = 0, \qquad f (r,0) = f_0 (r),
\end{align}
where the corrected BM Hamiltonian $\Hfull : H^2 (\mathbb{R}^2; \mathbb{C}^4) \to L^2 (\mathbb{R}^2; \mathbb{C}^4)$ is given by
\begin{align}\label{eq:H_eff}
    \Hfull := \HBM + \radNNN (\eps) \HBMNNN + \radpNN (\eps) \HBMgNN + \eps \HBMt,
\end{align}
with 
$\radNNN (\eps)$ and $\radpNN (\eps)$ both $o(1)$ and defined in \eqref{eq:rad_ang}.
Although $\Hfull$ depends on $\eps$ via these prefactors, 
the operators $\HBM, \HBMNNN, \HBMgNN$ and $\HBMt$ are all independent of $\eps$.





We note that the asymptotic order of the functions $\radNNN$ and $\radpNN$ are in general less than $1$, but greater than $0,$ in the limit $\eps \to 0$. The presence of these non integer-powers of $\eps$ allows our theory to apply to a large class of interlayer hopping functions; see Remark \ref{remark:ex_rad} for an example. 
The functions $\radNNN$ and $\radpNN$ are determined by the spectral decay (in $\eps$) of the interlayer hopping function.

The two main results of this paper are contained in Section \ref{sec:effective}. Roughly speaking, they state that given a sufficiently smooth initial condition $f_0$, the solution to the \emph{continuum} Schr\"odinger equation \eqref{eq:eff} can be used to construct an \emph{approximate} solution $\phph$ of the \emph{discrete} Schr\"odinger equation \eqref{eq:tb}. 
More specifically, 
Theorem \ref{thm:main} establishes that for any $t_0$,
$$\norm{\phph (\microtime)-\psi (\microtime)}_\cH \le C_{t_0} \eps^{2+\sdpt_-} \microtime, \qquad \qquad 0 \le \microtime \le t_0/\eps, \quad 0 < \eps < 1,$$
where (as before) $0 < \sdpt_- < \sdpt \le 1$ depends on the decay of the interlayer hopping function; see Assumption \ref{assumption:hperp2} where the relevant decay parameter $\sdpt$ is introduced. In Theorem \ref{thm:main2}, the approximate solution (now $\phi$ instead of $\phph$) is obtained by a multiscale approximation of the solution $f$ to \eqref{eq:eff}. 
The two theorems share the same order of convergence, though Theorem \ref{thm:main2} requires less regularity on $f_0$ (and fewer constraints on parameters of the tight-binding model) than Theorem \ref{thm:main}.
\srq{We stress that our multiscale analysis involves terms of $O (\zeta (\eps))$ and $O(\xi (\eps))$, which are typically not 
$O (\eps)$. The decay assumption \eqref{eq:neighbor_bds} on the interlayer hopping function ensures that $|\radNNN^2 (\eps)| + |\radpNN^2 (\eps)| = o(\eps)$, which in turn will guarantee that the multiscale expansion for $f$ from \eqref{eq:eff} has only four non-negligible terms; see Lemma \ref{lemma:ae2}.}
\srq{These terms are defined by four $\eps$-independent initial value problems \eqref{eq:it}.}

\srq{Our results improve the existing $\eps^{1+\sdpt_-} \microtime$ order of convergence from \cite{watson2023bistritzer} to $\eps^{2+\sdpt_-} \microtime$. 
We refer to Example \ref{ex:2} for a specific interlayer hopping function satisfying the assumptions of our theory, where an upper bound for the corresponding decay rate $\sdpt_-$ is explicitly calculated.}

In addition to obtaining a higher order of convergence, we also weaken assumptions on the underlying tight-binding model compared with \cite{watson2023bistritzer}; see, e.g. Remark \ref{remark:weaken_assumptions}. For example, 
our theory applies to all $2\pi/3$-rotation symmetric intralayer hopping functions that decay super-algebraically at infinity (see Assumption \ref{assumption:h}), which extends the existing results on nearest neighbor intralayer hopping functions. 
In particular, we prove with Corollary \ref{corr:Dirac} that provided $\alpha \ne 0$, any monolayer Hamiltonian whose hopping function satisfies Assumption \ref{assumption:h} has a Dirac cone at $k = K := \frac{4\pi}{3a}(1,0)$; see also Definition \ref{def:Dirac}. 
Note that the existence of this Dirac cone will play an essential role in proving Theorems \ref{thm:main2} and \ref{thm:main}. Regarding the \emph{interlayer} hopping function, 
our Assumption \ref{assumption:hperp2} allows for angular dependence which was absent from \cite[Assumption 2.1]{watson2023bistritzer}.


\subsection{Structure of paper}
We introduce the tight-binding model for twisted bilayer graphene in Section \ref{sec:tight-binding}, which can be written concisely in block form \eqref{eq:H}. 
With Corollary \ref{corr:Dirac}, we prove
that each diagonal block generically admits Dirac cones (see Definition \ref{def:Dirac}). A stronger result concerning the dispersion surfaces of each diagonal block near the Dirac points is given by Theorem \ref{thm:Dirac}.
Our main results (see also Section \ref{subsec:main}) are presented in Section \ref{sec:effective}. In Section \ref{sec:symmetries}, we establish symmetries of the tight-binding (Sections \ref{subsec:symmetries_mono} and \ref{subsec:symmetries_bilayer}) and effective models (Section \ref{subsection:symmetries_continuum}), and observe the direct correspondence between the discrete and continuum symmetries with Remark \ref{remark:correspondence}. Section \ref{sec:numerical} contains numerical simulations that illustrate our theory. We compute the electron dynamics in TBG for
parameter values outside the scope of our theoretical results, and still observe that the higher-order effective model is more accurate than the first-order model, with qualitative differences between their solutions. Concluding remarks are given in Section \ref{sec:conclusion}, and proofs of the results from Sections \ref{sec:effective} and \ref{sec:symmetries} are deferred to Appendix \ref{sec:proofs}.

\subsection{Related literature}
This work builds on the analysis from \cite{watson2023bistritzer}, which rigorously established the validity of the first-order BM model under appropriate initial conditions. The main contribution of the present paper is to show that higher-order terms improve the order of accuracy of the BM dynamics (with analytical results in Section \ref{sec:effective} and numerical simulations in Section \ref{sec:numerical}). Compared with \cite{watson2023bistritzer}, we consider a larger class of hopping functions, allowing for longer-range intralayer interactions and
interlayer hopping that depends on the displacement (not just distance) between two atoms. \srq{Another contribution of this paper is to prove the generic existence of monolayer Dirac cones under these more general assumptions; see Definition \ref{def:Dirac} and Corollary \ref{corr:Dirac}.}

The higher-order effective model in this paper is similar to the one derived by Vafek-Kang in \cite{vafek2023continuum}. Instead of working with specific values of physical parameters, we perform a rigorous asymptotic analysis in the limit of small twist angle and weak interlayer coupling. As in \cite{bistritzer2011moire, vafek2023continuum}, we assume that the Fourier transform of the interlayer hopping function decays rapidly away from the origin.  
The slowly varying fermion fields in \cite[Equation (12)]{vafek2023continuum} play the role of our wavepackets. 
Kang-Vafek also derive an effective model for relaxed atomic configurations; see \cite{kang2023pseudomagnetic}. Our ongoing work is extending our analysis to this setting.

Canc\`es-Garrigue-Gontier formally derive an effective second-order PDE model for TBG in \cite{cancesbm2023}. Their derivation uses variational methods and Density Functional Theory (DFT), and bypasses the tight-binding model that we consider in Section \ref{sec:tight-binding}. We refer to \cite{cances2023semiclassical} for a detailed spectral analysis of first-order BM-type models of TBG using pseudodifferential calculus. As opposed to the wavepacket setting considered here, Theorems 3.7 and 3.8 in \cite{cances2023semiclassical} establish explicit formulas for a normalized density of states corresponding to the DFT and BM Hamiltonians.

Our work is also related to the analysis of Fefferman-Weinstein in \cite{fefferman2014wave}, which rigorously derives an effective Dirac equation from a PDE model of monolayer graphene. As with \cite{watson2023bistritzer} and the present paper, 
the effective model in \cite{fefferman2014wave} emerges from initial conditions that are spectrally localized to the (monolayer) Dirac point. For an effective Dirac model of edge states in monolayer graphene-type materials, see \cite{FLW-ES-2015}.

\srq{Finally, we comment on the relationship between our paper and the work of Quan-Watson-Massatt in \cite{quan2024construction}. Both papers establish the validity of a continuum model for twisted bilayer graphene, starting from a discrete model. The continuum Hamiltonian proposed by Quan-Watson-Massatt depends nontrivially on various asymptotic parameters, while the effective operators $\HBM, \HBMNNN, \HBMgNN$ and $\HBMt$ in \eqref{eq:H_eff} are all independent of $\eps$. Quan-Watson-Massatt quantify the difference of the two models using a density of states, defined as the normalized trace of a Gaussian function of the Hamiltonian. Our results concern the \emph{time-dependent} problem \eqref{eq:tb}, whose solution $e^{-itH} \psi_0$ is not uniquely determined by the above density of states. Although we believe that regularizing the propagator $e^{-itH}$ could produce an accurate approximation of $\psi$ (for the initial conditions considered in this paper, and over long times), 
this question is outside the scope of the present paper.} 

\subsection{Notation}
\begin{itemize}
    \item The adjoint  of a linear operator $A$ is denoted by $A^\dagger$.
    \item The complex conjugate of a number $z \in \mathbb{C}$ is deonted by $\bar{z}$.
    \item We use the convention that $(v_1, \dots, v_d) \in \mathbb{R}^d$ is a column vector.
    \item If $v = (v_1, \dots, v_d) \in \mathbb{R}^d$, then $|v| := \sqrt{\sum_{j=1}^d |v_j|^2}$.
    \item $1 \{ \cdot \}$ is the indicator function.
    \item $C^\infty_b$: the space of smooth, bounded functions whose derivatives are all bounded.
    \item $\nabla u$ is the gradient of $u$ and $\nabla^2 u$ is the Hessian of $u$.
    \item $D_\alpha := \frac{1}{i} \partial_\alpha$ and $D := \frac{1}{i} \nabla$. 
    \item $\delta (\cdot)$ is the Dirac delta function, and $\delta_{ij}$ is the Kronecker delta function.
    \item $\hat{u} (k) := \int_{\mathbb{R}^d} e^{-i k \cdot x} u(x) {\rm d} x$ is the Fourier transform of $u$. 
    \item $H^N$ denotes the standard Sobolev space with norm $\norm{u}_{H^N} := \sqrt{\sum_{\alpha\in \mathbb{N}^d, |\alpha| \le N} \norm{\partial^\alpha u}_{L^2}^2}$.
    \item $\rot_\phi := \begin{pmatrix}
    \cos \phi & -\sin \phi\\
    \sin \phi & \cos \phi
    \end{pmatrix}$ is rotation counter-clockwise by $\phi$.
    \item $\mathbb{S} \subset \mathbb{R}^2$ is the unit circle.
\end{itemize}

\section{Tight-binding models}\label{sec:tight-binding}

\subsection{Monolayer tight-binding model} \label{subsec:mono}

We start by introducing a general tight-binding model for monolayer graphene. Let $a>0$, define
\begin{align}\label{eq:lattice_defn}
    a_1 := \left(\frac{1}{2}a,\frac{\sqrt{3}}{2}a\right), \qquad a_2:= \left(-\frac{1}{2}a,\frac{\sqrt{3}}{2}a\right),
\end{align}
\sq{and let $\tau^A, \tau^B \in \mathbb{R}^2$ such that $\tau^A - \tau^B = \left(0,-\frac{1}{\sqrt{3}}a\right)$.}
We let 
\begin{align}\label{eq:lattice}
    \cR := \{n_1 a_1 + n_2 a_2 : (n_1, n_2) \in \mathbb{Z}^2\}, \qquad \Gamma := \{\alpha_1 a_1 + \alpha_2 a_2 : (\alpha_1, \alpha_2) \in [0,1)^2\},
\end{align} 
respectively denote the direct lattice and fundamental (unit) cell. The atom locations are then given by the set $(\cR + \tau^A) \cup (\cR + \tau^B)$; see Figure \ref{fig:lattice} for an illustration.
\begin{figure}[h!]
    \centering
    \includegraphics[width=.75\textwidth]{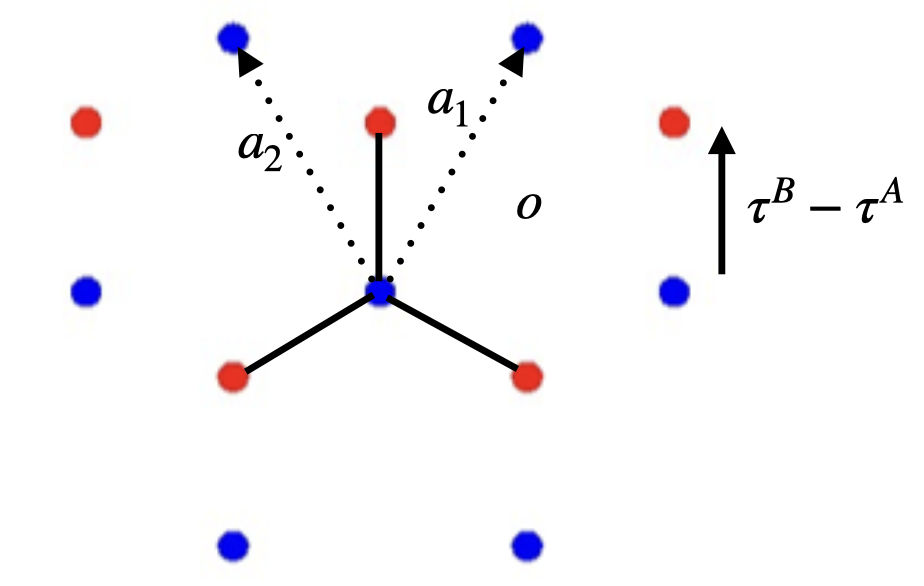}
    \caption{Monolayer graphene lattice with the blue and red dots respectively corresponding to the $A$ and $B$ sites. If the two right-most dots are indexed by $R = 0$ (so that their positions are $\tau^A$ and $\tau^B$), then
    the point $\co := \frac{1}{2} (\tau^A + \tau^B) - (\frac{a}{2}, 0)$ marks the center of the right 
    hexagon.}
    \label{fig:lattice}
\end{figure}
Define $\cR_\pm := \cR \pm (\tau^A - \tau^B)$, 
$\cRall := \cR \cup \cR_+ \cup \cR_-$, and $\cHm := \ell^2(\cR; \mathbb{C}^2)$.
The Hamiltonian $H : \cHm \to \cHm$ is the self-adjoint operator defined by
\begin{align}\label{eq:tight-binding}
    (H \varphi)^\sigma_R = \sum_{R' \in \cR} \sum_{\sigma ' \in \{A,B\}} h (R-R'+\tau^{\sigma, \sigma '}) \varphi^{\sigma'}_{R'},
\end{align}
where $\tau^{\sigma, \sigma'} := \tau^{\sigma} - \tau^{\sigma '}$ and $h : \cRall \rightarrow \mathbb{C}$ is the hopping function.
To ensure that $H$ is self-adjoint, we assume that $h(-r) = \overline{h(r)}$ for all $r \in \cRall$.
Throughout this paper, we will make the following
\begin{assumption}\label{assumption:h}
    Assume that for any $N > 0$, the function $h: \cRall \to \mathbb{C}$ satisfies $|r|^N |h(r)| \to 0$ as $|r| \to \infty$. Moreover, assume that $h (\rot_{2\pi/3} r) = h(r)$ for all $r \in \cRall$.
\end{assumption}
One can easily verify that $\rot_{2\pi/3} \cRall = \cRall$, hence the rotation condition (also known as ``rotation symmetry'') in Assumption \ref{assumption:h} is well defined.
The rapid decay of $h$ at infinity more than ensures that the sum in \eqref{eq:tight-binding} is absolutely convergent for any $\varphi \in \cHm$; it also guarantees that the Bloch transform $\tilde{h}^{\sigma, \sigma'}$ of $h$ defined by \eqref{eq:tildeH} is infinitely smooth.
These assumptions are realistic due to the exponential decay and symmetry properties of Wannier orbitals \cite{ashcroft1976solid}.


The reciprocal lattice vectors $b_1$ and $b_2$ are defined by the condition that $b_i \cdot a_j = 2\pi \delta_{ij}$ for $i,j \in \{1,2\}$, and are given explicitly by
\begin{align}\label{eq:reciprocal_lattice}
    b_1 = \frac{4\pi}{\sqrt{3} a} \left(\frac{\sqrt{3}}{2}, \frac{1}{2}\right), \qquad b_2 = \frac{4\pi}{\sqrt{3} a} \left(-\frac{\sqrt{3}}{2}, \frac{1}{2}\right).
\end{align}
The reciprocal lattice, $\cR^*$, and Brillouin zone, $\Gamma^*$, are then defined by 
\begin{align}\label{eq:cR_star}
    \cR^* := \{n_1 b_2 + n_2 b_2 : (n_1, n_2) \in \mathbb{Z}^2\}, \qquad \Gamma^* := \{\beta_1 b_2 + \beta_2 b_2 : (\beta_1, \beta_2) \in [0,1)^2\}.
\end{align}
We next define the Bloch transform and its inverse by
\begin{equation}\label{eq:bloch}
    \tilde{\varphi}^\sigma(k) := \left[ \mathcal{G} \varphi \right]^\sigma(k) := \sum_{R \in \mathcal{R}} e^{- i k \cdot (R + \tau^\sigma)} \varphi_R^\sigma, \quad \varphi_R^\sigma = \left[ \mathcal{G}^{-1} \tilde{\varphi} \right]^\sigma_R := \frac{1}{|\Gamma^*|} \inty{\Gamma^*}{}{ e^{i k \cdot (R + \tau^\sigma)} \tilde{\varphi}^\sigma(k) }{k}.
\end{equation}
Observe that the Bloch transform is quasi-periodic with respect to $\cR^*$, as
\begin{align*}
    \left[ \cG \varphi \right]^\sigma (k+G) = e^{-iG \cdot \tau^\sigma} \left[ \cG \varphi \right]^\sigma (k),\qquad G \in \cR^*.
\end{align*}
The Bloch transform of $H \varphi$ is
\begin{align}\label{eq:tildeH} 
    (\widetilde{H \varphi})^\sigma (k) = \sum_{\sigma' \in \{A,B\}} \tilde{h}^{\sigma,\sigma '} (k) \tilde{\varphi}^{\sigma '} (k), \qquad
    \tilde{h}^{\sigma, \sigma '} (k) := \sum_{R \in \cR} e^{-i(R+\tau^{\sigma, \sigma '}) \cdot k} h(R+\tau^{\sigma, \sigma '}),
\end{align}
where $\tilde{h}(k)$ is known as the Bloch Hamiltonian.
The rapid decay of $h$ in Assumption \ref{assumption:h} implies that $\tilde{h}^{\sigma \sigma'} \in C^\infty (\mathbb{R}^2; \mathbb{C})$.
We easily see from the definition \eqref{eq:tildeH} that
\begin{align}\label{eq:AABB}
    \tilde{h}^{A,A} (k) = \tilde{h}^{B,B} (k), \qquad k \in \mathbb{R}^2.
\end{align}
Moreover, 
the self-adjointness of $H$, i.e. $h (-r) = \overline{h(r)}$, implies that
\begin{align*}
    \tilde{h}^{\sigma',\sigma} (k) = \overline{\tilde{h}^{\sigma,\sigma'} (k)}, \qquad k \in \mathbb{R}^2.
\end{align*}




\subsection{Existence of monolayer Dirac cones}\label{subsec:existence}
In this section, we derive expressions for the functions $\tilde{h}^{\sigma, \sigma'}$ (defined in \eqref{eq:tildeH}) and their derivatives evaluated at the high-symmetry points
\begin{align}\label{eq:Dirac_points}
K := \frac{4\pi}{3a}(1,0), \qquad K' := -K.
\end{align}
\srq{These expressions will allow us to prove the generic existence of Dirac cones (see Definition \ref{def:Dirac} below) at $K$ and $K'$ for the tight-binding Hamiltonian \eqref{eq:tight-binding}; see Corollary \ref{corr:Dirac}.}
Following the standard convention, we will refer to the points $K$ and $K'$ as \emph{Dirac points}. 
Given that $$\tilde{h}^{\sigma,\sigma'} (k+G) = e^{-i G\cdot \tau^{\sigma,\sigma'} } \tilde{h}^{\sigma,\sigma'} (k), \qquad G \in \cR^*,$$ it is easy to extend the following results to the set of all Dirac points $\pm K + \cR^*$.

\begin{definition}\label{def:Dirac}
    We say that the tight-binding Hamiltonian \eqref{eq:tight-binding} has a \emph{Dirac cone} at $k=K_*$ if the $2 \times 2$ matrix $\tilde{h}$ defined in \eqref{eq:tildeH} has eigenvalues $\eshift \pm v_F |k-K_*| + O(|k-K_*|^2)$ as $|k-K_*| \to 0$, for some real numbers $\eshift$ and $v_F$ with $v_F \ne 0$. 
\end{definition}
The constants $\eshift$ and $v_F$ are known as the ``Dirac energy''  and ``Fermi velocity'' in the physics literature. The existence of Dirac cones has been analyzed in many contexts; see for example \cite{becker2024dirac, berkolaiko2018symmetry, fefferman2012honeycomb, fefferman2018honeycomb,malinovitch2025twistedbilayergraphenecommensurate}. In neutral graphene, the Fermi level occurs exactly at the Dirac energy.
\begin{lemma}\label{lemma:tilde_h_rot}
    Under Assumption \ref{assumption:h}, for all $k \in \mathbb{R}^2$, 
    \begin{align*}
        \tilde{h}^{\sigma \sigma} (\rot_{2\pi/3} (k-K) + K) = \tilde{h}^{\sigma \sigma} (k), &\qquad
        \tilde{h}^{AB} (\rot_{2\pi/3} (k-K) + K) = \sq{e^{-i2\pi/3}} \tilde{h}^{AB} (k),\\
        \tilde{h}^{\sigma \sigma} (\rot_{2\pi/3} (k-K') + K') = \tilde{h}^{\sigma \sigma} (k), &\qquad
        \tilde{h}^{AB} (\rot_{2\pi/3} (k-K') + K') = \sq{e^{i2\pi/3}} \tilde{h}^{AB} (k).
    \end{align*}
\end{lemma}
\begin{proof}
    We write
    \begin{align}\label{eq:rotAA}
        \tilde{h}^{\sigma \sigma} (\rot_{2\pi/3} (k-K) + K) = \sum_{R \in \cR} e^{-i R \cdot (\rot_{2\pi/3} (k-K) + K)} h(R) = \sum_{R \in \cR} e^{-i R \cdot (k-K + \rot_{-2\pi/3} K)} h(R),
    \end{align}
    where we changed variables $R \leftarrow \rot_{-2\pi/3} R$ and used the $2\pi/3$-rotation invariance of $h$ to obtain the last equality. One can then verify that $R \cdot (-K + \rot_{-2\pi/3} K) \in 2\pi \mathbb{Z}$ for all $R \in \cR$, hence \eqref{eq:rotAA} equals $\tilde{h}^{\sigma \sigma} (k)$.

    Similarly, we have
    \begin{align}\label{eq:rotAB}
    \begin{split}
        \tilde{h}^{AB} (\rot_{2\pi/3} (k-K) + K) &=\sum_{R \in \cR} e^{-i (R+\tau^{A,B}) \cdot (\rot_{2\pi/3} (k-K) + K)} h(R+\tau^{A,B})\\
        &= \sum_{R \in \cR} e^{-i (R+\tau^{A,B}) \cdot (k-K + \rot_{-2\pi/3} K)} h(R+\tau^{A,B}),
        \end{split}
    \end{align}
    with the second equality again following from $2\pi/3$-rotation invariance of $h$. Note that the sum in the second line is still over $\cR$ since $\rot_{-2\pi/3} (\cR + \tau^{A,B}) = \cR + \tau^{A,B}$. Using that $\tau^{A,B} \cdot (-K + \rot_{-2\pi/3} K) = \sq{2\pi/3}$, it follows that \eqref{eq:rotAB} equals $\sq{e^{-i2\pi/3}} \tilde{h}^{AB} (k)$.

    The proofs for rotations about $K'$ are similar.
\end{proof}
For $\sigma, \sigma' \in \{A,B\}$, define $\hcs^{\sigma \sigma'} :\mathbb{C} \to \mathbb{C}$ by $\hcs^{\sigma \sigma'} (q) = \tilde{h}^{\sigma \sigma'} ((q_1, q_2) + K)$, where $q = q_1 + i q_2$ with $q_1, q_2 \in \mathbb{R}$. Lemma \ref{lemma:tilde_h_rot} then implies that $\hcs^{\sigma \sigma} (e^{i2\pi/3} q) = \hcs^{\sigma \sigma} (q)$ and $\hcs^{AB} (e^{i2\pi/3} q) = \sq{e^{-i2\pi/3}} \hcs^{AB} (q)$. It immediately follows from Taylor expanding $\hcs^{\sigma \sigma'}$ about $0$ that for all $m, n \in \mathbb{N}_0$,
\begin{align}\label{eq:allowed_derivatives}
    \partial^m_q \partial^n_{\bar{q}} \hcs^{\sigma \sigma} (0) = 0 \quad \text{if} \quad n-m \notin 3 \mathbb{Z}, \qquad
    \partial^m_q \partial^n_{\bar{q}} \hcs^{AB} (0) = 0 \quad \text{if} \quad n-m \notin 3 \mathbb{Z}\sq{+}1,
\end{align}
where $\partial_q := \frac{1}{2} (\partial_{q_1} - i \partial_{q_2})$ and $\partial_{\bar{q}} := \frac{1}{2} (\partial_{q_1} + i \partial_{q_2})$. 
Using a parallel argument for $K'$, we have thus proven
\begin{theorem}\label{thm:Dirac}
    Under Assumption \ref{assumption:h}, for some $\alpha, \alpha', \alpha_o, \alpha_o' \in \mathbb{C}$ and $\alpha_d, \alpha'_d \in \mathbb{R}$, 
    \begin{align*}
        \tilde{h}^{A,B} (K) &= \tilde{h}^{A,B} (K') = 0, \quad \nabla \tilde{h}^{A,B} (K) = \alpha (1,-i), \quad \nabla \tilde{h}^{A,B} (K') = \alpha' (1,i),\\
        \nabla \tilde{h}^{\sigma, \sigma} (K) &= \nabla \tilde{h}^{\sigma, \sigma} (K') = (0,0), \quad
        \nabla^2 \tilde{h}^{A,B} (K) = \alpha_o \begin{pmatrix}
            1 & i\\
            i & -1
        \end{pmatrix}, \quad
        \nabla^2 \tilde{h}^{A,B} (K') = \alpha'_o \begin{pmatrix}
            -1 & i\\
            i & 1
        \end{pmatrix},
    \end{align*}
    where $I_2$ is the $2 \times 2$ identity matrix.
\end{theorem}

Combining the above theorem with \eqref{eq:AABB}, we have shown that
\begin{align}\label{eq:Dirac}
    \tilde{h}(K+q) = \tilde{h}^{AA} (K) I_2 + \begin{pmatrix}
        0 & \alpha (q_1 - i q_2)\\
        \bar{\alpha} (q_1 + i q_2) & 0
    \end{pmatrix} + O(|q|^2),
\end{align}
which proves that the matrix $\tilde{h}$ has eigenvalues $\tilde{h}^{AA} (K) \pm |\alpha| |q| + O (|q|^2)$.
Therefore, we have 
\begin{corollary}\label{corr:Dirac}
    Provided $\alpha \ne 0$, the tight-binding Hamiltonian \eqref{eq:tight-binding} has a Dirac cone at $k=K$ with Fermi velocity $v_F = |\alpha|$.
    Similarly, the Hamiltonian has a Dirac cone at $k=K'$ if $\alpha ' \ne 0$.
\end{corollary}
\begin{remark}
    For every tight-binding model of graphene in the physics literature, the function $h$ is real-valued (i.e., $[H,\cC]=0$ with $\cC f:= \bar{f}$ the complex conjugation operator), which implies that $\alpha' = -\bar{\alpha}$.
    Indeed,
    \begin{align*}
        \nabla \tilde{h}^{A,B} (k)= -i \sum_{R \in \cR} (R - \tau^B) e^{-ik\cdot (R - \tau^B)} h (R -\tau^B),
    \end{align*}
    so that
    \begin{align*}
        -\overline{\nabla \tilde{h}^{A,B} (K')}= -i \sum_{R \in \cR} (R - \tau^B) e^{-iK\cdot (R - \tau^B)} \overline{h (R -\tau^B)},
    \end{align*}
    with the above right-hand side equal to
    $\nabla \tilde{h}^{A,B} (K)$ if $h$ is real. 
\end{remark}

\subsection{Bilayer tight-binding model}\label{subsec:bilayer}

We now define a discrete model for twisted bilayer graphene that serves as the starting point for the analysis in this paper.
For $-\pi < \theta \le \pi$ and $j \in \{1,2\}$, set
\begin{align*}
    \cR_j := \rot_{\theta_j/2} \cR, \qquad \theta_j := (-1)^j \theta, \qquad \rot_\phi:= \begin{pmatrix}
    \cos \phi & -\sin \phi\\
    \sin \phi & \cos \phi
\end{pmatrix}, 
\end{align*}
where $\cR$ is the untwisted monolayer graphene lattice defined in \eqref{eq:lattice}. Thus $\cR_j$ denotes the lattice corresponding to layer $j$, and $\theta$ the relative twist angle between the two layers. 
Recall the definitions 
of $\tau^\sigma$ and $\Gamma$ in Section \ref{subsec:mono}, and for $\ls \in \Gamma$ define
\begin{align}\label{eq:tau_j}
\tau_j^\sigma := \rot_{\theta_j/2} \left(\tau^\sigma + (-1)^j \frac{\ls}{2}\right). 
\end{align}
The positions of the atoms in layer $j$ are then given by the set $(\cR_j + \tau^A_j) \cup (\cR_j + \tau^B_j)$, so that $\ls$ is the displacement of layer $2$ relative to layer $1$ before twisting.
Although the results in this paper are valid for any $\ls$, there are two natural choices. The first is $\ls = 0$, which corresponds to the untwisted graphene sheets aligned perfectly on top of each other. The second, known as the ``Bernal stacking'' configuration, is $\ls = \pm \tau^{A,B} \mod \cR$, which is the lowest-energy configuration of the untwisted system \cite{carr2018relaxation}.

We introduce the Hilbert space $\cH := \ell^2 (\cR_1; \mathbb{C}^2)\oplus \ell^2 (\cR_2; \mathbb{C}^2)$ with norm given by
\begin{align}\label{eq:cH}
    \norm{\varphi}_\cH := \left( \sum_{i\in \{1,2\}} \sum_{R_i \in \cR_i} \sum_{\sigma \in \{A,B\}} |(\varphi_i^\sigma)_{R_i}|^2 \right)^{1/2}.
\end{align}
We may write any function $\varphi \in \cH$ as a $4$-vector, $\varphi = (\varphi^A_1, \varphi^B_1, \varphi^A_2, \varphi^B_2)$, where each $\varphi^\sigma_j$ belongs to $\ell^2 (\cR_j; \mathbb{C})$.
We will also write $\varphi_j := (\varphi^A_j, \varphi^B_j) \in \ell^2 (\cR_j; \mathbb{C}^2)$ to denote the entries of $\varphi$ corresponding to layer $j$. 

\srq{The TBG tight-binding dynamics are modeled by the Schr\"odinger equation $$i \partial_\microtime \psi = H \psi, \qquad \psi \vert_{\microtime = 0} = \psi_0,$$ where the Hamiltonian $H : \cH \to \cH$ is a self-adjoint operator to be defined below.}
\srq{As described in Section \ref{subsec:main}, the analysis in this paper is restricted to a regime in which the two graphene sheets are \emph{weakly} coupled. We therefore introduce the relevant asymptotic parameter $0 < \eps < 1$, which will also control the width of the wave-packet envelope for $\psi_0$ (see Theorems \ref{thm:main2} and \ref{thm:main}) and the twist angle $\theta$ (see \eqref{eq:beta_def}, but in this section we can think of $\theta$ and $\eps$ as independent parameters). This is the same asymptotic regime that was analyzed in \cite{watson2023bistritzer}.}

\srq{The Hamiltonian $H$ is then defined by}


\begin{align}\label{eq:H}
    H \varphi= \begin{pmatrix}
        H_{11}& \Hperp\\
        \Hperp^{\dagger} & H_{22}
    \end{pmatrix}
    \begin{pmatrix}
        \varphi_1\\
        \varphi_2
    \end{pmatrix},
\end{align}
where
\begin{align}\label{eq:intra}
    (H_{jj}\varphi_j)^{\sigma}_{R_j} := \sum_{R_j' \in \cR_j} \sum_{\sigma'} h (\rot_{-\theta_j/2} (R_j - R_j' +\tau_j^{\sigma \sigma'})) (\varphi_j^{\sigma'})_{R'_j}, \qquad \tau_j^{\sigma \sigma'} := \tau_j^\sigma- \tau_j^{\sigma'}
\end{align}
and
\begin{align}\label{eq:Hperp}
    (\Hperp \varphi_2)^\sigma_{R_1} &:= \sum_{R_2 \in \cR_2} \sum_{\sigma'} \hperp (R_1 - R_2+\tau_1^{\sigma}-\tau_2^{\sigma'};\eps) (\varphi^{\sigma'}_{2})_{R_2},
\end{align}
are given in terms of an intralayer and interlayer hopping function $h$ and $\hperp$ satisfying the following assumptions.
The intralayer hopping function $h$ in \eqref{eq:intra} satisfies Assumption \ref{assumption:h} \sq{and the self-adjointness condition $h(-r) = \overline{h(r)}$}. Note that $\rot_{-\theta_j/2} (\cR_j +\tau_j^{\sigma \sigma'}) = \cR + \tau^{\sigma \sigma'}$, hence $h$ in \eqref{eq:intra} is evaluated only at points in $\cRall$ and so is well defined. 
The interlayer hopping function
$\hperp : \mathbb{R}^2 \times (0,1) \to \mathbb{C}$ in \eqref{eq:Hperp} satisfies 
the following
\begin{assumption}\label{assumption:hperp2}
    Fix $0 <\sdpt \le 1$. Assume $\hathperp (\cdot \; ; \eps) \in C^2 (\mathbb{R}^2;\mathbb{C})$ is of the form 
    \begin{align}\label{eq:sep}
        \hathperp (k;\eps) = \hathperprad (|k|; \eps) \hathperpang (\angvar{k}), \qquad k \in \mathbb{R}^2, \quad 0 < \eps < 1,
    \end{align}
    for some $\hathperprad (\cdot \; ; \eps) \in C^2 ([0, \infty); \mathbb{R})$ and $\hathperpang \in C^2 (\mathbb{S}; \mathbb{C})$,
    where $\angvar{k} := k/|k|$.
    Moreover, suppose 
    \begin{align}\label{eq:neighbor_bds}
        \hathperprad (|K|; \eps) = \eps, \qquad  |\hathperprad ' (|K|; \eps)| \le C \eps^{\frac{1+\sdpt}{2}}, \qquad |\hathperprad (2|K|; \eps)| \le C \eps^{\frac{3+\sdpt}{2}},
\end{align}
    and that there exists $0 < \swp \le 1$ such that
    \begin{equation}
    \begin{alignedat}{2}\label{eq:4bounds}
        |\hathperprad (|k|; \eps)| + |\hathperprad ' (|k|; \eps)| + |\hathperprad '' (|k|; \eps)| &\le C, \qquad &&|k| \ge 0,\\
        |\hathperprad (|k|; \eps)| + |\hathperprad ' (|k|; \eps)| + |\hathperprad '' (|k|; \eps)| &\le C (C \eps)^{\sdpt |k|/|K|}, \qquad &&|k| \ge (1-\swp) |K|,\\
        |\hathperprad (|k|; \eps)| + |\hathperprad ' (|k|; \eps)| &\le C (C\eps)^{\frac{1+\sdpt}{2}|k|/|K|}, \qquad &&|k| \ge (1-\swp) 2 |K|,\\
        |\hathperprad (|k|; \eps)| &\le C (C\eps)^{\frac{2+\sdpt}{\sqrt{7}}|k|/|K|}, \qquad &&|k| \ge (1-\swp) \sqrt{7} |K|,
    \end{alignedat}
    \end{equation}
    uniformly in $0 < \eps < 1$, where $\hathperprad '$ and $\hathperprad ''$ are the first and second derivatives of $\hathperprad(\cdot \; ; \eps)$.
\end{assumption}


The decomposition \eqref{eq:sep} and first two bounds in \eqref{eq:neighbor_bds} imply that
\begin{align}\label{eq:nabla_bd_K}
    |\nabla \hathperp (k; \eps)| \le C \eps^{\frac{1+\sdpt}{2}}, \qquad \qquad |k| = |K|, \quad 0 < \eps < 1.
\end{align}
Moreover, by \eqref{eq:4bounds} we have
\begin{equation}
    \begin{alignedat}{2}\label{eq:4bounds_full}
        \absval{\hathperp (k; \eps)} +\absval{\nabla \hathperp (k; \eps)} + \norm{\nabla^2 \hathperp (k; \eps)}_{{\rm F}} &\le C, \qquad && k \in \mathbb{R}^2 ,\\
        \absval{\hathperp (k; \eps)} +\absval{\nabla \hathperp (k; \eps)} +\norm{\nabla^2 \hathperp (k; \eps)}_{{\rm F}} &\le C (C\eps)^{\sdpt |k|/|K|}, \qquad &&|k| \ge (1-\swp) |K|,\\
        \absval{\hathperp (k; \eps)} + \absval{\nabla \hathperp (k; \eps)} &\le C (C\eps)^{\frac{1+\sdpt}{2}|k|/|K|}, \qquad &&|k| \ge (1-\swp) 2 |K|,\\
        \absval{\hathperp (k; \eps)} &\le C (C\eps)^{\frac{2+\sdpt}{\sqrt{7}}|k|/|K|}, \qquad &&|k| \ge (1-\swp) \sqrt{7} |K|,
    \end{alignedat}
\end{equation}
uniformly in $0 < \eps < 1$, where
the above derivatives are with respect to $k$, and 
$\norm{\cdot}_{{\rm F}}$ is the Frobenius norm.
\srq{We will use these bounds to approximate $\hathperp (k;\eps)$ by its (zeroth, first, or second order) Taylor series expansion \eqref{eq:filtered_Taylor} about the points $k = \rot_{\theta/2} (K+G)$ with $G \in \cR^*$, using the exponential decay of $\hathperp$ and its derivatives to control the error when $|G_2|$ is large. The last three bounds in \eqref{eq:4bounds_full} imply that 
\begin{align*}
    \eps^2 \norm{\nabla^2 \hathperp (k_1; \eps)}_{\rm F} + \eps |\nabla \hathperp (k_2; \eps)| + |\hathperp (k_3; \eps)| \le C \eps^{2+\sdpt}, \qquad (|k_1|, |k_2|, |k_3|) = (|K|, 2|K|, \sqrt{7}|K|),
\end{align*}
which will be used to obtain the almost $\eps^{2+\sdpt} \microtime$ order of convergence in Theorems \ref{thm:main2} and \ref{thm:main}.
The importance of these specific $|k_j|$ reflects the geometry of the monolayer reciprocal lattice, as the points $k \in \cR^*$ that are closest to $-K$ satisfy $|k+K| \in \{|K|, 2|K|, \sqrt{7}|K|\}$; see Figure \ref{fig:NNN}. As we will see in the proof of Lemma \ref{lemma:bilayer1}, controlling relevant derivatives of $\hathperp$ at these ``nearest neighbors'' is critical to establish our main results. Note that the effective continuum model \eqref{eq:eff} depends on $\hathperprad$ only through the functions $\radNNN (\eps)$ and $\radpNN (\eps)$, which are determined by $\hathperprad (2|K|;\eps)$ and $\hathperprad ' (|K|; \eps)$; see \eqref{eq:rad_ang}.}




\srq{We now show that the set of interlayer hopping functions satisfying Assumption \ref{assumption:hperp2} is nonempty.}
\begin{example} \label{ex:2}
    Fix $0 < \hoppingparam < \frac{\sqrt{3}}{2} |K|$ and $\lparam > 0$, and consider the hopping function
    \begin{align}\label{eq:example_hopping}
        \hperp (r;\eps) := \frac{e^{-\hoppingparam \sqrt{|r|^2+\ell^2 (\eps)}}}{\sqrt{|r|^2+\ell^2 (\eps)}}, \qquad \ell (\eps) := -\frac{1}{\sqrt{|K|^2+\hoppingparam^2}}\log \left( \frac{\lparam \eps \sqrt{|K|^2 + \hoppingparam^2}}{2\pi}\right).
    \end{align}
    The Fourier transform of $\hperp (\cdot \; ; \eps)$ is
    \begin{align*}
        \hathperp (k; \eps) = 2\pi \frac{e^{-\ell (\eps) \sqrt{|k|^2 + \hoppingparam^2}}}{\sqrt{|k|^2+\hoppingparam^2}} = 
        \frac{2\pi}{\sqrt{|k|^2+\hoppingparam^2}} \left( \frac{\lparam \eps \sqrt{|K|^2 + \hoppingparam^2}}{2\pi}\right)^{\sqrt{\frac{|k|^2+\hoppingparam^2}{|K|^2+\hoppingparam^2}}}
        =: \lparam \hathperprad (|k|; \eps),
    \end{align*}
    where $\hathperprad (|K|; \eps) = \eps$. The definition of $\hoppingparam$ implies that the constant
    \begin{align*}
        \sdpt_+ := \frac{\sqrt{7} |K|}{\sqrt{|K|^2+\hoppingparam^2}}-2
    \end{align*}
    is positive. Moreover, for any $0 < \sdpt < \sdpt_+$, we have
    \begin{align*}
        \frac{2 + \sdpt}{\sqrt{7}} \frac{|k|}{|K|} < \sqrt{\frac{|k|^2+\hoppingparam^2}{|K|^2+\hoppingparam^2}}, \qquad \qquad k \in \mathbb{R}^2,
    \end{align*}
    and thus
    \begin{align*}
        |\hathperprad (|k|; \eps)| + |\hathperprad '(|k|; \eps)| + |\hathperprad ''(|k|; \eps)| \le C (C\eps)^{\frac{2+\sdpt}{\sqrt{7}}|k|/|K|}, \qquad \qquad k \in \mathbb{R}^2, \quad 0 < \eps < 1.
    \end{align*}
    The bounds in \eqref{eq:neighbor_bds} and \eqref{eq:4bounds} follow, hence $\hathperp$ satisfies Assumption \ref{assumption:hperp2} as desired.
\end{example}
We refer to Appendix \ref{sec:examples} for other examples of interlayer hopping functions satisfying Assumption \ref{assumption:hperp2}. These examples clarify that 
the decay assumptions \eqref{eq:4bounds} capture an important physical phenomenon: locality of interlayer hopping in momentum space due to separation of the layers. Indeed, one can interpret $\eps$ as a parameter measuring the (out of plane) distance $\ell$ between the layers, with $\ell$ a monotonically decreasing function of $\eps$. Therefore, small $\eps$ would correspond to large $\ell$ and hence weakly interacting layers. In several applications of interest (see Examples \ref{ex:2} and \ref{ex:1}), the function $r \mapsto \hperp (r; \eps)$ is analytic in the strip $\mathbb{R} \times (-\ell,\ell) \subset \mathbb{C}$, with $\ell \sim |\log \eps|$ as $\eps \to 0$.  

The most popular choice of hopping function in applications \cite{castro2007}  has been the Slater-Koster interlayer hopping function \cite{slater1954simplified} given by
\begin{align}\label{eq:slater-koster}
    \hperp (\vec r; \eps) = V_{pp\pi}^0 \exp\left(-\frac{\sqrt{|\vec r|^2+\ell^2}-a_0}{r_0}\right)\left(\frac{|\vec r|^2}{|\vec r|^2+\ell^2}\right) + V_{pp\sigma}^0\exp\left(-\frac{\sqrt{|\vec r|^2+\ell^2}-\ell}{r_0}\right)\left(\frac{\ell^2}{|\vec r|^2+\ell^2}\right),
\end{align}
where $\ell = \ell (\eps)$ is still the interlayer separation, while $r_0 > 0$ and $V_{pp\pi}^0, V_{pp\sigma}^0, a_0 \in \mathbb{R}$ are $\eps$-independent parameters. Since $\hperp$ is radial, so is its Fourier transform, and thus the decomposition \eqref{eq:sep} trivially holds.
\srq{Moreover, we observe that $\hperp (\cdot \; ; \eps)$ is exponentially decaying and analytic in $\mathbb{R} \times (-\ell, \ell)$, which implies that for any $\ell_- < \ell$,
\begin{align}\label{eq:SK_decay}
    \absval{\hathperp (k; \eps)} +\absval{\nabla \hathperp (k; \eps)} + \norm{\nabla^2 \hathperp (k; \eps)}_{{\rm F}} \le C e^{-\ell_- |k|}, \qquad k \in \mathbb{R}^2.
\end{align}
But $\hathperp (k; \eps)$ is not known analytically, making it difficult to verify the estimates in Assumption \ref{assumption:hperp2}. In particular, it is unclear whether one could choose $\ell \gtrsim |\log \eps|$ such that 
$\hathperprad (|K|; \eps) = \eps$. 
This equality could be satisfied if, for example, $\absval{\hathperprad (|K|; \eps)} \ge c e^{-\ell \delta}$ for some $c, \delta > 0$; see \cite[Equation (2.68)]{watson2023bistritzer}. More recently, accurate interlayer hopping functions have been proposed by Fang-Kaxiras, see Remark \ref{remark:weaken_assumptions}. Although we cannot verify Assumption \ref{assumption:hperp2} for these hopping functions, we are confident that our assumption captures all of their essential features.
}
\srq{
\begin{remark}\label{remark:interlayer_hopping_strength}
    As we will see in the derivation of our effective continuum model (more specifically, the proof of Lemma \ref{lemma:tildeF12}), the assumption \eqref{eq:neighbor_bds} that $\hathperprad (|K|; \eps)=\eps$ implies that the strength of interlayer coupling is $O(\eps)$. Indeed, the factors $\lambda_j$ in the leading-order off-diagonal term $\hoppingT(r)$ from \eqref{eq:HBM_1} are given by $\lambda_j = \eps^{-1} \hathperp (\rot_{\pi j/3} K; \eps)$; see also \eqref{eq:rad_ang}. Thus the ratio of energy scales of the tight-binding Hamiltonian $H$ and $\eps$-independent first-order BM Hamiltonian $\HBM$ is $O(\eps)$. This will be made evident also by the $\eps$ scaling of the time variable in \eqref{eq:def_phi}. 
\end{remark}
}

\begin{remark}\label{remark:weaken_assumptions}
    The separability assumption in \eqref{eq:sep} is a natural way to extend the theory from \cite{watson2023bistritzer} to non-radial interlayer hopping functions. 
    That is, one could always take $\hathperpang$ to be constant, in which case $\hathperp$, and thus $\hperp$, would be radial.
    We note that the most realistic tight-binding models include angular dependence of the interlayer hopping function, as this is necessary to capture the overlap of Wannier orbitals of different layers; see \cite{Fang_Kaxiras_2016}.
\end{remark}
\begin{remark}\label{remark:linearity}
    By linearity, the results in this paper easily generalize to  
    a much larger class of 
    $C^2 (\mathbb{R}^2; \mathbb{C})$ hopping functions satisfying \eqref{eq:4bounds_full}. 
    Indeed, any such function can be 
    decomposed into its real and imaginary parts as
    \begin{align*}
        \hathperp (k; \eps) = \hathperp^{({\rm re})} (k; \eps) + i \hathperp^{({\rm im})} (k; \eps),
    \end{align*}
    where each $\hathperp^{(j)} \in C^2 (\mathbb{R}^2; \mathbb{R})$ has the Fourier decomposition
    \begin{align*}
        \hathperp^{(j)} (k; \eps) = \sum_{n =0}^\infty a_n^{(j)}(|k|; \eps) \cos (n \phi_k) + \sum_{n =1}^\infty b_n^{(j)}(|k|; \eps) \sin (n \phi_k), \qquad k = (|k| \cos \phi_k, |k| \sin \phi_k),
    \end{align*}    
    and thus can be written as a linear combination of functions of the form \eqref{eq:sep}. 
    In order to recover a BM-type model at leading order (that is, to prevent the interlayer blocks from dominating the intralayer Dirac term), the assumption that $\hathperprad (|K|; \eps) = \eps$ would then get replaced by 
    \begin{align*}
        a_n^{({\rm re})} (|K|; \eps)= O(\eps), \quad b_n^{({\rm re})} (|K|; \eps)= O(\eps), \quad a_n^{({\rm im})} (|K|; \eps)= O(\eps), \quad b_n^{({\rm im})} (|K|; \eps) = O(\eps) \qquad \text{for all } n.
    \end{align*}
    The other bounds from \eqref{eq:neighbor_bds} would become
    \begin{align*}
        (a_n^{(j)})' (|K|; \eps)= O(\eps^{\frac{1+\sdpt}{2}}), \quad (b_n^{(j)})' (|K|; \eps)= O(\eps^{\frac{1+\sdpt}{2}}), \quad a_n^{(j)} (2|K|; \eps)= O(\eps^{\frac{3+\sdpt}{2}}), \quad b_n^{(j)} (2|K|; \eps)= O(\eps^{\frac{3+\sdpt}{2}}). 
    \end{align*}
\end{remark}
\begin{remark}\label{remark:rv}
    In the physics literature, it is assumed that the interlayer hopping function is real-valued. This would require $\hathperp (-k;\eps) = \overline{\hathperp (k; \eps)}$ and $\hathperpang (-\angvar{k}) = \overline{\hathperpang (\angvar{k})}$.
\end{remark}


\begin{figure}
    \centering
    \includegraphics[scale=.75]{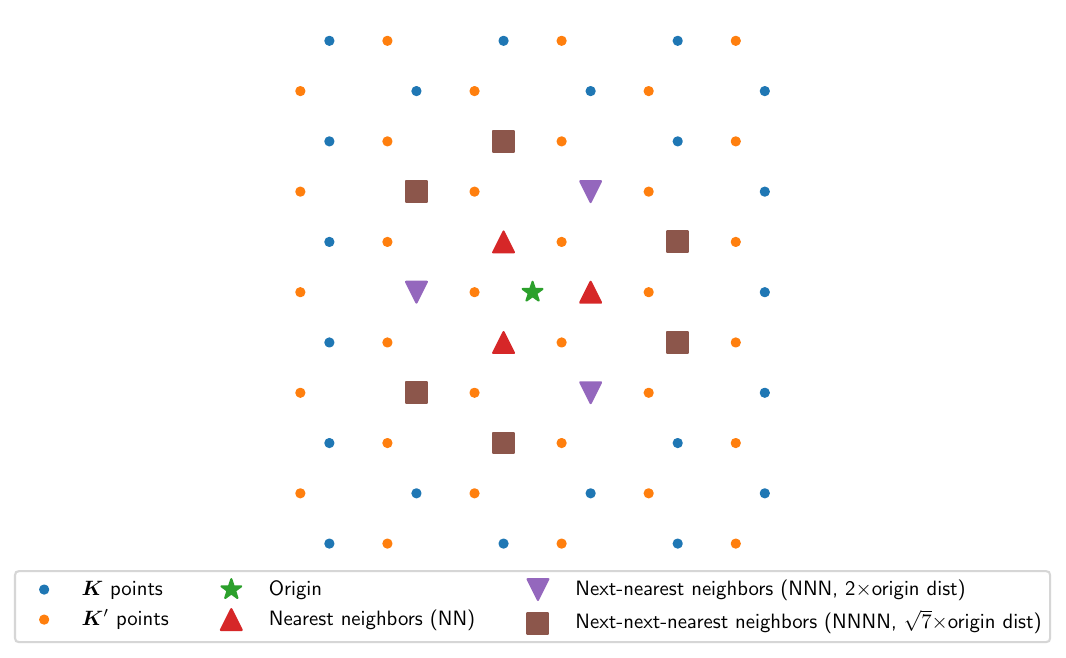}
    \caption{Geometry of the monolayer graphene reciprocal lattice.}
     \label{fig:NNN}
\end{figure}






\section{Effective continuum model}\label{sec:effective}
In this section, we present our two main results, Theorems \ref{thm:main2} and \ref{thm:main}. These theorems establish the validity of a second-order PDE model for TBG in the limit that $\theta$ and $\eps$ both go to $0$.
This asymptotic regime corresponds to small twist angles and weakly coupled layers, \srq{with the latter enforced by the decay of the interlayer hopping function in $\eps$; recall \eqref{eq:neighbor_bds}-\eqref{eq:4bounds}.}
\srq{For ease of exposition, we henceforth assume the existence of real 
constants $\beta \ne 0$ and $0 < \eps_0 < 1$ such that
\begin{align}\label{eq:beta_def}
\theta = 2 \sin^{-1} (\beta \eps/2), \qquad 0 < \eps < \eps_0.
\end{align}
This assumption immediately implies that
$ 
    2 \sin (\theta/2) = \beta \eps 
$
for all $0 < \eps < \eps_0$,
with $\theta \sim \beta \eps$ as $\eps \to 0$.}
A natural extension of the exact formula \eqref{eq:beta_def} for $\theta$ would be to assume that $\theta = \beta_1 \eps + \beta_2 \eps^2 + O(\eps^3)$ as $\eps \to 0$,
in which case the Hamiltonian $\HBMt$ (see below) would acquire an additional term involving $\beta_2$.

In Theorem \ref{thm:main2}, the effective model takes the form of a multiple-scales expansion; that is, a collection of four $\eps$-independent first-order Schr\"odinger equations whose solutions can be used to construct a second-order accurate approximation of the tight-binding dynamics. On the other hand, Theorem \ref{thm:main} presents the effective model as a single $\eps$-dependent second-order PDE. These results share the same order of convergence, though Theorem \ref{thm:main} requires slightly stronger assumptions on parameters of the tight-binding model and a higher degree of regularity of the initial data.

The proofs of Theorems \ref{thm:main2} and \ref{thm:main} are deferred to Appendix \ref{sec:proofs}.

\subsection{Multiscale expansion and main result}\label{subsec:multiscale}
This subsection presents an effective second-order model as a 
collection of four Schr\"odinger equations, each independent of the asymptotic parameter $\eps$. First, we introduce some notation.
Parameters of the effective model include the values $\hathperp (\rot_{2\pi j/3} K;\eps)$, $\hathperp (-2 \rot_{2\pi j/3} K;\eps)$ and $\nabla \hathperp(\rot_{2\pi j/3} K;\eps)$
for $j \in \{0,1,2\}$, with $K$ given by \eqref{eq:Dirac_points}. 
\srq{Recalling the decomposition $\hathperp (k;\eps) = \hathperprad (|k|; \eps) \hathperpang (\angvar{k})$ from \eqref{eq:sep},} it will thus be useful to define
\begin{align}\label{eq:rad_ang}
\begin{split}
    \radNNN (\eps) &:= \eps^{-1}\hathperprad (2|K|;\eps) \sq{\; \in \mathbb{R}},\qquad
     \radpNN (\eps) := \hathperprad ' (|K|;\eps) \sq{\; \in \mathbb{R}},\\
    \angs_i &:= \hathperpang (\rot_{\pi i/3} \angvar{K}) \sq{\; \in \mathbb{C}}, \qquad
    \angp_j := \frac{1}{|K|} \hathperpang ' (\rot_{2\pi j/3} \angvar{K}) \sq{\; \in \mathbb{C}},
\end{split}
\end{align}
for $i \in \{0,1,\dots, 5\}$ and $j \in \{0,1,2\}$, where $\angvar{K} := K/|K| = (1,0)$. \srq{Here, $\hathperprad '$ is the derivative of $\hathperprad (|k|; \eps)$ with respect to $|k|$, while $\hathperpang ' (\angvar{k}) := \lim_{\phi \to 0} \phi^{-1} (\hathperpang (\rot_\phi \angvar{k}) - \hathperpang (\angvar{k}))$ is the derivative of $\hathperpang$ with respect to counterclockwise rotations of $\angvar{k}$.}
\srq{Note that the $\lambda_i$ involve $\pi/3$ (rather than $2 \pi/3$) rotations of $\angvar{K}$ since $\hathperpang$ must be evaluated at $2\pi/3$ rotations of both $\angvar{K}$ and $-\angvar{K}$. These points correspond to the six triangles in Figure \ref{fig:NNN}.}

In addition to 
these $\lambda_i$ and $\mu_j$, the effective second-order model will involve the constant $\beta$ from \eqref{eq:beta_def} and parameters
$\alpha, \alpha_o \in \mathbb{C}$ and $\alpha_d \in \mathbb{R}$ in Theorem \ref{thm:Dirac}. 
Recall also the lattice constant $a$ from \eqref{eq:lattice_defn}, the reciprocal lattice vectors $b_1$ and $b_2$ defined by \eqref{eq:reciprocal_lattice}, and the interlayer displacement $\ls$ in \eqref{eq:tau_j}. Set
\begin{align}\label{eq:svecs}
    \svec_1 := \frac{4\pi \beta}{3a} (0,-1), \qquad 
    \svec_2 &:= \frac{4\pi \beta}{3a} \left(
        \frac{\sqrt{3}}{2}, \frac{1}{2}
    \right), \qquad
    \svec_3 := \frac{4\pi \beta}{3a} \left(-\frac{\sqrt{3}}{2}, \frac{1}{2}\right), 
\end{align}
and observe that $s_j = \rot_{2\pi/3} s_{j-1}$. Note also the difference between the layer $1$ and layer $2$ Dirac points corresponding to $K$ is 
\begin{align}\label{eq:Dirac_diff}
K_1 - K_2 := \rot_{-\theta/2}K - \rot_{\theta/2} K = \begin{pmatrix}
    0 & 2 \sin (\theta/2)\\
    -2 \sin (\theta/2) & 0
\end{pmatrix} K = \eps \svec_1,
\end{align}
with similar expressions relating differences of other corresponding Dirac points (obtained by translating $K$ in \eqref{eq:Dirac_diff} by monolayer reciprocal lattice vectors) to the other $\svec_j$; see Figure \ref{fig:bands} (left panel) or \cite[Figure 4]{watson2023bistritzer}.
Finally, we remark that differences between the $\svec_j$ are restricted to the lattice
\begin{align*}
    \cR_m^* = \{n_1 b_{m,1} + n_2 b_{m,2} : (n_1, n_2) \in \mathbb{Z}^2\}, \qquad b_{m,1} := \frac{4\pi \beta}{\sqrt{3} a} \begin{pmatrix}
        1/2\\
        -\sqrt{3}/2
    \end{pmatrix},
    \qquad b_{m,2} := \frac{4\pi \beta}{\sqrt{3} a} \begin{pmatrix}
        1/2\\
        \sqrt{3}/2
    \end{pmatrix},
\end{align*}
which is the reciprocal lattice to the moir\'e lattice 
$\cR_m$ defined in \eqref{eq:moire}. We will show with Proposition \ref{prop:continuum_symmetries} (\ref{symm:translation}) that our effective second-order model exhibits translational symmetry with respect to $\cR_m$.

We define the effective Hamiltonian as follows, recalling the shorthand $D_{r_j} := -i \partial_{r_j}$.
There are four scales that arise in our second-order approximation of the tight-binding dynamics. Each scale admits an effective Hamiltonian, which is a partial differential operator with smooth and periodic coefficients. We will denote these operators by
$$\HBM : H^1 \to L^2, \qquad \HBMNNN : L^2 \to L^2,\qquad \HBMgNN : H^1 \to L^2,\qquad \HBMt : H^2 \to L^2,$$ 
where $H^N := H^N (\mathbb{R}^2; \mathbb{C}^4)$ is the Sobolev space of order $N$.
The leading-order operator, $\HBM$, is 
given by
\begin{align}\label{eq:HBM_def}
\begin{split}
    \HBM &:= \begin{pmatrix}
        L & \hoppingT(r)\\
        \hoppingT^\dagger (r) & L
    \end{pmatrix}, \qquad
    L := \begin{pmatrix}
        0 & \alpha (D_{r_1} - i D_{r_2})\\
        \bar{\alpha} (D_{r_1} + i D_{r_2}) & 0
    \end{pmatrix},\\
    \hoppingT(r) &:= \frac{1}{|\Gamma|} \Bigg(
    \lambda_0 e^{-i \svec_1 \cdot r} \begin{pmatrix}
        1 & 1\\
        1 & 1
    \end{pmatrix} + 
    \lambda_2 e^{-i b_2 \cdot \ls} e^{-i \svec_2 \cdot r} \begin{pmatrix}
        1 & e^{-i2\pi/3}\\
        e^{i2\pi/3} & 1
    \end{pmatrix}\\
    & \hspace{5cm} 
    +\lambda_4 e^{ib_1 \cdot \ls}e^{-i \svec_3 \cdot r} \begin{pmatrix}
        1 & e^{i2\pi/3}\\
        e^{-i2\pi/3} & 1
    \end{pmatrix}\Bigg).
\end{split}
\end{align}
Note that if we take $\alpha \in \mathbb{R}$ and $\lambda_0 = \lambda_2 = \lambda_4$, then $\HBM$ reduces to the BM Hamiltonian \cite{bistritzer2011moire}.
As in \eqref{eq:Dirac}, the operator $L$ is obtained by Taylor expanding the Bloch transform of the intralayer hopping function about the monolayer Dirac point $K$. The periodic off-diagonal blocks model the leading-order coupling between the two layers.

We will see in the proof of Theorem \ref{thm:main2} that the factors $\lambda_0, \lambda_2, \lambda_4$ in \eqref{eq:HBM_def} arise from evaluating $\hathperp (k; \eps)$ at the points $k = \rot_{2\pi j/3} K$ for $j\in \{0,1,2\}$. These three points are the displacements between $-K$ and the three reciprocal lattice vectors that are closest to $-K$; see Figure \ref{fig:NNN}. The higher-order correction, $\HBMgNN$, to the off-diagonal blocks similarly consists of terms proportional to 
the leading-order contribution of $\nabla \hathperp$ evaluated at the same three ``nearest neighbor'' points. It is the first-order differential operator defined by
\begin{align*}
    \HBMgNN &:= \begin{pmatrix}
        0 & \hoppingT_{\nabla, {\rm NN}}\\
        \hoppingT^\dagger_{\nabla, {\rm NN}} & 0
    \end{pmatrix},\\
    \hoppingT_{\nabla, {\rm NN}} &:= \frac{1}{|\Gamma|} \Bigg(
    \lambda_0 e^{-i \svec_1 \cdot r} \begin{pmatrix}
        1 & 1\\
        1 & 1
    \end{pmatrix}(\angvar{K} \cdot D) + 
    \lambda_2 e^{-ib_2 \cdot \ls}e^{-i \svec_2 \cdot r} \begin{pmatrix}
        1 & \sq{e^{-i2\pi/3}}\\
        \sq{e^{i2\pi/3}} & 1
    \end{pmatrix} ((\rot_{2\pi/3} \angvar{K}) \cdot D)\\
    &\hspace{5cm}
    +\lambda_4 e^{ib_1 \cdot \ls}e^{-i \svec_3 \cdot r} \begin{pmatrix}
        1 & \sq{e^{i2\pi/3}}\\
        \sq{e^{-i2\pi/3}} & 1
    \end{pmatrix}((\rot_{4\pi/3} \angvar{K}) \cdot D)\Bigg),
\end{align*}
where $D := (D_{r_1}, D_{r_2})$.

We next define the bounded point-wise multiplication operator
\begin{align*}
    \HBMNNN &:= \begin{pmatrix}
        0 & \hoppingT_{{\rm NNN}} (r)\\
        \hoppingT^\dagger_{{\rm NNN}} (r) & 0
    \end{pmatrix},\\
    \hoppingT_{{\rm NNN}}(r) &:= \frac{1}{|\Gamma|} \Bigg(
    \lambda_3 e^{i (b_1 - b_2) \cdot \ls}e^{i 2\svec_1 \cdot r} \begin{pmatrix}
        1 & 1\\
        1 & 1
    \end{pmatrix} + 
    \lambda_5 e^{i (b_1 + b_2) \cdot \ls}e^{i 2\svec_2 \cdot r} \begin{pmatrix}
        1 & \sq{e^{-i2\pi/3}}\\
        \sq{e^{i2\pi/3}} & 1
    \end{pmatrix} \\
    &\hspace{5cm} + 
    \lambda_1 e^{-i (b_1 + b_2) \cdot \ls}e^{i 2\svec_3 \cdot r} \begin{pmatrix}
        1 & \sq{e^{i2\pi/3}}\\
        \sq{e^{-i2\pi/3}} & 1
    \end{pmatrix}\Bigg),
\end{align*}
which corresponds to evaluations of $\hathperp (k; \eps)$ at $k = \rot_{2\pi j/3} (-2K)$ for $j\in \{0,1,2\}$. These are the ``next-nearest neighbors'' to $-K$ on the reciprocal lattice, and are denoted by blue triangles in Figure \ref{fig:NNN}.
Recall that our decay assumption \eqref{eq:neighbor_bds} on $\hathperp$ ensures that these next-nearest neighbor evaluations are of higher order than the nearest neighbor values $\hathperp (\rot_{2\pi j/3} K; \eps) = O(\eps)$.

Finally, 
we define the second-order differential operator $\HBMt$ by
\begin{align*}
    \HBMt &:= \begin{pmatrix}
        \HBMsd - \frac{1}{2}\beta i \sigma_3 L & \hoppingT_2 + \Tu (r)\\
        \hoppingT^\dagger_2 + \Tu^\dagger (r) & \HBMsd + \frac{1}{2}\beta i \sigma_3 L
    \end{pmatrix}, \quad
    \HBMsd := \frac{1}{2}\begin{pmatrix}
        \alpha_d (D_{r_1}^2 + D_{r_2}^2) & \alpha_o (D^2_{r_1} - D^2_{r_2} + 2i D_{r_1 r_2})\\ 
        \overline{\alpha_o} (D^2_{r_1} - D^2_{r_2} - 2i D_{r_1 r_2}) & \alpha_d (D_{r_1}^2 + D_{r_2}^2)
    \end{pmatrix},\\
    \hoppingT_2 &:= \frac{1}{|\Gamma|} \Bigg(
    \mu_0 e^{-i \svec_1 \cdot r} \begin{pmatrix}
        1 & 1\\
        1 & 1
    \end{pmatrix}((\rot_{\pi/2} \angvar{K}) \cdot D) + 
    \mu_1 e^{-ib_2 \cdot \ls}e^{-i \svec_2 \cdot r} \begin{pmatrix}
        1 & \sq{e^{-i2\pi/3}}\\
        \sq{e^{i2\pi/3}} & 1
    \end{pmatrix} ((\rot_{7\pi/6} \angvar{K}) \cdot D)\\
    &\hspace{5cm}
    +\mu_2 e^{ib_1 \cdot \ls}e^{-i \svec_3 \cdot r} \begin{pmatrix}
        1 & \sq{e^{i2\pi/3}}\\
        \sq{e^{-i2\pi/3}} & 1
    \end{pmatrix}((\rot_{11 \pi/6} \angvar{K}) \cdot D)\Bigg),\\
    \Tu (r) &:= \frac{\beta |K|}{2 |\Gamma|}
    \Bigg(
    \mu_0 e^{-i \svec_1 \cdot r} \begin{pmatrix}
        1 & 1\\
        1 & 1
    \end{pmatrix}+ 
    \mu_1 e^{-ib_2 \cdot \ls}e^{-i \svec_2 \cdot r} \begin{pmatrix}
        1 & \sq{e^{-i2\pi/3}}\\
        \sq{e^{i2\pi/3}} & 1
    \end{pmatrix}\\
    &\hspace{5cm}
    +\mu_2 e^{ib_1 \cdot \ls}e^{-i \svec_3 \cdot r} \begin{pmatrix}
        1 & \sq{e^{i2\pi/3}}\\
        \sq{e^{-i2\pi/3}} & 1
    \end{pmatrix}\Bigg),
\end{align*}
with $L$ defined in \eqref{eq:HBM_def}. The operator $\HBMsd$ corresponds to the second-order term in the Taylor expansion \eqref{eq:Dirac}, while $\hoppingT_2$ contains the terms of the $\{ \nabla \hathperp (\rot_{2\pi j/3} K; \eps)\}_{j \in \{0,1,2\}}$ that were omitted from $\hoppingT_{\nabla, {\rm NN}}$. The operators $\mp\frac{1}{2} \beta i \sigma_3 L$ on the diagonal blocks and $\Tu (r)$ on the off-diagonal correct for the twisting of the monolayer Dirac points. Indeed, although the parameters of the effective model depend only on evaluations of hopping functions at the untwisted Dirac point $K$ (and its $\pi i/3$ rotations, $i \in \{0,1, \dots, 5\}$), the Dirac point for layer $j$ is in fact rotated by $\theta_j/2$; see \eqref{eq:Dirac_pt_rot}.

We observe that $\hoppingT$, $\hoppingT_{\rm NNN}$ and $\Tu$ are point-wise multiplication operators with $\hoppingT^\dagger$, $\hoppingT_{\rm NNN}^\dagger$ and $\Tu^\dagger$ their Hermitian conjugates. Although $\hoppingT_{\nabla, {\rm NN}}$ contains both variable coefficients and derivatives, its adjoint is simply
\begin{align*}
    \hoppingT^\dagger_{\nabla, {\rm NN}} &= \frac{1}{|\Gamma|} \Bigg(
    \overline{\lambda_0} e^{i \svec_1 \cdot r} \begin{pmatrix}
        1 & 1\\
        1 & 1
    \end{pmatrix}(\angvar{K} \cdot D) + 
    \overline{\lambda_2} e^{ib_2 \cdot \ls}e^{i \svec_2 \cdot r} \begin{pmatrix}
        1 & \sq{e^{i2\pi/3}}\\
        \sq{e^{-i2\pi/3}} & 1
    \end{pmatrix} ((\rot_{2\pi/3} \angvar{K}) \cdot D)\\
    &\hspace{4cm}
    +\overline{\lambda_4} e^{-ib_1 \cdot \ls}e^{i \svec_3 \cdot r} \begin{pmatrix}
        1 & \sq{e^{-i2\pi/3}}\\
        \sq{e^{i2\pi/3}} & 1
    \end{pmatrix}((\rot_{4\pi/3} \angvar{K}) \cdot D) \Bigg),
\end{align*}
as $[v_1 \cdot D, e^{i v_2 \cdot r}] = 0$ if $v_1 \cdot v_2 = 0$. Finally, the adjoint of $\hoppingT_2$ is
\begin{align}\label{eq:adjoint}
\begin{split}
    \hoppingT_2^\dagger &= \frac{1}{|\Gamma|} \Bigg(
    \overline{\mu_0} ((\rot_{\pi/2} \angvar{K}) \cdot D) e^{i \svec_1 \cdot r} \begin{pmatrix}
        1 & 1\\
        1 & 1
    \end{pmatrix}+ 
    \overline{\mu_1} e^{ib_2 \cdot \ls}((\rot_{7\pi/6} \angvar{K}) \cdot D)e^{i \svec_2 \cdot r} \begin{pmatrix}
        1 & \sq{e^{i2\pi/3}}\\
        \sq{e^{-i2\pi/3}} & 1
    \end{pmatrix}\\
    &\hspace{3cm}
    +\overline{\mu_2} e^{-ib_1 \cdot \ls}((\rot_{11\pi/6} \angvar{K}) \cdot D)e^{i \svec_3 \cdot r} \begin{pmatrix}
        1 & \sq{e^{-i2\pi/3}}\\
        \sq{e^{i2\pi/3}} & 1
    \end{pmatrix}\Bigg)\\
    &= \frac{1}{|\Gamma|} \Bigg(
    \overline{\mu_0} e^{i \svec_1 \cdot r} \begin{pmatrix}
        1 & 1\\
        1 & 1
    \end{pmatrix}((\rot_{\pi/2} \angvar{K}) \cdot D) + 
    \overline{\mu_1} e^{ib_2 \cdot \ls}e^{i \svec_2 \cdot r} \begin{pmatrix}
        1 & \sq{e^{i2\pi/3}}\\
        \sq{e^{-i2\pi/3}} & 1
    \end{pmatrix}((\rot_{7\pi/6} \angvar{K}) \cdot D)\\
    &\hspace{2cm}
    +\overline{\mu_2} e^{-ib_1 \cdot \ls}e^{i \svec_3 \cdot r} \begin{pmatrix}
        1 & \sq{e^{-i2\pi/3}}\\
        \sq{e^{i2\pi/3}} & 1
    \end{pmatrix}((\rot_{11\pi/6} \angvar{K}) \cdot D)\Bigg)
    - 2 \Tu^\dagger (r).
\end{split}
\end{align}

Our first main result is then the following
\begin{theorem}\label{thm:main2}
    Define $H: \cH \to \cH$ by \eqref{eq:H}-\eqref{eq:intra}-\eqref{eq:Hperp} with $\cH := \ell^2 (\cR_1; \mathbb{C}^2)\oplus \ell^2 (\cR_2; \mathbb{C}^2)$.
    Define the operators $\HBM, \HBMNNN, \HBMgNN, \HBMt$ as above.
    Suppose the intralayer hopping function $h$ satisfies Assumption \ref{assumption:h}, and the constant \srq{$0 < \sdpt \le 1$} 
    and interlayer hopping function $\hperp$ satisfy Assumption \ref{assumption:hperp2}.
    Assume $\alpha$ in \eqref{eq:HBM_def} is nonzero, and
    let \srq{$f_0 \in H^{6+\sdpt} (\mathbb{R}^2; \mathbb{C}^4)$}. 
    For $i\in \{1,2\}$, $\sigma \in \{A,B\}$ and $R_i \in \cR_i$, define $$(\psi_{i,0})^\sigma_{R_i} := \eps f^\sigma_{i,0} (\eps (R_i+\tau_i^\sigma)) e^{i K_i \cdot (R_i + \tau_i^\sigma)},$$
    where $K_i := \rot_{\theta_i/2} K$ is the Dirac point corresponding to layer $i$, and we recall that $\theta_i = (-1)^i \theta$.
    Suppose that the wave function $\psi := ((\psi_1)^\sigma_{R_1} (\microtime), (\psi_2)^\sigma_{R_2}(\microtime))$ solves the Schr\"odinger equation $i \partial_\microtime \psi = H \psi$ subject to the initial condition $(\psi_j)^\sigma_{R_j} (0) = (\psi_{j,0})^\sigma_{R_j}$.

    Define $\radNNN$ and $\radpNN$ by \eqref{eq:rad_ang}.
    Set $f := \tff{1} + \radNNN (\eps) \tff{{\rm NNN}} + \radpNN (\eps) \tff{\nabla, {\rm NN}} + \eps \tff{2}$, where the $\eps$-independent functions $f^{(j)}= f^{(j)} (r,\macrotime)$ 
    on the right-hand side are defined iteratively by
    \begin{align}\label{eq:it}
    \begin{split}
        (i \partial_\macrotime -\HBM) \tff{1} = 0, \qquad & \tff{1} (r,0) = f_0 (r),\\
        (i \partial_\macrotime -\HBM) \tff{{\rm NNN}} = \HBMNNN \tff{1}, \qquad & \tff{{\rm NNN}} (r,0) = 0,\\
        (i \partial_\macrotime -\HBM) \tff{\nabla, {\rm NN}} = \HBMgNN \tff{1}, \qquad & \tff{\nabla, {\rm NN}} (r,0) = 0,\\
        (i \partial_\macrotime -\HBM) \tff{2} = \HBMt\tff{1}, \qquad & \tff{2} (r,0) = 0.
    \end{split}
    \end{align}
    
    Set $\eshift := \tilde{h}^{A,A} (K)$. Then for any 
    \srq{$0 < \sdpt_- < \sdpt$},
    the wave function
    \begin{align}\label{eq:def_phi}
        (\mswf_{i})^\sigma_{R_i} (\microtime) := \eps \tf^\sigma_{i} (\eps (R_i+\tau_i^\sigma),\eps \microtime) e^{i K_i \cdot (R_i + \tau_i^\sigma)}e^{-i\eshift \microtime}
    \end{align}
    satisfies
    \begin{align}\label{eq:error2}
        \norm{\mswf (\microtime) -\psi (\microtime)}_\cH \le C \srq{\norm{f_0}_{H^{6+\sdpt}}} \eps^{1 + \sdpt_-} \left(\eps \microtime + \eps^{\sdpt -\sdpt_-} (\eps \microtime)^2 \right)
    \end{align}
    uniformly in $0 < \eps < 1$ and $\microtime \ge 0$.
\end{theorem}
For the proof of Theorem \ref{thm:main2}, see Appendix \ref{subsec:proof_thm_main2}. Note that 
\eqref{eq:neighbor_bds} implies that
\begin{align}\label{eq:xi_zeta_bd}
    |\radNNN (\eps)| + |\radpNN (\eps)| \le C \eps^{\frac{1+\sdpt}{2}}, \qquad \qquad 0 < \eps < 1,
\end{align}
thus the terms $\radNNN (\eps) \tff{{\rm NNN}}$ and $\radpNN (\eps) \tff{\nabla, {\rm NN}}$ are indeed of higher order than $\tff{1}$.  
\srq{We make no assumption on the size of $\radNNN (\eps)$ relative to $\radpNN (\eps)$.}

\srq{
\begin{remark}\label{remark:ex_rad}
    For the particular choice of interlayer hopping function \eqref{eq:example_hopping} from Example \ref{ex:2},
\begin{align*}
        \radNNN (\eps)  = C_0 \eps^{\sqrt{\frac{4|K|^2+\hoppingparam^2}{|K|^2+\hoppingparam^2}}-1}, \qquad \radpNN (\eps) = C_1 \eps
    \end{align*}
for some real constants $C_0$ and $C_1$. Moreover, in this example we have $0 < \sdpt < \sqrt{7} - 2$, making it impossible to obtain a convergence rate of $O (\eps^{\sqrt{7}}\microtime)$ in \eqref{eq:error2}.
But for any fixed $\delta > 0$, we can obtain a convergence rate of $O(\eps^{\sqrt{7} - \delta}\microtime)$ provided that $\hoppingparam$ is sufficiently small.
\end{remark}}
\begin{remark}
    The above theorem is an extension of \cite[Theorem 3.1]{watson2023bistritzer} to our second-order effective continuum model.
    In \cite{watson2023bistritzer}, it was shown (under stronger assumptions on the hopping functions) that the function 
    \begin{align*}
        (\phi^{\text{BM}}_i)^\sigma_{R_i} (\microtime) := \eps (\tff{1})^\sigma_{i} (\eps (R_i+\tau_i^\sigma),\eps \microtime) e^{i K_i \cdot (R_i + \tau_i^\sigma)}e^{-i\eshift \microtime}
    \end{align*}
    constructed from the Bistritzer-MacDonald wave-function $\tff{1}$ defined by \eqref{eq:it} satisfies
    \begin{align*}
        \norm{\mswf^{\text{BM}} (\microtime) -\psi (\microtime)}_\cH \le C\eps^{1+\eta_*} \microtime
    \end{align*}
    uniformly in $\microtime$, for some $0 < \eta_* < 1$. 
    Thus, over time scales of $O(\eps^{-1})$, Theorem \ref{thm:main2} improves the existing $O(\eps^{1+\eta_*} \microtime)$ rate of convergence to $O(\eps^{2 + \sdpt_-} \microtime)$.
\end{remark}
\begin{remark}
    The bound \eqref{eq:error2} implies that the continuum model produces an $O(\eps^{1 + \sdpt_-})$ approximation of the tight-binding dynamics over timescales of $O(\eps^{-1})$. Moreover, for any $t_0 > 0$,
    \begin{align*}
        \lim_{\eps \to 0} \; \sup \{\norm{\mswf (\microtime) -\psi (\microtime)}_\cH: 0 \le \microtime \le t_0\eps^{-(3+\sdpt_-)/2}\} = 0.
    \end{align*}
    Finally, we comment on the long-time behavior of the continuum approximation.
    Recall that $\psi$ solves a Schr\"odinger equation and thus $\norm{\psi(\microtime)}_\cH$ is independent of $\microtime$. However, Sobolev norms of $f$, and thus $\norm{\phi(\microtime)}_\cH$, might go to infinity as $\microtime \to \infty$ (see Lemma \ref{lemma:Sobolev_ms2}). 
    Indeed, \eqref{eq:error2} does not rule out the possibility that at fixed $\eps$, the error $\norm{\mswf (\microtime) -\psi (\microtime)}_\cH$ between the continuum and discrete models grows quadratically in $\microtime$.
\end{remark}

\subsection{An alternate perspective}\label{subsec:alternate}
Instead of the systematic multi-scale expansion in Theorem \ref{thm:main2}, one might instead write an effective model directly in terms of the operator 
\begin{align}\label{eq:Hfull}
    \Hfull := \HBM + \radNNN (\eps) \HBMNNN + \radpNN (\eps) \HBMgNN + \eps \HBMt.
\end{align}
The validity of such an approach requires the assumption that for any $s \ge 0$, there exist positive constants $\ellC$, $c$ and $C$ such that
\begin{align}\label{eq:ellipticity}
    c \eps^{s/2} \norm{u}_{H^{s}} \le \norm{((\Hbm)^2 + \ellC)^{s/4} u}_{L^2} \le C\norm{u}_{H^{s}}, \qquad u \in H^s.
\end{align}
In fact, the above upper bound is guaranteed to hold since $\Hbm$ is a second order differential operator and all derivatives of any of its coefficients are uniformly bounded in $\eps$. 
The lower bound in \eqref{eq:ellipticity}, which is an ellipticity assumption on $\Hbm$, holds provided the Fourier symbol of the leading-order operator
\begin{align*}
    \HBMsd = \frac{1}{2}\begin{pmatrix}
        \alpha_d (D_{r_1}^2 + D_{r_2}^2) & \alpha_o (D^2_{r_1} - D^2_{r_2} + 2i D_{r_1 r_2})\\ 
        \overline{\alpha_o} (D^2_{r_1} - D^2_{r_2} - 2i D_{r_1 r_2}) & \alpha_d (D_{r_1}^2 + D_{r_2}^2)
    \end{pmatrix} 
\end{align*}
grows quadratically at infinity.
The symbol of $S$ is
\begin{align*}
    \hat{S} (p_1, p_2) = \frac{1}{2}\begin{pmatrix}
        \alpha_d (p_1^2 + p_2^2) & \alpha_o (p_1+ip_2)^2\\ 
        \overline{\alpha_o} (p_1 - ip_2)^2 & \alpha_d (p_1^2 + p_2^2)
    \end{pmatrix},
\end{align*}
which has eigenvalues $\frac{1}{2} (\alpha_d \pm |\alpha_o|) (p_1^2 + p_2^2)$. Therefore, \eqref{eq:ellipticity} holds so long as $|\alpha_d| \ne |\alpha_o|$.

We now state our second main result.
\begin{theorem}\label{thm:main}
    Define $H: \cH \to \cH$ by \eqref{eq:H}-\eqref{eq:intra}-\eqref{eq:Hperp} with $\cH := \ell^2 (\cR_1; \mathbb{C}^2)\oplus \ell^2 (\cR_2; \mathbb{C}^2)$. 
    Assume the hopping functions $h$ and $\hperp$ respectively satisfy Assumptions \ref{assumption:h} and \ref{assumption:hperp2} for some \srq{$0 < \sdpt \le 1$} defined in the latter.
    Define $\Hfull$ by \eqref{eq:Hfull}, and assume the coefficients $\alpha_d$ and $\alpha_o$ appearing in $\HBMt$ have different absolute values, and that $\alpha$ in \eqref{eq:HBM_def} is nonzero.
    For $i\in \{1,2\}$ and $\sigma \in \{A,B\}$, let \srq{$f^\sigma_{i,0} \in H^{8+\sdpt} (\mathbb{R}^2; \mathbb{C})$}. 
    For $R_i \in \cR_i$, define $$(\psi_{i,0})^\sigma_{R_i} := \eps f^\sigma_{i,0} (\eps (R_i+\tau_i^\sigma)) e^{i K_i \cdot (R_i + \tau_i^\sigma)},$$
    as in Theorem \ref{thm:main2}.
    Suppose that the wave function 
    $\psi := ((\psi_1)^\sigma_{R_1}(\microtime), (\psi_2)^\sigma_{R_2}(\microtime))$ solves the Schr\"odinger equation $i \partial_\microtime \psi = H \psi$ subject to the initial condition $(\psi_j)^\sigma_{R_j} (0) = (\psi_{j,0})^\sigma_{R_j}$.
    Moreover, suppose the function $\phf := (\phf^\sigma_1(r,\macrotime), \phf^\sigma_2(r,\macrotime))$ satisfies the (continuous-in-space) Schr\"odinger equation $i \partial_\macrotime \phf = \Hfull \phf$ 
    with the initial condition $\phf^\sigma_j (r,0)= f^\sigma_{j,0} (r)$.
    Then for any \srq{$0 < \sdpt_- < \sdpt$}, 
    the wave function $$(\phph_{i})^\sigma_{R_i} (\microtime) := \eps \phf^\sigma_{i} (\eps (R_i+\tau_i^\sigma),\eps \microtime) e^{i K_i \cdot (R_i + \tau_i^\sigma)}e^{-i\eshift \microtime}$$ satisfies
    \begin{align}\label{eq:bd_thm_main}
        \srq{\norm{\phph (\microtime)-\psi (\microtime)}_\cH \le C \left( \norm{f_0}_{H^{6+\sdpt}} \eps^{1+\sdpt_-} \left(\eps \microtime + \eps^{\sdpt -\sdpt_-} (\eps \microtime)^2 \right) + \norm{f_0}_{H^{8+\sdpt}} \eps^{2 - (\sdpt - \sdpt_-)}(\eps \microtime)^3 \right)}
    \end{align}
    uniformly in $0 < \eps < 1$ and $\microtime \ge 0$.
\end{theorem}

We postpone the proof of Theorem \ref{thm:main} to Appendix \ref{subsec:proof_thm_main}. Although the order of convergence in Theorems \ref{thm:main2} and \ref{thm:main} is the same, the 
former makes no assumptions on the parameters $\alpha_d$ and $\alpha_o$, and requires less regularity of the initial data $f_0$.

\begin{remark}
    One could extend the derivation of the effective Hamiltonian \eqref{eq:Hfull} to any order \srq{by keeping more terms in the Taylor series expansions of the hopping functions \eqref{eq:intra_Taylor} and \eqref{eq:filtered_Taylor}, and evaluating $\hathperp (q; \eps)$ and its derivatives at larger $q = K_2+G_2$ ($G_2 \in \rot_{\theta/2} \cR^*$) in the latter}. Thus, given any $s >0$, our procedure could be used to obtain an $O (\eps^s)$ approximation of the tight-binding dynamics over time intervals of $O (\eps^{-1})$. However, we believe our second-order approximation already captures the mathematical insight of this approach without introducing the extra complexity of higher orders.
\end{remark}


\section{Symmetries}\label{sec:symmetries}
In this section, we identify symmetries of the tight-binding and continuum models for TBG from Sections \ref{subsec:bilayer} and \ref{subsec:alternate}.
\subsection{Monolayer tight-binding model}\label{subsec:symmetries_mono}
Set $\co := \frac{1}{2} (\tau^A + \tau^B) - (\frac{a}{2}, 0)$, which corresponds to the center of a hexagon on the lattice (see Figure \ref{fig:lattice}). Define $\Pi : \cR \to \cR$ and $\Pi^\sigma : \cR \to \cR$ by
\begin{align*}
    \Pi (R) := -R - \tau^A - \tau^B + 2 \co, \qquad \Pi^\sigma (R) := \rot_{2\pi/3} (R + \tau^\sigma - \co) - \tau^\sigma + \co.
\end{align*}

Let $\cP: \cH_1 \to \cH_1$ be the parity operator defined by $(\cP \psi)^{\sigma}_R := \psi^{\sigma^c}_{\Pi (R)}$, where $\sigma^c$ is the element in $\{A,B\}$ \emph{not} equal to $\sigma$. This is a $\pi$-rotation about the shifted $z$-axis whose corresponding symmetry is known as $\mathcal{C}_{2z}$ symmetry \cite{kaxiras2019quantum}. The center of rotation is the point $\co$.

Let $\cC :\cH_1 \to \cH_1$ be the complex conjugation operator, so that $(\cC \psi)^\sigma_R := \overline{\psi^\sigma_R}$. In the literature, this is often referred to as the time-reversal operator and denoted instead by $\mathcal{T}$.

Next, define the rotation operator $\rotop :\cH_1 \to \cH_1$ by $(\rotop \psi)^\sigma_R = \psi^\sigma_{\Pi^\sigma (R)}.$ This $2\pi/3$ rotation about $\co$ is often denoted $\mathcal{C}_{3z}$.

A direct calculation using the self-adjointness, $h(-r) = \overline{h(r)}$, of $H$ reveals that $[H, \cP \cC]= 0$. Similarly, the $2\pi/3$-rotation invariance of $h$ in Assumption \ref{assumption:h} implies that $[H, \rotop] = 0$.



\medskip

One might wonder if it is possible to generalize the hopping function in \eqref{eq:tight-binding} so that $H$ still obeys the $\cP \cC$ and $\rotop$ symmetries. That is, suppose instead that the self-adjoint operator $H$ is defined by
\begin{align}\label{eq:tight-binding_general}
    (H \psi)^\sigma_R = \sum_{R' \in \cR} \sum_{\sigma ' \in \{A,B\}} h^{\sigma \sigma'} (R-R') \psi^{\sigma '}_{R'},
\end{align}
where the hopping functions $h^{\sigma \sigma'}: \cR \to \mathbb{C}$ are chosen such that $[H, \cP \cC] = [H, \rotop] = 0$. The self-adjointness and $\cP \cC$ symmetry of $H$ respectively imply that
\begin{align}\label{eq:h_SA}
    h^{\sigma ' \sigma} (-r) = \overline{h^{\sigma \sigma'} (r)}, \qquad h^{\sigma (\sigma ')^c} (r) = \overline{h^{\sigma^c \sigma'} (-r)},
\end{align}
and thus $h^{AA} (r) = \overline{h^{AA} (-r)} = h^{BB} (r)$. Recalling the definitions of $\cR_\pm$ and $\cRall$ in Section \ref{subsec:mono}, since the sets $\cR, \cR_+, \cR_-$ are pairwise disjoint, there exists a function $h : \cRall \to \mathbb{C}$ such that $h^{\sigma \sigma'} (R) = h (R + \tau^\sigma - \tau^{\sigma'})$ for all $R \in \cR$. The self-adjointness of $H$ then implies that $h(-r) = \overline{h(r)}$, while the assumption $[H, \rotop] = 0$ implies that $h(\rot_{2\pi/3} r) = h(r)$.


This shows that if $H$ defined by \eqref{eq:tight-binding_general} commutes with $\cP \cC$ and $\rotop$, then the hopping function $h^{\sigma \sigma'} (R-R')$ must be of the form $h (R-R' + \tau^\sigma - \tau^{\sigma'})$ for some function $h: \cRall \to \mathbb{C}$ satisfying $h(-r) = \overline{h(r)}$ and $h(\rot_{2\pi/3} r) = h(r)$ for all $r \in \cRall$. In other words, the Hamiltonian defined by \eqref{eq:tight-binding} and Assumption \ref{assumption:h} 
covers the entire class of self-adjoint operators of the form \eqref{eq:tight-binding_general} that satisfy the above $\cP \cC$ and $\rotop$ symmetries.

\begin{remark}
    In any physical model of graphene (see, e.g. \cite{Fang_Kaxiras_2016,jung2013tight,slater1954simplified}), the function $h$ is real-valued, meaning that $H$ commutes with both $\cP$ and $\cC$ separately.
\end{remark}
\begin{remark}
    One could instead define the parity and rotation operators with respect to different centers of rotation. In particular, we could take $(\cP \psi)^\sigma_R := \psi^{\sigma^c}_{-R}$, which corresponds to a $\pi$-rotation about the midpoint of the $A$ and $B$ site in the $R=0$ unit cell, that is $\frac{1}{2} (\tau^A + \tau^B)$.
    The rotation operator could be defined by $(\rotop \psi)^B_R := \psi^B_{\rot_{2\pi/3} R}$ and $(\rotop \psi)^A_R := \psi^A_{\rot_{2\pi/3} (R + \tau^{A,B}) - \tau^{A,B}}$, which is a $2\pi/3$ rotation about the point $\tau^B$. With these alternate definitions, the statements in this section would still hold. 
    That is, if $H$ is given by \eqref{eq:tight-binding_general}, then $H$ is of the form \eqref{eq:tight-binding} with $h$ satisfying Assumption \ref{assumption:h} if and only if $[H, \cP \cC] = [H, \rotop] = 0$.
\end{remark}



\subsection{Bilayer tight-binding model}\label{subsec:symmetries_bilayer}
For ease of exposition, we restrict our discussion of symmetries of the bilayer tight-binding model \eqref{eq:H}-\eqref{eq:intra}-\eqref{eq:Hperp} to the parameter values 
\begin{align}\label{eq:particular}
    \tau^A := \left(\frac{a}{2}, \sq{-}\frac{a}{2\sqrt{3}}\right),\qquad \ls := (0,0).
\end{align}
By the definition of $\tau^B$ in Section \ref{subsec:mono}, this means $\tau^B = \left(\frac{a}{2}, \sq{\frac{a}{2\sqrt{3}}}\right).$
Our choice of $\tau^A$ and $\ls$ ensures a symmetric TBG lattice and corresponds to $\co = (0,0)$ in Section \ref{subsec:symmetries_mono}; see Figure \ref{fig:symmetry}.
\begin{figure}[h!]
    \centering
    \includegraphics[width=.45\textwidth]{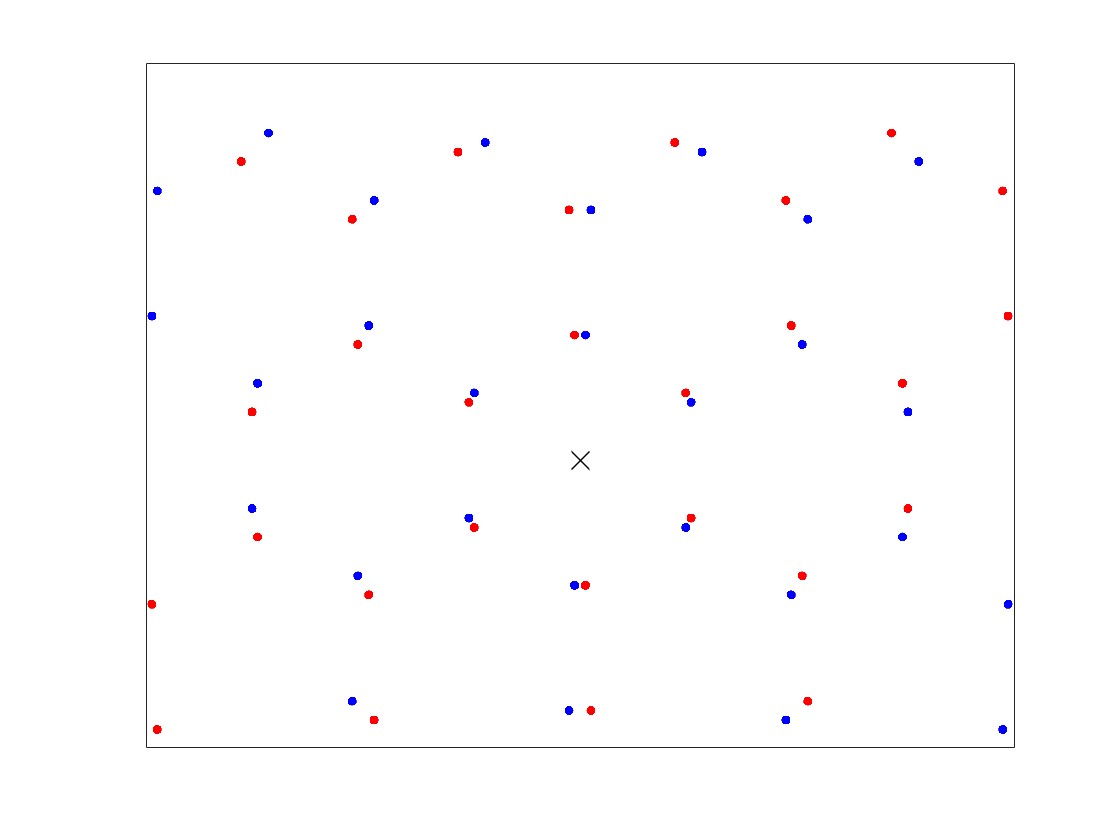}
    \includegraphics[width=.45\textwidth]{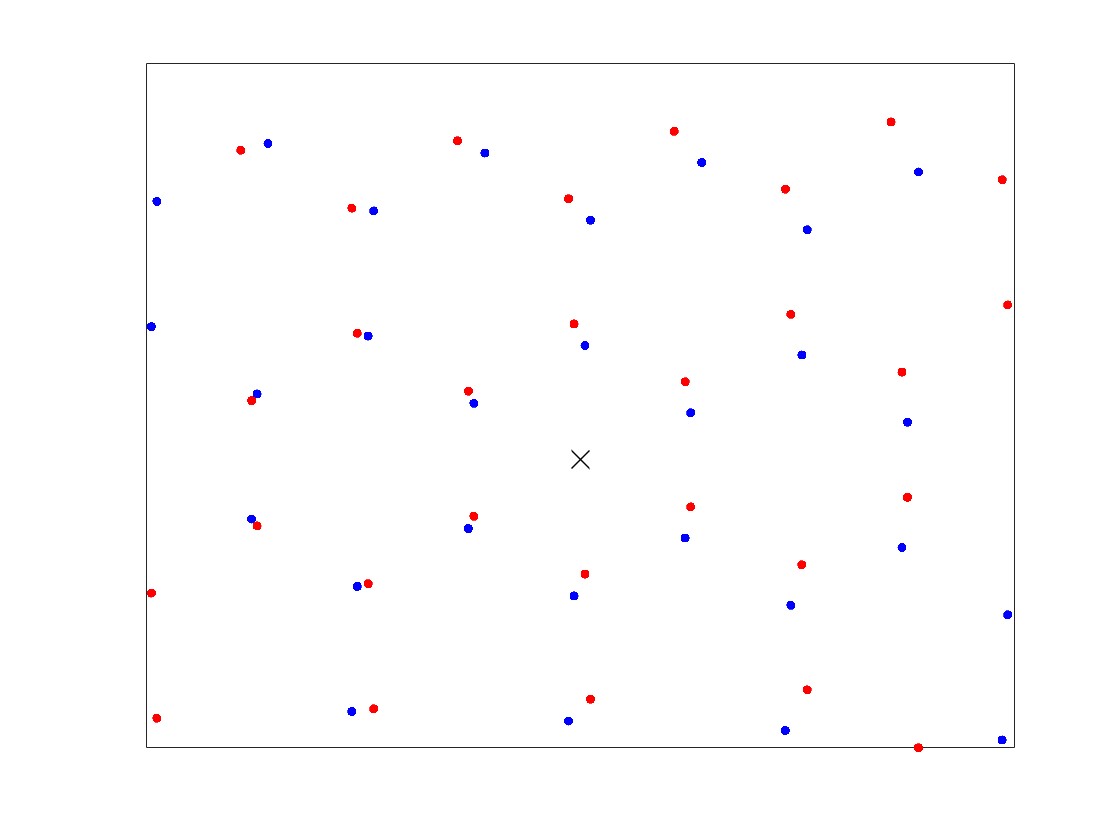}
    \caption{TBG lattice at twist angle $\theta = 5^\circ$ with the red and blue dots respectively corresponding to layers $1$ and $2$. The point ``$\times$'' marks the center of rotation. Left: parameters are given by \eqref{eq:particular}. Right: same value of $\tau^A$, but now $\ls := (0, a/10)$.}
    \label{fig:symmetry}
\end{figure}

\begin{remark}
    For arbitrary $\tau^A$ and $\ls$, the tight-binding Hamiltonian $H$ will in general not have the symmetries we identify below. However, the ergodic structure of the bilayer lattice at incommensurate twist angles implies that one could translate the lattice to achieve values of $\tau^A$ and $\ls$ arbitrarily close to those in \eqref{eq:particular}. 
    Therefore, given the smoothness of $\hperp$, there would exist corresponding symmetry operators that \emph{almost} commute with $H$ up to an arbitrarily small error.
\end{remark}

Note that Assumption \ref{assumption:hperp2} does not include any symmetries of the interlayer hopping function $\hperp$ and thus does not imply the usual symmetries associated with twisted bilayer graphene. 
We now collect additional assumptions on $\hperp$ and identify their corresponding symmetries. 
In the following proposition, it is assumed that $h$ and $\hperp$ respectively satisfy Assumptions \ref{assumption:h} and \ref{assumption:hperp2}, with $h(-r) = \overline{h (r)}$ for all $r \in \cRall$. 
\begin{proposition}\label{prop:symmetries}
    Define $H : \cH \to \cH$ by \eqref{eq:H}-\eqref{eq:intra}-\eqref{eq:Hperp}. 
    \begin{enumerate}
        \item \label{it:pc} \textbf{$\cP \cC $ symmetry.} Suppose $\hperp (-r;\eps) = \overline{\hperp (r;\eps)}$ for all $r \in \mathbb{R}^2$ and $0 < \eps < 1$. 
        For $j\in \{1,2\}$, define $\Pi_j : \cR_j \to \cR_j$ by
        \begin{align}\label{eq:Pi0}
            \Pi_j (R_j) := -R_j - \tau_j^A - \tau_j^B. 
        \end{align}
        For $j \in \{1,2\}$,
        define $\cP_j : \ell^2 (\cR_j; \mathbb{C}^2) \to \ell^2 (\cR_j; \mathbb{C}^2)$ and $\cC_j :\ell^2 (\cR_j; \mathbb{C}^2) \to \ell^2 (\cR_j; \mathbb{C}^2)$ by $(\cP_j \psi)^{\sigma}_{R_j} = \psi^{\sigma^c}_{\Pi_j (R_j)}$ and $(\cC_j \psi)^\sigma_{R_j} = \overline{\psi^\sigma_{R_j}}$, where $\sigma^c$ is the element in $\{A,B\}$ not equal to $\sigma$. Define $\cD : \cH \to \cH$ by $\cD := \diag (\cP_1 \cC_1, \cP_2 \cC_2)$. Then $[H, \cD] = 0$.  
        \item \label{it:rot} \textbf{Rotation symmetry.} Suppose $\hperp (\rot_{2\pi/3} r; \eps) = \hperp (r; \eps)$ for all $r \in \mathbb{R}^2$ and $0 < \eps < 1$. 
        For $\sigma \in \{A,B\}$ and $j \in \{1,2\}$, 
        define $\Pi^\sigma_j: \cR_j \to \cR_j$ by
        \begin{align}\label{eq:Pi}
            \Pi^\sigma_j (R_j) := \rot_{2\pi/3} (R_j + \tau_j^\sigma) - \tau_j^\sigma.
        \end{align}
        Define $\cU : \cH \to \cH$ by $((\cU \psi)^\sigma_j)_{R_j} = (\psi^\sigma_j)_{\Pi^\sigma_j (R_j)}.$
        Then $[H,\cU] = 0$.
        \item \textbf{Mirror symmetry.}  \label{it:M} Define the maps $m_x : \mathbb{R}^2 \to \mathbb{R}^2$ and $m_y: \mathbb{R}^2 \to \mathbb{R}^2$ by $m_x (r_1, r_2) = (-r_1, r_2)$ and $m_y (r_1, r_2) = (r_1, -r_2)$.
        Next, for $\sigma \in \{A,B\}$ and $j \in \{1,2\}$, define $\Pi^\sigma_{x,j} : \cR_j \to \cR_{j^c}$ and $\Pi^\sigma_{y,j}: \cR_j \to \cR_{j^c}$ by
        \begin{align*}
            \Pi^\sigma_{x,j} (R_j) := m_x (R_j + \tau^\sigma_j) - \tau^{\sigma}_{j^c}, \qquad \Pi^\sigma_{y,j} (R_j) := m_y(R_j + \tau^\sigma_j) - \tau^{\sigma^c}_{j^c},
        \end{align*}
        where $j^c$ is the element of $\{1,2\}$ not equal to $j$.
        Finally, define the ``mirror operators'' $\cM_x : \cH \to \cH$ and $\cM_y : \cH \to \cH$ by
        \begin{align*}
            ((\cM_x \psi)^\sigma_j)_{R_j} = (\psi^\sigma_{j^c})_{\Pi^\sigma_{x,j}(R_{j})}, \qquad
            ((\cM_y \psi)^\sigma_j)_{R_j} = (\psi^{\sigma^c}_{j^c})_{\Pi^\sigma_{y,j}(R_{j})}.
        \end{align*}
        Set $i \in \{x,y\}$, if $h (m_i (r)) = h(r)$ and $\hperp (m_i (r'); \eps) = \overline{\hperp (r'; \eps)}$ for all $r \in \cRall$, $r' \in \mathbb{R}^2$ and $0 < \eps < 1$, then $[H, \cM_i] = 0.$
    \end{enumerate}
\end{proposition}
For a proof of Proposition \ref{prop:symmetries}, see Appendix \ref{subsection:proofs_symmetries}.



\begin{remark}
    Above, the operators $\cP_j$ correspond to a $\pi$-rotation about the origin, while $\cU$ is a $2\pi/3$-rotation about the origin. The mirror operators $\cM_x$ and $\cM_y$ respectively correspond to reflections about the $y$- and $x$-axes.
\end{remark}

\begin{remark}\label{remark:different}
    Under different assumptions on $h$ and $\hperp$, the first result in Proposition \ref{prop:symmetries} can be strengthened. Indeed, if $h$ and $\hperp$ are real-valued (as is the case for graphene), then the operator $\cC := \diag (\cC_1, \cC_2)$ commutes with $H$. If, in addition, $\hperp (-r;\eps) = \hperp (r;\eps)$ for all $r \in \mathbb{R}^2$ and $0 < \eps < 1$, then $\cP := \diag (\cP_1, \cP_2)$ also commutes with $H$. We observe that these results, as well as part 3 of Proposition \ref{prop:symmetries}, in fact do not require the $2\pi/3$-rotation invariance of $h$ in Assumption \ref{assumption:h}.
\end{remark}


\subsection{Continuum model}\label{subsection:symmetries_continuum}
Let $\Hfull_0$ denote the operator $\Hfull$ from \eqref{eq:Hfull} when $\ls = 0$. The following proposition shows that for any $\ls$, $\Hfull$ and $\Hfull_0$ are related by a unitary transformation.
\begin{proposition}\label{prop:unitary}
    Set $\phi := \frac{1}{2} \ls \cdot K$ and $\fw := \frac{1}{\beta}\rot_{-\pi/2} \ls$. Define the unitary operator $\cU$ by
    \begin{align}\label{eq:unitary}
        \cU f (r) := \begin{pmatrix}
            e^{-i \phi}I_2 & 0\\
            0 & e^{i\phi} I_2
        \end{pmatrix} f(r-\fw).
    \end{align}
    Then $\cU \Hfull \cU^\dagger = \Hfull_0$.
\end{proposition}
We refer to Appendix \ref{subsection:proofs_symmetries} for a proof of Proposition \ref{prop:unitary}. The result will allow for symmetries of $\Hfull_0$ to extend to symmetries of $\Hfull$; see Proposition \ref{prop:continuum_symmetries} below.

\medskip

The effective Hamiltonian $\Hfull$ 
admits symmetries analogous to those for the bilayer tight-binding model. The continuum Hamiltonian also satisfies a translation symmetry that is absent from the tight-binding picture.
To make this more precise, define the ``moir\'e lattice'' \cite{Nam_Koshino_2017, watson2023bistritzer} by
\begin{align}\label{eq:moire}
    \cR_m := \{n_1 a_{m,1} + n_2 a_{m,2} : (n_1, n_2) \in \mathbb{Z}^2\}, \qquad a_{m,1} := \frac{a}{\beta} \begin{pmatrix}
        \sqrt{3}/2\\
        1/2
    \end{pmatrix},
    \qquad a_{m,2} := \frac{a}{\beta} \begin{pmatrix}
        -\sqrt{3}/2\\
        1/2
    \end{pmatrix},
\end{align}
where we recall the respective definitions of $a$ and $\beta$ in 
\eqref{eq:lattice_defn} and \eqref{eq:beta_def}.
For $v \in \mathbb{R}^2$ and $\svec_1$ given by \eqref{eq:svecs}, let $$\cT_v f (r) := \begin{pmatrix}
    I_2 & 0\\
    0 & e^{i \svec_1 \cdot v} I_2
\end{pmatrix}f(r-v)$$ denote the operator which translates a function $f$ by $v$ and then multiplies the third and fourth entries by a $v$-dependent phase. \srq{Observe that if $v \in \cR_m$, then $s_1 \cdot v \in \frac{2 \pi}{3} \mathbb{Z}$, meaning that $e^{i \svec_1 \cdot v} \in \{1, e^{i2\pi/3}, e^{-i2\pi/3}\}$.}
Define the operator which rotates $\mathbb{C}^4$-valued functions counterclockwise by $2\pi/3$ and then multiplies the second and fourth entries by a phase as
\begin{align*}
    \mathscr{R} f(r) := \diag (1, \sq{e^{i2\pi/3}}, 1, \sq{e^{i2\pi/3}}) f (\rot_{2\pi/3}^{\top} r).
\end{align*}
Analogous to Section \ref{subsec:symmetries_bilayer}, define the parity and complex conjugation operators
\begin{align}\label{eq:PC}
    \cP f(r) = f(-r), \qquad \cC f(r) = \overline{f(r)},
\end{align}
and compose these operators together with the operator which swaps sublattice as
\begin{align*}
    \cD := \cP \cC \begin{pmatrix}
        0 & 1 & 0 & 0\\
        1 & 0 & 0 & 0\\
        0 & 0 & 0 & 1\\
        0 & 0 & 1 & 0
    \end{pmatrix}.
\end{align*}
Finally, define the mirror operators by
\begin{align*}
    \cM_x f(r_1, r_2) := \begin{pmatrix}
        0 & 0 & 1 & 0\\
        0 & 0 & 0 & 1\\
        1 & 0 & 0 & 0\\
        0 & 1 & 0 & 0
    \end{pmatrix} \overline{f(-r_1, r_2)},
    \qquad \cM_y f(r_1, r_2) := \begin{pmatrix}
        0 & 0 & 0 & 1\\
        0 & 0 & 1 & 0\\
        0 & 1 & 0 & 0\\
        1 & 0 & 0 & 0
    \end{pmatrix} f(r_1, -r_2).
\end{align*}
We are now ready to state the assumptions needed for $\Hfull$ to commute with the above operators. These assumptions involve the constants $\alpha$, $\alpha_o$, $\lambda_i$ and $\mu_j$ defined in Theorem \ref{thm:Dirac} and \eqref{eq:rad_ang}.
\begin{proposition}\label{prop:continuum_symmetries}
    Define $\Hfull$ by \eqref{eq:Hfull} and $\cU$ by \eqref{eq:unitary}. 
    Assume that $h$ and $\hperp$ respectively satisfy Assumptions \ref{assumption:h} and \ref{assumption:hperp2}, and that $h(-r) = \overline{h (r)}$ for all $r \in \cRall$. Then
    \begin{enumerate}
        \item \textbf{Translation  symmetry.} \label{symm:translation} For any $v \in \cR_m$, $[\Hfull, \cT_v] = 0$.
        \item \textbf{$\cP \cC $ symmetry.} \label{symm:PC} 
        If $\lambda_i, \mu_j \in \mathbb{R}$ for all $i \in \{0,1, \dots, 5\}$ and $j \in \{0,1,2\}$, then $[\Hfull, \cU^\dagger \cD \cU] = 0$.
        \item \textbf{Rotation symmetry.} \label{symm:rotation} If $\lambda_0 = \lambda_2 = \lambda_4$, $\lambda_1, = \lambda_3 = \lambda_5$ and $\mu_0 = \mu_1 = \mu_2$, then $[\Hfull, \cU^\dagger \mathscr{R}\, \cU] = 0$.
        \item \textbf{Mirror symmetry.}\label{symm:mirror_x} 
        If $\alpha, \alpha_o \in \mathbb{R}$, $\lambda_1 = \lambda_5$, $\lambda_2 = \lambda_4$, $\mu_0 = 0$ and $-\mu_1 = \mu_2$, then $[\Hfull, \cU^\dagger \cM_x\, \cU] =0$.
        \item \textbf{Mirror symmetry.} \label{symm:mirror_y} If $\alpha, \alpha_o, \lambda_0, \lambda_3 \in \mathbb{R}$, $\overline{\lambda_1} = \lambda_5$, $\overline{\lambda_2} = \lambda_4$, $\mu_0 \in i \mathbb{R}$ and $\overline{\mu_1} = -\mu_2$, then $[\Hfull, \cU^\dagger \cM_y \, \cU] = 0$.
    \end{enumerate}
\end{proposition}


Next, we provide sufficient conditions for the assumptions of Proposition \ref{prop:continuum_symmetries} to hold.
\begin{proposition}\label{prop:sufficient}
    Suppose $h$ and $\hperp$ respectively satisfy Assumptions \ref{assumption:h} and \ref{assumption:hperp2}, and that $h(-r) = \overline{h (r)}$ for all $r \in \cRall$. Then
    \begin{enumerate}
        \item If $\hperp (-r; \eps) = \overline{\hperp (r; \eps)}$ for all $(r;\eps) \in \mathbb{R}^2 \times (0,1)$, then $\lambda_i, \mu_j \in \mathbb{R}$ for all $i \in \{0,1, \dots, 5\}$ and $j \in \{0,1,2\}$.
        \item If $\hperp (\rot_{2\pi/3} r; \eps) = \hperp (r; \eps)$ for all $(r;\eps) \in \mathbb{R}^2 \times (0,1)$, then $\lambda_0 = \lambda_2 = \lambda_4$, $\lambda_1, = \lambda_3 = \lambda_5$ and $\mu_0 = \mu_1 = \mu_2$.
        \item \label{it:sufficient_alpha} If $h(-r_1, r_2) = h(r_1, r_2)$ for all $(r_1, r_2) \in \cRall$, then $\alpha, \alpha_o \in \mathbb{R}$.
        \item \label{it:sufficient_mirror_x} If $\hperp (r_1, -r_2; \eps) = \hperp (r_1, r_2; \eps)$ for all $(r_1, r_2; \eps) \in \mathbb{R}^2 \times (0,1)$, then $\lambda_1 = \lambda_5$, $\lambda_2 = \lambda_4$, $\mu_0 = 0$ and $-\mu_1 = \mu_2$.
        \item \label{it:sufficient_mirror_y} If $\hperp (-r_1, r_2; \eps) = \overline{\hperp (r_1, r_2; \eps)}$ for all $(r_1, r_2; \eps) \in \mathbb{R}^2 \times (0,1)$, then $\lambda_0, \lambda_3 \in \mathbb{R}$, $\overline{\lambda_1} = \lambda_5$, $\overline{\lambda_2} = \lambda_4$, $\mu_0 \in i \mathbb{R}$ and $\overline{\mu_1} = -\mu_2$.
    \end{enumerate}
\end{proposition}
Our proofs of Propositions \ref{prop:continuum_symmetries} and \ref{prop:sufficient} can be found in Appendix \ref{subsection:proofs_symmetries}.

\begin{remark}
    The assumptions in Proposition \ref{prop:sufficient} are generally satisfied for models of TBG in the physics literature \cite{bistritzer2011moire, slater1954simplified}. The simplest examples are real-valued, radial hopping functions.
\end{remark}
\begin{remark}\label{remark:correspondence}
    As expected, there is a correspondence between symmetries of the discrete and continuum models. Indeed, Proposition \ref{prop:sufficient} demonstrates that provided $\hperp (-r; \eps) = \overline{\hperp (r; \eps)}$ for all $(r;\eps) \in \mathbb{R}^2 \times (0,1)$, the operator $\cD$ (in a slight abuse of notation, we use the same symbol in the discrete and continuum settings) is a symmetry of both models; see Propositions \ref{prop:symmetries} (\ref{it:pc}) and \ref{prop:continuum_symmetries} (\ref{symm:PC}). Similarly, $2\pi/3$-rotation invariance of the interlayer hopping function induces a rotation symmetry on the tight-binding and PDE Hamiltonians, as shown by Propositions \ref{prop:symmetries} (\ref{it:rot}) and \ref{prop:continuum_symmetries} (\ref{symm:rotation}).

    There are, however, also differences in the tight-binding and continuum symmetries. For example, the translation symmetry in Proposition \ref{prop:continuum_symmetries} (\ref{symm:translation}) is absent from the discrete model, as the latter is only moir\'e periodic. Moreover, we observe that the assumptions in Proposition \ref{prop:sufficient} (\ref{it:sufficient_alpha} and \ref{it:sufficient_mirror_y}) that imply the $\cM_y$ symmetry in Proposition \ref{prop:continuum_symmetries} (\ref{symm:mirror_y}) correspond to the discrete $\cM_x$ symmetry in Proposition \ref{prop:symmetries} (\ref{it:M}). Indeed, the assumptions that $h(-r_1, r_2) = h(r_1, r_2)$ and $\hperp (-r_1, r_2; \eps) = \overline{\hperp (r_1, r_2; \eps)}$ surprisingly imply a continuum mirror symmetry in the $r_2$ (not $r_1)$ variable, which demonstrates a $\pi/2$-rotation of the coordinates of our effective model relative to the tight-binding setting. See also the $\pi/2$ rotation in the definition of $\fw$ in Proposition \ref{prop:unitary}. 
    Similarly, if $h$ and $\hperp$ are real-valued, the conditions in Proposition \ref{prop:sufficient} (\ref{it:sufficient_alpha} and \ref{it:sufficient_mirror_x}) implying the continuum $\cM_x$ symmetry agree with those in Proposition \ref{prop:symmetries} (\ref{it:M}) for the discrete $\cM_y$ symmetry. (Given our universal assumption that $h(-r) = \overline{h(r)}$, the condition $h(-r_1, r_2) = h(r_1, r_2)$ is equivalent to $h(r_1, -r_2) = \overline{h(r_1, r_2)}$.) 
    Real-valued hopping functions are required since the continuum $\cM_x$ operator includes complex conjugation while the discrete $\cM_y$ does not. 
    
    Note that the Dirac point $K=\frac{4\pi}{3a}(1,0)$ is invariant with respect to reflection across the $x$-axis, while a reflection across the $y$-axis maps $K$ to $-K$. This explains the presence (resp. absence) of complex conjugation in the continuum $\cM_x$ (resp. $\cM_y$) symmetry operator.
\end{remark}

Finally, we introduce an emergent symmetry of the first order continuum model that does not exist in the tight-binding model and gets broken at higher orders. Define $\cM_x'$ by
\begin{align*}
    \cM_x' f(r_1, r_2) = \begin{pmatrix}
        0 & 1 & 0 & 0\\
        1 & 0 & 0 & 0\\
        0 & 0 & 0 & -1\\
        0 & 0 & -1 & 0
    \end{pmatrix} f(-r_1, r_2).
\end{align*}
\begin{proposition}[\textbf{Particle-hole symmetry of the BM Hamiltonian}]\label{prop:ph}
    If $\alpha \in \mathbb{R}$ and $\lambda_2 = \lambda_4$, then $$\HBM \cU^\dagger \cM'_x \cU + \cU^\dagger \cM'_x \cU \HBM = 0.$$
\end{proposition}
We refer to Appendix \ref{subsection:proofs_symmetries} for a proof. We recall Proposition \ref{prop:sufficient} (parts \ref{it:sufficient_alpha} and \ref{it:sufficient_mirror_x}), which states sufficient conditions for the above assumptions to hold.
\begin{remark}
    Proposition \ref{prop:ph} implies a symmetry for the spectrum of $\HBM$. Indeed, if $\HBM \psi = E \psi$ for some $E \in \mathbb{R}$ and $\psi : \mathbb{R}^2 \to \mathbb{C}^4$, then $\varphi := \cU^\dagger \cM'_x \cU \psi$ satisfies $\HBM \varphi = -E \varphi$. But there is no such symmetry in the tight-binding model and indeed $\Hfull \cU^\dagger \cM'_x \cU + \cU^\dagger \cM'_x \cU \Hfull \ne 0$ in general. The breaking of particle-hole symmetry is illustrated by Figure \ref{fig:bands} (right panel), where the first-order band structure is invariant with respect to reflections about the horizontal axis, while the second-order band structure is not.
\end{remark}

\section{Numerical results}\label{sec:numerical}

We test the approximation error of the continuum models of TBG with twist angle $ \theta = 1.05^\circ$, also called the magic angle because it is the
angle which gives a nearly flat band at the energy level of monolayer Dirac points\cite{watson2023bistritzer, bistritzer2011moire}. There are different tight-binding model parameters to model TBG, most notably a Slater-Koster type hopping function \cite{slater1954simplified, TramblydeLaissardiere_Mayou_Magaud_2010, Moon_Koshino_2012}, and a DFT-based hopping function \cite{Fang_Kaxiras_2016}. 

We start with the Slater-Koster hopping function \eqref{eq:slater-koster} with physical parameters $V^0_{pp\pi}=-2.7 \text{ eV}$, $V^0_{pp\sigma}=0.48 \text{ eV}$, $a_0 = a/\sqrt{3} = 1.42 \text{ \r A}$, $r_0 = 0.319 a_0$.  We set $\ell = 0$ for the monolayer tight-binding hopping function
\begin{align}
    h_\text{intra} (\vec r) := V_{pp\pi}^0 \exp\left(-\frac{|\vec r|-a_0}{r_0}\right).
\end{align}
We assume the distance between graphene layers is $\ell_0 = 3.5 \text{ \r{A}}$, and the interlayer hopping function is
\begin{align}
    h_\text{inter} (\vec r) := V_{pp\pi}^0 \exp\left(-\frac{\sqrt{|\vec r|^2+\ell_0^2}-a_0}{r_0}\right)\left(\frac{|\vec r|^2}{|\vec r|^2+\ell_0^2}\right) + V_{pp\sigma}^0\exp\left(-\frac{\sqrt{|\vec r|^2+\ell_0^2}-\ell_0}{r_0}\right)\left(\frac{\ell_0^2}{|\vec r|^2+\ell_0^2}\right).
\end{align}

%
We report that the intralayer component has the following parameters.
\begin{table}[h!]
\centering
\begin{tabular}{@{}cccc@{}}
\toprule
$\mu = \tilde{h}_\text{intra}^{A,A} (K)$ & $v =  \nabla \tilde{h}_\text{intra}^{A,B} (K) $ & $v_d = \nabla^2 \tilde{h}_\text{intra}^{A,A} (K)$ & $v_o = \nabla^2 \tilde{h}^{A,B}_\text{intra} (K)$  \\ \midrule
0.79 \text{eV} & 5.23 eV$\cdot$\r{A}       & -1.12 eV$\cdot\text{\r{A}}^2$        & -2.25 eV$\cdot\text{\r{A}}^2$       \\ \bottomrule
\end{tabular} 
\caption{Intralayer parameters for the physical continuum model computed from  the Slater-Koster hopping function. }\label{tab:intra_params}
\end{table}

We can compute the Fourier transform of the interlayer hopping function numerically using the Hankel transform. We report some Fourier coefficients in Table \ref{tab:inter_params}. The radial dependence also means $\mu_j = 0$ in \eqref{eq:rad_ang}.

\begin{table}[h!]
\centering
\begin{tabular}{@{}cccc@{}}
\toprule
 $w_1 = \frac{1}{|\Gamma|}\widehat{h_\text{inter}}(\vec K)$ & $w_2 = \frac{1}{|\Gamma|}\widehat{h_\text{inter}}(2\vec K)$ & $w_3 = \frac{1}{|\Gamma|}\widehat{h_\text{inter}}(\sqrt{7}\vec K)$  & $|w_1'|=\left|\frac{1}{|\Gamma|}\nabla\widehat{h_\text{inter}}(\vec K)\right|$ \\ \midrule
    110.85 meV   & 1.56 meV     & 0.06 meV   & 226.15 meV$\cdot$\r{A}   \\ \bottomrule
\end{tabular} 
\caption{Interlayer parameters $w_1,w_2,w_1'$ for the physical continuum model computed from the Slater-Koster hopping function. Note that although $w_1'$ is the largest parameter here, it appears in the Bistritzer-MacDonald model with a derivative. Since the wave-packet is assumed to vary on the moir\'e scale, its effect on the dynamics is much smaller; see Table \ref{tab:scale}. $w_3$ does not appear in the effective model, but is used in the error analysis. }\label{tab:inter_params}
\end{table}

To compare to existing physics literature, throughout this section we will use the Hamiltonian with physical units recovered from the non-dimensonalized model \eqref{eq:Hfull}. We first change variables $\vec r$ and $\vec t$ back to \begin{equation}
	X := \frac{\vec r}{\eps}\cdot \text{\r{A}}, \quad \mathscr S:= \frac{t \hbar}{\eps \cdot \text{eV}},
\end{equation}
and multiply both sides of \eqref{eq:Hfull} by $\eps$. 
The physical continuum Hamiltonian with order $j$ acts on the wave packet envelope function $\mathfrak f$ through
\begin{equation}
	i \hbar \partial_{\mathscr S} \mathfrak f(X, \mathscr S) = H_{\text{cont},j} \mathfrak f (\vec X, \mathscr S).
\end{equation}
After the change of variable $\mathfrak f$ also depends on the small parameter $\eps$. To ensure the regularity of $\mathfrak f$, we will restrict our numerics to a specific class of functions \eqref{eq:initial_data} with appropriate length scale. 

The first order continuum Hamiltonian is (choosing $\ls = 0$)
\begin{equation}\label{eq:H_cont_1}
   H_{\text{cont},1} := \begin{pmatrix} v \vec{\sigma} \cdot (- i \nabla_{\vec X}) &  \displaystyle  w_1\sum_{n = 1}^3  T_{n} e^{- i \mathfrak{s}_{n} \cdot \vec X} \\  \displaystyle  w_1\sum_{n = 1}^3  T_{n}^\dagger e^{i \mathfrak{s}_{n} \cdot \vec{X}} &  v \vec{\sigma} \cdot (- i \nabla_{\vec{X}}) \end{pmatrix},
\end{equation} 
where the energies $v$ is from Table \ref{tab:intra_params} and $w_1$ is from  Table \ref{tab:inter_params}. They are related to the dimensionless variables through
   $ v = \alpha \cdot \text{eV} \cdot \text{\r{A}}$ and $w_1 = \eps \cdot \lambda_0 \cdot \text{eV}$.
   $\sigma = (\sigma_1, \sigma_2)$ is a vector of Pauli matrices. The interlayer momentum shifts are $\mathfrak s_{1} = \eps \vec s_1 =  \frac{4\pi \sin(\theta/2)}{3a} (0,-1)$, $\mathfrak s_{2} = \rot_{2\pi/3} \mathfrak s_{1}$, and $\mathfrak s_{3} = \rot_{4\pi/3} \mathfrak s_{1}$, and
\begin{equation}
	T_{1} =  \begin{pmatrix}
        1 & 1\\
        1 & 1
    \end{pmatrix},  \quad 
    T_{2} = \begin{pmatrix}
        1 & e^{-i2\pi/3}\\
        e^{i2\pi/3} & 1
    \end{pmatrix}, \quad 
    T_{3} = \begin{pmatrix}
        1 & e^{i2\pi/3}\\
        e^{-i2\pi/3} & 1
    \end{pmatrix}.
\end{equation}
This reduces to the BM Hamiltonian introduced in 
\cite{bistritzer2011moire}. 

The second order terms include the next nearest-neighbor contributions
	\begin{equation}\label{eq:H_cont_NNN}
	 H_{\text{cont},2}^\text{(NNN)} = 
	    \begin{pmatrix}  0 &  \displaystyle w_2 \sum_{n=1}^{3}  T_{n} e^{i 2\mathfrak s_{n} \vec X}\\  
	    \displaystyle w_2 \sum_{n=1}^{3} T_{n}^\dagger e^{-i2\mathfrak s_{n} \vec X} & 0 \end{pmatrix},
	\end{equation}
	with $w_2 = \eps\radNNN(\eps) \cdot \lambda_1 \cdot \text{eV}$.
Also included are the momentum-dependent terms, or the non-local approximation
\begin{equation}\label{eq:H_cont_k-dep}
		H_{\text{cont},2}^{(\nabla, \text{NN})}  = 
    \begin{pmatrix}  0 &  \displaystyle  \sum_{n=1}^{3} \vec w_{n}' \cdot \left(-i\nabla_{\vec X} \right) T_{n} e^{-i\mathfrak s_{n} \vec X} \\  
    \displaystyle  \sum_{n=1}^3  \vec w_{n}' \cdot \left(-i\nabla_{\vec X} \right) T_{n}^\dagger e^{i\mathfrak s_{n} \vec X} & 0  \end{pmatrix},
	\end{equation}
where $w_{1}' = \frac{1}{|\Gamma|} \nabla \widehat{h_\text{inter}}(\vec K)$,  $w_{2}' = \rot_{2\pi/3} w_{1}'$, $w_{3}' = \rot_{4\pi/3} w_{1}'$. Their magnitudes are $|w_n'| = \eps \radpNN(\eps) \cdot \lambda_0 \cdot \text{eV} \cdot \text{\r A}$. Lastly, we include the second order corrections to the Dirac cone
\begin{equation}\label{eq:H_cont_intra_2}
		H_{\text{cont},2}^{(2)}= \frac{1}{2}
    \begin{pmatrix} 
    \begin{aligned}
    {\left[v_d (-i\nabla_{\vec X})^2 \right] Id} +  v_o \vec{\sigma} \cdot (-\partial^2_{X_1}+\partial^2_{X_2}, 2\partial_{X_1} \partial_{X_2}) \\
    + i\sin(\theta/2)  \sigma_3 \cdot v \vec{\sigma} \cdot (- i \nabla_{\vec X})
    \end{aligned}	
&  0 \\  0 & \begin{aligned}	{v_d (-i\nabla_x)^2 Id}  + v_o \vec{\sigma} \cdot (-\partial^2_{X_1}+\partial^2_{X_2}, 2\partial_{X_1} \partial_{X_2}) \\
+ i\sin(-\theta/2)  \sigma_3 \cdot v \vec{\sigma} \cdot (- i \nabla_{\vec X})
 \end{aligned}
     \end{pmatrix},
	\end{equation}
where $v_d = \alpha_d \cdot  \text{eV} \cdot \text{\r{A}}^2$ and $v_o = \alpha_o \cdot  \text{eV} \cdot \text{\r{A}}^2$, and their values are in Table \ref{tab:intra_params}.

Summing up these three terms gives the complete second order continuum model
\begin{equation}\label{eq:H_cont_2}
	  H_{\text{cont},2} := \begin{pmatrix}  
    \begin{aligned} v \vec{\sigma}_{ \theta/2} \cdot  (- i \nabla_{\vec X}) + {\frac{1}{2}v_d (-i\nabla_{\vec X})^2  Id}   \\ +  \frac{1}{2}v_o \vec{\sigma} \cdot (-\partial^2_{X_1}+\partial^2_{X_2}, 2\partial_{X_1} \partial_{X_2})
    \end{aligned}
	&   \begin{aligned}   w_1 \sum_{n=1}^{3}  T_{n} e^{-i\mathfrak s_{n} \cdot \vec X}  + w_2 \sum_{n=1}^{3}  T_{n} e^{i2\mathfrak s_{n} \cdot \vec X} \\+ \sum_{n=1}^{3} \vec w_{n}' \cdot \left(-i\nabla_{\vec X} \right) T_{n} e^{-i\mathfrak s_{n}\cdot \vec X} \\ \end{aligned}
 \\
       \begin{aligned}  w_1 \sum_{n=1}^{3}  T_{n}^\dagger e^{i\mathfrak s_{n} \cdot \vec X}  + w_2 \sum_{n=1}^{3}  T_{n}^\dagger e^{-i2\mathfrak s_{n} \cdot \vec X} \\+ \sum_{n=1}^{3} \vec w_{n}' \cdot \left(-i\nabla_{\vec X} \right) T_{n}^\dagger e^{i\mathfrak s_{n}\cdot \vec X} \end{aligned}
       & \begin{aligned} v \vec{\sigma}_{-\theta/2} \cdot  (- i \nabla_{\vec X}) + {\frac{1}{2}v_d (-i\nabla_{\vec X})^2 Id} \\  + \frac{1}{2}v_0 \vec{\sigma} \cdot (-\partial^2_{X_1}+\partial^2_{X_2}, 2\partial_{X_1} \partial_{X_2})
       \end{aligned}
     \end{pmatrix},
\end{equation}
with $\sigma_{\theta/2}:= [I - i \sin(\theta/2)\sigma_3] \cdot \sigma$ being the Dirac operator with a linearized rotation.

The length scale of the TBG moir\'e reciprocal unit cell is approximately  $ |\Delta \vec K| \approx 0.033 \text{\r{A}}^{-1}$. Suppose our choice of $\eps$ is on the same length scale as $ |\Delta \vec K|$, which is equivalent to the wave packet initial data spanning at least one moir\'e unit cell, we can estimate the energy-scale of the Hamiltonian. Table \ref{tab:scale} shows that with the physical parameters, there is a natural separation of scale for the terms in \eqref{eq:H_cont_1} and \eqref{eq:H_cont_2}. In particular, the first order model \eqref{eq:H_cont_1} include only $O(\eps)$ terms, which have energy of around 100 meV. The second order model  \eqref{eq:H_cont_2} includes $O(\eps^2)$ terms, which have energy less than 10 meV. Based on these estimations, there exists a range of $\eps$ we can choose to ensure these energy scales are separated. 

\begin{table}[h!]
\centering
\begin{tabular}{@{}ccccccc@{}}
\toprule
& \multicolumn{2}{c}{First order}  & \multicolumn{4}{c}{Second order}                                                         
\\ \cmidrule(l){2-3} \cmidrule(l){4-7} 
& Term   & Energy & Term   & Energy & Term & Energy \\ \midrule
Intralayer & $v |\Delta \vec K|$ & 170 meV                               & $v_d |\Delta \vec K|^2$       & 1 meV  & $v_o |\Delta \vec K|^2$ & 2 meV  \\
Interlayer  & $w_{1}$                     & 111 meV & $|\vec w'_{1} \cdot \Delta \vec K|$ & 7 meV  & $w_{2}$                       & 2 meV  \\ \bottomrule
\end{tabular}
\caption{Separation of energy scale using physical values. }\label{tab:scale}
\end{table}


To numerically investigate the approximation error of continuum models, we first need to generate reference solutions from the discrete tight-binding model. We generate the Hamiltonian by replacing $h$ in \eqref{eq:tight-binding}
with \eqref{eq:slater-koster}. The lattices of TBG are incommensurate and infinite, but the dynamics of the tight-binding model can be approximated by computations on a finite domain. Suppose we restrict all the interaction to a space contained in a ball with radius $R$, the truncation error decays exponentially with respect to $R$  \cite{kong2023modeling}.

We choose a sufficiently large $R$ so that the truncation error is small. In this way, we are able to study the non-dimensionalized relative approximation error $\eta_j$ of $j$-th order continuum models on the finite domain 
\begin{equation}\label{eq:relerror}
	\eta_j := \frac{\left\|\Psi^j_{\text{cont}} -\Psi_{\text{TB},R} \right\|_{\mathcal H_R}}{\left\|\Psi_{\text{TB},R} \right\|_{\mathcal H_R}}.
\end{equation}
For all numerical experiments, we use the truncated tight-binding model with $R = 86.60 \text{ \r{A}}$.  
 
 We can utilize the band structure of the BM Hamiltonian to excite wave-packets with wave numbers concentrated in momentum space on any band. Choose a point $\vec k $ in the moir\'e Brillouin zone (see Figure \ref{fig:bands})
 and band index $n$, the energy on the $n$-th band is the $n$-th eigenvalue $E_n(\vec k)$, with the corresponding eigenfunction $\Phi_n(X; \vec k)$. Then the wave packet initial data concentrated on $k_i$ on the $n$-th band is described by
\begin{equation}\label{eq:initial_data}
    \mathfrak f_0(\vec X) = c\cdot \Phi_n(\vec X; \vec k)\cdot G(\vec X), \quad G(\vec X) := e^{-\frac{|\vec X|^2}{2\sigma_r^2}}, 
\end{equation}
where $c$ is the overall normalization constant, and $G$ is the two-dimensional Gaussian function. We can control the wave-packet envelope length scale by setting $\sigma_r = \eps^{-1}$.
 The group velocity for the wave-packet envelope is $\nabla_{\vec k} E_n(\vec k)$.

We implement the first and second order continuum model, and compute the approximation error by comparing to the solution of tight-binding models. The approximation error for the second order model $\eta_2$ is much less than that of the first order model (see Figure \ref{fig:error_terms}). This result validates that the inclusion of higher order corrections captures more information about the tight-binding model.

\begin{figure}[h!]
    \centering
    \includegraphics[width=.3\textwidth]{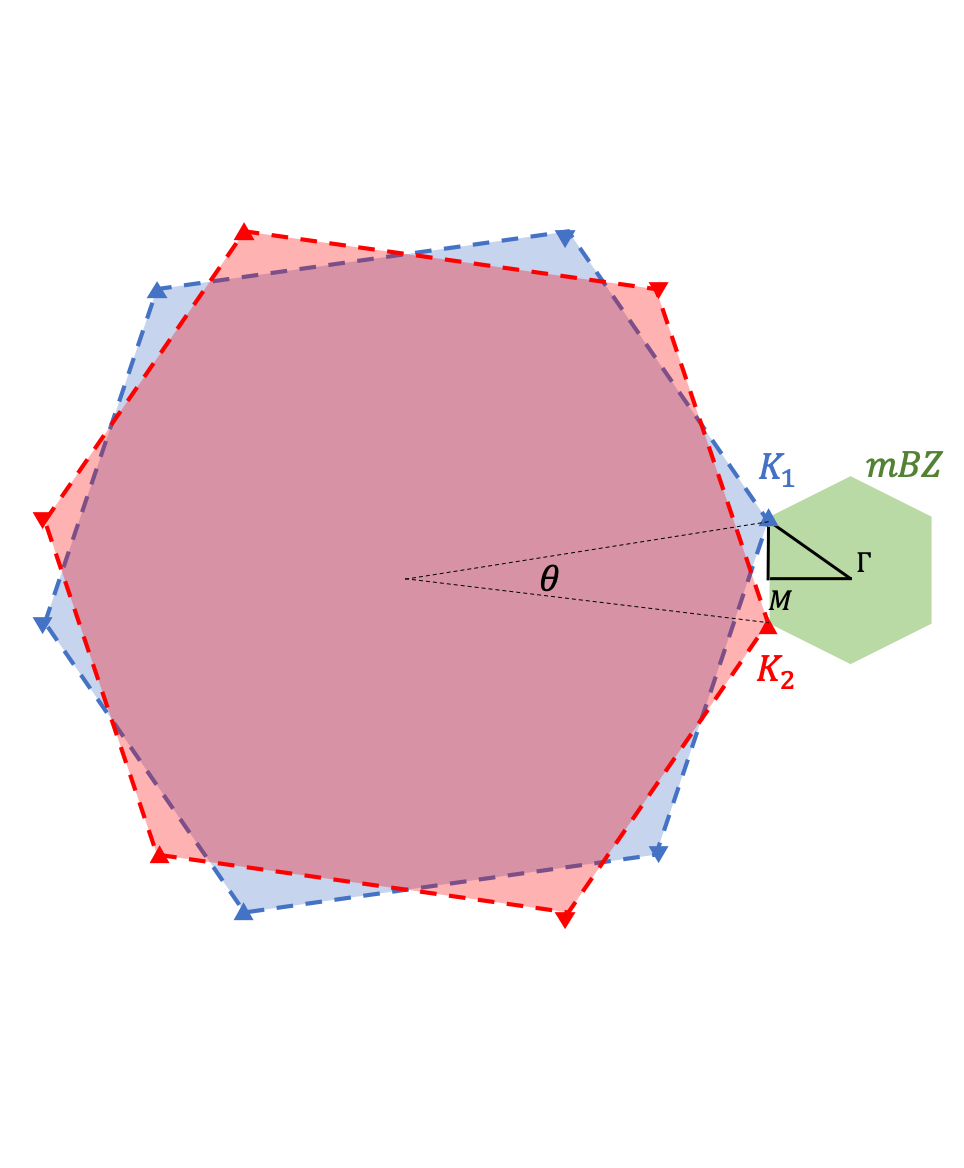}
    \includegraphics[width=.15\textwidth]{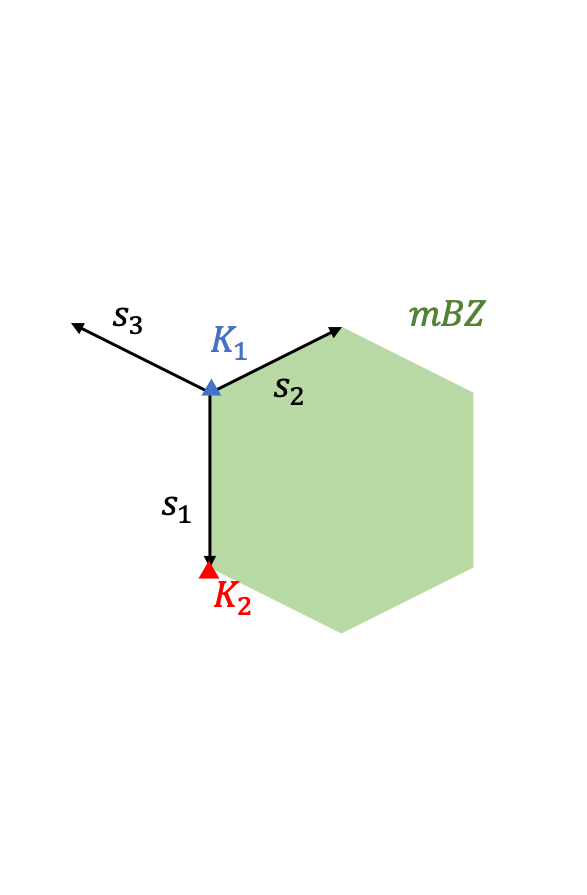}
    \hspace{2cm}
    \includegraphics[width=.33\textwidth]{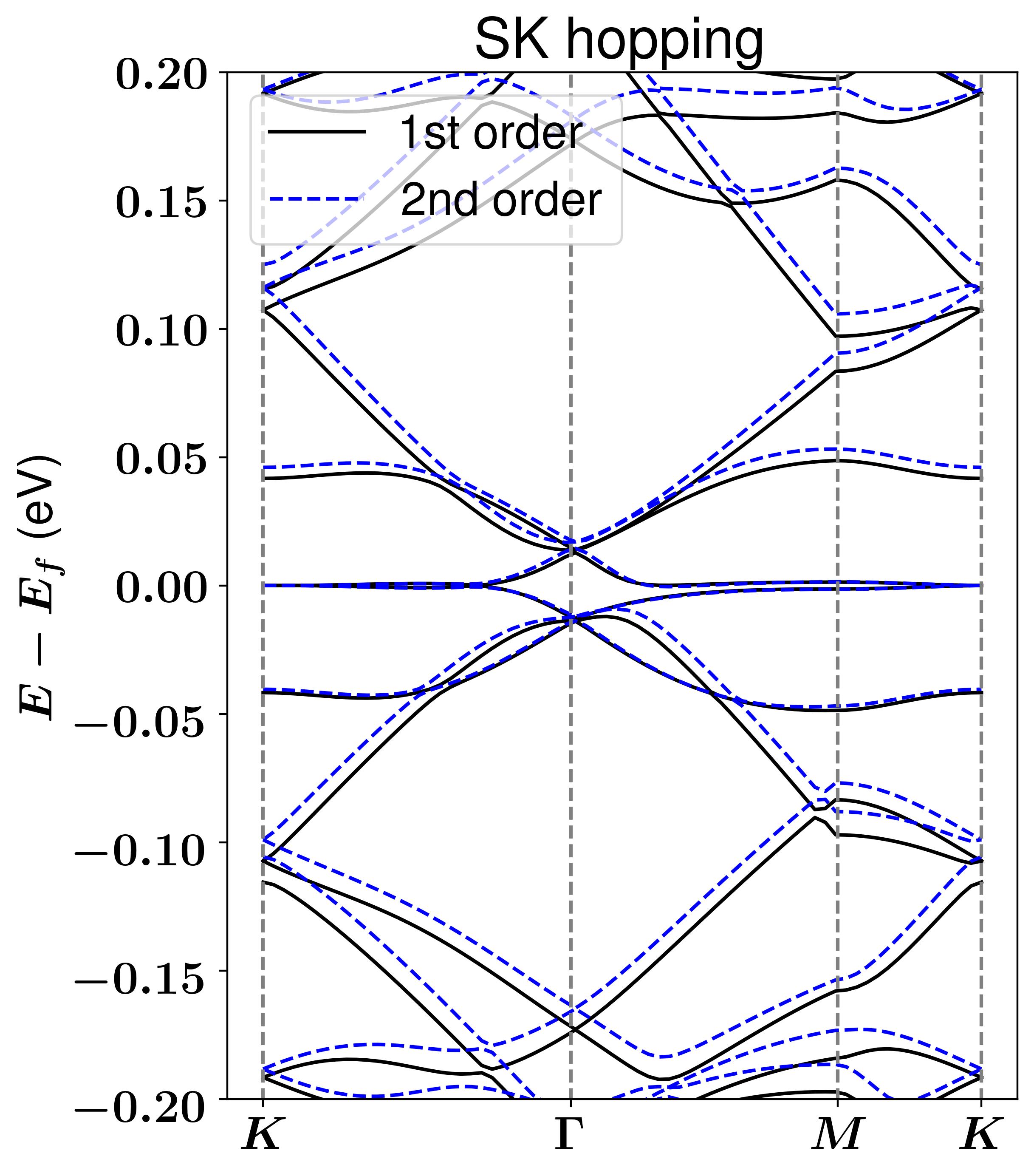}
    \caption{Left: Illustration of monolayer Brillouin zone (BZ) in red and blue for twisted bilayer graphene, high symmetry point $\vec K_1$ and $\vec K_2$, moir\'e Brillouin zone (mBZ) in green and moir\'e high symmetry points $\vec K := \vec K_1$, $\Gamma$ and $M$. 
    Right: The band structure showing the dispersion relation along a specific path in mBZ for first and second order BM continuum models, with Slater-Koster hopping function. 
    The second order corrections breaks the particle-hole symmetry of the first order model.}
    \label{fig:bands}
\end{figure}

\begin{figure}[h!]
    \centering
    \includegraphics[width=.45\textwidth]{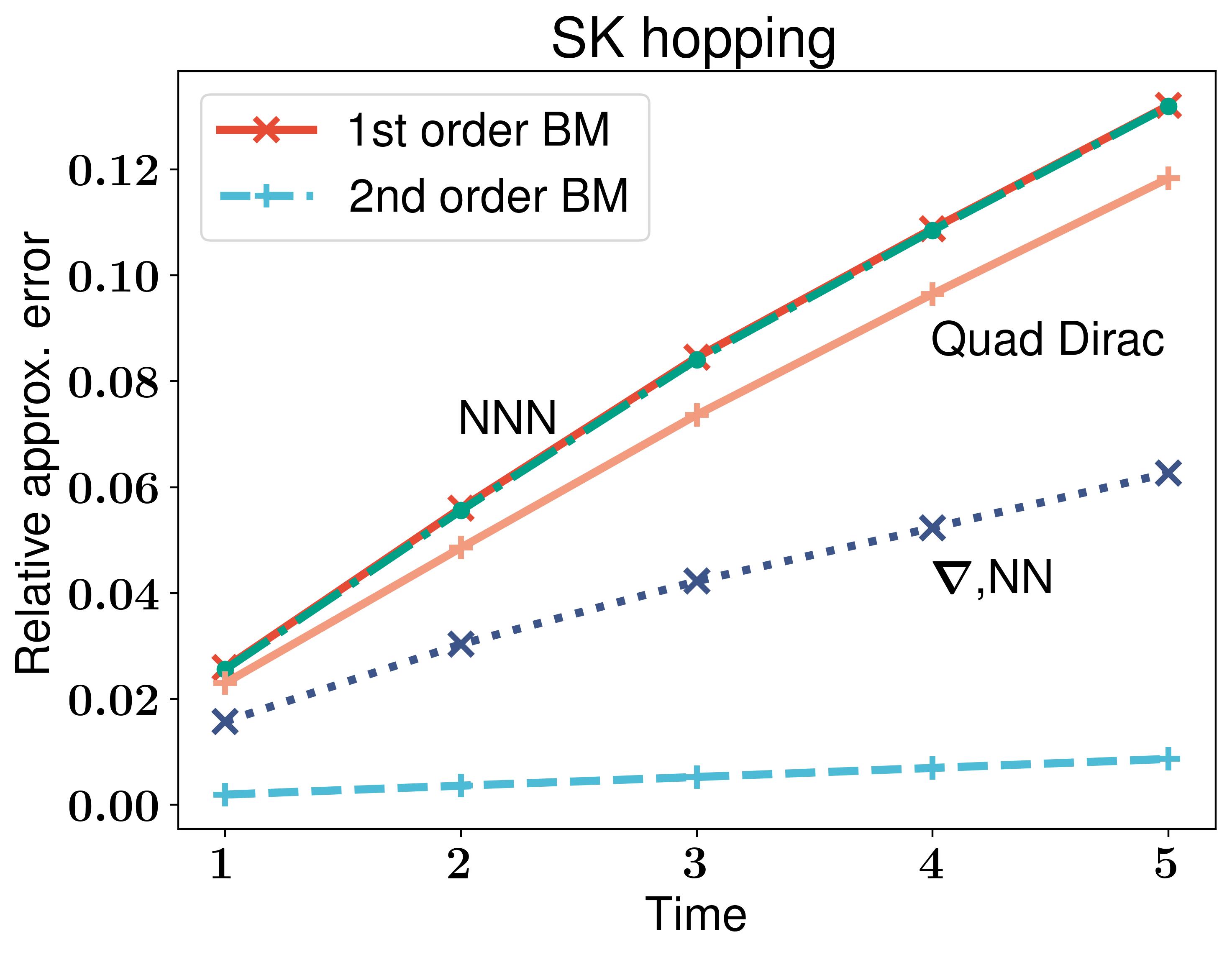}
    \caption{Relative approximation error \eqref{eq:relerror} of wave packet dynamics of first and second order BM model, when compared to the tight-binding reference dynamics. The inclusion of the second order correction terms greatly reduces the approximation error. Here $\eps=0.05$. When looking at improvements of individual terms, the non-local  ($\nabla$,NN) term \eqref{eq:H_cont_k-dep} has the smallest error, followed by quadratic Dirac term \eqref{eq:H_cont_intra_2}. The next nearest-neighbor (NNN) term \eqref{eq:H_cont_NNN} has only negligible improvement. 
    }
    \label{fig:error_terms}
\end{figure}

We also probe which term in the second order model has the largest contribution. Recall that the second order model includes three additional improvements as the first order model \eqref{eq:H_cont_NNN}, \eqref{eq:H_cont_k-dep} and \eqref{eq:H_cont_intra_2}.
All of these terms are of the order of between $O(\eps)$ and $O(\eps^2)$.  We can add one of these terms individually to the first order continuum model, and calculate the approximation error for that model. In this way, we are probing the magnitude of each individual term to see what is their contribution in the second order model. Figure \ref{fig:error_terms} shows that the momentum-dependent terms introduce the largest improvement, followed by the quadratic terms in the intralayer Dirac operator. The next nearest-neighbor terms does not have visible contributions for improvements, as the approximation error for the model with these corrections do not differ significantly from the original one.

We present time snippets of the wave packet dynamics for all three models with different initial data. All initial data are wave packets generated on selected bands using \eqref{eq:initial_data}, and the band structure is generated from the first order continuum model. For wave packets concentrated on the remote band, all models generate similar group velocity (see Figure \ref{fig:wave_packet_group_velocity}). Although the tight-binding model is aperiodic, we can still observe an approximate group velocity that originates from the effective continuum model which is periodic with respect to the moir\'e lattices. 

\begin{figure}[h!]
    \centering
    \includegraphics[width=.65\textwidth]{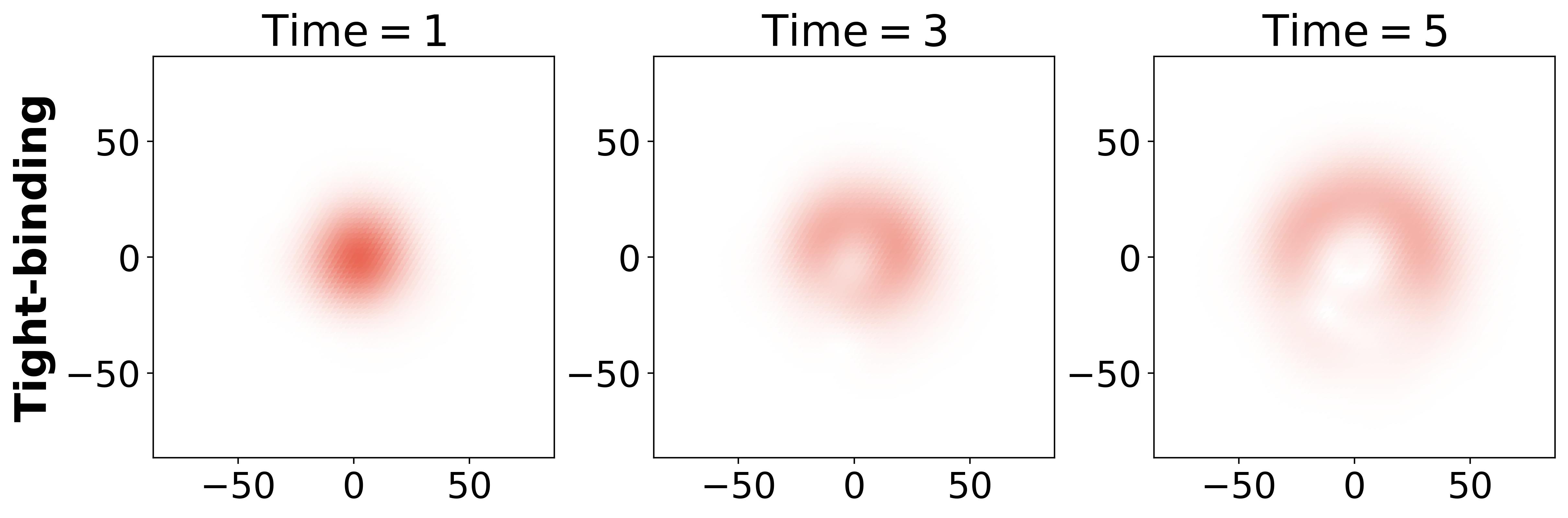}
    \includegraphics[width=.65\textwidth]{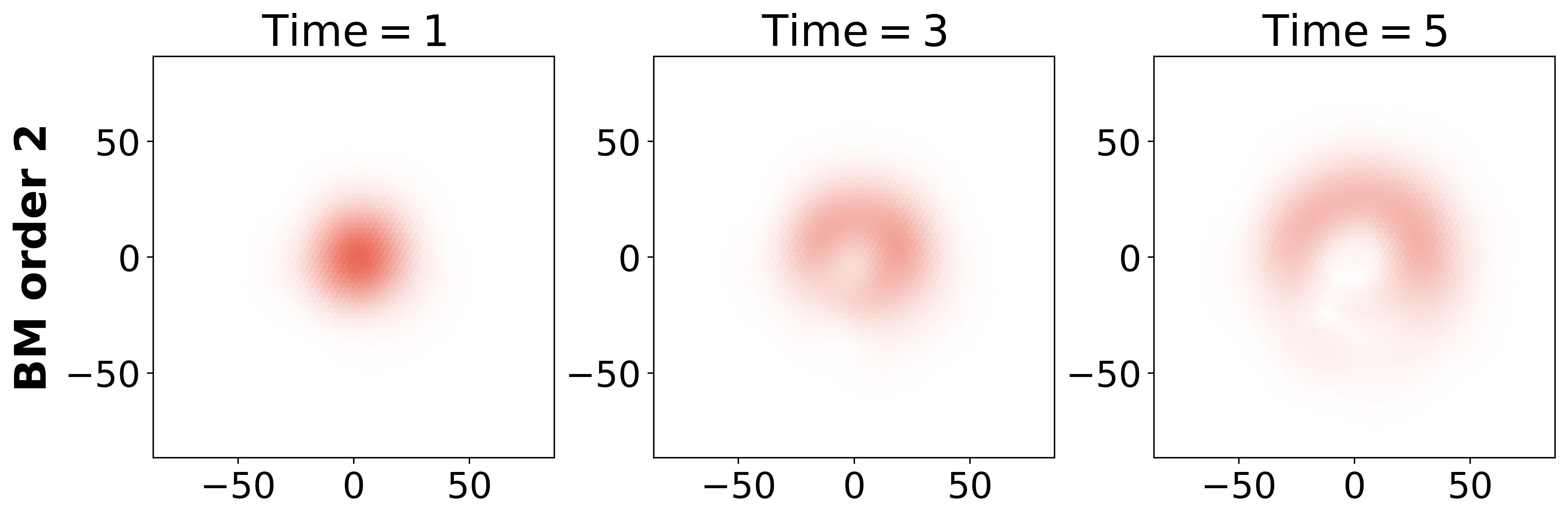}
    \includegraphics[width=.65\textwidth]{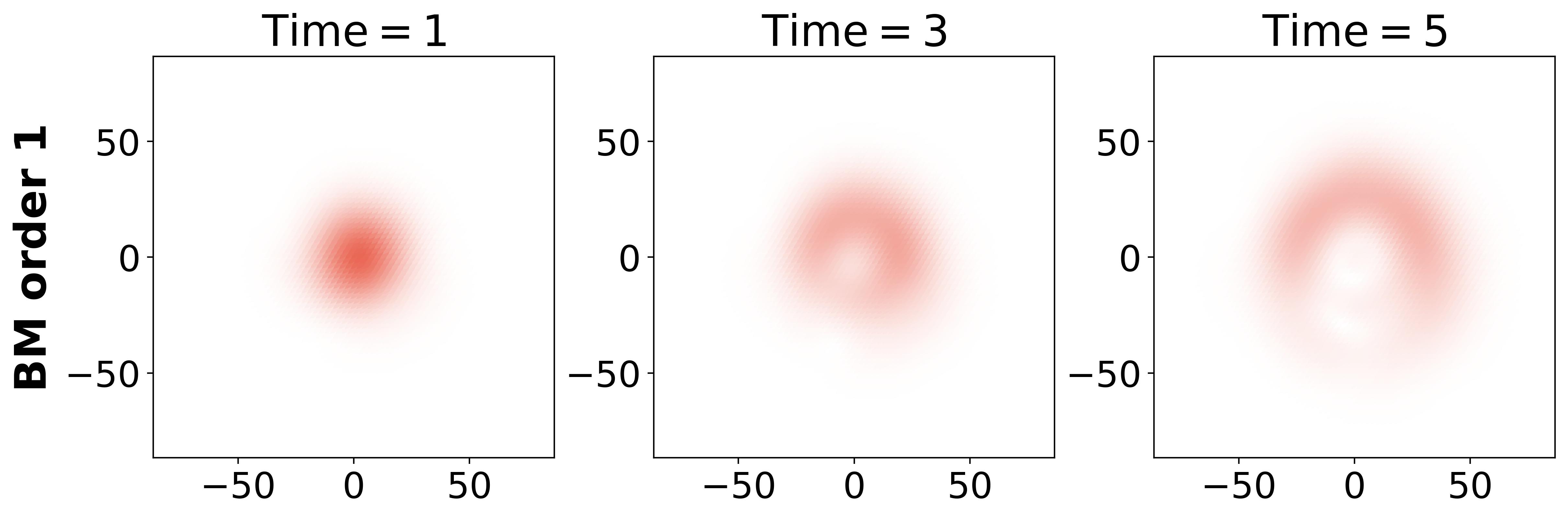}
    \caption{The modulus of the wave-function for the tight-binding model and the first and second order continuum models. The initial data is concentrated on a remote band. All models predict similar group velocities. Only layer 1 is presented, as the two layers have similar behaviour. The axes have units \r{A}, and one unit of time is $\hbar \cdot \mathrm{eV}^{-1} \approx 6.6 \times 10^{-16} \text{ s}$.}
    \label{fig:wave_packet_group_velocity}
\end{figure}

When the initial data is concentrated on the flat band at $\Gamma_m$, the wave packet propagates with near 0 group velocity (see Figure \ref{fig:wave_packet_spiral}). For each layer, the reference tight-binding dynamics show a spiral pattern with the orientation same as the rotation of that layer. The spiral is recovered in the second order continuum model, but lost in the first order model. Further investigation shows that within the three improvements from the first to second order, the spiral is generated from the momentum dependent interlayer hopping terms \eqref{eq:H_cont_k-dep}.

Lastly,  we want to probe the relation of relative approximation error with respect to the small scale $\eps$. The continuum models \eqref{eq:H_cont_1}, \eqref{eq:H_cont_2} depend on the twist angle $\theta$ and interlayer distance $\ell$. We propose they can be parametrized by $\eps$ through
\begin{equation}
	\theta(\eps) = \beta \eps, \quad \ell(\eps) = \gamma |\log\eps|.
\end{equation}
The choice of $\eps,\beta$ and $\gamma$ can be arbitrary. We choose $\eps_0 = 0.05$, and fix constants $\beta$ and $\gamma$ such that $\theta(\eps_0) = 1.05^\circ$ and $\ell(\eps_0) = 3.5\text{\r A}$ to match the estimated physical values of these parameters in twisted bilayer graphene at the magic angle. The continuum Hamiltonian at $\eps_0$ now has hopping energies as reported in Table \ref{tab:inter_params}. 
We then vary $\eps$, and generate the Hamiltonian according to our choice of $\theta(\eps)$ and $\ell(\eps)$. At a fixed small time, we plot the approximation error for both the first and second order models as a function of $\eps$ in Figure \ref{fig:error_epsilon}. The linear regression results show the error is approximately $O(\eps^{1.98})$ for the first order model, and $O(\eps^{2.98})$ for the second order model. The results are consistent with \cite[Theorem 3.1]{watson2023bistritzer} and Theorem \ref{thm:main}, which states the approximation error of the first order continuum model is $O(\eps^{1+\sdpt_*})$ and that of the second order continuum model is $O(\eps^{2 + \sdpt_-})$, where $0<\sdpt_*,\sdpt_-<1$, $\sdpt_*$ can be taken arbitrarily close to 1, and $\sdpt_-<\sdpt$ where $\sdpt$ depends on the decay of the Slater-Koster hopping function. 

Despite the apparent $\eps^3$ convergence of the second order model in Figure \ref{fig:error_epsilon}, we would not expect this to be the asymptotic behavior as $\eps \to 0$. Assuming the interlayer separation satisfies $\ell \sim |\log\eps|$, the spectral decay of the Slater-Koster interlayer hopping function in \eqref{eq:SK_decay} 
 would imply that $|\hathperp (k; \eps)| \le C \eps^{|k|/|K|}$ for all $\eps$ and $k$, with the factor $|K|^{-1}$ in the exponent explained by the requirement that $\hathperp (K; \eps) = O(\eps)$. The last line of \eqref{eq:4bounds_full} would then suggest that $\sdpt = \sqrt{7}-2$, 
implying $O(\eps^{\sqrt{7}})$ convergence of the second order model over timescales of $O(1)$.

Such discrepancies are due to the fact that the coefficient of the $O(\eps^{\sqrt{7}})$ terms is sufficiently small that these terms are negligible compared to the $O(\eps^3)$ terms for the range of $\eps$ values we consider in Figure \ref{fig:error_epsilon}.
We assume these error terms can be written as 
$	\mathrm{error}(\eps) = F_1 \eps^{\sqrt 7} + F_2 \eps^3$. The $\eps^{\sqrt{7}}$ error term arises from the fact that in the effective Hamiltonian we cutoff momentum space hopping at the NNN term. The coefficient of the $\eps^{\sqrt{7}}$ error term is thus determined by the magnitude of hopping to the third shell (NNNN) term, see Figure \ref{fig:NNN}, defined by the physical parameter $w_3$, see Table \ref{tab:inter_params}. More precisely, following the prescription that the magnitude of our $\eps$-dependent hopping function agrees with the physical values in Table \ref{tab:inter_params} when $\eps = \eps_0 = 0.05$, we have $F_1 = w_3 / (0.05)^{\sqrt{7}} \approx 0.166$ eV. We can compare this to a typical intralayer term which will appear in the error at $\eps^3$ order. The values of the first-order and second-order intralayer terms are consistent with coefficients $F_2$ on the order of $2$ eV, since we have energies $F_2 \eps_0 = 100 \meV$  and $F_2 \eps_0^2 = 5 \meV$, see Table \ref{tab:scale}. The size of a typical $\eps^3$ term is thus $F_2 \eps_0^3 = 2.5\times 10^{-4}\eV$. Since $F_1\eps_0^{\sqrt{7}} = 6\times 10^{-5}\eV$, we see that the $O(\eps^3)$ error dominates in the interval $[0.05, 0.175]$ presented in Figure \ref{fig:error_epsilon}. In fact, these errors are comparable only when $F_1 \eps^{\sqrt 7} = F_2 \eps^3$, for $\eps$ as low as $9\times 10^{-4}$. Note that we express error in this paragraph in terms of energies in order to compare with the physical energies in Tables \ref{tab:inter_params} and \ref{tab:scale}. The dimensionless error could be computed by multiplying by the typical timescale and dividing by $\hbar$.

\begin{figure}[h!]
    \centering
    \includegraphics[width=.45\textwidth]{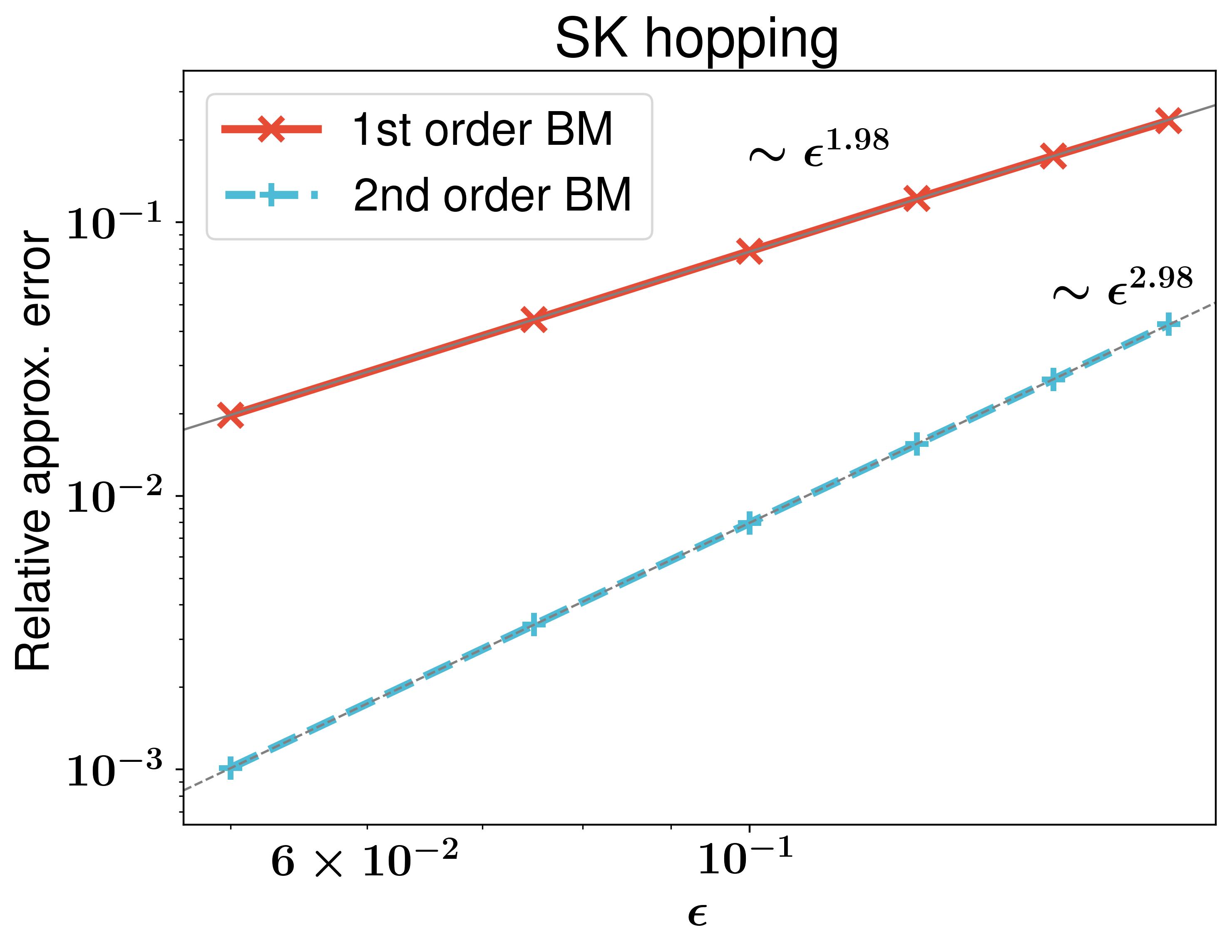}
\caption{Relative approximation error \eqref{eq:relerror} between BM and tight-binding dynamics for $\eps \in [0.05, 0.175].$ For a fixed $t=1$, the approximation error is approximately $O(\eps^{1.98})$ for the first order BM model, and $O(\eps^{2.98})$ for the second order BM model. Note that our estimates suggest that for sufficiently small $\eps$, the error should be $O(\eps^{\sqrt{7}})$ rather than $O(\eps^3)$ as observed here. We attribute the inconsistency to the smallness of the coefficient of the $O(\eps^{\sqrt{7}})$ terms in the expansion, which make these terms negligible compared to the $O(\eps^3)$ error terms for the range of $\eps$ values considered here. For more details, see the discussion in the final paragraphs of Section \ref{sec:numerical}. Simulations for smaller values of $\eps$ are omitted due to their prohibitive memory requirements.} 
    \label{fig:error_epsilon}
\end{figure}

\section{Conclusion}\label{sec:conclusion}
This paper rigorously establishes the validity of a second-order PDE model for TBG, \srq{extending the first-order result from \cite{watson2023bistritzer} to obtain a higher order of accuracy}. We show that in the limit of small twist angle and weakly interacting layers, the second-order model accurately captures the tight-binding dynamics of wave-packets whose Bloch transforms are localized to the monolayer Dirac point. We use a systematic multiple-scales expansion to construct a series of Sch\"odinger equations that are independent of our asymptotic parameter, and prove that each of these Schr\"odinger equations is (generically) well-posed. 

Our analysis is fairly general, \srq{as it applies to a larger class of hopping functions than those considered in \cite{watson2023bistritzer}.}
But it does not yet handle some models of TBG that can be found in the physics literature. In particular, some interlayer hopping functions \srq{(e.g. Slater-Koster \cite{slater1954simplified})} might fail to satisfy the specific bounds we impose on their Fourier transforms. Still, our numerical simulations demonstrate that the second-order PDE remains an accurate effective model for TBG dynamics even outside the requirements of our theory.

Finally, we verify the symmetries of the tight-binding and continuum models, and comment on the correspondence between the two. Most of these symmetries require realistic constraints on the hopping functions which are not needed for our main results. We find that the higher-order model 
breaks the ``accidental'' particle-hole symmetry of the first-order BM model. 
This is expected to be important in modeling the many-body physics of twisted bilayer graphene \cite{StubbsGroundState,FaulstichIBM}, which typically starts from adding electron-electron Coulomb interaction terms to moir\'e-scale continuum models like the models we justify here. 


The wavepacket framework from this paper could also be used to derive effective continuum models for non-TBG moir\'e materials \cite{MargetisTrilayer,zhu20trilayer,zhu2019moir}. At the expense of introducing more terms in the multiscale expansion, these effective models could in general achieve arbitrarily high orders of accuracy. 
Ongoing work 
includes a similar asymptotic analysis for a model of TBG that takes atomic relaxation into account. 
Any effective model capturing relaxation effects would involve more terms in the multiscale expansion, due to 
a slower decay rate of the Fourier transform of the interlayer hopping function; see \cite{kang2023pseudomagnetic, relaxedblg23}. 

\section{Acknowledgments}

ML's and TK's research was partially supported by Simons Targeted Grant Award No. 896630. AW's, TK's, and ML's research was also partially supported in part by NSF DMREF Award No. 1922165.  AW's research was supported in part by grant NSF DMS-2406981. 

ML's research was supported in part by grant NSF PHY-2309135 to the Kavli Institute for Theoretical Physics (KITP). ML and SQ also thank the Flatiron Institute for hospitality while (a portion of) this research was carried out. The Flatiron Institute is a division of the Simons Foundation.

\bibliographystyle{siam}
\bibliography{refs}

\appendix
\numberwithin{equation}{section}
\section{Proofs of main results}\label{sec:proofs}
\subsection{Theorem \ref{thm:main2}}\label{subsec:proof_thm_main2}
We will now prove Theorem \ref{thm:main2}. For this, we need the following four lemmas. Lemma  \ref{lemma:res2} is a general result which controls the norm of a solution to the forced Schr\"odinger equation. Lemma \ref{lemma:Sobolev_ms2} establishes that all Sobolev norms of each $\tff{j}$ are uniformly bounded in any bounded and fixed time interval. Lemma \ref{lemma:ae2} then uses this uniform boundedness to show that the function $\tf$ almost satisfies a continuum Schr\"odinger equation. Finally, Lemma \ref{lemma:Hs_to_l22} is a classical result that controls the discrete $\ell^2$ norm of a function by its (continuum) Sobolev norms.
\begin{lemma}\label{lemma:res2}
    Let $\genH$ be a symmetric linear operator on a Hilbert space $\gencH$.
    Suppose $\varphi = \varphi (\cdot, t)$ satisfies $i\partial_t \varphi = \genH \varphi + \resgen$ with zero initial condition $\varphi (\cdot, 0) \equiv 0$, where the residual $\resgen$ belongs to $\gencH$ for all $0 \le t \le t_0$, for some $t_0 > 0$.
    Then $\norm{\varphi (\cdot, t)}_{\gencH} \le \int_0^t \norm{\rho (\cdot, s)}_{\gencH} {\rm d} s$ for all $0 \le t\le t_0$.
\end{lemma}
\begin{proof}
    For conciseness, suppress the subscript $\gencH$ in the norms and inner-products below. We see that $$2i \norm{\varphi} \partial_t \norm{\varphi} =i\partial_t \norm{\varphi}^2 = (\varphi, i\partial_t \varphi) - (i\partial_t \varphi, \varphi) = (\varphi,\genH \varphi)+(\varphi,\resgen)-(\genH \varphi,\varphi)-(\resgen,\varphi) = 2i\Im (\varphi,\resgen),$$
    where we have used the self-adjointness of $\genH$ to justify the last equality. It follows that $\norm {\varphi} \partial_t \norm{\varphi} \le \norm{\varphi} \norm{\resgen}$, and thus $\norm{\varphi (\cdot, t)} \le \int_0^t \norm{\rho (\cdot, s)}{\rm d} s$ as desired.
\end{proof}
\begin{lemma}\label{lemma:Sobolev_ms2}
    Under the assumptions of Theorem \ref{thm:main2}, for any $N \ge 0$,
    \begin{align}\label{eq:Sobolev2}
    \begin{split}
        \norm{\tff{1} (\cdot, t)}_{H^{N}} \le C \norm{f_0}_{H^N}, &\qquad \norm{\tff{{\rm NNN}} (\cdot, t)}_{H^{N}}\le C t \norm{f_0}_{H^N},\\
        \norm{\tff{\nabla, {\rm NN}} (\cdot, t)}_{H^{N}}\le C t \norm{f_0}_{H^{N+1}},&\qquad \norm{\tff{2} (\cdot, t)}_{H^{N}}\le C t \norm{f_0}_{H^{N+2}},
    \end{split}
    \end{align}
    uniformly in $t \ge 0$.
\end{lemma}
\begin{proof}
    By \eqref{eq:HBM_def}, the assumption that $\alpha \ne 0$ implies that $\HBM$ is an elliptic first-order differential operator. This means that for any $N \ge 0$, there exist positive constants $\mu$, $c$ and $C$ such that
    \begin{align}\label{eq:elliptic_M2}
        c \norm{u}_{H^N} \le \norm{\left(\left(\HBM\right)^2 + \mu\right)^{N/2} u}_{L^2} \le C\norm{u}_{H^N}
    \end{align}
    uniformly in $u \in H^N$.
    It follows that
    \begin{align*}
        \norm{\tff{1} (\cdot, t)}_{H^{N}} &= \norm{e^{-it\HBM} f_0}_{H^{N}} \le c^{-1} \norm{((\HBM)^2 + \mu)^{N/2} e^{-it\HBM} f_0}_{L^2} = c^{-1} \norm{((\HBM)^2 + \mu)^{N/2} f_0}_{L^2}\\
        &= c^{-1} C \norm{f_0}_{H^{N}}
    \end{align*}
    is bounded uniformly in $t \ge 0$.
    Moreover, since $\HBM$ is symmetric, we have that for all $j \ne 1$
    \begin{align*}
        i \partial_t \norm{((\HBM)^2 + \mu)^{N/2} \tff{j}}^2_{L^2} &= 2 i \Im \left( ((\HBM)^2 + \mu)^{N/2} \tff{j}, ((\HBM)^2 + \mu)^{N/2} i \partial_t \tff{j} \right)\\
        &= 2 i \Im \left( ((\HBM)^2 + \mu)^{N/2} \tff{j}, ((\HBM)^2 + \mu)^{N/2}\HBMvar{j}\tff{1}\right).
    \end{align*}
    By the regularity of the coefficients of the $\HBMvar{j}$, we have
    \begin{align}\label{eq:reg_Hj}
        \norm{\HBMNNN u}_{H^N} \le C \norm{u}_{H^{N}}, \qquad
        \norm{\HBMgNN u}_{H^N} \le C \norm{u}_{H^{N+1}}, \qquad
        \norm{\HBMt u}_{H^N} \le C \norm{u}_{H^{N+2}}
    \end{align}
    uniformly in $u$, which combined with the upper bound in \eqref{eq:elliptic_M2} implies
    \begin{align*}
        \partial_t \norm{((\HBM)^2 + \mu)^{N/2} \tff{j}}^2_{L^2} \le C \norm{((\HBM)^2 + \mu)^{N/2} \tff{j}}_{L^2} \norm{\tff{1}}_{H^{N_j}},
    \end{align*}
    where $N_{\rm NNN} = N$, $N_{\nabla, \rm NN} = N+1$ and $N_2 = N+2$.
    It follows that
    \begin{align*}
        \norm{((\HBM)^2 + \mu)^{N/2} \tff{j} (\cdot, t)}_{L^2} \le C \int_0^t \norm{\tff{1} (\cdot, t')}_{H^{N_j}}{\rm d} t' \le C t\norm{f_0}_{H^{N_j}}
    \end{align*}
    uniformly in $t \ge 0$. By the lower bound in \eqref{eq:elliptic_M2}, we conclude that
    \begin{align*}
        \norm{\tff{j}(\cdot, t)}_{H^N} \le C t\norm{f_0}_{H^{N_j}}
    \end{align*}
    and the proof is complete.
\end{proof}

\begin{lemma}\label{lemma:ae2}
    Under the assumptions of Theorem \ref{thm:main2}, define $\Hfull := \HBM + \radNNN (\eps) \HBMNNN + \radpNN (\eps) \HBMgNN + \eps \HBMt$ and $\res := (i \partial_t - \Hfull) \tf$. Then for any $N \ge 0$,
    \begin{align*}
        \norm{\res (\cdot, t)}_{H^N} \le C t \left( \eps^{1+\sdpt} \norm{f_0}_{H^{N+2}} + \eps^{\frac{3+\sdpt}{2}}\norm{f_0}_{H^{N+3}} + \eps^2 \norm{f_0}_{H^{N+4}}\right)
    \end{align*}
    uniformly in $t \ge 0$.
\end{lemma}
\begin{proof}
    It follows from \eqref{eq:it} that
    \begin{align*}
        \res = \left(\radNNN (\eps) \HBMNNN + \radpNN (\eps) \HBMgNN + \eps \HBMt \right) \left( \radNNN (\eps) \tff{{\rm NNN}} + \radpNN (\eps) \tff{\nabla, {\rm NN}} + \eps \tff{2}\right).
    \end{align*}
    By \eqref{eq:Sobolev2} and \eqref{eq:reg_Hj}, this means
    \begin{align*}
        \norm{\res (\cdot, t)}_{H^N} \le C t \Big( |\radNNN (\eps)|^2 \norm{f_0}_{H^N} +|\radNNN (\eps)\radpNN (\eps)| \norm{f_0}_{H^{N+1}}&+ (|\radpNN (\eps)|^2 + \eps |\radNNN (\eps)|)\norm{f_0}_{H^{N+2}}\\
        &+ \eps |\radpNN (\eps)|\norm{f_0}_{H^{N+3}} + \eps^2 \norm{f_0}_{H^{N+4}}\Big).
    \end{align*}
    Using \eqref{eq:xi_zeta_bd} and the fact that $\frac{1 + \sdpt}{2} < \frac{4+2\sdpt}{\sqrt{7}}-1$ (recall that $0 < \sdpt \le \sqrt{7} - 2$), we verify that
    \begin{align*}
        |\radNNN (\eps)|^2 + |\radNNN (\eps)\radpNN (\eps)| + |\radpNN (\eps)|^2 + \eps |\radNNN (\eps)| \le C \eps^{1 + \sdpt} \qquad \text{and} \qquad \eps |\radpNN (\eps)| \le C\eps^{\frac{3+\sdpt}{2}},
    \end{align*}
    which completes the result.
\end{proof}

With an abuse of notation, $u$ below will denote either a function defined on all of $\mathbb{R}^2$ or its restriction to a lattice.
\begin{lemma}\label{lemma:Hs_to_l22}
    Let $\cR \subset \mathbb{R}^2$ be a lattice. 
    For any $s>2$, there exists a positive constant $C_s$ such that $\norm{u}_{\ell^2 (\cR)} \le C_s \norm{u}_{H^s (\mathbb{R}^2)}$ uniformly in $u \in H^s (\mathbb{R}^2)$.
\end{lemma}
\begin{proof}
    Let $\cR^*$ denote the reciprocal lattice and $\Gamma^*$ the Brillouin zone.
    We know that
    \begin{align*}
        \norm{u}_{\ell^2 (\cR)} = \norm{\tilde{u}}_{L^2 (\Gamma^*)}, \qquad
        \tilde{u} (k) = (2\pi)^{-2} \sum_{G \in \cR^*} \hat{u} (k+G),
    \end{align*}
    with $\tilde{u}$ and $\hat{u}$ respectively the Bloch and Fourier transforms of $u$. It follows that
    \begin{align*}
        \norm{u}_{\ell^2 (\cR)}^2 = (2\pi)^{-4} \sum_{G_1, G_2 \in \cR^*} \int_{\Gamma^*} \overline{\hat{u} (k+G_1)} \hat{u} (k+G_2) {\rm d}k,
    \end{align*}
    which after a change of variables can be written as
    \begin{align*}
        \norm{u}_{\ell^2 (\cR)}^2 = (2\pi)^{-4} \sum_{G_1, G_2 \in \cR^*} \int_{\Gamma^*} \overline{\hat{u} (k+G_1)} \hat{u} (k+G_1+G_2) {\rm d}k.
    \end{align*}
    Combining the sum over $G_1$ with the integral over $k$, we obtain
    \begin{align}\label{eq:l2_u2}
        \norm{u}_{\ell^2 (\cR)}^2 = (2\pi)^{-4} \sum_{G \in \cR^*} \int_{\mathbb{R}^2} \overline{\hat{u} (k)} \hat{u} (k+G) {\rm d}k.
    \end{align}
    With $\aver{x} := \sqrt{1 + x^2}$, we next multiply and divide the integrand in \eqref{eq:l2_u2} by $\aver{k}^s \aver{k+G}^s$ and use the fact that $\aver{k}^{-s} \aver{k+G}^{-s} \le C \aver{G}^{-s}$ to obtain
    \begin{align*}
        \norm{u}_{\ell^2 (\cR)}^2 \le C \sum_{G \in \cR^*} \int_{\mathbb{R}^2} \aver{k}^s \overline{\hat{u} (k)}\aver{k+G}^s \hat{u} (k+G) \aver{G}^{-s}{\rm d}k,
    \end{align*}
    which by Cauchy-Schwarz implies
    \begin{align*}
        \norm{u}_{\ell^2 (\cR)}^2 \le C \norm{\aver{\cdot}^s \hat{u} (\cdot)}_{L^2 (\mathbb{R}^2)}^2 \sum_{G \in \cR^*}\aver{G}^{-s} \le C \norm{u}^2_{H^s (\mathbb{R}^2)}.
    \end{align*}
    This completes the proof.
\end{proof}
\sq{Before proving Theorem \ref{thm:main2}, let us introduce some notation. For $j=1,2$, set $\theta_j := (-1)^j \theta$ as before. Recall the definition $K := \frac{4\pi}{3a}(1,0)$ from 
\eqref{eq:Dirac_points} and for $j=1,2$ define
\begin{align}\label{eq:Dirac_pt_rot}
    K_j := \rot_{\theta_j/2} K
\end{align}
a Dirac point corresponding to layer $j$. Similarly, the reciprocal lattice and Brillouin zone for layer $j$ are respectively given by
\begin{align*}
    \cR_j^* := \rot_{\theta_j/2} \cR^*, \qquad \Gamma_j^* := \rot_{\theta_j/2} \Gamma^*,
\end{align*}
where $\cR^*$ and $\Gamma^*$ defined in \eqref{eq:cR_star} are the unrotated reciprocal lattice and Brillouin zone.}

\sq{We will denote wave functions in the Bloch domain $L^2 (\Gamma_1^*; \mathbb{C}^2) \oplus L^2 (\Gamma_2^*; \mathbb{C}^2)$ by $\tilde{\varphi} = (\tilde{\varphi}^A_1, \tilde{\varphi}^B_1, \tilde{\varphi}^A_2, \tilde{\varphi}^B_2)$ with each $\varphi_j^\sigma \in L^2 (\Gamma_j^*; \mathbb{C})$, or more concisely as $\tilde{\varphi} = (\tilde{\varphi}_1, \tilde{\varphi}_2)$ with $\tilde{\varphi}_j \in L^2 (\Gamma_j^*; \mathbb{C}^2)$. With 
$k_j \in \Gamma_j^*$ and $R_j \in \cR_j$ for $j=1,2$, the (unitary) Bloch transform and its inverse are then defined by
\begin{align*}
    [\cG \varphi] (k_1, k_2) := \begin{pmatrix}
        [\cG_1 \varphi_1] (k_1)\\
        [\cG_2 \varphi_2] (k_2)
    \end{pmatrix},
    \qquad
    [\cG^{-1} \tilde{\varphi}]_{R_1, R_2} := \begin{pmatrix}
        [\cG_1^{-1} \tilde{\varphi}_1]_{R_1}\\
        [\cG_2^{-1} \tilde{\varphi}_2]_{R_2}
    \end{pmatrix},
\end{align*}
where
\begin{align*}
    \left[ \mathcal{G}_j \varphi_j \right]^\sigma(k) := \sum_{R_j \in \mathcal{R}_j} e^{- i k \cdot (R_j + \tau_j^\sigma)} \varphi_{R_j}^\sigma, \qquad \left[ \mathcal{G}_j^{-1} \tilde{\varphi}_j \right]^\sigma_{R_j} := \frac{1}{|\Gamma^*|} \inty{\Gamma_j^*}{}{ e^{i k \cdot (R_j + \tau_j^\sigma)} \tilde{\varphi}_j^\sigma(k) }{k}.
\end{align*}
We will often use the shorthand $\tilde{\varphi}_j := \cG_j \varphi_j$.}
\begin{proof}[Proof of Theorem \ref{thm:main2}]
    We apply Lemma \ref{lemma:res2} to the operator $H : \cH \to \cH$ and function $\varphi := \phi - \psi$. Since $i\partial_\microtime \psi = H\psi$ and $(\psi_j)^\sigma_{R_j} (0) = (\phi_j)^\sigma_{R_j}$, it suffices to show that $i\partial_\microtime \phi - H \phi = \resgen$, where 
    \begin{align}\label{eq:resgen_suffice}
        \norm{\resgen (\microtime)}_\cH \le C \srq{\norm{f_0}_{H^{6+\sdpt}}} \eps^{1 + \sdpt_-} \left(\eps + \eps^{\sdpt -\sdpt_-} \eps^2 \microtime \right)
    \end{align}
    uniformly in $\microtime \ge 0$.
    It follows from the definition \eqref{eq:def_phi} of $\phi$ that
    \begin{align}\label{eq:partial_t_phi}
        i (\partial_\microtime \phi_j)^\sigma_{R_j} (\microtime) = \eps^2 i \partial_t f^\sigma_{j} (\eps (R_j+\tau_j^\sigma),\eps \microtime) e^{i K_j \cdot (R_j + \tau_j^\sigma)} e^{-i\eshift \microtime}
        + \eshift (\phi_j)^\sigma_{R_j} (\microtime).
    \end{align}
    We then use Lemmas \ref{lemma:ae2} and \ref{lemma:Hs_to_l22} to write $i \partial_t f = \Hfull f + \res$, where
     \begin{align}\label{eq:res_neglect}
     \begin{split}
        \norm{\eps^2 \res_j (\eps (\cdot \; +\tau_j^\sigma),\eps \microtime)e^{i K_j \cdot (\cdot \; + \tau_j^\sigma)}e^{-i\eshift \microtime}}_{\ell^2 (\cR_j)} &\le 
        \srq{C \norm{\eps^2 \res_j (\eps (\cdot \; +\tau_j^\sigma),\eps \microtime)e^{i K_j \cdot (\cdot \; + \tau_j^\sigma)}e^{-i\eshift \microtime}}_{H^{2+\sdpt} (\mathbb{R}^2)}}\\
        &\srq{\le C\eps \microtime \,\eps^{2+\sdpt}\norm{f_0}_{H^{6+\sdpt}}}
    \end{split}
    \end{align}
    uniformly in $\microtime \ge 0$. Combining \eqref{eq:partial_t_phi} and \eqref{eq:res_neglect}, we obtain
    \begin{align}\label{eq:partial_t_phi2}
        i (\partial_\microtime \phi_j)^\sigma_{R_j} (\microtime) = \eps^2 \{(\Hfull [f])^\sigma_{j} (\eps (R_j+\tau_j^\sigma),\eps \microtime)\} e^{i K_j \cdot (R_j + \tau_j^\sigma)} e^{-i\eshift \microtime}
        + \eshift (\phi_j)^\sigma_{R_j} (\microtime) + O_{\ell^2 (\cR_j)} (\eps \microtime \,\eps^{2 + \sdpt})
    \end{align}
    uniformly in $\microtime \ge 0$. 
    It thus remains to control $\norm{F - H \phi}_\cH$, where
    \begin{align*}
        (F_j)^\sigma_{R_j}(\microtime) := \eps^2 \{(\Hfull [f])^\sigma_{j} (\eps (R_j+\tau_j^\sigma),\eps \microtime)\} e^{i K_j \cdot (R_j + \tau_j^\sigma)} e^{-i\eshift \microtime}
        + \eshift (\phi_j)^\sigma_{R_j} (\microtime)
    \end{align*}
    is the first two terms on the right-hand side of \eqref{eq:partial_t_phi2}.

    It will now be most convenient to set $j$ above to $1$ and control the intra- and inter-layer terms separately (the parallel argument for $j=2$ will be omitted). We begin with the intralayer terms. 
    Define
    \begin{align*}
        (F_{11})^\sigma_{R_1}(\microtime) := \eps^2 \{(\Hfull_{11} [f_1])^\sigma (\eps (R_1+\tau_1^\sigma),\eps \microtime)\} e^{i K_1 \cdot (R_1 + \tau_1^\sigma)} e^{-i\eshift \microtime}
        + \eshift (\phi_1)^\sigma_{R_1} (\microtime),
    \end{align*}
    with the first term on the above right-hand side the contribution to $F_1$ from the diagonal block of $\Hfull$,
    \begin{align}\label{eq:Hfull_diag}
        \Hfull_{11} := L + \eps (\HBMsd - \sq{\frac{1}{2}}\beta i \sigma_3 L).
    \end{align}
    Define the auxiliary $\eps$-dependent symmetric operator $\HBMtwist : H^2 (\mathbb{R}^2; \mathbb{C}^2) \to L^2(\mathbb{R}^2; \mathbb{C}^2)$ by
    \begin{align}\label{eq:Stwist_def}
        \HBMtwist := \begin{pmatrix}
            \frac{1}{2}\eps \alpha_d (D_{r_1}^2 + D_{r_2}^2) & e^{-i\theta/2} \alpha (D_{r_1} - i D_{r_2}) + \frac{1}{2}\eps e^{i\theta}\alpha_o (D^2_{r_1} - D^2_{r_2} + 2i D_{r_1 r_2})\\
            * & \frac{1}{2}\eps \alpha_d (D_{r_1}^2 + D_{r_2}^2)
        \end{pmatrix},
    \end{align}
    where the constants $\alpha, \alpha_o \in \mathbb{C}$ and $\alpha_d \in \mathbb{R}$ are defined in Theorem \ref{thm:Dirac}. Observe that $\HBMtwist$ is the operator $L + \eps \HBMsd$ with $\alpha$ and $\alpha_o$ respectively replaced by $e^{-i\theta/2} \alpha$ and $e^{i\theta} \alpha_o$. Set
    \begin{align}\label{eq:wfntwist_def}
        (\wfntwist)_{R_1}^\sigma(\microtime) := \eps^2 \{(\HBMtwist [f_1])^\sigma (\eps (R_1+\tau_1^\sigma),\eps \microtime)\} e^{i K_1 \cdot (R_1 + \tau_1^\sigma)} e^{-i\eshift \microtime}
        + \eshift (\phi_1)^\sigma_{R_1} (\microtime).
    \end{align}
    We will bound $F_{11} - \wfntwist$ and $\wfntwist - H_{11} \phi_1$ starting with the latter. Define $\tilde{h}_1^{\sigma, \sigma'} (k) := \tilde{h}^{\sigma,\sigma'} (\rot_{\theta/2} k)$, where the function $\tilde{h}^{\sigma,\sigma'}$ is the Bloch Hamiltonian defined in \eqref{eq:tildeH}.
    The Bloch transform of $H_{11}\phi_1$ is
    \begin{align}\label{eq:Bloch_transform_diag}
        (\widetilde{H_{11}\phi_1})^\sigma (k,\microtime) = \frac{1}{\eps  |\Gamma|}e^{-i \eshift \microtime} \sum_{G_1 \in \cR_1^*} \sum_{\sigma' \in \{A,B\}} \tilde{h}_1^{\sigma, \sigma'} (k)e^{iG_1 \cdot \tau_1^{\sigma'}}\hat{f}_1^{\sigma'} \left(\frac{k-K_1+G_1}{\eps},\eps \microtime\right),
    \end{align}
    see Lemma \ref{lemma:Bloch_transform_diag} for a derivation. Lemma \ref{lemma:Bloch_transform_twist} states that the Bloch transform 
    \begin{align}\label{eq:Bloch_transform_twist}
        (\widetilde{\wfntwist})^\sigma (k,\microtime) = \frac{1}{\eps  |\Gamma|}e^{-i \eshift \microtime} \sum_{G_1 \in \cR_1^*} \sum_{\sigma' \in \{A,B\}} \tilde{\mathscr{h}}_{1}^{\sigma, \sigma'} (k;K_1 - G_1)e^{iG_1 \cdot \tau_1^{\sigma'}}\hat{f}_1^{\sigma'} \left(\frac{k-K_1+G_1}{\eps},\eps \microtime\right)
    \end{align}
    of $\wfntwist$ is nearly identical, only $\tilde{h}_1^{\sigma, \sigma'}$ gets replaced by its second-order Taylor expansion about $K_1 - G_1$,
    \begin{align}\label{eq:intra_Taylor}
    \begin{split}
        \tilde{\mathscr{h}}_{1}^{\sigma, \sigma'} (k;K_1 - G_1) := \tilde{h}^{\sigma,\sigma'}_1 (K_1 - G_1) &+ (k-K_1 + G_1) \nabla \tilde{h}^{\sigma,\sigma'}_1 (K_1 - G_1)\\
        &+ \frac{1}{2}(k-K_1 + G_1) \cdot \nabla^2 \tilde{h}^{\sigma,\sigma'}_1 (K_1 - G_1)(k-K_1 + G_1).
    \end{split}
    \end{align}
    \sq{We conclude by Lemma \ref{lemma:monolayer1} that 
    \begin{align}\label{eq:intra_diff1}
        \norm{\wfntwist (\microtime) - H_{11} \phi_1 (\microtime)}_{\ell^2 (\cR_1)} \le 
        \srq{C \eps^3 \norm{f_1 (\cdot, \eps \microtime)}_{H^{4+\sdpt}}}.
    \end{align}}
    We next write
    \begin{align*}
        (F_{11} - \wfntwist)^\sigma_{R_1} (\microtime) = \eps^2 \{(\Hfull_{11} [f_1] - \HBMtwist [f_1])^\sigma (\eps (R_1+\tau_1^\sigma),\eps \microtime)\} e^{i K_1 \cdot (R_1 + \tau_1^\sigma)} e^{-i\eshift \microtime},
    \end{align*}
    and see directly from the definitions \eqref{eq:Hfull_diag} and \eqref{eq:Stwist_def} that
    \begin{align*}
        \Hfull_{11} - \HBMtwist = 
        \begin{pmatrix}
            0 & (1-e^{-i\theta/2} + \sq{\frac{1}{2}}\eps i \beta) \alpha (D_{r_1} - i D_{r_2}) + \frac{1}{2}\eps (1-e^{i\theta})\alpha_o (D^2_{r_1} - D^2_{r_2} + 2i D_{r_1 r_2})\\
            * & 0
        \end{pmatrix}.
    \end{align*}
    Using 
    \srq{the assumption \eqref{eq:beta_def} that $|\theta| \le C \eps$ as $\eps \to 0$,} we have
    \begin{align*}
        \Hfull_{11} - \HBMtwist = 
        \eps^2
        \begin{pmatrix}
            0 & \alpha^\sharp (\eps) (D_{r_1} - i D_{r_2}) + \alpha_o^\sharp (\eps) (D^2_{r_1} - D^2_{r_2} + 2i D_{r_1 r_2})\\
            * & 0
        \end{pmatrix}
    \end{align*}
    for some $\alpha^\sharp, \alpha_o^\sharp \in C^\infty_b (\mathbb{R}; \mathbb{C})$ the space of smooth, bounded functions whose derivatives are all bounded.
    Applying Lemma \ref{lemma:Hs_to_l22}, it follows that for any $\delta > 0$,
    \begin{align}\label{eq:intra_diff2}
        \norm{F_{11} (\microtime) - \wfntwist (\microtime)}_{\ell^2 (\cR_1)} \le C_\delta \eps \norm{\Hfull_{11} [f_1](\cdot, \eps \microtime) - \HBMtwist [f_1] (\cdot, \eps \microtime)}_{H^{2+\delta}} \le C_\delta \eps^3 \norm{f_1 (\cdot, \eps \microtime)}_{H^{4+\delta}}.
    \end{align}
    \sq{Combining \eqref{eq:intra_diff1} and \eqref{eq:intra_diff2}, we have shown that 
    \begin{align}\label{eq:intra_bd}
        \norm{F_{11} (\microtime) - H_{11} \phi_1 (\microtime)}_{\ell^2 (\cR_1)} \le
        \srq{C \eps^3 \norm{f_1 (\cdot, \eps \microtime)}_{H^{4+\sdpt}}}.
    \end{align}}

    We will now control the difference $F_{12} - H_{12} \phi_2$ corresponding to the interlayer terms, where
    \begin{align*}
        (F_{12})^\sigma_{R_1}(\microtime) := \eps^2 \{(\Hfull_{12} [f_2])^\sigma (\eps (R_1+\tau_1^\sigma),\eps \microtime)\} e^{i K_1 \cdot (R_1 + \tau_1^\sigma)} e^{-i\eshift \microtime},
    \end{align*}
    and $\Hfull_{12} := \hoppingT(r) + \radNNN (\eps) \hoppingT_{{\rm NNN}}(r) +\radpNN (\eps) \hoppingT_{\nabla, {\rm NN}}+\eps (\hoppingT_2 + \Tu (r))$ is the off-diagonal block of $\Hfull$. Recall that $\angvar{K} := K/|K|$ and set $\angvar{K}_2 := K_2/|K|$. \sq{Define the $\eps$-dependent operator $\HBMptwist : H^1 (\mathbb{R}^2; \mathbb{C}^2) \to L^2 (\mathbb{R}^2; \mathbb{C}^2)$ by
    \begin{align*}
        \HBMptwist := \hoppingT^\theta (r) + \radNNN (\eps)\hoppingT^\theta_{{\rm NNN}} (r) + \radpNN (\eps) \hoppingT^\theta_{\nabla, {\rm NN}} + \eps \hoppingT^\theta_2,
    \end{align*}
    where each operator $\hoppingT^\theta_j$ is the operator $T_j$ with the $\lambda_i=\hathperpang (\rot_{\pi i/3} \angvar{K})$ and $\mu_i = \frac{1}{|K|} \hathperpang ' (\rot_{2\pi i/3} \angvar{K})$ respectively replaced by $\hathperpang (\rot_{\pi i/3} \angvar{K}_2)$ and $\frac{1}{|K|} \hathperpang ' (\rot_{2\pi i/3} \angvar{K}_2)$. Set 
    \begin{align}\label{eq:wfnptwist}
        (\wfnptwist)^\sigma_{R_1} (\microtime) := \eps^2 \{(\HBMptwist [f_2])^\sigma (\eps (R_1+\tau_1^\sigma),\eps \microtime)\} e^{i K_1 \cdot (R_1 + \tau_1^\sigma)} e^{-i\eshift \microtime}.
    \end{align}
    We will bound $F_{12} -\wfnptwist$ and $\wfnptwist - H_{12} \phi_2$ separately, starting with the latter.}

    As before, this is easiest to analyze in the Bloch domain, where we verify with Lemma \ref{lemma:tildeHperp} that
    \begin{align}\label{eq:tildeHperp}
    \begin{split}
        (\widetilde{\Hperp \phi_2})^\sigma (k_1, \microtime)
        &= \frac{1}{\eps |\Gamma|^2}e^{-i\eshift \microtime}\sum_{G_2, G_2' \in \cR_2^*} \sum_{\sigma'\in \{A,B\}} e^{i(\cG_1 (G_2) \cdot \tau^\sigma_1 + (G'_2- G_2) \cdot \tau^{\sigma'}_2)} \hathperp (k_1 + \cG_1 (G_2); \eps)\\
        &\hspace{6cm} \times \hat{f}^{\sigma'}_2 \Big(\frac{k_1 + \cG_1 (G_2) - G_2 - K_2 + G'_2}{\eps}, \eps \microtime \Big),
    \end{split}
    \end{align}
    where $\cG_1 (G_2)$ is the unique layer-$1$ reciprocal lattice vector $G_1 \in \cR_1^*$ such that $k_1 + G_1 - G_2 \in \Gamma_2^*$.
    For $k,q \in \mathbb{R}^2$ and $1\{\cdot\}$ the indicator function, let
    \begin{align}\label{eq:filtered_Taylor}
        \hat{\mathscr{h}}_{12} (k,q;\eps) := \hathperp (q;\eps)\, 1\{ |q| \le 2 |K| \}+ (k-q) \cdot \nabla \hathperp (q;\eps) \, 1 \{|q| \le |K|\}
    \end{align}
    denote a filtered first-order Taylor expansion of $\hathperp (k;\eps)$ about $k=q$. Then, as shown in Lemma \ref{lemma:tildeF12},
    \begin{align}\label{eq:tildeF12}
    \begin{split}
        (\widetilde{\sq{\wfnptwist}})^\sigma (k_1,\microtime) &= \frac{1}{\eps |\Gamma|^2}e^{-i\eshift \microtime}\sum_{G_2, G_2' \in \cR_2^*} \sum_{\sigma'\in \{A,B\}} e^{i(\cG_1 (G_2) \cdot \tau^\sigma_1 + (G'_2- G_2) \cdot \tau^{\sigma'}_2)}\\
        &\hspace{1.0cm} \times \hat{\mathscr{h}}_{12} (k_1 + \cG_1 (G_2), K_2 + G_2 - G'_2;\eps)
        \hat{f}^{\sigma'}_2 \Big(\frac{k_1 + \cG_1 (G_2) - G_2 - K_2 + G'_2}{\eps},\eps \microtime \Big).
    \end{split}
    \end{align}
    Observe that $\widetilde{\sq{\wfnptwist}}$ and $\widetilde{\Hperp \phi_2}$ are nearly identical, only the factor of $\hathperp(k_1 + \cG_1 (G_2);\eps)$ in \eqref{eq:tildeHperp} gets replaced by $\hat{\mathscr{h}}_{12} (k_1 + \cG_1 (G_2), K_2+G_2-G_2';\eps)$ in \eqref{eq:tildeF12}. \sq{We then use Lemma \ref{lemma:bilayer1} to show that
    \begin{align}\label{eq:bilayer1}
    \norm{\sq{\wfnptwist}(\microtime) - \Hperp \phi_2 (\microtime)}_{\ell^2 (\cR_1)} \le \srq{C \eps^{2 + \sdpt_-}\norm{f_2 (\cdot, \eps \microtime)}_{H^{4+\sdpt}}}.
    \end{align}}

    Using that $|K_2 - K| = O (\eps)$, we will next control $$F_{12}(\microtime) - \wfnptwist (\microtime) = \eps^2 \{(\Hfull_{12} [f_2] - \HBMptwist [f_2])^\sigma (\eps (R_1+\tau_1^\sigma),\eps \microtime)\} e^{i K_1 \cdot (R_1 + \tau_1^\sigma)} e^{-i\eshift \microtime}.$$
    Define $\dtwist_1 := \hoppingT (r) + \eps \Tu (r) - \hoppingT^\theta (r)$ and $\dtwist_2 := \radNNN (\eps)(\hoppingT_{{\rm NNN}} (r) - \hoppingT^\theta_{{\rm NNN}} (r)) + \radpNN (\eps) (\hoppingT_{\nabla, {\rm NN}} - \hoppingT^\theta_{\nabla, {\rm NN}}) + \eps (\hoppingT_2 - \hoppingT^\theta_2)$ so that
    \begin{align*}
        \Hfull_{12} - \HBMptwist = \dtwist_1 + \dtwist_2.
    \end{align*}
    Considering $\dtwist_2$ first,
    we use the definitions of $\hoppingT_{{\rm NNN}} (r)$ and $\hoppingT^\theta_{{\rm NNN}} (r)$ to write
    \begin{align*}
        \hoppingT_{{\rm NNN}} (r) - \hoppingT^\theta_{{\rm NNN}} (r) &=
        \frac{1}{|\Gamma|} \Bigg(
    (\hathperpang (\rot_{\pi} \angvar{K})-\hathperpang (\rot_{\pi} \angvar{K}_2))\sq{e^{i (b_1 - b_2) \cdot \ls}}e^{\frac{-i8\pi \beta}{3a} r_2} \begin{pmatrix}
        1 & 1\\
        1 & 1
    \end{pmatrix}\\
    &\hspace{-0.5cm} + 
    (\hathperpang (\rot_{5\pi/3} \angvar{K}) - \hathperpang (\rot_{5\pi/3} \angvar{K}_2)) \sq{e^{i (b_1 + b_2) \cdot \ls}}e^{\frac{i8\pi \beta}{3a} (\frac{\sqrt{3}}{2} r_1 + \frac{1}{2} r_2)} \begin{pmatrix}
        1 & \sq{e^{-i2\pi/3}}\\
        \sq{e^{i2\pi/3}} & 1
    \end{pmatrix} \\
    &\hspace{-0.5cm} + 
    (\hathperpang (\rot_{\pi/3} \angvar{K}) - \hathperpang (\rot_{\pi/3} \angvar{K}_2)) \sq{e^{-i (b_1 + b_2) \cdot \ls}}e^{\frac{-i8\pi \beta}{3a} (\frac{\sqrt{3}}{2} r_1 - \frac{1}{2} r_2)} \begin{pmatrix}
        1 & \sq{e^{i2\pi/3}}\\
        \sq{e^{-i2\pi/3}} & 1
    \end{pmatrix}\Bigg).
    \end{align*}
    Since $\angvar{K}_2 - \angvar{K} = (\rot_{\theta/2} - I_2) \angvar{K}$ with 
    \srq{$|\theta| \le C \eps$}, the regularity of $\hathperpang$ in Assumption \ref{assumption:hperp2} implies that
    \begin{align*}
        |\hoppingT_{{\rm NNN}} (r) - \hoppingT^\theta_{{\rm NNN}} (r)| \le C \eps
    \end{align*}
    uniformly in $r \in \mathbb{R}^2$ and $0 < \eps < 1$. All derivatives (with respect to $r$) of $\hoppingT_{{\rm NNN}} (r) - \hoppingT^\theta_{{\rm NNN}} (r)$ satisfy the same bound. Hence, for any $N \ge 0$, the operator norm of $\hoppingT_{{\rm NNN}} (r) - \hoppingT^\theta_{{\rm NNN}} (r) : H^N (\mathbb{R}^2; \mathbb{C}^2) \to H^N (\mathbb{R}^2; \mathbb{C}^2)$ is also bounded by $C \eps$. The same logic implies that
    \begin{align*}
        \norm{\hoppingT_{\nabla, {\rm NN}} - \hoppingT^\theta_{\nabla, {\rm NN}}}_{H^{N+1} (\mathbb{R}^2; \mathbb{C}^2) \to H^N (\mathbb{R}^2; \mathbb{C}^2)} + \norm{\hoppingT_2 - \hoppingT_2^\theta}_{H^{N+1} (\mathbb{R}^2; \mathbb{C}^2) \to H^N (\mathbb{R}^2; \mathbb{C}^2)} \le C \eps
    \end{align*}
    uniformly in $\eps$, where the domain $H^{N+1}$ is used instead of $H^N$ because the above operators are first-order differential operators. Recalling the bounds \eqref{eq:xi_zeta_bd} on $\radpNN (\eps)$ and $\radNNN (\eps)$, it follows that for any $N \ge 0$,
    \begin{align}\label{eq:dtwist2}
        \norm{\eps \{(\dtwist_2 [f_2])^\sigma (\cdot,\eps \microtime)\}}_{H^N} \le C\eps^{2 + \frac{1+\sdpt}{2}} \norm{f_2 (\cdot, \eps \microtime)}_{H^{N+1}}.
    \end{align}
    

    We will next prove a similar bound for $\dtwist_1$. We see that
    \begin{align}\label{eq:interlayer_expansion}
    \begin{split}
        \hoppingT(r) + \eps \Tu (r) &= \frac{1}{|\Gamma|} \Bigg(
    (\lambda_0 + \frac{\eps \beta |K|}{2}\mu_0) e^{\frac{i4\pi \beta}{3a} r_2} \begin{pmatrix}
        1 & 1\\
        1 & 1
    \end{pmatrix}\\
    &\hspace{2cm} + 
    (\lambda_2+ \frac{\eps \beta |K|}{2}\mu_1) \sq{e^{-i b_2 \cdot \ls}} e^{\frac{i4\pi \beta}{3a} (-\frac{\sqrt{3}}{2} r_1 - \frac{1}{2} r_2)} \begin{pmatrix}
        1 & \sq{e^{-i2\pi/3}}\\
        \sq{e^{i2\pi/3}} & 1
    \end{pmatrix} \\
    & \hspace{3cm} +(\lambda_4 + \frac{\eps \beta |K|}{2}\mu_2) \sq{e^{ib_1 \cdot \ls}}e^{\frac{i4\pi \beta}{3a} (\frac{\sqrt{3}}{2} r_1 - \frac{1}{2} r_2)} \begin{pmatrix}
        1 & \sq{e^{i2\pi/3}}\\
        \sq{e^{-i2\pi/3}} & 1
    \end{pmatrix}\Bigg).
    \end{split}
    \end{align}
    On the other hand,
    using that $2 \sin (\theta/2) = \beta \eps$ \srq{and $\cos (\theta/2)-1 = O(\eps^2)$}, we find that
    \begin{align*}
        K_2 - K
        = \frac{\eps \beta}{2} \begin{pmatrix}
            0 & -1\\
            1 & 0
        \end{pmatrix} K + O (\eps^2),
    \end{align*}
    and thus
    \begin{align*}
        (\rot_{2\pi j/3} (K_2- K)) \cdot \nabla \hathperp (\rot_{2\pi j/3} K; \eps) = \frac{\eps \beta}{2} \left( \begin{pmatrix}
            0 & -1\\
            1 & 0
        \end{pmatrix} K \right) \cdot \rot_{-2\pi j/3} \nabla \hathperp (\rot_{2\pi j/3} K; \eps) + O (\eps^{2+\frac{1+\sdpt}{2}})
    \end{align*}
    by the \srq{decay of $\nabla \hathperp (\rot_{2\pi j/3} K; \eps)$ in \eqref{eq:nabla_bd_K}}.
    \srq{The decomposition $\hathperp (k;\eps) = \hathperprad (|k|; \eps) \hathperpang (\angvar{k})$ from \eqref{eq:sep} implies that the gradient of $\hathperp (\cdot \; ; \eps)$ is given explicitly by
    \begin{align}\label{eq:nabla_hathperp}
        \nabla \hathperp (k; \eps) = \frac{1}{|k|} \hathperprad (|k|; \eps) \hathperpang ' (\angvar{k}) \rot_{\pi/2} \angvar{k} + \hathperprad ' (|k|; \eps) \hathperpang (\angvar{k}) \angvar{k}, \qquad k \in \mathbb{R}^2, 
    \end{align}
    hence the definition \eqref{eq:rad_ang} of the $\lambda_i$ and $\mu_j$ implies that}
    \begin{align*}
        \left( \begin{pmatrix}
            0 & -1\\
            1 & 0
        \end{pmatrix} K \right) \cdot \rot_{-2\pi j/3} \nabla \hathperp (\rot_{2\pi j/3} K; \eps)&= \eps \mu_j \left( \begin{pmatrix}
            0 & -1\\
            1 & 0
        \end{pmatrix} K \right) \cdot \rot_{\pi/2} \angvar{K}
        + \radpNN (\eps)\lambda_{2j} \left( \begin{pmatrix}
            0 & -1\\
            1 & 0
        \end{pmatrix} K \right) \cdot \angvar{K}\\
        &= \eps |K| \mu_j.
    \end{align*}
    Since $\hathperp (\rot_{2\pi j/3} K; \eps) = \eps \lambda_{2j}$, we conclude that $$\lambda_{2j} + \frac{\eps \beta |K|}{2}\mu_j = \frac{1}{\eps} \left(\hathperp (\rot_{2\pi j/3} K; \eps) + (\rot_{2\pi j/3} (K_2- K)) \cdot \nabla \hathperp (\rot_{2\pi j/3} K; \eps)\right) + O(\eps^{1+\frac{1+\sdpt}{2}}).$$
    We recognize the Taylor expansion of $\hathperp$ on the above right-hand side and recall that the operator 
    $\hoppingT^\theta (r)$ can be obtained from $\hoppingT(r) + \eps \Tu (r)$ by replacing $\lambda_{2j} + \frac{\eps \beta |K|}{2}\mu_j$ in \eqref{eq:interlayer_expansion} by $\frac{1}{\eps}\hathperp (\rot_{2\pi j/3} K_2; \eps)$. The decay of $\nabla^2 \hathperp$ in 
    \srq{\eqref{eq:4bounds_full}} then implies that for any $N \ge 0$,
    \begin{align*}
        \norm{\eps \{(\dtwist_1 [f_2])^\sigma (\cdot,\eps \microtime)\}}_{H^N} \le C\eps^{2+\sdpt} \norm{f_2 (\cdot, \eps \microtime)}_{H^{N}}.
    \end{align*}
    Recalling \eqref{eq:dtwist2} and Lemma \ref{lemma:Hs_to_l22}, we have shown that for any $\delta > 0$,
    \begin{align*}
        \norm{F_{12} (\microtime) - \wfnptwist (\microtime)}_{\ell^2 (\cR_1)} \le \norm{\eps (\Hfull_{12} [f_2] - \HBMptwist [f_2])^\sigma (\cdot,\eps \microtime)}_{H^{2+\delta}}\le C\eps^{2+\sdpt} \norm{f_2 (\cdot, \eps \microtime)}_{H^{3+\delta}}.
    \end{align*}
    Combining the above bound with \eqref{eq:bilayer1}, we conclude that
    \begin{align*}
        \norm{F_{12}(\microtime) - H_{12} \phi_2(\microtime)}_{\ell^2 (\cR_1)} \le \srq{C \eps^{2 + \sdpt_-}\norm{f_2 (\cdot, \eps \microtime)}_{H^{4+\sdpt}}}
    \end{align*}

    Recalling \eqref{eq:intra_bd} and the identities
    \begin{align*}
        F_1 (\microtime) = F_{11} (\microtime) + F_{12} (\microtime), \qquad (H \phi)_1 (\microtime) = H_{11} \phi_1 (\microtime) + H_{12} \phi_2 (\microtime),
    \end{align*}
    \srq{we have now shown that
    \begin{align}\label{eq:bd1_t}
        \norm{F_{1} (\microtime) - (H \phi)_1 (\microtime)}_{\ell^2 (\cR_1)} &\le
        C \eps^{2 + \sdpt_-}\norm{f (\cdot, \eps \microtime)}_{H^{4+\sdpt}}.
    \end{align}}
    By Lemma \ref{lemma:Sobolev_ms2}, for any $N \ge 0$, $f = \tff{1} + \radNNN (\eps) \tff{{\rm NNN}} + \radpNN (\eps) \tff{\nabla, {\rm NN}} + \eps \tff{2}$ satisfies
    \begin{align*}
        \norm{f (\cdot, t)}_{H^N} \le C \left( \norm{f_0}_{H^N} + |\radNNN (\eps)| t \norm{f_0}_{H^N} + |\radpNN (\eps)| t \norm{f_0}_{H^{N+1}} + \eps t \norm{f_0}_{H^{N+2}}\right) \le C \norm{f_0}_{H^{N+2}} \left(1 + \eps^{\frac{1 + \sdpt}{2}} t \right)
    \end{align*}
    uniformly in $0 < \eps < 1$ and $t \ge 0$. \sq{Therefore, \eqref{eq:bd1_t} implies that
    \begin{align*}
        \norm{F_{1} (\microtime) - (H \phi)_1 (\microtime)}_{\ell^2 (\cR_1)} \le \srq{C \norm{f_0}_{H^{6+\sdpt}} \eps^{2 + \sdpt_-} \left(1 + \eps^{\frac{3 + \sdpt}{2}} \microtime \right)}
    \end{align*}
    uniformly in $0 < \eps < 1$ and $\microtime \ge 0$.} Combining this with \eqref{eq:res_neglect}, we obtain that $\resgen_1 := i \partial_\microtime \phi_1 - (H\phi)_1$ satisfies
    \begin{align*}
        \norm{\resgen_1 (\microtime)}_{\ell^2 (\cR_1)} \le \srq{C \norm{f_0}_{H^{6+\sdpt}} \eps^{2 + \sdpt_-} \left(1 + \eps^{1+\sdpt -\sdpt_-} \microtime \right)}
    \end{align*}
    uniformly in $0 < \eps < 1$ and $\microtime \ge 0$.
    A parallel argument establishes the same bound for $\resgen_2 := i \partial_\microtime \phi_2 - (H\phi)_2$.
    Thus we have verified \eqref{eq:resgen_suffice} and the proof is complete.
\end{proof}

\begin{lemma}[Derivation of \eqref{eq:Bloch_transform_diag}]\label{lemma:Bloch_transform_diag}
    The Bloch transform of $H_{11}\phi_1$ is
    \begin{align*}
        (\widetilde{H_{11}\phi_1})^\sigma (k,\microtime) = \frac{1}{\eps  |\Gamma|}e^{-i \eshift \microtime} \sum_{G_1 \in \cR_1^*} \sum_{\sigma' \in \{A,B\}} \tilde{h}_1^{\sigma, \sigma'} (k)e^{iG_1 \cdot \tau_1^{\sigma'}}\hat{f}_1^{\sigma'} \left(\frac{k-K_1+G_1}{\eps},\eps \microtime\right).
    \end{align*}
\end{lemma}
\begin{proof}
    A direct calculation using the translation-invariance of $H_{11}$ reveals that
    \begin{align}\label{eq:translation_invariance}
        (\widetilde{H_{11}\phi_1})^\sigma (k,\microtime) = \sum_{\sigma' \in \{A,B\}} \tilde{h}_1^{\sigma, \sigma'} \tilde{\phi}_1^{\sigma'} (k, \microtime).
    \end{align}
    By the definition \eqref{eq:def_phi} of $\phi$, we have
    \begin{align*}
        \tilde{\phi}_1^\sigma (k,\microtime) = \eps e^{-i\eshift \microtime} \sum_{R_1 \in \cR_1} e^{-ik \cdot (R_1 + \tau^\sigma_1)} f_1^\sigma (\eps (R_1 + \tau_1^\sigma), \eps \microtime) e^{iK_1 \cdot (R_1 + \tau_1^\sigma)}.
    \end{align*}
    Writing $f_1^\sigma$ in terms of its Fourier transform, we obtain
    \begin{align*}
        \tilde{\phi}_1^\sigma (k,\microtime) = \eps e^{-i\eshift \microtime} \sum_{R_1 \in \cR_1} e^{-ik \cdot (R_1 + \tau^\sigma_1)} \frac{1}{(2\pi)^2}\int_{\mathbb{R}^2} \hat{f}_1^\sigma (p, \eps \microtime) e^{ip \cdot \eps (R_1 + \tau_1^\sigma)} {\rm d}p \, e^{iK_1 \cdot (R_1 + \tau_1^\sigma)},
    \end{align*}
    which by Poisson's summation formula becomes
    \begin{align}\label{eq:phi_Bloch}
    \begin{split}
        \tilde{\phi}_1^\sigma (k,\microtime) &=\frac{\eps}{|\Gamma|} e^{-i\eshift \microtime} e^{-ik \cdot \tau_1^\sigma} \int_{\mathbb{R}^2} \sum_{G_1 \in \cR_1^*} \delta (\eps p + K_1 - k - G_1)\hat{f}_1^\sigma (p, \eps \microtime) e^{i \eps p \cdot \tau_1^\sigma} {\rm d} p \, e^{iK_1 \cdot \tau_1^\sigma}\\
        &= \frac{1}{\eps |\Gamma|} e^{-i\eshift \microtime} \sum_{G_1 \in \cR_1^*} \hat{f}_1^\sigma \left(\frac{k-K_1+G_1}{\eps}, \eps \microtime\right) e^{i G_1 \cdot \tau_1^\sigma}.
    \end{split}
    \end{align}
    The lemma then follows from 
    plugging 
    \eqref{eq:phi_Bloch} into \eqref{eq:translation_invariance}.
\end{proof}

\begin{lemma}[Derivation of \eqref{eq:Bloch_transform_twist}]\label{lemma:Bloch_transform_twist}
    The Bloch transform of $\wfntwist$ is
    \begin{align*}
        (\widetilde{\wfntwist})^\sigma (k,\microtime) = \frac{1}{\eps  |\Gamma|}e^{-i \eshift \microtime} \sum_{G_1 \in \cR_1^*} \sum_{\sigma' \in \{A,B\}} \tilde{\mathscr{h}}_{1}^{\sigma, \sigma'} (k;K_1 - G_1)e^{iG_1 \cdot \tau_1^{\sigma'}}\hat{f}_1^{\sigma'} \left(\frac{k-K_1+G_1}{\eps},\eps \microtime\right).
    \end{align*}
\end{lemma}
\begin{proof}
    Recall that by Theorem \ref{thm:Dirac}, the constants $\alpha, \alpha_o \in \mathbb{C}$ and $\alpha_d \in \mathbb{R}$ satisfy 
    \begin{align}\label{eq:alphas}
        \nabla \tilde{h}^{\sigma, \sigma}(K) = (0,0), \quad
        \nabla \tilde{h}^{A,B} (K) = \alpha (1,-i), \quad \nabla^2 \tilde{h}^{\sigma,\sigma} (K) = \alpha_d I_2, \quad \nabla^2 \tilde{h}^{A,B} (K) &= \alpha_o \begin{pmatrix}
            1 & i\\
            i & -1
        \end{pmatrix}.
    \end{align}
    It follows from 
    the definition $\tilde{h}_1^{\sigma, \sigma'} (k) := \tilde{h}^{\sigma,\sigma'} (\rot_{\theta/2} k)$ that 
    $\nabla \tilde{h}_1^{\sigma, \sigma'} (k) = \rot_{-\theta/2}\nabla \tilde{h}^{\sigma, \sigma'} (\rot_{\theta/2} k)$ and $\nabla^2 \tilde{h}_1^{\sigma, \sigma'} (k) = \rot_{-\theta/2}\nabla^2 \tilde{h}^{\sigma, \sigma'} (\rot_{\theta/2} k) \rot_{\theta/2}$ for all $k$. Plugging in $k=K_1 = \rot_{-\theta/2} K$ and applying \eqref{eq:alphas}, we obtain
    \begin{align*}
        \nabla \tilde{h}_1^{\sigma, \sigma}(K_1) = (0,0), &\qquad
        \nabla \tilde{h}_1^{A,B} (K_1) = e^{-i\theta/2} \alpha (1,-i), \\
        \nabla^2 \tilde{h}_1^{\sigma,\sigma} (K_1) = \alpha_d I_2, &\qquad \nabla^2 \tilde{h}_1^{A,B} (K_1) = e^{i\theta} \alpha_o \begin{pmatrix}
            1 & i\\
            i & -1
        \end{pmatrix}.
    \end{align*}
    It follows from the definition \eqref{eq:Stwist_def} that $\HBMtwist$ is given by
    \begin{align}\label{eq:Stwist_given}
        \HBMtwist^{\sigma,\sigma'}[u](r) = -i \nabla \tilde{h}_1^{\sigma,\sigma'}(K_1)\cdot \nabla_r u (r)- \frac{1}{2} \eps \nabla_r \cdot \nabla^2 \tilde{h}_1^{\sigma,\sigma'}(K_1) \nabla_r u(r).
    \end{align}
    Recall that $\eshift := \tilde{h}^{A,A} (K) = \tilde{h}^{B,B} (K)$, where the definition was provided in the statement of Theorem \ref{thm:main2} and the equality follows \eqref{eq:AABB}. We also have $\tilde{h}^{A,B} (K) = 0$ by Theorem \ref{thm:Dirac}. Applying the above definition of $\tilde{h}^{\sigma,\sigma'}_1 (k)$, we conclude that $\tilde{h}_1^{\sigma, \sigma'} (K_1) = \eshift \delta_{\sigma,\sigma'}$. Plugging this and \eqref{eq:Stwist_given} into \eqref{eq:wfntwist_def}, we obtain
    \begin{align*}
        (\wfntwist)_{R_1}^\sigma(\microtime) = \eps e^{i K_1 \cdot (R_1 + \tau_1^\sigma)} e^{-i\eshift \microtime}\sum_{\sigma' \in \{A,B\}}\Big( \tilde{h}_1^{\sigma,\sigma'}(K_1)f_1^{\sigma'} (r,\eps \microtime)&- i\eps \nabla \tilde{h}_1^{\sigma,\sigma'}(K_1)\cdot \nabla_r f_1^{\sigma'} (r,\eps \microtime)\\
        &\hspace{-1cm} - \frac{1}{2} \eps^2 \nabla_r \cdot \nabla^2 \tilde{h}_1^{\sigma,\sigma'}(K_1) \nabla_r f_1^{\sigma'} (r,\eps \microtime)\Big) \Big \vert _{r = \eps (R_1+\tau_1^\sigma)}.
    \end{align*}
    It follows from the definition
    $
        (\widetilde{\wfntwist})^\sigma(k,\microtime):= \sum_{R_1 \in \cR_1} e^{-ik \cdot (R_1+\tau_1^\sigma)}(\wfntwist)_{R_1}^\sigma(\microtime)
    $
    and the Poisson summation formula that
    \begin{align*}
        (\widetilde{\wfntwist})^\sigma(k,\microtime)&=\frac{1}{\eps|\Gamma|} e^{-i\eshift \microtime}\int_{\mathbb{R}^2} \sum_{G_1 \in \cR_1^*} \sum_{\sigma' \in \{A,B\}} e^{i G_1 \cdot \tau_1^{\sigma}} e^{-i(k-K_1 + G_1) \cdot r/\eps}\Big( \tilde{h}_1^{\sigma,\sigma'}(K_1)f_1^{\sigma'} (r,\eps \microtime)\\
        &\hspace{2cm} - i\eps \nabla \tilde{h}_1^{\sigma,\sigma'}(K_1)\cdot \nabla_r f_1^{\sigma'} (r,\eps \microtime)
        - \frac{1}{2} \eps^2 \nabla_r \cdot \nabla^2 \tilde{h}_1^{\sigma,\sigma'}(K_1) \nabla_r f_1^{\sigma'} (r,\eps \microtime)\Big) {\rm d}r.
    \end{align*}
    Integrating by parts in $r$, this becomes
    \begin{align*}
        (\widetilde{\wfntwist})^\sigma(k,\microtime)&=\frac{1}{\eps|\Gamma|} e^{-i\eshift \microtime}\int_{\mathbb{R}^2} \sum_{G_1 \in \cR_1^*} \sum_{\sigma' \in \{A,B\}} e^{i G_1 \cdot \tau_1^{\sigma}} e^{-i(k-K_1 + G_1) \cdot r/\eps}\Big( \tilde{h}_1^{\sigma,\sigma'}(K_1)\\
        &+ (k-K_1 + G_1) \cdot \nabla \tilde{h}_1^{\sigma,\sigma'}(K_1)+ \frac{1}{2} (k-K_1+G_1) \cdot \nabla^2 \tilde{h}_1^{\sigma,\sigma'}(K_1) (k-K_1 + G_1) \Big) f_1^{\sigma'} (r,\eps \microtime) {\rm d}r.
    \end{align*}
    We now recognize that the integral over $r$ yields the Fourier transform of $f_1^{\sigma'}$, that is
    \begin{align*}
        (\widetilde{\wfntwist})^\sigma(k,\microtime)&=\frac{1}{\eps|\Gamma|} e^{-i\eshift \microtime} \sum_{G_1 \in \cR_1^*} \sum_{\sigma' \in \{A,B\}} e^{i G_1 \cdot \tau_1^{\sigma}} \Big( \tilde{h}_1^{\sigma,\sigma'}(K_1)+ (k-K_1 + G_1) \cdot \nabla \tilde{h}_1^{\sigma,\sigma'}(K_1)\\
        &\hspace{2cm} + \frac{1}{2} (k-K_1+G_1) \cdot \nabla^2 \tilde{h}_1^{\sigma,\sigma'}(K_1) (k-K_1 + G_1) \Big) \hat{f}_1^{\sigma'} \left(\frac{k-K_1+G_1}{\eps},\eps \microtime\right).
    \end{align*}
    The lemma then follows from the fact that $\tilde{h}_1^{\sigma, \sigma'} (k-G_1) = e^{i G_1 \cdot \tau_1^{\sigma,\sigma'}} \tilde{h}_1^{\sigma, \sigma'} (k)$ for all $k \in \mathbb{R}^2$ and $G_1 \in \cR_1^*$.
\end{proof}

\begin{lemma}[Proof of \eqref{eq:intra_diff1}]\label{lemma:monolayer1}
    \srq{With $0< \sdpt \le 1$ defined in Theorem \ref{thm:main2}, we have}
    \begin{align}\label{eq:intra_diff1_pf}
        \norm{\wfntwist (\microtime) - H_{11} \phi_1 (\microtime)}_{\ell^2 (\cR_1)} \le 
        \srq{C \eps^3 \norm{f_1 (\cdot, \eps \microtime)}_{H^{4+\sdpt}}}.
    \end{align}
\end{lemma}
\begin{proof}
    By the identity $\norm{u}_{\ell^2 (\cR_1)} = |\Gamma^*|^{-1/2} \norm{\tilde{u}}_{L^2 (\Gamma_1^*)}$, it suffices to 
    verify \eqref{eq:intra_diff1_pf} in the Bloch domain. Lemmas \ref{lemma:Bloch_transform_diag} and \ref{lemma:Bloch_transform_twist} imply that
    \begin{align}\label{eq:diff_Bloch1}
        \begin{split}
        \widetilde{\wfntwist^\sigma} (k, \microtime) - \widetilde{(H_{11} \phi_1)^\sigma} (k,\microtime) &= \frac{1}{\eps  |\Gamma|}e^{-i \eshift \microtime} \sum_{G_1 \in \cR_1^*} \sum_{\sigma' \in \{A,B\}} \left( \tilde{h}_1^{\sigma,\sigma'} (k) - \tilde{\mathscr{h}}_{1}^{\sigma, \sigma'} (k;K_1 - G_1)\right) \\
        &\hspace{5cm} \times e^{iG_1 \cdot \tau_1^{\sigma'}}\hat{f}_1^{\sigma'} \left(\frac{k-K_1+G_1}{\eps},\eps \microtime\right).
        \end{split}
    \end{align}
    The rapid decay of $h$ in Assumption \ref{assumption:h} easily implies that for any $\sigma, \sigma' \in \{A,B\}$, the function $\tilde{h}_1^{\sigma, \sigma'} \in C^\infty_b$ is smooth with all derivatives bounded. Hence, by Taylor's theorem,
    \begin{align*}
        \sq{\left| \tilde{h}_1^{\sigma,\sigma'} (k) - \tilde{\mathscr{h}}_{1}^{\sigma, \sigma'} (k;K_1 - G_1) \right| \le C 
        \chi (k-K_1 + G_1), \qquad \chi (x) := \min \{|x|^2, |x|^3\}.}
    \end{align*}
    Plugging this bound into \eqref{eq:diff_Bloch1}, we obtain
    \begin{align*}
        \norm{\widetilde{\wfntwist} (\cdot, \microtime) - \widetilde{H_{11} \phi_1} (\cdot, \microtime)}^2_{L^2 (\Gamma_1^*)} &\le \frac{C}{\eps^2} \int_{\Gamma_1^*}
        \sum_{G_1, G_1' \in \cR_1^*} \sum_{\sigma', \sigma'' \in \{A,B\}}\sq{\chi(k-K_1 + G_1) \chi(k-K_1 + G_1')}\\
        & \hspace{2cm} \times 
        \left| \hat{f}_1^{\sigma'} \left(\frac{k-K_1+G_1}{\eps},\eps \microtime\right)\right| \left| \hat{f}_1^{\sigma''} \left(\frac{k-K_1+G_1'}{\eps},\eps \microtime\right)\right| {\rm d} k.
    \end{align*}
    Combining the integral over $k$ with the sum over $G_1$, the above right-hand side is equal to
    \begin{align}\label{eq:combining}
        \frac{C}{\eps^2} \int_{\mathbb{R}^2}
        \sum_{G_1 \in \cR_1^*} \sum_{\sigma', \sigma'' \in \{A,B\}}\sq{\chi(k-K_1) \chi(k-K_1 + G_1)}
        \left| \hat{f}_1^{\sigma'} \left(\frac{k-K_1}{\eps},\eps \microtime\right)\right| \left| \hat{f}_1^{\sigma''} \left(\frac{k-K_1+G_1}{\eps},\eps \microtime\right)\right| {\rm d} k.
    \end{align}
    \sq{Using that $\chi (x) \le |x|^3$,} the term $G_1 = 0$ is bounded by $C \eps^6 \norm{f_1 (\cdot, \eps \microtime)}^2_{H^3}$.
    To handle the case $G_1 \ne 0$, let $\Omega_1, \Omega_2 \subset \mathbb{R}^2$ such that $\mathbb{R}^2 = \Omega_1 \cup \Omega_2$ and
    \begin{align*}
        2\Bigg|\frac{k - K_1}{\eps}\Bigg| &\ge |G_1|/\eps \qquad \text{for all} \quad k \in \Omega_1, \qquad \text{and}\\
        2\Bigg|\frac{k - K_1 + G_1}{\eps}\Bigg| &\ge |G_1|/\eps \qquad \text{for all} \quad k \in \Omega_2.
    \end{align*}
    \srq{Using that $\chi (x) \le |x|^2$, it follows that the $G_1 \ne 0$ contribution of \eqref{eq:combining} is bounded by
    \begin{align*}
        &C \eps^2
        \sum_{\vec{G}_1 \ne 0} \sum_{\sigma', \sigma''} \int_{\mathbb{R}^2} \Bigg|\frac{\vec k - \vec K_1}{\eps}\Bigg|^{2} \Bigg|\frac{\vec k - \vec K_1 + \vec G_1}{\eps}\Bigg|^{2} \Bigg| \hat{f}_1^{\sigma'} \Big( \frac{\vec k - \vec K_1}{\eps}, \eps \microtime\Big) \Bigg| \; \Bigg| \hat{f}_1^{\sigma''} \Big( \frac{\vec k - \vec K_1 + \vec G_1}{\eps}, \eps \microtime \Big) \Bigg| {\rm d} k\\
        &\le C\eps^2 \sum_{G_1 \ne 0}\sum_{\sigma', \sigma''} \int_{\Omega_1} \Big(\frac{2\eps}{|G_1|}\Big)^{2+\sdpt}\Bigg|\frac{\vec k - \vec K_1}{\eps}\Bigg|^{4+\sdpt} \Bigg|\frac{\vec k - \vec K_1 + \vec G_1}{\eps}\Bigg|^{2} \Bigg| \hat{f}_1^{\sigma'} \Big( \frac{\vec k - \vec K_1}{\eps}, \eps \microtime \Big) \Bigg| \; \Bigg| \hat{f}_1^{\sigma''} \Big( \frac{\vec k - \vec K_1 + \vec G_1}{\eps}, \eps \microtime\Big) \Bigg| {\rm d} k\\
        &+ C\eps^2 \sum_{G_1 \ne 0}\sum_{\sigma', \sigma''} \int_{\Omega_2} \Big(\frac{2\eps}{|G_1|}\Big)^{2+\sdpt}\Bigg|\frac{\vec k - \vec K_1}{\eps}\Bigg|^{2} \Bigg|\frac{\vec k - \vec K_1 + \vec G_1}{\eps}\Bigg|^{4+\sdpt} \Bigg| \hat{f}_1^{\sigma'} \Big( \frac{\vec k - \vec K_1}{\eps}, \eps \microtime\Big) \Bigg| \; \Bigg| \hat{f}^{\sigma''}_1 \Big( \frac{\vec k - \vec K_1 + \vec G_1}{\eps}, \eps \microtime\Big) \Bigg| {\rm d} k\\
        &\le
        C\eps^2 \sum_{\vec{G}_1 \ne 0} \Big( \frac{\eps}{|\vec{G}_1|}\Big)^{2+\sdpt}\eps^2 \norm{f_1 (\cdot, \eps \microtime)}_{H^{4+\sdpt}} \norm{f_1 (\cdot, \eps \microtime)}_{H^{2}},
    \end{align*}}  
    with the above sum over two-dimensional vectors $G_1$ finite since $\sdpt > 0$. Plugging the above bounds for the $G_1 = 0$ and $G_1 \ne 0$ cases into \eqref{eq:combining}, we conclude that
    \srq{\begin{align*}
        \norm{\widetilde{\wfntwist} (\cdot, \microtime) - \widetilde{H_{11} \phi_1} (\cdot, \microtime)}^2_{L^2 (\Gamma_1^*)} &\le C (\eps^6 \norm{f_1 (\cdot, \eps \microtime)}^2_{H^3} + \eps^{6+\sdpt}\norm{f_1 (\cdot, \eps \microtime)}_{H^{4+\sdpt}} \norm{f_1 (\cdot, \eps \microtime)}_{H^{2}})\\
        &\le C \eps^6 \norm{f_1 (\cdot, \eps \microtime)}_{H^{4+\sdpt}}^2
    \end{align*}}
    and the proof is complete.
\end{proof}

\begin{lemma}[Derivation of \eqref{eq:tildeHperp}]\label{lemma:tildeHperp}
    The Bloch transform of $\Hperp \phi_2$ is
    \begin{align}\label{eq:tildeHperp_pf}
    \begin{split}
        (\widetilde{\Hperp \phi_2})^\sigma (k_1, \microtime)
        &= \frac{1}{\eps |\Gamma|^2}e^{-i\eshift \microtime}\sum_{G_2, G_2' \in \cR_2^*} \sum_{\sigma'\in \{A,B\}} e^{i(\cG_1 (G_2) \cdot \tau^\sigma_1 + (G'_2- G_2) \cdot \tau^{\sigma'}_2)} \hathperp (k_1 + \cG_1 (G_2); \eps)\\
        &\hspace{6cm} \times \hat{f}^{\sigma'}_2 \Big(\frac{k_1 + \cG_1 (G_2) - G_2 - K_2 + G'_2}{\eps}, \eps \microtime \Big).
    \end{split}
    \end{align}
\end{lemma}
\begin{proof}
    By equation (44) in \cite{watson2023bistritzer}, we know that
    \begin{align*}
        (\widetilde{\Hperp \phi_2})^\sigma (k_1, \microtime)= \frac{1}{|\Gamma|}\sum_{G_2 \in \cR_2^*} \sum_{G_1 \in \cR_1^*}\int_{\Gamma_2^*} \sum_{\sigma'} e^{i(G_1 \cdot \tau^\sigma_1 - G_2 \cdot \tau^{\sigma'}_2)} \delta (k_1+G_1-k_2-G_2) \hathperp (k_2 + G_2; \eps) \tilde{\phi}^{\sigma'}_2 (k_2, \microtime) {\rm d} k_2.
    \end{align*}
    Observe that $k_1$ and $k_2$ in the above integrand are respectively restricted to $\Gamma_1^*$ and $\Gamma_2^*$. Therefore, for every $G_2 \in \cR_2^*$, there exists a unique $G_1 \in \cR_1^*$ such that the delta-function does not identically vanish.
    We can therefore write
    \begin{align}\label{eq:tildeHperp_pf2}
    \begin{split}
        (\widetilde{\Hperp \phi_2})^\sigma (k_1, \microtime)&= 
        \frac{1}{|\Gamma|} \sum_{G_2 \in \cR_2^*} \int_{\Gamma_2^*} \sum_{\sigma'} e^{i(\cG_1 (G_2) \cdot \tau^\sigma_1 - G_2 \cdot \tau^{\sigma'}_2)} \delta (k_1+\cG_1 (G_2)-k_2-G_2) \\
        &\hspace{7cm} \times \hathperp (k_2 + G_2; \eps) \tilde{\phi}^{\sigma'}_2 (k_2, \microtime) {\rm d} k_2,
    \end{split}
    \end{align}
    where $\cG_1 (G_2)$ denotes this unique value of $G_1$ for a given $G_2$. The layer-$2$ analogue of \eqref{eq:phi_Bloch} is
    \begin{align*}
        \tilde{\phi}_2^\sigma (k_2,\microtime) =
        \frac{1}{\eps |\Gamma|} e^{-i\eshift \microtime} \sum_{G_2 \in \cR_2^*} \hat{f}_2^\sigma \left(\frac{k_2-K_2+G_2}{\eps}, \eps \microtime\right) e^{i G_2 \cdot \tau_2^\sigma},
    \end{align*}
    which after carrying out the integral over $k_2$ in \eqref{eq:tildeHperp_pf2} implies \eqref{eq:tildeHperp_pf}.
\end{proof}

\begin{lemma}[Derivation of \eqref{eq:tildeF12}]\label{lemma:tildeF12}
    The Bloch transform of $\sq{\wfnptwist}$ is
    \begin{align}\label{eq:tildeF12_pf}
    \begin{split}
        (\widetilde{\sq{\wfnptwist}})^\sigma (k_1,\microtime) &= \frac{1}{\eps |\Gamma|^2}e^{-i\eshift \microtime}\sum_{G_2, G_2' \in \cR_2^*} \sum_{\sigma'\in \{A,B\}} e^{i(\cG_1 (G_2) \cdot \tau^\sigma_1 + (G'_2- G_2) \cdot \tau^{\sigma'}_2)}\\
        &\hspace{0.5cm} \times \hat{\mathscr{h}}_{12} (k_1 + \cG_1 (G_2), K_2 + G_2 - G'_2; \eps)
        \hat{f}^{\sigma'}_2 \Big(\frac{k_1 + \cG_1 (G_2) - G_2 - K_2 + G'_2}{\eps},\eps \microtime \Big).
    \end{split}
    \end{align}
\end{lemma}
\begin{proof}
    Recall \srq{from \eqref{eq:beta_def}} that $\sq{2} \sin (\theta/2) = \beta \eps$ and thus $$\eps^{-1} (\rot_{\theta/2} - \rot_{-\theta/2}) = \beta \begin{pmatrix}
        0 & -1\\
        1 & 0
    \end{pmatrix}.$$ 
    \sq{Recalling the expression \eqref{eq:nabla_hathperp} for $\nabla \hathperp (k; \eps)$, 
    it follows from the definition 
    \eqref{eq:wfnptwist} of $\HBMptwist$ that}
    \begin{align*}
        (\sq{\wfnptwist})^\sigma_{R_1}(\microtime) &= 
        \frac{\eps}{ |\Gamma|} e^{-i\eshift \microtime} e^{iK_1 \cdot (R_1 + \tau_1^\sigma)} \sum_{G\in \cR^*}\sum_{\sigma' \in \{A,B\}}e^{i (K + G) \cdot \eps^{-1} (\rot_{-\theta/2}-\rot_{\theta/2})\eps (R_1 + \tau_1^\sigma)} e^{iG \cdot (\tau^\sigma - \tau^{\sigma'})} \sq{e^{-iG \cdot \ls}} \\
        &\hspace{3cm} \times \Big( \hathperp (\rot_{\theta/2} (K + G);\eps ) 1\{|G+K| \le 2|K|\}f^{\sigma'}_2 (r, \eps \microtime)\\
        &\hspace{3.5cm} - i \eps \nabla \hathperp (\rot_{\theta/2} (K + G);\eps ) 1\{|G+K| \le |K|\}\nabla_r f^{\sigma'}_2 (r, \eps \microtime)\Big) \Big \vert_{r = \eps (R_1 +\tau_1^\sigma)},
    \end{align*}
    where $\cR^*, K$ and $\tau^\sigma$ respectively are the \emph{unrotated} reciprocal lattice, Dirac point and sublattice shift. \sq{To simplify the exponential factors, we write
    \begin{align*}
        &K_1 \cdot (R_1 + \tau_1^\sigma) + (K+G) \cdot (\rot_{-\theta/2} - \rot_{\theta/2}) (R_1 + \tau_1^\sigma) + G\cdot (\tau^\sigma - \tau^{\sigma'})-G \cdot \ls\\
        &\hspace{7cm} = (K_2 + G_2) \cdot (R_1 + \tau^\sigma_1) - G \cdot \tau^{\sigma'} - \frac{1}{2}G \cdot \ls - G_1 \cdot R_1,
    \end{align*}
    where 
    $G_1 = \rot_{-\theta/2} G$ and $G_2 = \rot_{\theta/2} G$. Recalling that $G_1 \cdot R_1 \in 2\pi \mathbb{Z}$ and $G \cdot \tau^{\sigma'} = G_2 \cdot \tau_2^{\sigma'} - \frac{1}{2} G \cdot \ls$,}
    it then follows from the definition
    $
        (\widetilde{\sq{\wfnptwist}})^\sigma (k_1,\microtime) := \sum_{R_1 \in \cR_1} e^{-ik_1 \cdot (R_1 + \tau_1^\sigma)}(\sq{\wfnptwist})^\sigma_{R_1}(\microtime)
    $
    that
    \begin{align*}
        (\widetilde{\sq{\wfnptwist}})^\sigma (k_1,\microtime) &= 
        \frac{\eps}{ |\Gamma|} e^{-i\eshift \microtime} \sum_{R_1 \in \cR_1} \sum_{G_2\in \cR_2^*}\sum_{\sigma' \in \{A,B\}}e^{-i(k_1 - K_2 - G_2) \cdot (R_1 + \tau_1^\sigma)} e^{-iG_2 \cdot \tau_2^{\sigma'}}\\
        &\hspace{3cm} \times \Big( \hathperp (K_2 + G_2;\eps ) 1\{|G_2+K_2| \le 2|K_2|\}f^{\sigma'}_2 (r, \eps \microtime)\\
        &\hspace{3.5cm} - i \eps \nabla \hathperp (K_2 + G_2;\eps ) 1\{|G_2+K_2| \le |K_2|\}\nabla_r f^{\sigma'}_2 (r, \eps \microtime)\Big) \Big \vert_{r = \eps (R_1 +\tau_1^\sigma)}.
    \end{align*}
    By the Poisson summation formula, this is equal to
    \begin{align*}
        &\frac{1}{\eps |\Gamma|^2}e^{-i\eshift \microtime}\int_{\mathbb{R}^2} \sum_{G_1\in \cR_1^*}\sum_{G_2\in \cR_2^*}\sum_{\sigma' \in \{A,B\}}e^{i(G_1 \cdot \tau^\sigma_1 - G_2 \cdot \tau^{\sigma'}_2)}e^{-i (k_1+G_1 - G_2 -K_2) \cdot r/\eps}\\
        &\hspace{4.0cm} \times \Big( \hathperp (K_2 + G_2;\eps ) 1\{|G_2+K_2| \le 2|K_2|\}f^{\sigma'}_2 (r, \eps \microtime)\\
        &\hspace{4.5cm} - i \eps \nabla \hathperp (K_2 + G_2;\eps ) 1\{|G_2+K_2| \le |K_2|\}\nabla_r f^{\sigma'}_2 (r, \eps \microtime)\Big){\rm d} r.
    \end{align*}
    Changing variables in the sums and using that $\cG_1 : \cR_2^* \rightarrow \cR_1^*$ is a bijection, we obtain
    \begin{align*}
        (\widetilde{\sq{\wfnptwist}})^\sigma (k_1,\microtime) &=
        \frac{1}{\eps |\Gamma|^2}e^{-i\eshift \microtime}\int_{\mathbb{R}^2} \sum_{G_2, G_2' \in \cR_2^*}\sum_{\sigma' \in \{A,B\}}e^{i(\cG_1 (G_2) \cdot \tau^\sigma_1 - (G_2-G_2') \cdot \tau^{\sigma'}_2)}e^{-i (k_1+\cG_1 (G_2) - G_2 +G_2' -K_2) \cdot r/\eps}\\
        &\hspace{2cm} \times \Big( \hathperp (K_2 + G_2 - G_2';\eps ) 1\{|G_2 - G_2'+K_2| \le 2|K_2|\}f^{\sigma'}_2 (r, \eps \microtime)\\
        &\hspace{2.5cm} - i \eps \nabla \hathperp (K_2 + G_2 - G_2';\eps ) 1\{|G_2 - G_2'+K_2| \le |K_2|\}\nabla_r f^{\sigma'}_2 (r, \eps \microtime)\Big){\rm d} r.
    \end{align*}
    Integrating by parts in $r$, this becomes
    \begin{align*}
        (\widetilde{\sq{\wfnptwist}})^\sigma (k_1,\microtime) &=
        \frac{1}{\eps |\Gamma|^2}e^{-i\eshift \microtime}\int_{\mathbb{R}^2} \sum_{G_2, G_2' \in \cR_2^*}\sum_{\sigma' \in \{A,B\}}e^{i(\cG_1 (G_2) \cdot \tau^\sigma_1 - (G_2-G_2') \cdot \tau^{\sigma'}_2)}e^{-i (k_1+\cG_1 (G_2) - G_2 +G_2' -K_2) \cdot r/\eps}\\
        &\hspace{5.5cm} \times \hat{\mathscr{h}}_{12} (k_1 + \cG_1 (G_2), K_2 + G_2 - G'_2;\eps) f^{\sigma'}_2 (r, \eps \microtime){\rm d} r.
    \end{align*}
    Evaluating the integral over $r$, we obtain \eqref{eq:tildeF12_pf} as desired.
\end{proof}

\begin{lemma}[Proof of \eqref{eq:bilayer1}]\label{lemma:bilayer1}
    Take 
    \srq{$0 < \sdpt_- < \sdpt \le 1$} as in Theorem \ref{thm:main2}. Then
    \begin{align}\label{eq:bilayer1_pf}
        \norm{\sq{\wfnptwist}(\microtime) - \Hperp \phi_2 (\microtime)}_{\ell^2 (\cR_1)} \le C \eps^{2 + \sdpt_-}\sq{\norm{f_2 (\cdot, \eps \microtime)}_{H^{4+\sdpt}}}
    \end{align}
    uniformly in $0 < \eps < 1$ and $\microtime \ge 0$.
\end{lemma}
\begin{proof}
    Since $\norm{u}_{\ell^2 (\cR_1)} = |\Gamma^*|^{-1/2} \norm{\tilde{u}}_{L^2 (\Gamma_1^*)}$ for any $u \in \ell^2 (\cR_1)$, it suffices to show that 
    \begin{align}\label{eq:bilayer1_suffices}
        \norm{\widetilde{\sq{\wfnptwist}}(\cdot,\microtime) - \widetilde{\Hperp \phi_2} (\cdot,\microtime)}_{L^2 (\Gamma_1^*)} \le C \eps^{2 + \sdpt_-}\sq{\norm{f_2 (\cdot, \eps \microtime)}_{H^{4+\sdpt}}}
    \end{align}
    uniformly in $0 < \eps < 1$ and $\microtime \ge 0$.
    Lemmas \ref{lemma:tildeHperp} and \ref{lemma:tildeF12} imply that
    \begin{align*}
        &\widetilde{(\sq{\wfnptwist})^\sigma}(k_1,\microtime) - \widetilde{(\Hperp \phi_2)^\sigma} (k_1,\microtime) =
        \frac{1}{\eps |\Gamma|^2}e^{-i\eshift \microtime}\sum_{G_2, G_2' \in \cR_2^*} \sum_{\sigma'\in \{A,B\}} e^{i(\cG_1 (G_2) \cdot \tau^\sigma_1 + (G'_2- G_2) \cdot \tau^{\sigma'}_2)}\\
        &\hspace{0.5cm} \times (\hat{\mathscr{h}}_{12} (k_1 + \cG_1 (G_2), K_2 + G_2 - G'_2;\eps) -\hathperp (k_1 + \cG_1 (G_2); \eps))
        \hat{f}^{\sigma'}_2 \Big(\frac{k_1 + \cG_1 (G_2) - G_2 - K_2 + G'_2}{\eps},\eps \microtime \Big).
    \end{align*}
    Recall that $\hat{\mathscr{h}}_{12} (k,q)$ vanishes if $|q|$ is sufficiently large, hence it will be useful to split the above sum over $G_2, G_2'$ into subsets of $\cR_2^* \oplus \cR_2^*$ depending on $|G_2 - G_2'+K_2|$. To this end, define the sets $\Lambda_1, \Lambda_2, \Lambda_3 \subset \cR_2^* \oplus \cR_2^*$ by
    \begin{align*}
        \Lambda_j := \{(G_2, G_2') \in \cR_2^* \oplus \cR_2^* : G_2 - G_2' \in \Lambda_j^1\},
    \end{align*}
    where the sets $\Lambda_1^1, \Lambda_2^1, \Lambda_3^1 \subset \cR_2^*$ are given by
    \begin{align*}
        \Lambda_1^1 &:= \{G_2 \in \cR_2^*: |G_2+ K_2| = |K_2|\},\qquad
        \Lambda_2^1 := \{G_2 \in \cR_2^*: |G_2+ K_2| = 2|K_2|\},\\
        \Lambda_3^1 &:= \{G_2 \in \cR_2^*: |G_2+ K_2| \ge \sqrt{7}|K_2|\}.
    \end{align*}
    It is a well-known property of the honeycomb lattice that $\Lambda_1^1 \cup \Lambda_2^1 \cup \Lambda_3^1 = \cR_2^*$, and thus $\Lambda_1 \cup \Lambda_2 \cup \Lambda_3 = \cR_2^* \oplus \cR_2^*$.
    We can then write
    \begin{align}\label{eq:Delta_sum}
        \widetilde{\sq{\wfnptwist}}(k_1,\microtime) - \widetilde{\Hperp \phi_2} (k_1,\microtime) = \Delta_1 (k_1, \microtime) + \Delta_2 (k_1, \microtime) + \Delta_3 (k_1, \microtime),
    \end{align}
    where for each $j$,
    \begin{align}\label{eq:Deltaj_def}
    \begin{split}
        \Delta_j^\sigma (k_1, \microtime) &:= \frac{1}{\eps |\Gamma|^2}e^{-i\eshift \microtime}\sum_{(G_2, G_2') \in \Lambda_j} \sum_{\sigma'\in \{A,B\}} e^{i(\cG_1 (G_2) \cdot \tau^\sigma_1 + (G'_2- G_2) \cdot \tau^{\sigma'}_2)}\\
        &\hspace{-2.25cm} \times (\hat{\mathscr{h}}_{12} (k_1 + \cG_1 (G_2), K_2 + G_2 - G'_2;\eps) -\hathperp (k_1 + \cG_1 (G_2); \eps))
        \hat{f}^{\sigma'}_2 \Big(\frac{k_1 + \cG_1 (G_2) - G_2 - K_2 + G'_2}{\eps},\eps \microtime \Big).
    \end{split}
    \end{align}
    Regarding $\Delta_1$, we 
    recall the definition \eqref{eq:filtered_Taylor} of $\hat{\mathscr{h}}_{12}$ which states that
    \begin{align}\label{eq:Taylor1}
    \begin{split}
        &\hat{\mathscr{h}}_{12} (k_1 + \cG_1 (G_2), K_2 + G_2 - G'_2;\eps) \\
        &\hspace{2cm} =
        \hathperp (K_2 + G_2 - G'_2;\eps) + (k_1 + \cG_1 (G_2) -(K_2 + G_2 - G'_2)) \cdot \nabla \hathperp (K_2 + G_2 - G'_2;\eps)
    \end{split}
    \end{align}
    for all $(G_2, G_2') \in \Lambda_1$. Similalry, the terms $\Delta_2$ and $\Delta_3$ can be written explicitly as
    \begin{align}\label{eq:Deltas}
    \begin{split}
        \Delta_2^\sigma (k_1, \microtime) &= \frac{1}{\eps |\Gamma|^2}e^{-i\eshift \microtime}\sum_{(G_2, G_2') \in \Lambda_2} \sum_{\sigma'\in \{A,B\}} e^{i(\cG_1 (G_2) \cdot \tau^\sigma_1 + (G'_2- G_2) \cdot \tau^{\sigma'}_2)}\\
        &\hspace{0.0cm} \times (\hathperp (K_2 + G_2 - G'_2;\eps) -\hathperp (k_1 + \cG_1 (G_2); \eps))
        \hat{f}^{\sigma'}_2 \Big(\frac{k_1 + \cG_1 (G_2) - G_2 - K_2 + G'_2}{\eps},\eps \microtime \Big)\\
        \Delta_3^\sigma (k_1, \microtime) &= -\frac{1}{\eps |\Gamma|^2}e^{-i\eshift \microtime}\sum_{(G_2, G_2') \in \Lambda_3} \sum_{\sigma'\in \{A,B\}} e^{i(\cG_1 (G_2) \cdot \tau^\sigma_1 + (G'_2- G_2) \cdot \tau^{\sigma'}_2)}\\
        &\hspace{4.0cm} \times \hathperp (k_1 + \cG_1 (G_2); \eps)
        \hat{f}^{\sigma'}_2 \Big(\frac{k_1 + \cG_1 (G_2) - G_2 - K_2 + G'_2}{\eps},\eps \microtime \Big).
    \end{split}
    \end{align}
    We will prove the lemma by bounding each of the $\Delta_j$ separately. 

    \medskip
    \textbf{I. We begin with $\Delta_3$.} After changing variables $G_1 \leftarrow -\cG_1 (G_2)$, $G_2 \leftarrow G_2-G_2'$ in the sum, we obtain
    \begin{align*}
        \Delta_3^\sigma (k_1, \microtime) &= -\frac{1}{\eps |\Gamma|^2}e^{-i\eshift \microtime}\sum_{G_1 \in \cR_1^*} \sum_{G_2 \in \Lambda_3^1} \sum_{\sigma'\in \{A,B\}} e^{-i(G_1 \cdot \tau^\sigma_1 + G_2 \cdot \tau^{\sigma'}_2)}\\
        &\hspace{4.5cm} \times \hathperp (k_1 - G_1; \eps)
        \hat{f}^{\sigma'}_2 \Big(\frac{k_1 - G_1 - G_2 - K_2}{\eps},\eps \microtime \Big),
    \end{align*}
    which implies that
    \begin{align*}
        \norm{\Delta_3 (\cdot, \microtime)}^2_{L^2 (\Gamma_1^*)} &\le 
        \frac{C}{\eps^2}\int_{\Gamma_1^*} \sum_{G_1,G_1' \in \cR_1^*} \sum_{G_2,G_2' \in \Lambda_3^1}\sum_{\sigma,\sigma'\in \{A,B\}}\left| \hathperp (k_1 - G_1; \eps) \right| \left| \hathperp (k_1 - G_1'; \eps) \right|\\
        &\hspace{2cm} \times 
        \left| \hat{f}^{\sigma}_2 \Big(\frac{k_1 - G_1 - G_2 - K_2}{\eps},\eps \microtime \Big) \right|
        \left| \hat{f}^{\sigma'}_2 \Big(\frac{k_1 - G_1' - G_2' - K_2}{\eps},\eps \microtime \Big) \right| {\rm d}k_1.
    \end{align*}
    Combining the sum over $G_1$ with the integral over $k_1$, we obtain
    \begin{align*}
        \norm{\Delta_3 (\cdot, \microtime)}^2_{L^2 (\Gamma_1^*)} &\le 
        \frac{C}{\eps^2}\int_{\mathbb{R}^2} \sum_{G_1' \in \cR_1^*} \sum_{G_2,G_2' \in \Lambda_3^1}\sum_{\sigma,\sigma'\in \{A,B\}}\left| \hathperp (k_1; \eps) \right| \left| \hathperp (k_1 - G_1'; \eps) \right|\\
        &\hspace{2cm} \times 
        \left| \hat{f}^{\sigma}_2 \Big(\frac{k_1 - G_2 - K_2}{\eps},\eps \microtime \Big) \right|
        \left| \hat{f}^{\sigma'}_2 \Big(\frac{k_1 - G_1' - G_2' - K_2}{\eps},\eps \microtime \Big) \right| {\rm d}k_1,
    \end{align*}
    which after the change of variables $k_1 \leftarrow k_1 - G_2 - K_2$ becomes
    \begin{align}\label{eq:beta2}
    \begin{split}
        \norm{\Delta_3 (\cdot, \microtime)}^2_{L^2 (\Gamma_1^*)} &\le 
        \frac{C}{\eps^2}\int_{\mathbb{R}^2} \sum_{G_1 \in \cR_1^*} \sum_{G_2,G_2' \in \Lambda_3^1}\sum_{\sigma,\sigma'\in \{A,B\}}\left| \hathperp (k_1+G_2+K_2; \eps) \right| \\
        &\hspace{0.0cm} \times 
        \left| \hathperp (k_1+G_2+K_2- G_1; \eps) \right|
        \left| \hat{f}^{\sigma}_2 \Big(\frac{k_1}{\eps},\eps \microtime \Big) \right|
        \left| \hat{f}^{\sigma'}_2 \Big(\frac{k_1+G_2 - G_1 - G_2'}{\eps},\eps \microtime \Big) \right| {\rm d}k_1.
    \end{split}
    \end{align}
    \srq{Take $0 < \swp \le 1$ from Assumption \ref{assumption:hperp2}. Without loss, assume $\swp$ is sufficiently small} such that $\sdpt (1-\nu) - 2\nu \ge \sdpt_-$.
    For fixed $\vec{G}_1, \vec{G}_2, \vec{G}'_2$, define
    \begin{align*}
        S_1^0 :&= \{\vec{k} : |\vec{k}| < |\vec{G}_1 - \vec{G}_2 + \vec{G}'_2|/2 \}, \qquad \qquad
        S_2^0 := \{k:|\vec{k}| < \nu |\vec{K}_2 + \vec{G}_2| \}, \\
        S_3^0 :&= \{k:|\vec{k}- \vec{G}_1 + \vec{G}_2 - \vec{G}'_2| < \nu |\vec{K}_2 + \vec{G}'_2| \}, \qquad \qquad S_j^1 := S_j^c, \quad j \in \{1,2,3\}.
    \end{align*}
    We will partition $\mathbb{R}^2$ into the eight sets $S_{i_1, i_2, i_3} := S_1^{i_1} \cap S_2^{i_2} \cap S_3^{i_3}$, with the $i_j$ in $\{0,1\}$.
    By a straightforward application of the triangle inequality, we see that
    \begin{align}\label{eq:S_bounds}
        \begin{split}
        |k + K_2 + G_2| \ge (1-\nu) |K_2+G_2| \quad &\text{for all} \; k \in S_2^0, \\
        |k+K_2-G_1+G_2| \ge (1-\nu) |K_2+G_2'| \quad &\text{for all} \; k \in S_3^0.
        \end{split}
    \end{align}
    Set $\srp := \sdpt-\sdpt_-$.
    \srq{By the decay of $\hathperp$ in \eqref{eq:4bounds_full},} 
    it follows from \eqref{eq:S_bounds} that for all $G_2, G_2' \in \Lambda_3^1$,
    \begin{align}\label{eq:normell2}
    \begin{split}
        \left| \hathperp (k+G_2+K_2; \eps) \right| \le C \eps^{\frac{2+\sdpt}{\sqrt{7}} (1-\nu)|K_2+G_2|/|K|} \le C \eps^{(2+\sdpt) (1-\nu)}|K_2+G_2|^{-2-\srp} \quad &\text{for all} \; k \in S_2^0,\\
        \left| \hathperp (k+G_2+K_2- G_1; \eps) \right|\le C \eps^{\frac{2+\sdpt}{\sqrt{7}} (1-\nu)|K_2+G_2'|/|K|}\le C \eps^{(2+\sdpt) (1-\nu)}|K_2+G_2'|^{-2-\srp} \quad &\text{for all} \; k \in S_3^0,
        \end{split}
    \end{align}
    where we used the fact that $|K_2 + G_2| \ge \sqrt{7}|K|$ to obtain the last inequality on each line. 

    With $I$ denoting the right-hand side of \eqref{eq:beta2},
    we now write $I = C\sum_{i_1, i_2, i_3} I_{i_1, i_2, i_3}$, where
    \begin{align*}
        I_{i_1, i_2, i_3} &:= \frac{1}{\eps^2}\int_{S_{i_1,i_2,i_3}} \sum_{G_1 \in \cR_1^*} \sum_{G_2,G_2' \in \Lambda_3^1}\sum_{\sigma,\sigma'\in \{A,B\}}\left| \hathperp (k+G_2+K_2; \eps) \right| \\
        &\hspace{0.5cm} \times 
        \left| \hathperp (k+G_2+K_2- G_1; \eps) \right|
        \left| \hat{f}^{\sigma}_2 \Big(\frac{k}{\eps},\eps \microtime \Big) \right|
        \left| \hat{f}^{\sigma'}_2 \Big(\frac{k+G_2 - G_1 - G_2'}{\eps},\eps \microtime \Big) \right| {\rm d}k
    \end{align*}
    restricts the integral over $k$ to $S_{i_1, i_2, i_3}$.
    If $\vec{k} \in S_{0,0,0}$, then the above integrand is bounded by 
    \begin{align*}
        &C \frac{1 + 2^{2+\delta} |\vec{k} - \vec{G}_1 + \vec{G}_2 -\vec{G}'_2|^{2+\delta}}{1 + |\vec{G}_1 - (\vec{G}_2 - \vec{G}'_2) |^{2+\delta}}\eps^{2(2+\sdpt) (1-\nu)}
        \absval{\vec{K}_2+\vec{G}_2}^{-2-\delta}
        \absval{\vec{K}_2+\vec{G}'_2}^{-2-\delta}\\
        &\hspace{7cm}\times
        \left| \hat{f}^{\sigma}_2 \left( \frac{\vec{k}}{\eps}, \eps \microtime\right) \right| \left| \hat{f}^{\sigma'}_2 \left( \frac{\vec{k} - \vec{G}_1 + \vec{G}_2 -\vec{G}'_2}{\eps}, \eps \microtime\right) \right|.
    \end{align*}
    Indeed, the factors in the first line are justified by \eqref{eq:normell2} and the fact that
    \begin{align}\label{eq:S10_bound}
        |\vec{G}_1 - (\vec{G}_2 - \vec{G}'_2) | \le 2|\vec{k} - \vec{G}_1 + \vec{G}_2 -\vec{G}'_2| \quad \text{for all} \; k \in S^0_1.
    \end{align}
    It follows that
    \begin{align}\label{eq:I_000_2}
    \begin{split}
        I_{0,0,0} &\le C \eps^{-2} \eps^{2(2+\sdpt)(1-\nu)}\sum_{\vec{G}_1\in \cR_1^*} \sum_{G_2, G_2' \in \Lambda_3^1}\sum_{\sigma,\sigma' \in \{A,B\}} \int_{S_{0,0,0}}
        \frac{1 + 2^{2+\delta} |\vec{k} - \vec{G}_1 + \vec{G}_2 -\vec{G}'_2|^{2+\delta}}{1 + |\vec{G}_1 - (\vec{G}_2 - \vec{G}'_2) |^{2+\delta}}\\
        &\hspace{2cm} \times 
        \absval{\vec{K}_2+\vec{G}_2}^{-2-\delta}
        \absval{\vec{K}_2+\vec{G}'_2}^{-2-\delta}\left| \hat{f}^{\sigma}_2 \left( \frac{\vec{k}}{\eps}, \eps \microtime \right) \right| \left| \hat{f}^{\sigma'}_2 \left( \frac{\vec{k} - \vec{G}_1 + \vec{G}_2 -\vec{G}'_2}{\eps}, \eps \microtime\right) \right|
        {\rm d} k\\
        &\le C \eps^{2(1+\sdpt)(1-\nu)-2\nu} \eps \norm{f_2 (\cdot, \eps \microtime)}_{L^2} (\eps \norm{f_2 (\cdot, \eps \microtime)}_{L^2} + \eps^{3+\delta} 2^{2+\delta} \norm{f_2 (\cdot, \eps \microtime)}_{H^{2+\delta}})\\
        &\hspace{6cm} \times\sum_{\vec{G}_1\in \cR_1^*} \sum_{G_2, G_2' \in \Lambda_3^1} 
        \frac{\absval{\vec{K}_2+\vec{G}_2}^{-2-\delta}
        \absval{\vec{K}_2+\vec{G}'_2}^{-2-\delta}}{1 + |\vec{G}_1 - (\vec{G}_2 - \vec{G}'_2) |^{2+\delta}}\\
        &\le C \eps^{1 + (1+\sdpt)(1-\nu)-\nu}\norm{f_2 (\cdot, \eps \microtime)}_{L^2}\\
        &\hspace{2cm} \times (\eps^{1 + (1+\sdpt)(1-\nu)-\nu} \norm{f_2 (\cdot, \eps \microtime)}_{L^2} + \eps^{3+\delta+(1+\sdpt)(1-\nu)-\nu} \norm{f_2 (\cdot, \eps \microtime)}_{H^{2+\delta}}).
        \end{split}
    \end{align}
    Indeed, the sum over $G_1, G_2, G'_2$ is finite because $\sum_{G_1 \in \cR_1^*} \frac{1}{1 + |\vec{G}_1 - (\vec{G}_2 - \vec{G}'_2) |^{2+\delta}} \le C\sum_{G_1\in \cR_1^*} \frac{1}{1 + |\vec{G}_1|^{2+\delta}}$ uniformly in $G_2,G_2'\in\cR_2^*$, and $\sum_{G_2\in \cR_2^*} \absval{K_2 + G_2}^{-2-\delta} < \infty$ (recall that $\absval{K_2+G_2} \ge \absval{K}$ for all $G_2 \in \cR_2^*$).
    Since our choice of $\nu$ ensures that $(1+\sdpt)(1-\nu)-\nu \ge 1+\sdpt_-$, we get
    \begin{align*}
        I_{0,0,0} \le C \eps^{2 + \sdpt_-}\norm{f_2 (\cdot, \eps \microtime)}_{L^2} (\eps^{2 + \sdpt_-} \norm{f_2 (\cdot, \eps \microtime)}_{L^2} + \eps^{4+\delta+\sdpt_-} \norm{f_2 (\cdot, \eps \microtime)}_{H^{2+\delta}}).
    \end{align*}

    Similarly, using the \srq{uniform boundedness of $\hathperp$ in \eqref{eq:4bounds_full} and the} fact that \sq{$\nu^{2+\sdpt_-} |\vec{K}_2 + \vec{G}_2|^{2+\sdpt_-}\le |k|^{2+\sdpt_-}$} for all $k \in S_2^1$, we find
    \begin{align}\label{eq:I_010_2}
    \begin{split}
        I_{0,1,0} &\le C \eps^{-2} \sum_{G_1 \in \cR_1^*}
        \sum_{\vec{G}_2, \vec{G}'_2 \in \Lambda_3^1} \sum_{\sigma,\sigma' \in \{A,B\}}\int_{S_{0,1,0}}\frac{1 + 2^{2+\delta} |\vec{k} - \vec{G}_1 + \vec{G}_2 -\vec{G}'_2|^{2+\delta}}{1 + |\vec{G}_1 - (\vec{G}_2 - \vec{G}'_2) |^{2+\delta}}
        \sq{\frac{|\vec{k}|^{2+\sdpt_-}}{\nu^{2+\sdpt_-} |\vec{K}_2 + \vec{G}_2|^{2+\sdpt_-}}}\\
        &\hspace{2.5cm} \times 
        \eps^{(2+\sdpt) (1-\nu)} \absval{\vec{K}_2+\vec{G}'_2}^{-2-\srp}\left| \hat{f}^{\sigma}_2 \left( \frac{\vec{k}}{\eps}, \eps \microtime \right) \right| \left| \hat{f}^{\sigma'}_2 \left( \frac{\vec{k} - \vec{G}_1 + \vec{G}_2 -\vec{G}'_2}{\eps}, \eps \microtime \right) \right|
        {\rm d} k\\
        &\le C \eps^{\sdpt (1-\nu) - 2\nu} \sq{\nu^{-2-\sdpt_-}}
        \sum_{G_1 \in \cR_1^*}
        \sum_{\vec{G}_2, \vec{G}'_2 \in \Lambda_3^1} \sum_{\sigma,\sigma' \in \{A,B\}}
        \frac{\sq{\absval{K_2 + G_2}^{-2-\sdpt_-}} \absval{K_2 + G_2'}^{-2-\srp}}{1+\absval{G_1 - (G_2 - G_2')}^{2+\srp}}
        \\
        & \hspace{2.5cm} \times\norm{ \sq{\absval{\cdot}^{2+\sdpt_-}} \hat{f}^{\sigma}_2 \left( \frac{\cdot}{\eps}, \eps \microtime \right)}_{L^2 (\mathbb{R}^2)} \Bigg( \norm{\hat{f}^{\sigma'}_2 \left( \frac{\cdot - G_1 + G_2 - G'_2}{\eps}, \eps \microtime\right)}_{L^2 (\mathbb{R}^2)} \\
        & \hspace{3.5cm} + 2^{2+\srp} \norm{\absval{\cdot - G_1 + G_2 - G'_2}^{2+\srp} \hat{f}^{\sigma '}_2 \left( \frac{\cdot - G_1 + G_2 - G'_2}{\eps}, \eps \microtime\right)}_{L^2 (\mathbb{R}^2)} \Bigg)\\
        &= C \eps^{\sdpt (1-\nu) - 2\nu} \sq{\nu^{-2-\sdpt_-}} 
         \sum_{G_1 \in \cR_1^*}
        \sum_{\vec{G}_2, \vec{G}'_2 \in \Lambda_3^1} \sum_{\sigma,\sigma' \in \{A,B\}}
        \frac{\sq{\absval{K_2 + G_2}^{-2-\sdpt_-}} \absval{K_2 + G_2'}^{-2-\srp}}{1+\absval{G_1 - (G_2 - G_2')}^{2+\srp}}
        \\
        & \hspace{0.5cm} \times \norm{ \sq{\absval{\cdot}^{2+\sdpt_-}} \hat{f}^{\sigma}_2 \left( \frac{\cdot}{\eps}, \eps \microtime \right)}_{L^2 (\mathbb{R}^2)} \Bigg( \norm{\hat{f}^{\sigma'}_2 \left( \frac{\cdot}{\eps}, \eps \microtime\right)}_{L^2 (\mathbb{R}^2)}
        + 2^{2+\srp} \norm{\absval{\cdot}^{2+\srp} \hat{f}^{\sigma '}_2 \left( \frac{\cdot}{\eps}, \eps \microtime\right)}_{L^2 (\mathbb{R}^2)} \Bigg).
        \end{split}
    \end{align}
    As before, the $L^2$-norms are independent of $G_1, G_2, G'_2$, and the sum of the first factor is finite. Therefore,
    \begin{align*}
        I_{0,1,0} &\le C \eps^{\sdpt (1-\nu) - 2\nu} \sq{\nu^{-2-\sdpt_-}}\sum_{\sigma,\sigma' \in \{A,B\}}\norm{ \sq{\absval{\cdot}^{2+\sdpt_-}} \hat{f}^{\sigma}_2 \left( \frac{\cdot}{\eps}, \eps \microtime \right)}_{L^2 (\mathbb{R}^2)}\\
        &\hspace{4cm} \times \Bigg( \norm{\hat{f}^{\sigma'}_2 \left( \frac{\cdot}{\eps}, \eps \microtime\right)}_{L^2 (\mathbb{R}^2)}
        + 2^{2+\srp} \norm{\absval{\cdot}^{2+\srp} \hat{f}^{\sigma '}_2 \left( \frac{\cdot}{\eps}, \eps \microtime\right)}_{L^2 (\mathbb{R}^2)} \Bigg)\\
        &= C \eps^{\sdpt (1-\nu) - 2\nu} \sq{\nu^{-2-\sdpt_-}}\sum_{\sigma,\sigma' \in \{A,B\}}\sq{\eps^{3+\sdpt_-} \norm{ f^{\sigma}_2 \left( \cdot, \eps \microtime \right)}_{H^{2+\sdpt_-} (\mathbb{R}^2)}}\\
        &\hspace{4cm} \times \Bigg(\eps \norm{f^{\sigma'}_2 \left(\cdot, \eps \microtime\right)}_{L^2 (\mathbb{R}^2)}
        + 2^{2+\srp}\eps^{3+\srp} \norm{f^{\sigma '}_2 \left( \cdot, \eps \microtime\right)}_{H^{2+\delta} (\mathbb{R}^2)} \Bigg).
    \end{align*}
    Recalling that $\sdpt (1-\nu) - 2\nu \ge \sdpt_-$ \srq{and $\sdpt = \srp + \sdpt_-$}, we obtain
    \begin{align*}
        I_{0,1,0} \le C \sq{\eps^{3+\sdpt_-} \norm{ f_2 \left( \cdot, \eps \microtime \right)}_{H^{2+\sdpt_-} (\mathbb{R}^2)}}\Bigg(\eps^{1+\sdpt_-} \norm{f_2 \left(\cdot, \eps \microtime\right)}_{L^2 (\mathbb{R}^2)}
        + \srq{\eps^{3+\sdpt}} \norm{f_2 \left( \cdot, \eps \microtime\right)}_{H^{2+\delta} (\mathbb{R}^2)} \Bigg).
    \end{align*}

    To bound $I_{0,1,1}$, we recall also that \sq{$\nu^{2+\sdpt_-} |K_2 + G_2'|^{2+\sdpt_-} \le |k-G_1 + G_2 - G_2'|^{2+\sdpt_-}$} for all $k \in S_3^1$. Hence
    \begin{align}\label{eq:I_011_2}
    \begin{split}
        I_{0,1,1} &\le C \eps^{-2} \sum_{G_1 \in \cR_1^*}
        \sum_{\vec{G}_2, \vec{G}'_2 \in \Lambda_3^1} \sum_{\sigma,\sigma' \in \{A,B\}}\int_{S_{0,1,1}}\frac{1 + 2^{2+\delta} |\vec{k} - \vec{G}_1 + \vec{G}_2 -\vec{G}'_2|^{2+\delta}}{1 + |\vec{G}_1 - (\vec{G}_2 - \vec{G}'_2) |^{2+\delta}}
        \sq{\frac{|\vec{k}|^{2+\sdpt_-}}{\nu^{2+\sdpt_-} |\vec{K}_2 + \vec{G}_2|^{2+\sdpt_-}}}\\
        &\hspace{2.5cm} \times 
        \sq{\frac{|k-G_1 + G_2 - G_2'|^{2+\sdpt_-}}{\nu^{2+\sdpt_-} |K_2 + G_2'|^{2+\sdpt_-}}} \left| \hat{f}^{\sigma}_2 \left( \frac{\vec{k}}{\eps}, \eps \microtime \right) \right| \left| \hat{f}^{\sigma'}_2 \left( \frac{\vec{k} - \vec{G}_1 + \vec{G}_2 -\vec{G}'_2}{\eps}, \eps \microtime \right) \right|
        {\rm d} k\\
        &\le C \eps^{-2} \sq{\nu^{-4-2\sdpt_-}}
        \sum_{G_1 \in \cR_1^*}
        \sum_{\vec{G}_2, \vec{G}'_2 \in \Lambda_3^1} \sum_{\sigma,\sigma' \in \{A,B\}}
        \frac{\sq{\absval{K_2 + G_2}^{-2-\sdpt_-} \absval{K_2 + G_2'}^{-2-\sdpt_-}}}{1+\absval{G_1 - (G_2 - G_2')}^{2+\srp}} \\
        & \hspace{1.0cm} \times\norm{ \sq{\absval{\cdot}^{2+\sdpt_-}} \hat{f}^{\sigma}_2 \left( \frac{\cdot}{\eps}, \eps \microtime \right)}_{L^2 (\mathbb{R}^2)} \Bigg( \sq{\absval{\cdot - G_1 + G_2 - G'_2}^{2+\sdpt_-}}\norm{\hat{f}^{\sigma'}_2 \left( \frac{\cdot - G_1 + G_2 - G'_2}{\eps}, \eps \microtime\right)}_{L^2 (\mathbb{R}^2)} \\
        & \hspace{3.25cm} + 2^{2+\srp} \norm{\sq{\absval{\cdot - G_1 + G_2 - G'_2}^{4+\srp+\sdpt_-}} \hat{f}^{\sigma '}_2 \left( \frac{\cdot - G_1 + G_2 - G'_2}{\eps}, \eps \microtime\right)}_{L^2 (\mathbb{R}^2)} \Bigg)\\
        &\le C \eps^{-2} \sq{\nu^{-4-2\sdpt_-} \eps^{3+\sdpt_-} \norm{f_2 (\cdot, \eps \microtime)}_{H^{2+\sdpt_-}} \left( \eps^{3+\sdpt_-} \norm{f_2 (\cdot, \eps \microtime)}_{H^{2+\sdpt_-}} + \eps^{5+\srp+\sdpt_-} \norm{f_2 (\cdot, \eps \microtime)}_{H^{4+\srp+\sdpt_-}}\right)}\\
        &=C \srq{\eps^{2+\sdpt_-} \norm{f_2 (\cdot, \eps \microtime)}_{H^{2+\sdpt_-}} \left( \eps^{2+\sdpt_-} \norm{f_2 (\cdot, \eps \microtime)}_{H^{2+\sdpt_-}} + \eps^{4+\sdpt} \norm{f_2 (\cdot, \eps \microtime)}_{H^{4+\sdpt}}\right)},
    \end{split}
    \end{align}
    where we absorbed the ($\eps$-independent) constant $\nu^{-4-2\srp}$ into $C$ in the last line, \srq{and recall that $\sdpt = \srp+ \sdpt_-$}.

    The term $I_{0,0,1}$ is handled similarly to $I_{0,1,0}$. 
    Indeed, the factors
    \begin{align*}
        \sq{\frac{|\vec{k}|^{2+\sdpt_-}}{\nu^{2+\sdpt_-} |\vec{K}_2 + \vec{G}_2|^{2+\sdpt_-}}} \qquad \text{and} \qquad \eps^{(2+\sdp) (1-\nu)} \absval{\vec{K}_2+\vec{G}'_2}^{-2-\srp}
    \end{align*}
    on the first and second lines of \eqref{eq:I_010_2} respectively get replaced by
    \begin{align*}
        \sq{\frac{|k-G_1 + G_2 - G'_2|^{2+\sdpt_-}}{\nu^{2+\sdpt_-} |K_2 + G'_2|^{2+\sdpt_-}}} \qquad \text{and} \qquad \eps^{(2+\sdp) (1-\nu)} \absval{\vec{K}_2+\vec{G}_2}^{-2-\srp}
    \end{align*}
    to yield
    \begin{align}\label{eq:I_001_2}
    \begin{split}
        I_{0,0,1} &\le C \eps^{-2} \sum_{G_1 \in \cR_1^*}
        \sum_{\vec{G}_2, \vec{G}'_2 \in \Lambda_3^1} \sum_{\sigma,\sigma' \in \{A,B\}}\int_{S_{0,1,0}}\frac{1 + 2^{2+\delta} |\vec{k} - \vec{G}_1 + \vec{G}_2 -\vec{G}'_2|^{2+\delta}}{1 + |\vec{G}_1 - (\vec{G}_2 - \vec{G}'_2) |^{2+\delta}}\\
        &\hspace{-1.0cm} \times 
        \sq{\frac{|k-G_1 + G_2 - G'_2|^{2+\sdpt_-}}{\nu^{2+\sdpt_-} |K_2 + G'_2|^{2+\sdpt_-}}}
        \eps^{(2+\sdpt) (1-\nu)} \absval{\vec{K}_2+\vec{G}_2}^{-2-\srp}\left| \hat{f}^{\sigma}_2 \left( \frac{\vec{k}}{\eps}, \eps \microtime \right) \right| \left| \hat{f}^{\sigma'}_2 \left( \frac{\vec{k} - \vec{G}_1 + \vec{G}_2 -\vec{G}'_2}{\eps}, \eps \microtime \right) \right|
        {\rm d} k,
        \end{split}
    \end{align}
    which leads to
    \begin{align*}
        I_{0,0,1} &\le C\sq{\nu^{-2-\sdpt_-}} \eps^{\sdpt (1-\nu) - 2\nu} \eps \norm{f_2 (\cdot, \eps \microtime)}_{L^2} \sq{\left( \eps^{3+\sdpt_-} \norm{f_2 (\cdot, \eps \microtime)}_{H^{2+\sdpt_-}} + \eps^{5+\srp+\sdpt_-} \norm{f_2 (\cdot, \eps \microtime)}_{H^{4+\srp+\sdpt_-}}\right)}.
    \end{align*}
    Recalling our choice of $\sdpt (1-\nu) - 2\nu \ge \sdpt_-$ \srq{and $\delta = \sdpt - \sdpt_-$}, this implies that
    \begin{align*}
        I_{0,0,1} &\le C\eps^{1+\sdpt_-} \norm{f_2 (\cdot, \eps \microtime)}_{L^2} \srq{\left( \eps^{3+\sdpt_-} \norm{f_2 (\cdot, \eps \microtime)}_{H^{2+\sdpt_-}} + \eps^{5+\sdpt} \norm{f_2 (\cdot, \eps \microtime)}_{H^{4+\sdpt}}\right)}.
    \end{align*}
    The arguments for $I_{1,i_2, i_3}$ are parallel to those for $I_{0,i_2,i_3}$.
    To bound $I_{1,i_2,i_3}$, replace the factor $$\frac{1 + 2^{2+\srp} |\vec{k} - \vec{G}_1 + \vec{G}_2 -\vec{G}'_2|^{2+\srp}}{1 + |\vec{G}_1 - (\vec{G}_2 - \vec{G}'_2) |^{2+\srp}}$$ appearing on the first lines of \eqref{eq:I_000_2}, \eqref{eq:I_010_2}, \eqref{eq:I_011_2} and \eqref{eq:I_001_2} by
    $$\frac{1 + 2^{2+\srp} |\vec{k}|^{2+\srp}}{1 + |\vec{G}_1 - (\vec{G}_2 - \vec{G}'_2) |^{2+\srp}},$$ 
    as the latter is greater than or equal to $1$ for all $k$ in $S_1^1$. Following the steps above, we then get
    \begin{align*}
        I_{1,0,0} &\le C \eps^{2 + \sdpt_-}\norm{f_2 (\cdot, \eps \microtime)}_{L^2} (\eps^{2 + \sdpt_-} \norm{f_2 (\cdot, \eps \microtime)}_{L^2} + \eps^{4+\delta+\sdpt_-} \norm{f_2 (\cdot, \eps \microtime)}_{H^{2+\delta}}),\\
        I_{1,1,0} &\le C\eps^{1+\sdpt_-} \norm{f_2 (\cdot, \eps \microtime)}_{L^2} \sq{\left( \eps^{3+\sdpt_-} \norm{f_2 (\cdot, \eps \microtime)}_{H^{2+\sdpt_-}} + \eps^{5+\sdpt} \norm{f_2 (\cdot, \eps \microtime)}_{H^{4+\sdpt}}\right)},\\
        I_{1,1,1} &\le C \sq{\eps^{2+\sdpt_-} \norm{f_2 (\cdot, \eps \microtime)}_{H^{2+\sdpt_-}} \left( \eps^{2+\sdpt_-} \norm{f_2 (\cdot, \eps \microtime)}_{H^{2+\sdpt_-}} + \eps^{4+\sdpt} \norm{f_2 (\cdot, \eps \microtime)}_{H^{4+\sdpt}}\right)},\\
        I_{1,0,1} &\le C \sq{\eps^{3+\sdpt_-} \norm{ f_2 \left( \cdot, \eps \microtime \right)}_{H^{2+\sdpt_-}}}\left(\eps^{1+\sdpt_-} \norm{f_2 \left(\cdot, \eps \microtime\right)}_{L^2}
        + \sq{\eps^{3+\sdpt}} \norm{f_2 \left( \cdot, \eps \microtime\right)}_{H^{2+\delta}} \right).
    \end{align*}
    Combining the bounds on all the $I_{i_1, i_2, i_3}$, we have shown that 
    \begin{align}\label{eq:Delta3_bd}
        \norm{\Delta_3 (\cdot, \microtime)}_{L^2 (\Gamma_1^*)}\le C \eps^{2+\sdpt_-} \srq{\norm{f_2 (\cdot, \eps \microtime)}_{H^{4+\sdpt}}}.
    \end{align}

    \medskip

    \textbf{II. We next bound $\Delta_2$.} Recalling the expression \eqref{eq:Deltas} for $\Delta_2$, 
    it follows that
    \begin{align*}
        \Delta_2^\sigma (k_1, \microtime) &= \frac{1}{\eps |\Gamma|^2}e^{-i\eshift \microtime}\sum_{G_1 \in \cR_1^*} \sum_{G_2 \in \Lambda_2^1} \sum_{\sigma'\in \{A,B\}} e^{-i(G_1 \cdot \tau^\sigma_1 + G_2 \cdot \tau^{\sigma'}_2)}\\
        &\hspace{0.5cm} \times (\hathperp (K_2 + G_2;\eps) -\hathperp (k_1 -G_1; \eps))
        \hat{f}^{\sigma'}_2 \Big(\frac{k_1 -G_1 - G_2 - K_2}{\eps},\eps \microtime \Big).
    \end{align*}
    By Taylor's theorem,
    \begin{align*}
        \left \vert\hathperp (K_2 + G_2;\eps) -\hathperp (k_1 -G_1; \eps) \right \vert \le |k_1 - G_1 - K_2 - G_2| \gradmax (K_2 + G_2, k_1 - G_1; \eps),
    \end{align*}
    where the (scalar-valued) function $\gradmax$ defined by
    \begin{align*}
        \gradmax (k,q;\eps) := \max_{0 \le s \le 1} \left \vert \nabla \hathperp (s q + (1-s) k; \eps) \right \vert
    \end{align*}
    is the maximum absolute value 
    of $\nabla \hathperp (\cdot, \eps)$ over the line segment connecting $k$ and $q$.
    It follows that
    \begin{align*}
        \norm{\Delta_2 (\cdot, \microtime)}^2_{L^2 (\Gamma_1^*)} &\le 
        \frac{C}{\eps^2}\int_{\Gamma_1^*} \sum_{G_1,G_1' \in \cR_1^*} \sum_{G_2,G_2' \in \Lambda_2^1}\sum_{\sigma,\sigma'\in \{A,B\}}|k_1 - G_1 - K_2 - G_2| \gradmax (K_2 + G_2, k_1 - G_1; \eps)\\
        &\hspace{2cm} \times |k_1 - G_1' - K_2 - G_2'| \gradmax (K_2 + G_2', k_1 - G_1'; \eps)\\
        &\hspace{2.5cm} \times
        \left| \hat{f}^{\sigma}_2 \Big(\frac{k_1 - G_1 - G_2 - K_2}{\eps},\eps \microtime \Big) \right|
        \left| \hat{f}^{\sigma'}_2 \Big(\frac{k_1 - G_1' - G_2' - K_2}{\eps},\eps \microtime \Big) \right| {\rm d}k_1.
    \end{align*}
    Combining the sum over $G_1$ with the integral over $k_1$, 
    we obtain
    \begin{align*}
        \norm{\Delta_2 (\cdot, \microtime)}^2_{L^2 (\Gamma_1^*)} &\le 
        \frac{C}{\eps^2}\int_{\mathbb{R}^2} \sum_{G_1' \in \cR_1^*} \sum_{G_2,G_2' \in \Lambda_2^1}\sum_{\sigma,\sigma'\in \{A,B\}}|k_1- K_2 - G_2| \gradmax (K_2 + G_2, k_1; \eps)|k_1 - G_1' - K_2 - G_2'| \\
        &\hspace{0.75cm} \times 
        \gradmax (K_2 + G_2', k_1 - G_1'; \eps)
        \left| \hat{f}^{\sigma}_2 \Big(\frac{k_1 - G_2 - K_2}{\eps},\eps \microtime \Big) \right|
        \left| \hat{f}^{\sigma'}_2 \Big(\frac{k_1 - G_1' - G_2' - K_2}{\eps},\eps \microtime \Big) \right| {\rm d}k_1.
    \end{align*}
    Changing variables $k_1 \leftarrow k_1 - G_2 - K_2$, we are left with
    \begin{align}\label{eq:beta2_2}
    \begin{split}
        \norm{\Delta_2 (\cdot, \microtime)}^2_{L^2 (\Gamma_1^*)} &\le 
        \frac{C}{\eps^2}\int_{\mathbb{R}^2} \sum_{G_1 \in \cR_1^*} \sum_{G_2,G_2' \in \Lambda_2^1}\sum_{\sigma,\sigma'\in \{A,B\}}|k_1| \gradmax (K_2 + G_2, k_1 + K_2 + G_2; \eps)\\
        &\hspace{2cm} \times 
        |k_1 - G_1 +G_2 - G_2'| 
        \gradmax (K_2 + G_2', k_1 + K_2 + G_2 - G_1; \eps)\\
        &\hspace{3.5cm} \times \left| \hat{f}^{\sigma}_2 \Big(\frac{k_1}{\eps},\eps \microtime \Big) \right|
        \left| \hat{f}^{\sigma'}_2 \Big(\frac{k_1 - G_1 - G_2'+G_2}{\eps},\eps \microtime \Big) \right| {\rm d}k_1.
    \end{split}
    \end{align}
    As above, 
    \srq{assume without loss that $0 < \swp \le  1$ from Assumption \ref{assumption:hperp2} satisfies} $\sdpt (1-\nu) - 2\nu \ge \sdpt_-$, and for fixed $\vec{G}_1, \vec{G}_2, \vec{G}'_2$ define
    \begin{align*}
        S_1^0 :&= \{\vec{k} : |\vec{k}| < |\vec{G}_1 - \vec{G}_2 + \vec{G}'_2|/2 \}, \qquad \qquad
        S_2^0 := \{k:|\vec{k}| < \nu |\vec{K}_2 + \vec{G}_2| \}, \\
        S_3^0 :&= \{k:|\vec{k}- \vec{G}_1 + \vec{G}_2 - \vec{G}'_2| < \nu |\vec{K}_2 + \vec{G}'_2| \}, \qquad \qquad S_j^1 := S_j^c, 
        \qquad S_{i_1, i_2, i_3} := S_1^{i_1} \cap S_2^{i_2} \cap S_3^{i_3},
    \end{align*}
    where $j \in \{1,2,3\}$ and $i_1, i_2, i_3 \in \{0,1\}$. 
    As in \eqref{eq:S_bounds}, the triangle inequality implies that
    \begin{align}\label{eq:S_bounds2}
    \begin{split}
        |s (k+K_2+G_2) + (1-s) (K_2 + G_2)| \ge (1-\nu) |K_2+G_2| \quad &\text{for all} \; k \in S_2^0 \; \text{and} \; 0 \le s \le 1, \\
        |s (k+K_2-G_1+G_2) + (1-s) (K_2 + G_2')| \ge (1-\nu) |K_2+G_2'| \quad &\text{for all} \; k \in S_3^0 \; \text{and} \; 0 \le s \le 1.
    \end{split}
    \end{align}
    \srq{The decay of $\nabla \hathperp$ in \eqref{eq:4bounds_full}} then implies that for all $G_2, G_2' \in \Lambda_2^1$,
    \begin{align}\label{eq:normell2_2}
    \begin{split}
        \left| \gradmax (K_2 + G_2, k+K_2 + G_2; \eps) \right| \le C \eps^{\frac{1+\sdpt}{2} (1-\nu)|K_2+G_2|/|K|} = C \eps^{(1+\sdpt) (1-\nu)} \quad &\text{for all} \; k \in S_2^0,\\
        \left| \gradmax (K_2 + G_2', k + K_2 + G_2 - G_1; \eps) \right|\le C \eps^{\frac{1+\sdpt}{2} (1-\nu)|K_2+G_2'|/|K|}= C \eps^{(1+\sdpt) (1-\nu)} \quad &\text{for all} \; k \in S_3^0,
        \end{split}
    \end{align}
    where to obtain the equalities we recall that $|K_2 + G_2|/|K| = 2$ for all $G_2 \in \Lambda_2^1$.

    Let $I^{(2)}$ denote the right-hand side of \eqref{eq:beta2_2} and write $I^{(2)} = C\sum_{i_1, i_2, i_3} I^{(2)}_{i_1, i_2, i_3}$, where
    \begin{align}\label{eq:integrand_22}
    \begin{split}
        I^{(2)}_{i_1, i_2, i_3} &:= \frac{1}{\eps^2}\int_{S_{i_1,i_2,i_3}} \sum_{G_1 \in \cR_1^*} \sum_{G_2,G_2' \in \Lambda_2^1}\sum_{\sigma,\sigma'\in \{A,B\}}|k| \gradmax (K_2 + G_2, k + K_2 + G_2; \eps)\\
        &\hspace{2cm} \times 
        |k - G_1 +G_2 - G_2'| 
        \gradmax (K_2 + G_2', k + K_2 + G_2 - G_1; \eps)\\
        &\hspace{3.5cm} \times \left| \hat{f}^{\sigma}_2 \Big(\frac{k}{\eps},\eps \microtime \Big) \right|
        \left| \hat{f}^{\sigma'}_2 \Big(\frac{k - G_1 - G_2'+G_2}{\eps},\eps \microtime \Big) \right| {\rm d}k
    \end{split}
    \end{align}
    restricts the integral over $k$ to $S_{i_1, i_2, i_3}$. As before, we will bound each $I^{(2)}_{i_1, i_2, i_3}$ separately \srq{and set $\srp := \sdpt-\sdpt_-$}.
    If $k \in S_{0,0,0}$, then \eqref{eq:normell2_2} implies that the integrand in \eqref{eq:integrand_22} is bounded by
    \begin{align}\label{eq:I2_000}
    \begin{split}
        &C \frac{1 + 2^{2+\delta} |\vec{k} - \vec{G}_1 + \vec{G}_2 -\vec{G}'_2|^{2+\delta}}{1 + |\vec{G}_1 - (\vec{G}_2 - \vec{G}'_2) |^{2+\delta}}\eps^{2(1+\sdpt) (1-\nu)}\, |k| \; |k - G_1 +G_2 - G_2'|\\
        &\hspace{7cm} \times 
        \left| \hat{f}^{\sigma}_2 \Big(\frac{k}{\eps},\eps \microtime \Big) \right|
        \left| \hat{f}^{\sigma'}_2 \Big(\frac{k - G_1 - G_2'+G_2}{\eps},\eps \microtime \Big) \right|,
    \end{split}
    \end{align}
    meaning that
    \begin{align*}
        I^{(2)}_{0,0,0} &\le C \eps^{2\sdpt (1-\nu) - 2\nu} \eps^{2} \norm{f_2 (\cdot, \eps \microtime)}_{H^1}
        (\eps^{2} \norm{f_2 (\cdot, \eps \microtime)}_{H^1} + \eps^{4+\delta} \norm{f_2 (\cdot, \eps \microtime)}_{H^{3+\delta}})\\
        &\hspace{7cm} \times 
        \sum_{G_1 \in \cR_1^*} \sum_{G_2,G_2' \in \Lambda_2^1}\frac{1}{1 + |\vec{G}_1 - (\vec{G}_2 - \vec{G}'_2) |^{2+\delta}}.
    \end{align*}
    Recall that $\Lambda_2^1$ is a finite set, hence the sum on the second line is finite. Our choice of $\nu$ then implies that
    \begin{align*}
        I^{(2)}_{0,0,0} \le C \eps^{2+\sdpt_-} \norm{f_2 (\cdot, \eps \microtime)}_{H^1}
        (\eps^{2+\sdpt_-} \norm{f_2 (\cdot, \eps \microtime)}_{H^1} + \eps^{4+\delta+\sdpt_-} \norm{f_2 (\cdot, \eps \microtime)}_{H^{3+\delta}}).
    \end{align*}

    \srq{The uniform bound on $\nabla \hathperp$ in \eqref{eq:4bounds_full} implies that $\gradmax$ is uniformly bounded in all of its arguments.} Using also that $\nu^{2+\srp} |\vec{K}_2 + \vec{G}_2|^{2+\srp}\le |k|^{2+\srp}$ for all $k \in S_2^1$, we write
    \begin{align}\label{eq:I2_010}
    \begin{split}
        I^{(2)}_{0,1,0} &\le \frac{1}{\eps^2}\int_{S_{0,1,0}} \sum_{G_1 \in \cR_1^*} \sum_{G_2,G_2' \in \Lambda_2^1}\sum_{\sigma,\sigma'\in \{A,B\}}\frac{1 + 2^{2+\delta} |\vec{k} - \vec{G}_1 + \vec{G}_2 -\vec{G}'_2|^{2+\delta}}{1 + |\vec{G}_1 - (\vec{G}_2 - \vec{G}'_2) |^{2+\delta}}
        |k| \frac{|k|^{2+\srp}}{\nu^{2+\srp} |\vec{K}_2 + \vec{G}_2|^{2+\srp}}\\
        &\hspace{2cm} \times 
        |k - G_1 +G_2 - G_2'| 
        \eps^{(1+\sdpt)(1-\nu)}\left| \hat{f}^{\sigma}_2 \Big(\frac{k}{\eps},\eps \microtime \Big) \right|
        \left| \hat{f}^{\sigma'}_2 \Big(\frac{k - G_1 - G_2'+G_2}{\eps},\eps \microtime \Big) \right| {\rm d}k,
    \end{split}
    \end{align}
    which after applying Cauchy-Schwarz to the integral implies
    \begin{align*}
        I^{(2)}_{0,1,0} &\le C \nu^{-2-\delta} \eps^{-2} \eps^{(1+\sdpt)(1-\nu)} \eps^{4+\srp} \norm{f_2 (\cdot, \eps \microtime)}_{H^{3+\delta}}
        (\eps^2 \norm{f_2 (\cdot, \eps \microtime)}_{H^{1}} + \eps^{4+\delta}\norm{f_2 (\cdot, \eps \microtime)}_{H^{3+\delta}}),
    \end{align*}
    where we again absorbed a finite sum over $G_1, G_2, G_2'$ into the constant $C$. Our \srq{smallness} assumption on $\nu$ 
    implies that $\eps^{-2} \eps^{(1+\sdpt)(1-\nu)} \le \eps^{-1 + \sdpt_-}$, and thus
    \begin{align*}
        I^{(2)}_{0,1,0} \le C \eps^{3+\srp} \norm{f_2 (\cdot, \eps \microtime)}_{H^{3+\delta}}
        (\eps^{2+\sdpt_-} \norm{f_2 (\cdot, \eps \microtime)}_{H^{1}} + \eps^{4+\delta+\sdpt_-}\norm{f_2 (\cdot, \eps \microtime)}_{H^{3+\delta}}).
    \end{align*}

    Since \sq{$\nu^{1+\sdpt_-} |K_2 + G_2'|^{1+\sdpt_-} \le |k-G_1 + G_2 - G_2'|^{1+\sdpt_-}$} for all $k \in S_3^1$,
    \begin{align}\label{eq:I2_011}
    \begin{split}
        I^{(2)}_{0,1,1} &\le \frac{1}{\eps^2}\int_{S_{0,1,1}} \sum_{G_1 \in \cR_1^*} \sum_{G_2,G_2' \in \Lambda_2^1}\sum_{\sigma,\sigma'\in \{A,B\}}\frac{1 + 2^{2+\delta} |\vec{k} - \vec{G}_1 + \vec{G}_2 -\vec{G}'_2|^{2+\delta}}{1 + |\vec{G}_1 - (\vec{G}_2 - \vec{G}'_2) |^{2+\delta}}
        |k| \frac{|k|^{2+\srp}}{\nu^{2+\srp} |\vec{K}_2 + \vec{G}_2|^{2+\srp}}\\
        &\hspace{0cm} \times 
        |k - G_1 +G_2 - G_2'| 
        \sq{\frac{|k-G_1 + G_2 - G_2'|^{1+\sdpt_-}}{\nu^{1+\sdpt_-} |K_2 + G_2'|^{1+\sdpt_-}}}
        \left| \hat{f}^{\sigma}_2 \Big(\frac{k}{\eps},\eps \microtime \Big) \right|
        \left| \hat{f}^{\sigma'}_2 \Big(\frac{k - G_1 - G_2'+G_2}{\eps},\eps \microtime \Big) \right| {\rm d}k \\
        &\le C \sq{\nu^{-3-\delta-\sdpt_-}} \eps^{-2} \eps^{4+\srp} \norm{f_2 (\cdot, \eps \microtime)}_{H^{3+\srp}} \sq{(\eps^{3+\sdpt_-} \norm{f_2 (\cdot, \eps \microtime)}_{H^{2+\sdpt_-}} + \eps^{5+\srp+\sdpt_-} \norm{f_2 (\cdot, \eps \microtime)}_{H^{4+\srp+\sdpt_-}})}\\
        &= C\eps^{3+\srp} \norm{f_2 (\cdot, \eps \microtime)}_{H^{3+\srp}} \srq{(\eps^{2+\sdpt_-} \norm{f_2 (\cdot, \eps \microtime)}_{H^{2+\sdpt_-}} + \eps^{4+\sdpt} \norm{f_2 (\cdot, \eps \microtime)}_{H^{4+\sdpt}})}.
    \end{split}
    \end{align}
    
    To bound $I^{(2)}_{0,0,1}$, replace the factor
    $\frac{|\vec{k}|^{2+\srp}}{\nu^{2+\srp} |\vec{K}_2 + \vec{G}_2|^{2+\srp}}$
    on the right-hand side of \eqref{eq:I2_010} by $$\sq{\frac{|k-G_1 + G_2 - G_2'|^{1+\sdpt_-}}{\nu^{1+\sdpt_-} |K_2 + G_2'|^{1+\sdpt_-}}}$$
    to obtain
    \begin{align*}
        I^{(2)}_{0,0,1} &\le C \sq{\nu^{-1-\sdpt_-}}\eps^{-2}\eps^{(1+\sdpt)(1-\nu)} \eps^2 \norm{f_2 (\cdot, \eps \microtime)}_{H^{1}}
        \sq{(\eps^{3+\sdpt_-} \norm{f_2 (\cdot, \eps \microtime)}_{H^{2+\sdpt_-}} + \eps^{5+\srp+\sdpt_-} \norm{f_2 (\cdot, \eps \microtime)}_{H^{4+\srp+\sdpt_-}})}\\
        &\le C\eps^{2+\sdpt_-} \norm{f_2 (\cdot, \eps \microtime)}_{H^{1}}
        \srq{(\eps^{2+\sdpt_-} \norm{f_2 (\cdot, \eps \microtime)}_{H^{2+\sdpt_-}} + \eps^{4+\sdpt} \norm{f_2 (\cdot, \eps \microtime)}_{H^{4+\sdpt}})}.
    \end{align*}
    
    We have now controlled all of the $I^{(2)}_{0,i_2,i_3}$. To handle $I^{(2)}_{1,i_2,i_3}$, replace the factor $$\frac{1 + 2^{2+\srp} |\vec{k} - \vec{G}_1 + \vec{G}_2 -\vec{G}'_2|^{2+\srp}}{1 + |\vec{G}_1 - (\vec{G}_2 - \vec{G}'_2) |^{2+\srp}}$$ appearing on the first lines of \eqref{eq:I2_000}, \eqref{eq:I2_010} and \eqref{eq:I2_011} by
    $$\frac{1 + 2^{2+\srp} |\vec{k}|^{2+\srp}}{1 + |\vec{G}_1 - (\vec{G}_2 - \vec{G}'_2) |^{2+\srp}},$$ 
    which is greater than or equal to $1$ for all $k$ in $S_1^1$. Following the above arguments, we conclude that
    \begin{align*}
        I^{(2)}_{1,0,0} &\le C \eps^{2+\sdpt_-} \norm{f_2 (\cdot, \eps \microtime)}_{H^1}
        (\eps^{2+\sdpt_-} \norm{f_2 (\cdot, \eps \microtime)}_{H^1} + \eps^{4+\delta+\sdpt_-} \norm{f_2 (\cdot, \eps \microtime)}_{H^{3+\delta}}),\\
        I^{(2)}_{1,1,0} &\le C\eps^{2+\sdpt_-} \norm{f_2 (\cdot, \eps \microtime)}_{H^{1}}
        \srq{(\eps^{2+\sdpt_-} \norm{f_2 (\cdot, \eps \microtime)}_{H^{2+\sdpt_-}} + \eps^{4+\sdpt} \norm{f_2 (\cdot, \eps \microtime)}_{H^{4+\sdpt}})},\\
        I^{(2)}_{1,1,1} &\le C\eps^{3+\srp} \norm{f_2 (\cdot, \eps \microtime)}_{H^{3+\srp}} \srq{(\eps^{2+\sdpt_-} \norm{f_2 (\cdot, \eps \microtime)}_{H^{2+\sdpt_-}} + \eps^{4+\sdpt} \norm{f_2 (\cdot, \eps \microtime)}_{H^{4+\sdpt}})}
    \end{align*}
    and (after making the appropriate adjustments to \eqref{eq:I2_010})
    \begin{align*}
        I^{(2)}_{1,0,1} &\le \frac{1}{\eps^2}\int_{S_{1,0,1}} \sum_{G_1 \in \cR_1^*} \sum_{G_2,G_2' \in \Lambda_2^1}\sum_{\sigma,\sigma'\in \{A,B\}}\frac{1 + 2^{2+\delta} |\vec{k}|^{2+\delta}}{1 + |\vec{G}_1 - (\vec{G}_2 - \vec{G}'_2) |^{2+\delta}}
        |k| \frac{|k-G_1 + G_2 - G_2'|^{2+\srp}}{\nu^{2+\srp} |K_2 + G_2'|^{2+\srp}}\\
        &\hspace{2cm} \times 
        |k - G_1 +G_2 - G_2'| 
        \eps^{(1+\sdpt)(1-\nu)}\left| \hat{f}^{\sigma}_2 \Big(\frac{k}{\eps},\eps \microtime \Big) \right|
        \left| \hat{f}^{\sigma'}_2 \Big(\frac{k - G_1 - G_2'+G_2}{\eps},\eps \microtime \Big) \right| {\rm d}k\\
        &\le C \nu^{-2-\srp} \eps^{-2} \eps^{(1+\sdpt)(1-\nu)} \eps^{4+\srp} \norm{f_2 (\cdot, \eps \microtime)}_{H^{3+\srp}}
        (\eps^2 \norm{f_2(\cdot, \eps \microtime)}_{H^1} + \eps^{4+\srp} \norm{f_2 (\cdot, \eps \microtime)}_{H^{3+\srp}})\\
        &\le C\eps^{3+\srp} \norm{f_2 (\cdot, \eps \microtime)}_{H^{3+\srp}}
        (\eps^{2+\sdpt_-} \norm{f_2(\cdot, \eps \microtime)}_{H^1} + \eps^{4+\srp+\sdpt_-} \norm{f_2 (\cdot, \eps \microtime)}_{H^{3+\srp}}).
    \end{align*}
    
    With our bounds on the $I^{(2)}_{i_1, i_2, i_3}$, we have shown that 
    \begin{align}\label{eq:Delta2_bd}
        \norm{\Delta_2 (\cdot, \microtime)}_{L^2 (\Gamma_1^*)}\le C \eps^{2+\sdpt_-} \srq{\norm{f_2 (\cdot, \eps \microtime)}_{H^{4+\sdpt}}}.
    \end{align}

    \medskip

    \textbf{III. Finally, we bound $\Delta_1$.} By the definition \eqref{eq:Deltaj_def}-\eqref{eq:Taylor1} of $\Delta_1$, it follows that
    \begin{align*}
        \Delta_1^\sigma (k_1, \microtime) 
        &= \frac{1}{\eps |\Gamma|^2}e^{-i\eshift \microtime}\sum_{G_1 \in \cR_1^*} \sum_{G_2 \in \Lambda_1^1} \sum_{\sigma'\in \{A,B\}} e^{-i(G_1 \cdot \tau^\sigma_1 + G_2 \cdot \tau^{\sigma'}_2)}\Big(\hathperp (K_2 + G_2;\eps)\\
        & + (k_1 -G_1 -(K_2 + G_2)) \cdot \nabla \hathperp (K_2 + G_2;\eps)
        -\hathperp (k_1 -G_1; \eps)\Big)
        \hat{f}^{\sigma'}_2 \Big(\frac{k_1 -G_1 - G_2 - K_2}{\eps},\eps \microtime \Big).
    \end{align*}
    As before, Taylor's theorem implies that
    \begin{align*}
        &\left \vert \hathperp (K_2 + G_2;\eps)+ (k_1 -G_1 -(K_2 + G_2)) \cdot \nabla \hathperp (K_2 + G_2;\eps)
        -\hathperp (k_1 -G_1; \eps)\right \vert\\
        &\hspace{7cm} \le \frac{1}{2} |k_1 - G_1 - K_2 - G_2|^2 \gradmaxtwo (K_2 + G_2, k_1 - G_1; \eps),
    \end{align*}
    where
    \begin{align*}
        \gradmaxtwo (k,q;\eps) := \max_{0 \le s \le 1} \norm{\nabla^2 \hathperp (s q + (1-s) k; \eps)}_{{\rm F}}
    \end{align*}
    is the maximum Frobenius norm of $\nabla^2 \hathperp (\cdot, \eps)$ over the line segment connecting $k$ and $q$.
    It follows that
    \begin{align*}
        \norm{\Delta_1 (\cdot, \microtime)}^2_{L^2 (\Gamma_1^*)} &\le 
        \frac{C}{\eps^2}\int_{\Gamma_1^*} \sum_{G_1,G_1' \in \cR_1^*} \sum_{G_2,G_2' \in \Lambda_1^1}\sum_{\sigma,\sigma'\in \{A,B\}}|k_1 - G_1 - K_2 - G_2|^2 \gradmaxtwo (K_2 + G_2, k_1 - G_1; \eps)\\
        &\hspace{2cm} \times |k_1 - G_1' - K_2 - G_2'|^2 \gradmaxtwo (K_2 + G_2', k_1 - G_1'; \eps)\\
        &\hspace{2.5cm} \times
        \left| \hat{f}^{\sigma}_2 \Big(\frac{k_1 - G_1 - G_2 - K_2}{\eps},\eps \microtime \Big) \right|
        \left| \hat{f}^{\sigma'}_2 \Big(\frac{k_1 - G_1' - G_2' - K_2}{\eps},\eps \microtime \Big) \right| {\rm d}k_1.
    \end{align*}
    Combining the sum over $G_1$ with the integral over $k_1$, we get
    \begin{align*}
        \norm{\Delta_1 (\cdot, \microtime)}^2_{L^2 (\Gamma_1^*)} &\le 
        \frac{C}{\eps^2}\int_{\mathbb{R}^2} \sum_{G_1' \in \cR_1^*} \sum_{G_2,G_2' \in \Lambda_1^1}\sum_{\sigma,\sigma'\in \{A,B\}}|k_1- K_2 - G_2|^2 \gradmaxtwo (K_2 + G_2, k_1; \eps)|k_1 - G_1' - K_2 - G_2'|^2 \\
        &\hspace{0.25cm} \times \gradmaxtwo (K_2 + G_2', k_1 - G_1'; \eps)
        \left| \hat{f}^{\sigma}_2 \Big(\frac{k_1 - G_2 - K_2}{\eps},\eps \microtime \Big) \right|
        \left| \hat{f}^{\sigma'}_2 \Big(\frac{k_1 - G_1' - G_2' - K_2}{\eps},\eps \microtime \Big) \right| {\rm d}k_1,
    \end{align*}
    which after the change of variables $k_1 \leftarrow k_1 - G_2 - K_2$ yields
    \begin{align}\label{eq:beta2_1}
    \begin{split}
        \norm{\Delta_1 (\cdot, \microtime)}^2_{L^2 (\Gamma_1^*)} &\le 
        \frac{C}{\eps^2}\int_{\mathbb{R}^2} \sum_{G_1 \in \cR_1^*} \sum_{G_2,G_2' \in \Lambda_1^1}\sum_{\sigma,\sigma'\in \{A,B\}}|k_1|^2 \gradmaxtwo (K_2 + G_2, k_1 + G_2 + K_2; \eps)\\
        &\hspace{2cm} \times
        |k_1 - G_1 + G_2 - G_2'|^2 \gradmaxtwo (K_2 + G_2', k_1 - G_1 + G_2 + K_2; \eps)\\
        &\hspace{3.5cm} \times
        \left| \hat{f}^{\sigma}_2 \Big(\frac{k_1}{\eps},\eps \microtime \Big) \right|
        \left| \hat{f}^{\sigma'}_2 \Big(\frac{k_1 - G_1 - G_2' + G_2}{\eps},\eps \microtime \Big) \right| {\rm d}k_1.
    \end{split}
    \end{align}
    As before, let \srq{$0 < \nu \le 1$ from Assumption \ref{assumption:hperp2} be sufficiently small} such that $\sdpt (1-\nu) - 2\nu \ge \sdpt_-$, and for fixed $\vec{G}_1, \vec{G}_2, \vec{G}'_2$ define
    \begin{align*}
        S_1^0 :&= \{\vec{k} : |\vec{k}| < |\vec{G}_1 - \vec{G}_2 + \vec{G}'_2|/2 \}, \qquad \qquad
        S_2^0 := \{k:|\vec{k}| < \nu |\vec{K}_2 + \vec{G}_2| \}, \\
        S_3^0 :&= \{k:|\vec{k}- \vec{G}_1 + \vec{G}_2 - \vec{G}'_2| < \nu |\vec{K}_2 + \vec{G}'_2| \}, \qquad \qquad S_j^1 := S_j^c, 
        \qquad S_{i_1, i_2, i_3} := S_1^{i_1} \cap S_2^{i_2} \cap S_3^{i_3},
    \end{align*}
    where $j \in \{1,2,3\}$ and $i_1, i_2, i_3 \in \{0,1\}$. By \eqref{eq:S_bounds2} and 
    \srq{the decay of $\nabla^2 \hathperp$ in \eqref{eq:4bounds_full},} we conclude that
    for all $G_2, G_2' \in \Lambda_1^1$,
    \begin{align}\label{eq:normell2_1}
    \begin{split}
        \left| \gradmaxtwo (K_2 + G_2, k+K_2 + G_2; \eps) \right| \le C \eps^{\sdpt (1-\nu)|K_2+G_2|/|K|} = C \eps^{\sdpt (1-\nu)} \quad &\text{for all} \; k \in S_2^0,\\
        \left| \gradmaxtwo (K_2 + G_2', k + K_2 + G_2 - G_1; \eps) \right|\le C \eps^{\sdpt (1-\nu)|K_2+G_2'|/|K|}= C \eps^{\sdpt (1-\nu)} \quad &\text{for all} \; k \in S_3^0,
        \end{split}
    \end{align}
    where to obtain the equalities we recall that $|K_2 + G_2|/|K| = 1$ for all $G_2 \in \Lambda_1^1$.

    Let $I^{(1)}$ denote the right-hand side of \eqref{eq:beta2_1} and write $I^{(1)} = C\sum_{i_1, i_2, i_3} I^{(1)}_{i_1, i_2, i_3}$, where
    \begin{align*}
        I^{(1)}_{i_1,i_2,i_3} &:= \frac{1}{\eps^2}\int_{S_{i_1,i_2,i_3}} \sum_{G_1 \in \cR_1^*} \sum_{G_2,G_2' \in \Lambda_1^1}\sum_{\sigma,\sigma'\in \{A,B\}}|k|^2 \gradmaxtwo (K_2 + G_2, k + G_2 + K_2; \eps)\\
        &\hspace{2cm} \times
        |k - G_1 + G_2 - G_2'|^2 \gradmaxtwo (K_2 + G_2', k - G_1 + G_2 + K_2; \eps)\\
        &\hspace{3.5cm} \times
        \left| \hat{f}^{\sigma}_2 \Big(\frac{k}{\eps},\eps \microtime \Big) \right|
        \left| \hat{f}^{\sigma'}_2 \Big(\frac{k - G_1 - G_2' + G_2}{\eps},\eps \microtime \Big) \right| {\rm d}k
    \end{align*}
    restricts the integral to $k \in S_{i_1, i_2, i_3}$. We will again bound each of the $I^{(1)}_{i_1, i_2,i_3}$ separately \srq{and set $\srp := \sdpt-\sdpt_-$}. 
    Following familiar arguments, we write
    \begin{align}\label{eq:I1_000}
    \begin{split}
        I^{(1)}_{0,0,0}&\le \frac{1}{\eps^2}\int_{S_{0,0,0}} \sum_{G_1 \in \cR_1^*} \sum_{G_2,G_2' \in \Lambda_1^1}\sum_{\sigma,\sigma'\in \{A,B\}}
        \frac{1 + 2^{2+\delta} |\vec{k} - \vec{G}_1 + \vec{G}_2 -\vec{G}'_2|^{2+\delta}}{1 + |\vec{G}_1 - (\vec{G}_2 - \vec{G}'_2) |^{2+\delta}} |k|^2 \eps^{\sdpt (1-\nu)}\\
        &\hspace{2cm} \times |k - G_1 + G_2 - G_2'|^2 \eps^{\sdpt (1-\nu)}
        \left| \hat{f}^{\sigma}_2 \Big(\frac{k}{\eps},\eps \microtime \Big) \right|
        \left| \hat{f}^{\sigma'}_2 \Big(\frac{k - G_1 - G_2' + G_2}{\eps},\eps \microtime \Big) \right| {\rm d}k\\
        &\le C \eps^{-2} \eps^{2\sdpt (1-\nu)} \eps^3 \norm{f_2 (\cdot, \eps \microtime)}_{H^2}
        (\eps^3 \norm{f_2 (\cdot, \eps \microtime)}_{H^2} + \eps^{5+\srp}\norm{f_2 (\cdot, \eps \microtime)}_{H^{4+\srp}})\\
        &\le C \eps^{2+\sdpt_-} \norm{f_2 (\cdot, \eps \microtime)}_{H^2}
        (\eps^{2+\sdpt_-} \norm{f_2 (\cdot, \eps \microtime)}_{H^2} + \eps^{4+\srp+\sdpt_-}\norm{f_2 (\cdot, \eps \microtime)}_{H^{4+\srp}}),
    \end{split}
    \end{align}
    where we used \eqref{eq:S10_bound} and \eqref{eq:normell2_1} to obtain the first inequality, Cauchy-Schwarz and the fact that $\Lambda_1^1$ is a finite set for the second, and the assumption that $\sdpt (1-\nu) \ge \sdpt_-$ for the third.

    Similarly, using the fact that $\nu^{2+\srp} |\vec{K}_2 + \vec{G}_2|^{2+\srp}\le |k|^{2+\srp}$ for all $k \in S_2^1$, 
    we obtain
    \begin{align}\label{eq:I1_010}
    \begin{split}
        I^{(1)}_{0,1,0}&\le \frac{1}{\eps^2}\int_{S_{0,1,0}} \sum_{G_1 \in \cR_1^*} \sum_{G_2,G_2' \in \Lambda_1^1}\sum_{\sigma,\sigma'\in \{A,B\}}
        \frac{1 + 2^{2+\delta} |\vec{k} - \vec{G}_1 + \vec{G}_2 -\vec{G}'_2|^{2+\delta}}{1 + |\vec{G}_1 - (\vec{G}_2 - \vec{G}'_2) |^{2+\delta}} |k|^2 
        \frac{|k|^{2+\srp}}{\nu^{2+\srp} |\vec{K}_2 + \vec{G}_2|^{2+\srp}}\\
        &\hspace{2cm} \times |k - G_1 + G_2 - G_2'|^2 \eps^{\sdpt (1-\nu)}
        \left| \hat{f}^{\sigma}_2 \Big(\frac{k}{\eps},\eps \microtime \Big) \right|
        \left| \hat{f}^{\sigma'}_2 \Big(\frac{k - G_1 - G_2' + G_2}{\eps},\eps \microtime \Big) \right| {\rm d}k\\
        &\le C \eps^{-2} \eps^{\sdpt (1-\nu)} \eps^{5+\srp} \norm{f_2 (\cdot, \eps \microtime)}_{H^{4+\srp}} 
        (\eps^{3} \norm{f_2 (\cdot, \eps \microtime)}_{H^{2}} + \eps^{5+\srp} \norm{f_2 (\cdot, \eps \microtime)}_{H^{4+\srp}})\\
        &\le C \eps^{4+\srp} \norm{f_2 (\cdot, \eps \microtime)}_{H^{4+\srp}} 
        (\eps^{2+\sdpt_-} \norm{f_2 (\cdot, \eps \microtime)}_{H^{2}} + \eps^{4+\srp+\sdpt_-} \norm{f_2 (\cdot, \eps \microtime)}_{H^{4+\srp}}),
    \end{split}
    \end{align}
    while the fact that $\sq{\nu^{\sdpt_-} |K_2 + G_2'|^{\sdpt_-} \le |k-G_1 + G_2 - G_2'|^{\sdpt_-}}$ for all $k \in S_3^1$ implies
    \begin{align}\label{eq:I1_011}
    \begin{split}
        I^{(1)}_{0,1,1}&\le \frac{1}{\eps^2}\int_{S_{0,1,1}} \sum_{G_1 \in \cR_1^*} \sum_{G_2,G_2' \in \Lambda_1^1}\sum_{\sigma,\sigma'\in \{A,B\}}
        \frac{1 + 2^{2+\delta} |\vec{k} - \vec{G}_1 + \vec{G}_2 -\vec{G}'_2|^{2+\delta}}{1 + |\vec{G}_1 - (\vec{G}_2 - \vec{G}'_2) |^{2+\delta}} |k|^2 
        \frac{|k|^{2+\srp}}{\nu^{2+\srp} |\vec{K}_2 + \vec{G}_2|^{2+\srp}}\\
        &\hspace{0.0cm} \times |k - G_1 + G_2 - G_2'|^2 
        \sq{\frac{|k-G_1 + G_2 - G_2'|^{\sdpt_-}}{\nu^{\sdpt_-} |K_2 + G_2'|^{\sdpt_-}}}
        \left| \hat{f}^{\sigma}_2 \Big(\frac{k}{\eps},\eps \microtime \Big) \right|
        \left| \hat{f}^{\sigma'}_2 \Big(\frac{k - G_1 - G_2' + G_2}{\eps},\eps \microtime \Big) \right| {\rm d}k\\
        &\le C \eps^{-2} \eps^{5+\srp} \norm{f_2 (\cdot, \eps \microtime)}_{H^{4+\srp}}
        \sq{(\eps^{3+\sdpt_-} \norm{f_2 (\cdot, \eps \microtime)}_{H^{2+\sdpt_-}} + \eps^{5+\srp+\sdpt_-} \norm{f_2 (\cdot, \eps \microtime)}_{H^{4+\srp+\sdpt_-}})}\\
        &=C \eps^{4+\srp} \norm{f_2 (\cdot, \eps \microtime)}_{H^{4+\srp}}
        \srq{(\eps^{2+\sdpt_-} \norm{f_2 (\cdot, \eps \microtime)}_{H^{2+\sdpt_-}} + \eps^{4+\sdpt} \norm{f_2 (\cdot, \eps \microtime)}_{H^{4+\sdpt}})}
    \end{split}
    \end{align}
    and
    \begin{align}\label{eq:I1_001}
    \begin{split}
        I^{(1)}_{0,0,1}&\le \frac{1}{\eps^2}\int_{S_{0,0,1}} \sum_{G_1 \in \cR_1^*} \sum_{G_2,G_2' \in \Lambda_1^1}\sum_{\sigma,\sigma'\in \{A,B\}}
        \frac{1 + 2^{2+\delta} |\vec{k} - \vec{G}_1 + \vec{G}_2 -\vec{G}'_2|^{2+\delta}}{1 + |\vec{G}_1 - (\vec{G}_2 - \vec{G}'_2) |^{2+\delta}} |k|^2 
        \eps^{\sdpt (1-\nu)}\\
        &\hspace{0.0cm} \times |k - G_1 + G_2 - G_2'|^2 
        \sq{\frac{|k-G_1 + G_2 - G_2'|^{\sdpt_-}}{\nu^{\sdpt_-} |K_2 + G_2'|^{\sdpt_-}}}
        \left| \hat{f}^{\sigma}_2 \Big(\frac{k}{\eps},\eps \microtime \Big) \right|
        \left| \hat{f}^{\sigma'}_2 \Big(\frac{k - G_1 - G_2' + G_2}{\eps},\eps \microtime \Big) \right| {\rm d}k\\
        &\le C \eps^{-2} \eps^{\sdpt (1-\nu)}\eps^3 \norm{f_2 (\cdot, \eps \microtime)}_{H^{2}}
        \sq{(\eps^{3+\sdpt_-}\norm{f_2 (\cdot, \eps \microtime)}_{H^{2+\sdpt_-}}+\eps^{5+\srp+\sdpt_-}\norm{f_2 (\cdot, \eps \microtime)}_{H^{4+\srp+\sdpt_-}})}\\
        &\le C \eps^{2+\sdpt_-} \norm{f_2 (\cdot, \eps \microtime)}_{H^{2}}
        \srq{(\eps^{2+\sdpt_-}\norm{f_2 (\cdot, \eps \microtime)}_{H^{2+\sdpt_-}}+\eps^{4+\sdpt}\norm{f_2 (\cdot, \eps \microtime)}_{H^{4+\sdpt}})}.
    \end{split}
    \end{align}
    \srq{The first inequality in each of \eqref{eq:I1_010}-\eqref{eq:I1_001} also used the uniform boundedness of $\gradmaxtwo$, which in turn follows from the uniform bound on $\nabla^2 \hathperp$ in \eqref{eq:4bounds_full}.}

    As before, to bound the $I^{(1)}_{1,i_2,i_3}$, 
    replace the factor $$\frac{1 + 2^{2+\srp} |\vec{k} - \vec{G}_1 + \vec{G}_2 -\vec{G}'_2|^{2+\srp}}{1 + |\vec{G}_1 - (\vec{G}_2 - \vec{G}'_2) |^{2+\srp}}$$ appearing on the (first) right-hand sides of \eqref{eq:I1_000}, \eqref{eq:I1_010}, \eqref{eq:I1_011} and \eqref{eq:I1_001} by
    \begin{align}\label{eq:frac1}
    \frac{1 + 2^{2+\srp} |\vec{k}|^{2+\srp}}{1 + |\vec{G}_1 - (\vec{G}_2 - \vec{G}'_2) |^{2+\srp}}, 
    \end{align}
    as \eqref{eq:frac1} is greater than or equal to $1$ for all $k$ in $S_1^1$. It then follows that
    \begin{align*}
        I^{(1)}_{1,0,0} &\le C \eps^{2+\sdpt_-} \norm{f_2 (\cdot, \eps \microtime)}_{H^2}
        (\eps^{2+\sdpt_-} \norm{f_2 (\cdot, \eps \microtime)}_{H^2} + \eps^{4+\srp+\sdpt_-}\norm{f_2 (\cdot, \eps \microtime)}_{H^{4+\srp}}),\\
        I^{(1)}_{1,1,0} &\le C \eps^{2+\sdpt_-} \norm{f_2 (\cdot, \eps \microtime)}_{H^{2}}
        \srq{(\eps^{2+\sdpt_-}\norm{f_2 (\cdot, \eps \microtime)}_{H^{2+\sdpt_-}}+\eps^{4+\sdpt}\norm{f_2 (\cdot, \eps \microtime)}_{H^{4+\sdpt}})},\\
        I^{(1)}_{1,1,1} &\le C \eps^{4+\srp} \norm{f_2 (\cdot, \eps \microtime)}_{H^{4+\srp}}
        \srq{(\eps^{2+\sdpt_-} \norm{f_2 (\cdot, \eps \microtime)}_{H^{2+\sdpt_-}} + \eps^{4+\sdpt} \norm{f_2 (\cdot, \eps \microtime)}_{H^{4+\sdpt}})},\\
        I^{(1)}_{1,0,1} &\le C \eps^{4+\srp} \norm{f_2 (\cdot, \eps \microtime)}_{H^{4+\srp}} 
        (\eps^{2+\sdpt_-} \norm{f_2 (\cdot, \eps \microtime)}_{H^{2}} + \eps^{4+\srp+\sdpt_-} \norm{f_2 (\cdot, \eps \microtime)}_{H^{4+\srp}}).
    \end{align*}

    Using all the above bounds on the $I^{(1)}_{i_1, i_2, i_3}$, we obtain that 
    \begin{align}\label{eq:Delta1_bd}
        \norm{\Delta_1 (\cdot, \microtime)}_{L^2 (\Gamma_1^*)}\le C \eps^{2+\sdpt_-} \srq{\norm{f_2 (\cdot, \eps \microtime)}_{H^{4+\sdpt}}}.
    \end{align}

    \medskip

    Combining \eqref{eq:Delta_sum}, \eqref{eq:Delta3_bd}, \eqref{eq:Delta2_bd} and \eqref{eq:Delta1_bd}, we have proven \eqref{eq:bilayer1_suffices} as desired.
\end{proof}

\subsection{Theorem \ref{thm:main}}\label{subsec:proof_thm_main}
\begin{proof}[Proof of Theorem \ref{thm:main}]
In the interest of brevity, we will prove the theorem when $\zeta (\eps) = \xi (\eps) = \eps$. The general case follows from the same arguments, only it becomes more cumbersome to keep track of the various orders of $\eps$.

Recall the definitions of $f = \tff{1} + \radNNN (\eps) \tff{{\rm NNN}} + \radpNN (\eps) \tff{\nabla, {\rm NN}} + \eps \tff{2}$ and $\mswf$ in Theorem \ref{thm:main2}, which states that $\srq{\norm{\mswf (\microtime) -\psi (\microtime)}_\cH \le C \norm{f_0}_{H^{6+\sdpt}} \eps^{1 + \sdpt_-} \left(\eps \microtime + \eps^{\sdpt -\sdpt_-} (\eps \microtime)^2 \right)}$ uniformly in $0 < \eps < 1$ and $\microtime \ge 0$.
It remains to bound $\norm{\chi (\microtime) - \phi (\microtime)}_\cH$.

    Combining terms of the same order, 
    set $$\HBMta := \HBMNNN + \HBMgNN + \HBMt, \qquad \tffta := \tff{{\rm NNN}} + \tff{\nabla, {\rm NN}} + \tff{2}.$$
    Let $\tff{3}$ be the solution to
    \begin{align*}
        (i \partial_t - \HBM) \tff{3} = \HBMta \tffta, \qquad \tff{3} (r,0) = 0,
    \end{align*}
    and define $\mathring{f} := f + \eps^2 \tff{3}$.
    We will separately control the differences $\chi - \mathring{\phi}$ and $\mathring{\phi} - \phi$, where
    \begin{align*}
        (\mathring{\mswf}_{i})^\sigma_{R_i} (\microtime) := \eps \mathring{\tf}^\sigma_{i} (\eps (R_i+\tau_i^\sigma),\eps \microtime) e^{i K_i \cdot (R_i + \tau_i^\sigma)}e^{-i\eshift \microtime}
    \end{align*}
    replaces $f$ in \eqref{eq:def_phi} by $\mathring{f}$.

    \medskip
    
    We begin with $\chi - \mathring{\phi}$. The definition of $\mathring{f}$ implies that
    \begin{align}\label{eq:tilde_res}
        (i \partial_t - \Hbm) \mathring{f} = -\eps^3 \HBMta \tff{3}.
    \end{align}
    We now follow the proof of Lemma \ref{lemma:Sobolev_ms2} to obtain a uniform bound on the Sobolev norms of $\tff{3}$. First, \eqref{eq:reg_Hj} implies that for any $N \ge 0$, $\norm{\HBMta u}_{H^N} \le C \norm{u}_{H^{N+2}}$ uniformly in $u$. Writing
    \begin{align*}
        i \partial_t \norm{((\HBM)^2 + \mu)^{N/2} \tff{3}}^2_{L^2} =2 i \Im \left( ((\HBM)^2 + \mu)^{N/2} \tff{3}, ((\HBM)^2 + \mu)^{N/2}\HBMta\tffta \right)
    \end{align*}
    for $\mu > 0$ as in
    \eqref{eq:elliptic_M2}, it follows that
    \begin{align*}
        \partial_t \norm{((\HBM)^2 + \mu)^{N/2} \tff{3}}^2_{L^2} \le C \norm{((\HBM)^2 + \mu)^{N/2} \tff{3}}_{L^2} \norm{\tffta}_{H^{N+2}}.
    \end{align*}
    Since $\norm{\tffta (\cdot, t)}_{H^{N+2}} \le Ct \norm{f_0}_{H^{N+4}}$ by \eqref{eq:Sobolev2}, we obtain
    \begin{align*}
        \norm{((\HBM)^2 + \mu)^{N/2} \tff{3} (\cdot, t)}_{L^2} \le C t^2 \norm{f_0}_{H^{N+4}}.
    \end{align*}
    Applying the lower bound in \eqref{eq:elliptic_M2}, we conclude that
    \begin{align}\label{eq:f3}
        \norm{\tff{3} (\cdot, t)}_{H^N} \le C t^2 \norm{f_0}_{H^{N+4}}
    \end{align}
    uniformly in $t \ge 0$. Applying \eqref{eq:tilde_res}, it follows that $\mathring{\res} := (i \partial_t - \Hbm) \mathring{f}$ satisfies
    \begin{align}\label{eq:tilde_res_bound}
        \norm{\mathring{\res} (\cdot, t)}_{H^N} \le Ct^2 \eps^3 \norm{f_0}_{H^{N+6}}
    \end{align}
    uniformly in $t \ge 0$.
    
    Let $s > 0$, take $\ellC > 0$ satisfying \eqref{eq:ellipticity} and set $\dpm := \mathring{f} - g$. Since $((\Hbm)^2 + \ellC)^{s/4}$ commutes with $\Hbm$, we know that $\dpmH{s} := ((\Hbm)^2 + \ellC)^{s/4} \dpm$ satisfies $(i\partial_t - \Hbm) \dpmH{s} = ((\Hbm)^2 + \ellC)^{s/4}\mathring{\rho}$, where
    \begin{align*}
        \norm{((\Hbm)^2 + \ellC)^{s/4}\mathring{\res} (\cdot, t)}_{L^2} \le C\norm{\mathring{\res} (\cdot, t)}_{H^{s}} \le C t^2 \eps^{3} \norm{f_0}_{H^{s+6}}, \qquad 0 < \eps < 1, \quad t \ge 0,
    \end{align*}
    with the first and second inequalities respectively following from \eqref{eq:ellipticity} and \eqref{eq:tilde_res_bound}.
    Since $\Hbm$ is symmetric with respect to the $L^2$ inner product, we conclude from Lemma \ref{lemma:res2} that
    \begin{align*}
        \norm{\dpmH{s} (\cdot, t)}_{L^2} \le C t^3 \eps^{3} \norm{f_0}_{H^{s+6}}, \qquad 0 < \eps < 1, \quad t \ge 0.
    \end{align*}
    Applying the lower bound in \eqref{eq:ellipticity}, we obtain
    \begin{align*}
        \norm{\dpm (\cdot, t)}_{H^s} \le C t^3 \eps^{3 - s/2}\norm{f_0}_{H^{s+6}}, \qquad 0 < \eps < 1, \quad t \ge 0.
    \end{align*}
    \srq{Set $\delta := \sdpt - \sdpt_-$.} Writing 
    \begin{align*}
        (\phph_{i})^\sigma_{R_i} (\microtime) - (\mathring{\mswf}_{i})^\sigma_{R_i} (\microtime) = -\eps \varphi^\sigma_{i} (\eps (R_i+\tau_i^\sigma),\eps \microtime) e^{i K_i \cdot (R_i + \tau_i^\sigma)}e^{-i\eshift \microtime}
    \end{align*}
    and choosing $s>2$ such that $3 - s/2 \ge 2-\delta$ and $\srq{s\le 2+\sdpt}$,
    we conclude by Lemma \ref{lemma:Hs_to_l22} that
    \begin{align}\label{eq:full_bd1}
        \norm{\chi (\microtime) - \mathring{\phi} (\microtime)}_\cH \le C (\eps \microtime)^3 \eps^{2-\delta}\srq{\norm{f_0}_{H^{8+\sdpt}}}, \qquad 0 < \eps < 1, \quad \microtime \ge 0.
    \end{align}

    We now bound $\mathring{\phi} - \phi$. Recall that $\mathring{f} - f = \eps^2 \tff{3}$ with $\tff{3}$ satisfying \eqref{eq:f3}. It follows that
    \begin{align*}
        (\mathring{\mswf}_{i})^\sigma_{R_i} (\microtime) - (\mswf_{i})^\sigma_{R_i} (\microtime) = \eps^3 (\tff{3})^\sigma_{i} (\eps (R_i+\tau_i^\sigma),\eps \microtime) e^{i K_i \cdot (R_i + \tau_i^\sigma)}e^{-i\eshift \microtime}
    \end{align*}
    satisfies
    \begin{align}\label{eq:full_bd2}
        \norm{\mathring{\mswf}(\microtime) - \mswf (\microtime)}_\cH \le C\eps^2 \srq{\norm{\tff{3}(\cdot, \eps \microtime)}_{H^{2+\sdpt}}} \le C (\eps \microtime)^2 \eps^2 \srq{\norm{f_0}_{H^{6+\sdpt}}}, \qquad 0 < \eps < 1, \quad \microtime \ge 0.
    \end{align}
    The result then follows from combining \eqref{eq:full_bd1} and \eqref{eq:full_bd2} with Theorem \ref{thm:main2}.
\end{proof}

\subsection{Propositions \ref{prop:symmetries}-\ref{prop:ph}}\label{subsection:proofs_symmetries}
\begin{proof}[Proof of Proposition \ref{prop:symmetries}]
    \begin{enumerate}
        \item 
        Note that $\Pi_j$ is a bijection with $\Pi_j^{-1} = \Pi_j$.
        The self-adjointness condition $h(-r) = \overline{h(r)}$ implies that $[H_{jj}, \cC_j \cP_j] = 0$ for $j=1,2$. Next, the assumption that $\hperp (-r) = \overline{\hperp (r)}$ implies
        \begin{align*}
            (\Hperp \cC_2 \cP_2 \psi)^\sigma_{R_1} = \sum_{R_2 \in \cR_2} \sum_{\sigma' \in \{A,B\}} \hperp (R_1 - \Pi_2 (R_2) + \tau^\sigma_1 - \tau^{(\sigma')^c}_2; \eps) \overline{\psi^{\sigma'}_{R_2}}
        \end{align*}
        and
        \begin{align*}
            (\cC_1 \cP_1 \Hperp \psi)^{\sigma}_{R_1} = \sum_{R_2 \in \cR_2} \sum_{\sigma' \in \{A,B\}} \hperp (-\Pi_1 (R_1) +R_2 - \tau^{\sigma^c}_1 + \tau^{\sigma'}_2; \eps) \overline{\psi^{\sigma'}_{R_2}}.
        \end{align*}
        The definition \eqref{eq:Pi0} of the $\Pi_j$ implies that $R_1 - \Pi_2 (R_2) + \tau^\sigma_1 - \tau^{(\sigma')^c}_2 = -\Pi_1 (R_1) +R_2 - \tau^{\sigma^c}_1 + \tau^{\sigma'}_2$, and thus $\Hperp \cC_2 \cP_2 = \cC_1 \cP_1 \Hperp$.
        Since the $\cC_j$ and $\cP_j$ are self-adjoint and satisfy $[\cC_j, \cP_j] = 0$, it follows that $\cC_2 \cP_2 \Hperp^\dagger = \Hperp^\dagger \cC_1 \cP_1$. Thus we have shown that $[H, \cD] = 0$.
        \item Observe that $\Pi_j^\sigma$ is a bijection with $(\Pi^\sigma_j)^{-1} (R_j) = \rot_{-2\pi/3} (R_j + \tau_j^\sigma) - \tau_j^\sigma.$ For $j \in \{1,2\}$, let $\cU_j : \ell^2 (\cR_j; \mathbb{C}^2) \to \ell^2 (\cR_j; \mathbb{C}^2)$ such that $\cU = \diag (\cU_1, \cU_2)$. The 
        $2\pi/3$-rotation invariance of $h$ in Assumption \ref{assumption:h} implies that $[H_{jj}, \cU_j] = 0$ for $j=1,2$. Moreover, a direct calculation reveals that
        \begin{align*}
            (\Hperp \cU_2 \psi)^\sigma_{R_1} = \sum_{R_2 \in \cR_2} \sum_{\sigma' \in \{A,B\}} \hperp (R_1 - \rot_{-2\pi/3} (R_2 + \tau^{\sigma'}_2) + \tau^\sigma_1; \eps) \psi^{\sigma'}_{R_2}
        \end{align*}
        and
        \begin{align*}
            (\cU_1 \Hperp \psi)^\sigma_{R_1} = \sum_{R_2 \in \cR_2} \sum_{\sigma' \in \{A,B\}} \hperp (\rot_{2\pi/3} (R_1 +\tau_1^\sigma) -R_2 - \tau^{\sigma'}_2; \eps) \psi^{\sigma'}_{R_2}.
        \end{align*}
        The assumption that $\hperp (\rot_{2\pi/3} r; \eps) = \hperp (r; \eps)$ then implies that $\Hperp \cU_2 = \cU_1 \Hperp$. Since the $\cU_j$ are unitary and satisfy $\cU_j^3 = \text{Id}_{\ell^2 (\cR_j; \mathbb{C}^2)}$, it follows that $\Hperp^{\dagger} \cU_1 = \cU_2 \Hperp^{\dagger}$. We conclude that $[H, \cU] = 0$ and the proof is complete.
        \item We will prove that $[H, \cM_x] = 0$; the argument for $\cM_y$ is similar. Writing
        \begin{align*}
            \cM_x = \begin{pmatrix}
                0 & \cM_{x,12}\\
                \cM_{x,12}^\dagger & 0
            \end{pmatrix}
        \end{align*}
        with $\cM_{x,12} : \ell^2 (\cR_2; \mathbb{C}^2) \to \ell^2 (\cR_1; \mathbb{C}^2)$ defined by $(\cM_{x,12} \psi)^\sigma_{R_1} = \psi^\sigma_{\Pi^\sigma_{x,1}(R_{1})}$, it follows that
        \begin{align*}
            H \cM_x = \begin{pmatrix}
                H_{12} \cM^\dagger_{x,12} & H_{11} \cM_{x,12}\\
                H_{22} \cM_{x,12}^\dagger & H_{12}^\dagger \cM_{x,12}
            \end{pmatrix}, \qquad
            \cM_x H = \begin{pmatrix}
                \cM_{x,12} H^\dagger_{12} & \cM_{x,12} H_{22}\\
                \cM_{x,12}^\dagger H_{11} & \cM^\dagger_{x,12} H_{12}
            \end{pmatrix}.
        \end{align*}
        Thus it suffices to show that $H_{12} \cM^\dagger_{x,12} = \cM_{x,12} H^\dagger_{12}$ and $H_{11} \cM_{x,12} = \cM_{x,12} H_{22}$. We have
        \begin{align}\label{eq:HM12}
        \begin{split}
            (H_{12} \cM^\dagger_{x,12} \psi)^\sigma_{R_1} &= \sum_{R_2 \in \cR_2} \sum_{\sigma' \in \{A,B\}} \hperp (R_1 - R_2 + \tau^\sigma_1 - \tau^{\sigma'}_2; \eps) \psi^{\sigma'}_{\Pi^{\sigma'}_{x,2} (R_2)}\\
            &= \sum_{R_1 \in \cR'_1} \sum_{\sigma' \in \{A,B\}} \hperp (R_1 - \Pi^{\sigma'}_{x,1} (R'_1) + \tau^\sigma_1 - \tau^{\sigma'}_2; \eps) \psi^{\sigma'}_{R'_1},
        \end{split}
        \end{align}
        where the last equality follows from the fact that $\Pi^\sigma_{x,2} :\cR_2 \to \cR_1$ is a bijection with $(\Pi^\sigma_{x,2})^{-1} = \Pi^\sigma_{x,1}$. Next, we write
        \begin{align}\label{eq:M12H}
            (\cM_{x,12} H_{12}^\dagger \psi)^\sigma_{R_1} = \sum_{R_1 \in \cR'_1} \sum_{\sigma' \in \{A,B\}} \overline{\hperp (\Pi^{\sigma}_{x,1} (R_1) - R'_1 + \tau^\sigma_2 - \tau^{\sigma'}_1; \eps)} \psi^{\sigma'}_{R'_1},
        \end{align}
        and verify that the arguments of $\hperp$ in \eqref{eq:HM12} and \eqref{eq:M12H} satisfy
        \begin{align*}
            R_1 - \Pi^{\sigma'}_{x,1} (R'_1) + \tau^\sigma_1 - \tau^{\sigma'}_2 &= R_1 + \tau_1^\sigma - m_x (R'_1 + \tau_1^{\sigma'}), \\
            \Pi^{\sigma}_{x,1} (R_1) - R'_1 + \tau^\sigma_2 - \tau^{\sigma'}_1 &=
            m_x (R_1 +\tau_1^\sigma) - R'_1 - \tau^{\sigma'}_1.
        \end{align*}
        Therefore, our assumption that $\hperp (m_x (r)) = \overline{\hperp (r)}$ implies that $H_{12} \cM^\dagger_{x,12} = \cM_{x,12} H^\dagger_{12}$ as desired.

        It remains to prove that $H_{11} \cM_{x,12} = \cM_{x,12} H_{22}$.
        We find by a direct calculation that
        \begin{align*}
            (H_{11} \cM_{x,12} \psi)_{R_1}^\sigma &= \sum_{R_2 \in \cR_2} \sum_{\sigma' \in \{A,B\}} h(\rot_{\theta/2} (R_1 - \Pi^{\sigma'}_{x,2} (R_2) + \tau^{\sigma, \sigma'}_1)) \psi^{\sigma'}_{R_2},\\
            (\cM_{x,12} H_{22} \psi)^\sigma_{R_1} &= \sum_{R_2 \in \cR_2} \sum_{\sigma' \in \{A,B\}} h(\rot_{-\theta/2} ( \Pi^{\sigma}_{x,1} (R_1) - R_2 + \tau^{\sigma, \sigma'}_2)) \psi^{\sigma'}_{R_2},
        \end{align*}
        where
        \begin{align*}
            R_1 - \Pi^{\sigma'}_{x,2} (R_2) + \tau^{\sigma, \sigma'}_1 &= R_1 + \tau^\sigma_1 - m_x (R_2 + \tau^{\sigma'}_2),\\
            \Pi^{\sigma}_{x,1} (R_1) - R_2 + \tau^{\sigma, \sigma'}_2 &= m_x (R_1 + \tau^\sigma_1) - R_2 - \tau^{\sigma'}_2.
        \end{align*}
        The result then follows from the observation that $m_x \rot_{\theta/2} = \rot_{-\theta/2} m_x$ and our assumption that $h(m_x (r)) = h(r)$.
    \end{enumerate}
\end{proof}

\begin{proof}[Proof of Proposition \ref{prop:unitary}]
    The diagonal blocks $L + \eps \HBMsd \mp \frac{\eps}{2} \beta i \sigma_3 L$ of $\Hfull$ are independent of $\ls$ and commute with the diagonal blocks of $\cU$. Thus it remains to show that
    $\cU \Hoff \cU^\dagger = \Hoff_0$, where
    \begin{align*}
        \Hoff := \begin{pmatrix}
            0 & \hoppingT (r)\\
            \hoppingT^\dagger (r) & 0
        \end{pmatrix} + \radNNN (\eps) \HBMNNN + \radpNN (\eps) \HBMgNN + \eps \begin{pmatrix}
            0 & \hoppingT_2 + \Tu (r)\\
            \hoppingT_2^\dagger + \Tu^\dagger (r) & 0
        \end{pmatrix}
    \end{align*}
    consists of the off-diagonal blocks of $\Hfull$, and $\Hoff_0$ is the operator $\Hoff$ with $\ls = 0$.
    To this end, it is useful to
    define the functions $\tj_j : \mathbb{R}^2 \to \mathbb{C}^{2 \times 2}$ by
    \begin{align}\label{eq:tj}
    \begin{split}
        \tj_0 (r) &= |\Gamma|^{-1}e^{\frac{i4\pi \beta}{3a} r_2} \begin{pmatrix}
        1 & 1\\
        1 & 1
    \end{pmatrix}, \qquad 
    \tj_1 (r) = |\Gamma|^{-1} e^{\frac{i4\pi \beta}{3a} (-\frac{\sqrt{3}}{2} r_1 - \frac{1}{2} r_2)} \begin{pmatrix}
        1 & \sq{e^{-i2\pi/3}}\\
        \sq{e^{i2\pi/3}} & 1
    \end{pmatrix},\\
    \tj_2 (r) &= |\Gamma|^{-1} e^{\frac{i4\pi \beta}{3a} (\frac{\sqrt{3}}{2} r_1 - \frac{1}{2} r_2)} \begin{pmatrix}
        1 & \sq{e^{i2\pi/3}}\\
        \sq{e^{-i2\pi/3}} & 1
    \end{pmatrix},
    \end{split}
    \end{align}
    so that 
    \begin{align}\label{eq:Tj}
    \begin{split}
    \hoppingT(r) &= \lambda_0 \tj_0 (r) + \lambda_2 e^{-i b_2 \cdot \ls}\tj_1 (r) + \lambda_4 e^{ib_1 \cdot \ls}\tj_2 (r), \\
    \hoppingT_{{\rm NNN}} (r) &= \lambda_3 e^{i (b_1 - b_2) \cdot \ls} \tj_0 (-2r) + \lambda_5 e^{i (b_1 + b_2) \cdot \ls} \tj_1 (-2r)+ 
    \lambda_1 e^{-i (b_1 + b_2) \cdot \ls} \tj_2 (-2r),\\
    \hoppingT_{\nabla, {\rm NN}} &= \lambda_0 \tj_0 (r) D_{r_1} + \lambda_2 e^{-ib_2 \cdot \ls} \tj_1 (r) (-\frac{1}{2}D_{r_1} + \frac{\sqrt{3}}{2} D_{r_2})+ \lambda_4 e^{ib_1 \cdot \ls}\tj_2 (r) (-\frac{1}{2}D_{r_1} - \frac{\sqrt{3}}{2} D_{r_2}),\\
    \hoppingT_2 &= \mu_0 \tj_0 (r) D_{r_2} + \mu_1 e^{-ib_2 \cdot \ls} \tj_1 (r) (-\frac{\sqrt{3}}{2}D_{r_1} - \frac{1}{2} D_{r_2}) + \mu_2 e^{ib_1 \cdot \ls} \tj_2 (r)(\frac{\sqrt{3}}{2}D_{r_1} - \frac{1}{2} D_{r_2}),\\
    \Tu (r) &= \frac{\beta |K|}{2} (\mu_0 \tj_0 (r) + \mu_1 e^{-i b_2 \cdot \ls}\tj_1 (r) + \mu_2 e^{ib_1 \cdot \ls}\tj_2 (r)).
    \end{split}
    \end{align} 
    A direct calculation using the identities
    \begin{align}\label{eq:phi_w}
        \phi = \frac{1}{6} \ls \cdot (b_1 - b_2), \qquad \fw = \frac{a}{4\pi \beta} (\sqrt{3} (b_1 + b_2) \cdot \ls, (-b_1 + b_2) \cdot \ls)
    \end{align}
    reveals that 
    \begin{align*}
        e^{-2i\phi} \tj_0 (r-\fw) = \tj_0 (r), \qquad
        e^{-2i\phi}e^{-ib_2 \cdot \ls} \tj_1 (r-\fw) = \tj_1 (r), \qquad
        e^{-2i\phi}e^{ib_1 \cdot \ls} \tj_2 (r-\fw) = \tj_2 (r).
    \end{align*}
    Making the dependence of $\hoppingT$ on $\ls$ explicit, it follows that
    \begin{align}\label{eq:off}
        \cU \begin{pmatrix}
            0 & \hoppingT (r;\ls)\\
            \hoppingT^\dagger (r;\ls) & 0
        \end{pmatrix} \cU^\dagger = \begin{pmatrix}
            0 & e^{-2i\phi} \hoppingT (r-\fw; \ls)\\
            e^{2i\phi} \hoppingT^\dagger (r-\fw; \ls) & 0
        \end{pmatrix} = \begin{pmatrix}
            0 & \hoppingT(r; 0)\\
            \hoppingT^\dagger (r;0) & 0
        \end{pmatrix}.
    \end{align}
    The above still holds if $\hoppingT$ is replaced by $\hoppingT_{\nabla, {\rm NN}}$, $\hoppingT_2$ or $\Tu$, as
    the operators $D_{r_j}$ commute with translation. 
    Finally, one can verify by \eqref{eq:phi_w} that
    \begin{align*}
        e^{-2i\phi}e^{i(b_1-b_2) \cdot \ls} \tj_0 (r+2\fw) = \tj_0 (r), \quad
        e^{-2i\phi}e^{i(b_1+b_2) \cdot \ls} \tj_1 (r+2\fw) = \tj_1 (r), \quad
        e^{-2i\phi}e^{-i(b_1+b_2) \cdot \ls} \tj_2 (r+2\fw) = \tj_2 (r);
    \end{align*}
    hence $\hoppingT$ in \eqref{eq:off} can also be replaced by $\hoppingT_{{\rm NNN}}$ and the proof is complete.
\end{proof}

\begin{proof}[Proof of Propositon \ref{prop:continuum_symmetries}]
By Proposition \ref{prop:unitary}, it suffices to prove the result when $\ls = 0$.
    \begin{enumerate}
        \item 
        Let $v \in \cR_m$. We observe that the diagonal blocks of $\Hfull$ are differential operators with constant coefficients, and thus commute with $\cT_v$. It remains to show that the off-diagonal blocks also commute with $\cT_v$. We have $v = n_1 a_{m,1} + n_2 a_{m,2}$ for some $(n_1, n_2) \in \mathbb{Z}^2$, which implies that $$\hoppingT(r - v) = e^{-i2\pi (n_1 + n_2)/3} \hoppingT(r),$$
        hence
        \begin{align*}
        &\left[
            \begin{pmatrix}
                0 & \hoppingT(r)\\
                \hoppingT^\dagger (r) & 0
            \end{pmatrix}, \cT_v \right] f(r)\\
            &\hspace{2cm} = 
            \begin{pmatrix}
                0 & (e^{-i\frac{4\pi \beta}{3a}v_2} - e^{-i2\pi (n_1 + n_2)/3}) \hoppingT(r)\\
                (1 - e^{-i\frac{4\pi \beta}{3a}v_2} e^{i2\pi (n_1 + n_2)/3})\hoppingT^\dagger (r) & 0
            \end{pmatrix} f(r-v).
        \end{align*}
        Since $v_2 = \frac{(n_1 + n_2) a}{2\beta}$, we conclude that
        $$\left[
            \begin{pmatrix}
                0 & \hoppingT(r)\\
                \hoppingT^\dagger (r) & 0
            \end{pmatrix}, \cT_v \right] = 0.$$
        The other off-diagonal blocks of $\Hfull$ are handled similarly.
        \item We write $\cD = \diag (\cP \cC \sigma_1, \cP \cC \sigma_1)$ and observe that
        \begin{align}\label{eq:LPC}
        \begin{split}
            L \cP \cC \sigma_1 f (r) = L \begin{pmatrix} \overline{f^B (-r)}\\
            \overline{f^A (-r)}\end{pmatrix} = 
            \begin{pmatrix}
            -\alpha (D_{r_1} - i D_{r_2})\bar{f}^A (-r)\\
            -\bar{\alpha} (D_{r_1} + i D_{r_2}) \bar{f}^B (-r)
            \end{pmatrix}&= \cP \cC \begin{pmatrix}
                \bar{\alpha} (D_{r_1} + i D_{r_2}) f^A (r)\\
                \alpha (D_{r_1} - i D_{r_2}) f^B (r)
            \end{pmatrix}\\
            &= \cP \cC \sigma_1 L f(r)
        \end{split}
        \end{align}
        while the functions defined in \eqref{eq:tj} satisfy $\sigma_1 \tj^\dagger_j (-r) \sigma_1 = \tj_j (r)$. Using \eqref{eq:Tj} with $\ls = 0$ and the fact that $[D_{r_j}, \cP \cC] = 0$, our assumption that $\lambda_i, \mu_j \in \mathbb{R}$ implies that the operators $\HBM, \HBMNNN, \HBMgNN$ 
        all commute with $\cD$, as do the off-diagonal blocks of $\HBMt$. 
        To handle the diagonal blocks of $\HBMt$, we write
        \begin{align*}
            \cS \cP \cC \sigma_1 f(r) = \cS \begin{pmatrix}
                \overline{f^B (-r)}\\
                \overline{f^A (-r)}
            \end{pmatrix} &= \frac{1}{2} \begin{pmatrix}
                \alpha_d (D_{r_1}^2 + D_{r_2}^2) \bar{f}^B (-r) + \alpha_o (D_{r_1}^2 - D_{r_2}^2 + 2 i D_{r_1 r_2} \bar{f}^A (-r)\\
                \overline{\alpha_0} (D_{r_1}^2 - D_{r_2}^2 - 2i D_{r_1 r_2}) \bar{f}^B (-r) + \alpha_d (D_{r_1}^2 + D_{r_2}^2) \bar{f}^A (-r)
            \end{pmatrix}\\ 
            &= 
            \cP \cC \sigma_1 \cS f(r),
        \end{align*}
        where the last equality follows from the assumption that $\alpha_d \in \mathbb{R}$.
        Moreover, using \eqref{eq:LPC}, it follows that
        \begin{align*}
            \cP \cC \sigma_1 i \sigma_3 L = \cP \cC \sigma_2 L = - \sigma_2 \cP \cC L = i \sigma_3 \sigma_1 \cP \cC L = i \sigma_3 L \cP \cC \sigma_1,
        \end{align*}
        which completes the proof.
        \item Define the operator $\mathscr{R}_1 f(r) := \diag (1, e^{i2\pi/3}) f(\rot^\top_{2\pi/3} r)$ so that $\mathscr{R} = \diag (\mathscr{R}_1, \mathscr{R}_1)$. We then have
        \begin{align*}
            L \mathscr{R}_1 f (r) = \begin{pmatrix}\alpha e^{i 2\pi/3} (1,-i) \cdot \rot_{2\pi/3} \nabla f^B (\rot_{2\pi/3}^\top r)\\
            \bar{\alpha} (1,i) \cdot \rot_{2\pi/3} \nabla f^A (\rot_{2\pi/3}^\top r)\end{pmatrix} = \begin{pmatrix}
                \alpha (1,-i) \cdot \nabla f^B (\rot_{2\pi/3}^\top r)\\
                \bar{\alpha} e^{i2\pi/3} (1,i) \cdot \nabla f^A (\rot_{2\pi/3}^\top r)
            \end{pmatrix} = \mathscr{R}_1 L f (r)
        \end{align*}
        and
        \begin{align}\label{eq:rot_tj}
            \diag (1, e^{i2\pi/3}) \tj_j (\rot^\top_{2\pi/3} r) \diag (1, e^{-i2\pi/3}) = \tj_{j+1} (r), \qquad j = 0,1,2,
        \end{align}
        where it is understood that $\tj_3 := \tj_0$. Since $\lambda_0 = \lambda_2 = \lambda_4$, it follows from \eqref{eq:Tj} with $\ls = 0$ that
        \begin{align*}
            [\HBM, \mathscr{R}] = \diag([L, \mathscr{R}_1], [L, \mathscr{R}_1]) + \begin{pmatrix}
                0 & [\hoppingT(r), \mathscr{R}_1]\\
                [\hoppingT^\dagger (r), \mathscr{R}_1] & 0
            \end{pmatrix} = 0,
        \end{align*}
        as $[\hoppingT(r), \mathscr{R}_1] f(r) = (\hoppingT(r) \diag (1, e^{i 2\pi/3}) - \diag (1, e^{i 2\pi/3})\hoppingT(\rot^\top_{2\pi/3} r)) f(\rot^\top_{2\pi/3} r)$ vanishes by \eqref{eq:rot_tj}.
        Since the maps $r \mapsto -2r$ and $r \mapsto \rot^\top_{2\pi/3}r$ commute with each other, we similarly conclude that $[\HBMNNN, \mathscr{R}] = 0$. Writing $D = (D_1, D_2)$, $e_1 = (1,0)$ and
        \begin{align*}
            \hoppingT_{\nabla, {\rm NN}} = \lambda_0 \tj_0 (r) e_1 \cdot D + \lambda_2 \tj_1 (r) (\rot_{2\pi/3} e_1)\cdot D + \lambda_4 \tj_2 (r) (\rot^\top_{2\pi/3} e_1)\cdot D,
        \end{align*}
        it follows that
        \begin{align*}
            &\hoppingT_{\nabla, {\rm NN}} \mathscr{R}_1 f(r)= \Big( \lambda_0 \tj_0 (r) \diag (1, e^{i2\pi/3}) (e_1 \cdot \rot_{2\pi/3} D)\\
            &\hspace{2.5cm} + \lambda_2 \tj_1 (r) \diag (1, e^{i2\pi/3}) (e_1 \cdot D) + \lambda_4 \tj_2 (r)\diag (1, e^{i2\pi/3})(e_1 \cdot \rot^\top_{2\pi/3} D)\Big) f(\rot^\top_{2\pi/3}r)
        \end{align*}
        and
        \begin{align*}
            &\mathscr{R}_1 \hoppingT_{\nabla, {\rm NN}} f(r)= \Big( \lambda_0 \diag (1, e^{i2\pi/3}) \tj_0 (\rot^\top_{2\pi/3} r) (e_1 \cdot D) \\
            &\hspace{0.25cm} + \lambda_2 \diag (1, e^{i2\pi/3}) \tj_1 (\rot^\top_{2\pi/3} r) (e_1 \cdot \rot^{\top}_{2\pi/3}D) + \lambda_4 \diag (1, e^{i2\pi/3})\tj_2 (\rot^\top_{2\pi/3} r)(e_1 \cdot \rot_{2\pi/3} D)\Big) f(\rot^\top_{2\pi/3}r),
        \end{align*}
        hence $[\hoppingT_{\nabla, {\rm NN}}, \mathscr{R}_1] = 0$ by \eqref{eq:rot_tj} and the assumption that $\lambda_0 = \lambda_2 = \lambda_4$. The same argument establishes that
        \begin{align*}
            \hoppingT^\dagger_{\nabla, {\rm NN}} = \overline{\lambda_0} \tj^\dagger_0 (r) e_1 \cdot D + \overline{\lambda_2} \tj^\dagger_1 (r) (\rot_{2\pi/3} e_1)\cdot D + \overline{\lambda_4} \tj^\dagger_2 (r) (\rot^\top_{2\pi/3} e_1)\cdot D
        \end{align*}
        also commutes with $\mathscr{R}_1$, and thus $[\HBMgNN, \mathscr{R}] = 0$. Similarly, with $e_2 = (0,1)$,
        \begin{align*}
            \hoppingT_2 = \mu_0 \tj_0 (r) e_2 \cdot D + \mu_1 \tj_1 (r) (\rot_{2\pi/3} e_2) \cdot D + \mu_2 \tj_2 (r)(\rot^\top_{2\pi/3} e_2) \cdot D
        \end{align*}
        and $\Tu (r)$ commute with $\mathscr{R}_1$, as does $\hoppingT^\dagger_2 + \Tu^\dagger (r)$. Moreover,
        \begin{align*}
            \HBMsd \mathscr{R}_1 f (r) = \frac{1}{2} \begin{pmatrix}
                -\alpha_d \Delta f^A (\rot^\top_{2\pi/3} r) + \alpha_o e^{i2\pi/3}(-i, 1) \cdot \rot_{2\pi/3} \nabla^2 f^B (\rot^\top_{2\pi/3} r) \rot^\top_{2\pi/3} (-i,1)\\
                \overline{\alpha_o} (i, 1) \cdot \rot_{2\pi/3} \nabla^2 f^A (\rot^\top_{2\pi/3} r) \rot^\top_{2\pi/3} (i,1) - \alpha_d e^{i2\pi/3} \Delta f^B (\rot^\top_{2\pi/3} r)
            \end{pmatrix},
        \end{align*}
        which by $\rot^\top_{2\pi/3} (-i,1) = e^{i2\pi/3} (-i,1)$ and $\rot^\top_{2\pi/3} (i,1) = e^{-i2\pi/3} (i,1)$ implies that
        \begin{align*}
            \HBMsd \mathscr{R}_1 f (r) = \frac{1}{2} \begin{pmatrix}
                -\alpha_d \Delta f^A (\rot^\top_{2\pi/3} r) + \alpha_o (-i, 1) \cdot \nabla^2 f^B (\rot^\top_{2\pi/3} r) (-i,1)\\
                \overline{\alpha_o}e^{i2\pi/3} (i, 1) \cdot \nabla^2 f^A (\rot^\top_{2\pi/3} r) (i,1) - \alpha_d e^{i2\pi/3} \Delta f^B (\rot^\top_{2\pi/3} r)
            \end{pmatrix}= \mathscr{R}_1 \HBMsd f(r).
        \end{align*}
        Using that
        $\mathscr{R}_1$ commutes with both $L$ and $\sigma_3$ to handle the other diagonal term, we conclude that $[\HBMt, \mathscr{R}] = 0$, and thus $[\Hfull, \mathscr{R}] = 0$ as desired.
        \item We write
        \begin{align*}
            \cM_x = \begin{pmatrix}
                0 & \cC \sm_x\\
                \cC \sm_x & 0
            \end{pmatrix}, \qquad \sm_x f (r_1, r_2) := f(-r_1, r_2),
        \end{align*}
        with $\cC$ the complex conjugation operator defined in \eqref{eq:PC}. We begin with the diagonal blocks of $\Hfull$. Observe that
        \begin{align*}
            \left[ \diag (L, L), \cM_x \right] = \begin{pmatrix}
                0 & [L, \cC \sm_x]\\
                [L, \cC \sm_x] & 0
            \end{pmatrix} = 
            \begin{pmatrix}
                0 & 0\\
                0 & 0
            \end{pmatrix},
        \end{align*}
        as the assumption $\alpha \in \mathbb{R}$ implies that
        \begin{align}\label{eq:LCM}
            L \cC \sm_x f(r) = \begin{pmatrix}
                \alpha (-D_{r_1} - i D_{r_2}) \overline{f^B (-r_1, r_2)}\\
                \bar{\alpha} (-D_{r_1} + i D_{r_2}) \overline{f^A (-r_1, r_2)}
            \end{pmatrix} = \cC \sm_x L f(r).
        \end{align}
        Next, we write
        \begin{align*}
            \diag (-i \sigma_3 L, i \sigma_3 L) \cM_x &= \begin{pmatrix}
                0 & -i\sigma_3 L \cC \sm_x\\
                i\sigma_3 L \cC \sm_x & 0
            \end{pmatrix}, \\
            \cM_x \diag (-i \sigma_3 L, i \sigma_3 L) &= \begin{pmatrix}
                0 & \cC \sm_x i \sigma_3 L\\
                -\cC \sm_x i \sigma_3 L & 0
            \end{pmatrix}
        \end{align*}
        and use \eqref{eq:LCM} to conclude that
        \begin{align*}
            -i\sigma_3 L \cC \sm_x = -i \sigma_3 \cC \sm_x L = -i \sigma_3 \cC \sm_x (-i\sigma_3) i \sigma_3 L = -i\sigma_3 i \sigma_3 \cC \sm_x i \sigma_3 L = \cC \sm_x i \sigma_3 L;
        \end{align*}
        therefore, $[\diag (-i \sigma_3 L, i \sigma_3 L), \cM_x] = 0$.
        The remaining diagonal blocks satisfy
        \begin{align*}
            \left[ \diag (\HBMsd, \HBMsd), \cM_x \right] = \begin{pmatrix}
                0 & [\HBMsd, \cC \sm_x]\\
                [\HBMsd, \cC \sm_x] & 0
            \end{pmatrix}= 
            \begin{pmatrix}
                0 & 0\\
                0 & 0
            \end{pmatrix},
        \end{align*}
        since
        \begin{align*}
            \HBMsd \cC \sm_x f(r) = \frac{1}{2} \begin{pmatrix}
                \alpha_d (D_{r_1}^2 + D_{r_2}^2) \overline{f^A (-r_1, r_2)} + \alpha_o (D_{r_1}^2 - D_{r_2}^2 - 2iD_{r_1 r_2} \overline{f^B (-r_1, r_2)}\\
                \overline{\alpha_o} (D_{r_1}^2 - D_{r_2}^2 + 2iD_{r_1 r_2}) \overline{f^A (-r_1, r_2)} + \alpha_d (D_{r_1}^2 + D_{r_2}^2) \overline{f^B (-r_1, r_2)}
            \end{pmatrix}
            = \cC \sm_x \HBMsd f(r)
        \end{align*}
        by the assumption that $\alpha_o \in \mathbb{R}$.
        Therefore, we have shown that
        \begin{align*}
            [\diag (L + \eps \HBMsd - \frac{\eps}{2} \beta i\sigma_3 L, L + \eps \HBMsd + \frac{\eps}{2} \beta i\sigma_3 L), \cM_x] = 0.
        \end{align*}
        
        For the off-diagonal blocks, we use that
        \begin{align}\label{eq:tj_transpose}
            \tj_0^\top (-r_1, r_2) = \tj_0 (r_1, r_2), \qquad \tj_1^\top (-r_1, r_2) = \tj_2 (r_1, r_2),
        \end{align}
        which by \eqref{eq:Tj} and the assumption that $\lambda_2 = \lambda_4$ implies that $\hoppingT (r_1, r_2)=\hoppingT^\top (-r_1, r_2)$.
        Since
        \begin{align*}
            \begin{pmatrix}
                0 & \hoppingT(r)\\
                \hoppingT^\dagger (r) & 0
            \end{pmatrix} \cM_x = \diag (\hoppingT(r) \cC \sm_x, \hoppingT^\dagger (r) \cC \sm_x), \;
            \cM_x \begin{pmatrix}
                0 & \hoppingT(r)\\
                \hoppingT^\dagger (r) & 0
            \end{pmatrix} =
            \diag (
                \cC \sm_x \hoppingT^\dagger (r),
                \cC \sm_x \hoppingT (r))
        \end{align*}
        with
        \begin{align*}
            \hoppingT(r_1, r_2) \cC \sm_x = \cC \sm_x \sm_x \cC \hoppingT(r_1, r_2) \cC \sm_x = \cC \sm_x \overline{\hoppingT (-r_1, r_2)}= \cC \sm_x (\hoppingT^\top (-r_1, r_2))^\dagger = \cC \sm_x \hoppingT^\dagger (r_1, r_2),
        \end{align*}
        we conclude that
        \begin{align*}
            \left[ \begin{pmatrix}
                0 & \hoppingT(r)\\
                \hoppingT^\dagger (r) & 0
            \end{pmatrix}, \cM_x\right] = 0.
        \end{align*}
        Since the maps $r \mapsto -2r$ and $(r_1, r_2) \mapsto (-r_1, r_2)$ commute, 
        the assumption that $\lambda_1 = \lambda_5$ similarly implies that $[\HBMNNN, \cM_x] = 0$.
        For the next off-diagonal contribution, we have
        \begin{align*}
            \cC \sm_x \tj_0^\dagger (r_1, r_2) D_{r_1} \sm_x \cC &= \tj_0^\top (-r_1, r_2) D_{r_1}, \\
            \cC \sm_x \tj_1^\dagger (r_1, r_2) (-\frac{1}{2}D_{r_1} + \frac{\sqrt{3}}{2} D_{r_2}) \sm_x \cC &= \tj_1^\top (-r_1, r_2) (-\frac{1}{2}D_{r_1} - \frac{\sqrt{3}}{2} D_{r_2}),\\
            \cC \sm_x \tj_2^\dagger (r_1, r_2) (-\frac{1}{2}D_{r_1} - \frac{\sqrt{3}}{2} D_{r_2}) \sm_x \cC &= \tj_2^\top (-r_1, r_2) (-\frac{1}{2}D_{r_1} + \frac{\sqrt{3}}{2} D_{r_2}),
        \end{align*}
        which by \eqref{eq:Tj}, \eqref{eq:tj_transpose} and the assumption that $\lambda_2= \lambda_4$ implies that $\cC \sm_x \hoppingT^\dagger_{\nabla, {\rm NN}} \sm_x \cC = \hoppingT_{\nabla, {\rm NN}}$. Since
        \begin{align*}
            \HBMgNN \cM_x &= \diag (\hoppingT_{\nabla, {\rm NN}}, \cC \sm_x, \hoppingT^\dagger_{\nabla, {\rm NN}}\cC \sm_x),\quad
            \cM_x \HBMgNN = \diag (\cC \sm_x \hoppingT^\dagger_{\nabla, {\rm NN}}, \cC \sm_x \hoppingT_{\nabla, {\rm NN}}),
        \end{align*}
        we conclude that
        $
            \left[ \HBMgNN, \cM_x\right] = 0
        $
        as desired.
        \sq{It remains to analyze the second-order off-diagonal blocks. To this end,
        we write
        \begin{align*}
            \left[ \begin{pmatrix}
                0 & \hoppingT_2 + \Tu (r) \\
                \hoppingT_2^\dagger + \Tu^\dagger (r)
            \end{pmatrix}, \cM_x \right] = \diag (\cA, -\cC \sm_x \cA \sm_x \cC),
        \end{align*}
        where $\cA := (\hoppingT_2 + \Tu) \cC \sm_x - \cC \sm_x (\hoppingT_2^\dagger +\Tu^\dagger)$. Thus it suffices to show that $\cA = 0$, or equivalently, 
        \begin{align}\label{eq:offdiag2}
            \cC \sm_x (\hoppingT_2^\dagger +\Tu^\dagger) \sm_x \cC = \hoppingT_2 + \Tu.
        \end{align}
        By \eqref{eq:adjoint}, we have
        \begin{align}\label{eq:dagger_tj}
        \begin{split}
            \hoppingT_2^\dagger &= \overline{\mu_0} \tj_0 ^\dagger (r) D_{r_2} + \overline{\mu_1} \tj_1^\dagger (r)(-\frac{\sqrt{3}}{2}D_{r_1} - \frac{1}{2} D_{r_2}) + \overline{\mu_2} \tj_2^\dagger (r)(\frac{\sqrt{3}}{2}D_{r_1} - \frac{1}{2} D_{r_2}) - 2 \Tu^\dagger,\\
            \Tu^\dagger &= \frac{\beta |K|}{2} (\overline{\mu_0} \tj_0^\dagger (r) + \overline{\mu_1} \tj_1^\dagger (r) + \overline{\mu_2} \tj_2^\dagger (r)).
        \end{split}
        \end{align}
        Using the identities
        \begin{align*}
            \cC \sm_x \tj_0^\dagger (r_1, r_2) D_{r_2} \sm_x \cC &= -\tj_0^\top (-r_1, r_2) D_{r_2}, \\
            \cC \sm_x \tj_1^\dagger (r_1, r_2) (-\frac{\sqrt{3}}{2}D_{r_1} - \frac{1}{2} D_{r_2}) \sm_x \cC &= \tj_1^\top (-r_1, r_2) (-\frac{\sqrt{3}}{2}D_{r_1} + \frac{1}{2} D_{r_2}),\\
            \cC \sm_x \tj_2^\dagger (r_1, r_2) (\frac{\sqrt{3}}{2}D_{r_1} - \frac{1}{2} D_{r_2}) \sm_x \cC &= \tj_2^\top (-r_1, r_2) (\frac{\sqrt{3}}{2}D_{r_1} + \frac{1}{2} D_{r_2}),
        \end{align*}
        we obtain by \eqref{eq:tj_transpose} that
        \begin{align*}
            &\cC \sm_x (\hoppingT_2^\dagger +\Tu^\dagger) \sm_x \cC \\
            &\hspace{1cm} = -\mu_0 \tj_0 (r) D_{r_2} + \mu_1 \tj_2 (r) (-\frac{\sqrt{3}}{2}D_{r_1} + \frac{1}{2} D_{r_2}) + \mu_2 \tj_1 (r) (\frac{\sqrt{3}}{2}D_{r_1} + \frac{1}{2} D_{r_2}) - \cC \sm_x \Tu^\dagger \sm_x \cC,\\
            &\cC \sm_x \Tu^\dagger \sm_x \cC = \frac{\beta |K|}{2} (-\mu_0 \tj_0 (r) + \mu_1 \tj_2 (r) + \mu_2 \tj_1 (r)).
        \end{align*}
        Using the assumptions $\mu_0 = 0$ and $\mu_1 = -\mu_2$, we apply \eqref{eq:Tj} to conclude that \eqref{eq:offdiag2} holds.}

        \item We write 
        \begin{align*}
            \cM_y = \begin{pmatrix}
                0 & \sigma_1 \sm_y\\
                \sigma_1 \sm_y & 0
            \end{pmatrix}, \qquad \sm_y f(r_1, r_2) := f(r_1, -r_2)
        \end{align*}
        and, using the assumption that $\alpha \in \mathbb{R}$ to justify the last equality below,
        \begin{align*}
            L\sigma_1 \sm_y f(r) = L \begin{pmatrix}
                f^B(r_1, -r_2)\\
                f^A (r_1, -r_2)
            \end{pmatrix} =
            \begin{pmatrix}\alpha (D_{r_1} + i D_{r_2})f^A (r_1, -r_2)\\
            \bar{\alpha}(D_{r_1} - i D_{r_2})f^B (r_1, -r_2)\end{pmatrix} = \sigma_1 \sm_y L f(r);
        \end{align*}
        hence $[\diag (L, L), \cM_y]$. Moreover, we observe that
        \begin{align*}
            \diag (-\sigma_3 L, \sigma_3 L) \cM_y &= \begin{pmatrix}
                0 & -\sigma_3 L \sigma_1 \sm_y\\
                \sigma_3 L \sigma_1 \sm_y & 0
            \end{pmatrix},\\
            \cM_y \diag (-\sigma_3 L, \sigma_3 L) &= \begin{pmatrix}
                0 & \sigma_1 \sm_y \sigma_3 L\\
                -\sigma_1 \sm_y \sigma_3 L & 0
            \end{pmatrix},
        \end{align*}
        where $-\sigma_3 L \sigma_1 \sm_y = -\sigma_3 \sigma_1 \sm_y L = \sigma_1 \sigma_3 \sm_y L = \sigma_1 \sm_y \sigma_3 L$; hence $[\diag (-\sigma_3 L, \sigma_3 L), \cM_y] = 0$. Next,
        \begin{align*}
            \diag (\HBMsd, \HBMsd) \cM_y = \begin{pmatrix}
                0 & \HBMsd \sigma_1 \sm_y\\
                \HBMsd \sigma_1 \sm_y & 0
            \end{pmatrix}, \qquad
            \cM_y \diag (\HBMsd, \HBMsd) = \begin{pmatrix}
                0 & \sigma_1 \sm_y\HBMsd \\
                \sigma_1 \sm_y\HBMsd & 0
            \end{pmatrix},
        \end{align*}
        where
        \begin{align*}
            \HBMsd \sigma_1 \sm_y f(r) = \frac{1}{2} \begin{pmatrix}
                \alpha_d (D_{r_1}^2 + D_{r_2}^2) f^B(r_1, -r_2) + \alpha_o (D_{r_1}^2 - D_{r_2}^2 - 2iD_{r_1 r_2}) f^A(r_1, -r_2) \vspace{0.2cm}\\
                \overline{\alpha_o} (D_{r_1}^2 - D_{r_2}^2 + 2iD_{r_1 r_2}) f^B(r_1, -r_2) + \alpha_d (D_{r_1}^2 + D_{r_2}^2)f^A(r_1, -r_2)
            \end{pmatrix} = \sigma_1 \sm_y \HBMsd f(r)
        \end{align*}
        since $\alpha_o \in \mathbb{R}$. Thus we have shown that the diagonal blocks of $\Hfull$ commute with $\cM_y$; that is,
        \begin{align*}
            [\diag (L + \eps \HBMsd - \frac{\eps}{2} \beta i\sigma_3 L, L + \eps \HBMsd + \frac{\eps}{2} \beta i\sigma_3 L), \cM_y] = 0.
        \end{align*}
        
        For the interlayer coupling terms, we use the identities
        \begin{align}\label{eq:tj_y}
            \sigma_1 \tj^\dagger_0 (r_1, -r_2) \sigma_1 = \tj_0 (r_1, r_2), \qquad \sigma_1 \tj^\dagger_1 (r_1, -r_2) \sigma_1 = \tj_2 (r_1, r_2).
        \end{align}
        Recalling \eqref{eq:Tj}, since
        \begin{align*}
            \begin{pmatrix}
                0 & \hoppingT(r)\\
                \hoppingT^\dagger (r) & 0
            \end{pmatrix} \cM_y &= \diag (\hoppingT(r) \sigma_1 \sm_y, \hoppingT^\dagger (r) \sigma_1 \sm_y), \\
            \cM_y \begin{pmatrix}
                0 & \hoppingT(r)\\
                \hoppingT^\dagger (r) & 0
            \end{pmatrix} &= \diag (\sigma_1 \sm_y \hoppingT^\dagger (r), \sigma_1 \sm_y \hoppingT(r)),
        \end{align*}
        our assumptions that $\lambda_0 \in \mathbb{R}$ and $\overline{\lambda_2} = \lambda_4$ imply that $$\left[\begin{pmatrix}
                0 & \hoppingT(r)\\
                \hoppingT^\dagger (r) & 0
            \end{pmatrix}, \cM_y \right] = 0.$$
        Since the maps $(r_1, r_2) \mapsto (r_1, - r_2)$ and
        $r \mapsto - 2r$ commute with each other, a parallel argument using the assumptions $\lambda_3 \in \mathbb{R}$ and $\overline{\lambda_1} = \lambda_5$ reveals that $[\HBMNNN, \cM_y] = 0$. Moreover, \eqref{eq:tj_y} implies that
        \begin{align*}
            &\sigma_1 \sm_y \tj_0^\dagger (r) D_{r_1} \sm_y \sigma_1 = \tj_0 (r) D_{r_1}, \quad \sigma_1 \sm_y \tj_1^\dagger (r) (-\frac{1}{2} D_{r_1} + \frac{\sqrt{3}}{2} D_{r_2}) \sm_y \sigma_1 = \tj_2 (r)(-\frac{1}{2} D_{r_1} - \frac{\sqrt{3}}{2} D_{r_2})\\
            & \sigma_1 \sm_y \tj_2^\dagger (r) (-\frac{1}{2} D_{r_1} - \frac{\sqrt{3}}{2} D_{r_2}) \sm_y \sigma_1 = \tj_1 (r)(-\frac{1}{2} D_{r_1} + \frac{\sqrt{3}}{2} D_{r_2}),
        \end{align*}
        from which we conclude
        \begin{align*}
            [\HBMgNN, \cM_y] = \diag (\hoppingT_{\nabla, {\rm NN}} \sigma_1 \sm_y - \sigma_1 \sm_y \hoppingT^\dagger_{\nabla, {\rm NN}},\; \hoppingT^\dagger_{\nabla, {\rm NN}} \sigma_1 \sm_y - \sigma_1 \sm_y \hoppingT_{\nabla, {\rm NN}}) = 0.
        \end{align*}
        \sq{The remaining off-diagonal terms satisfy
        \begin{align*}
            \left[ \begin{pmatrix}
                0 & \hoppingT_2 + \Tu (r) \\
                \hoppingT_2^\dagger + \Tu^\dagger (r)
            \end{pmatrix}, \cM_y \right] = \diag (\cB, -\sigma_1 \sm_y \cB \sm_y \sigma_1),
        \end{align*}
        where $\cB := (\hoppingT_2 + \Tu) \sigma_1 \sm_y - \sigma_1 \sm_y (\hoppingT_2^\dagger +\Tu^\dagger)$. Therefore we must show that $\cB = 0$, or equivalently,
        \begin{align}\label{eq:offdiag2y}
            \sigma_1 \sm_y (\hoppingT_2^\dagger +\Tu^\dagger) \sigma_1 \sm_y = \hoppingT_2 + \Tu.
        \end{align}
        We verify that
        \begin{align*}
            \sigma_1\sm_y \tj_0^\dagger (r_1, r_2) D_{r_2} \sm_y \sigma_1 &= -\sigma_1 \tj_0^\dagger (r_1, -r_2) \sigma_1 D_{r_2}, \\
            \sigma_1\sm_y \tj_1^\dagger (r_1, r_2) (-\frac{\sqrt{3}}{2}D_{r_1} - \frac{1}{2} D_{r_2}) \sm_y \sigma_1 &= \sigma_1 \tj_1^\dagger (r_1, -r_2)\sigma_1 (-\frac{\sqrt{3}}{2}D_{r_1} + \frac{1}{2} D_{r_2}),\\
            \sigma_1\sm_y \tj_2^\dagger (r_1, r_2) (\frac{\sqrt{3}}{2}D_{r_1} - \frac{1}{2} D_{r_2}) \sm_y \sigma_1 &= \sigma_1 \tj_2^\dagger (r_1, -r_2) \sigma_1 (\frac{\sqrt{3}}{2}D_{r_1} + \frac{1}{2} D_{r_2}),
        \end{align*}
        then conclude by \eqref{eq:dagger_tj} and \eqref{eq:tj_y} that
        \begin{align*}
            &\sigma_1 \sm_y (\hoppingT_2^\dagger +\Tu^\dagger) \sm_y \sigma_1\\
            &\hspace{1cm} = - \overline{\mu_0}\tj_0(r) D_{r_2} + \overline{\mu_1} \tj_2 (r) (-\frac{\sqrt{3}}{2}D_{r_1} + \frac{1}{2} D_{r_2}) + \overline{\mu_2} \tj_1 (r) (\frac{\sqrt{3}}{2}D_{r_1} + \frac{1}{2} D_{r_2}) - \sigma_1 \sm_y \Tu^\dagger \sm_y \sigma_1,\\
            &\sigma_1 \sm_y \Tu^\dagger \sm_y \sigma_1 = \frac{\beta |K|}{2} (\overline{\mu_0} \tj_0 (r) + \overline{\mu_1} \tj_2 (r) + \overline{\mu_2} \tj_1 (r)).
        \end{align*}
        Our assumptions $\mu_0 \in i \mathbb{R}$ and $\overline{\mu_1} = -\mu_2$ then imply that \eqref{eq:offdiag2y} holds, which completes the proof.
        }
    \end{enumerate}
\end{proof}

\begin{proof}[Proof of Proposition \ref{prop:sufficient}]
    \begin{enumerate}
        \item The condition $\hperp (-r; \eps) = \overline{\hperp (r; \eps)}$ immediately implies that the Fourier transform of $\hperp$ is real-valued. We conclude by \eqref{eq:sep} that $\hathperpang$ is also real-valued, hence the result follows from the definitions of the $\lambda_i$ and $\mu_j$ in \eqref{eq:rad_ang}.
        \item The assumed rotation invariance of $\hperp$ implies that $\hathperp (\rot_{2\pi/3} k; \eps) = \hathperp (k; \eps)$ for all $(k; \eps) \in \mathbb{R}^2 \times (0,1)$, thus the result follows from \eqref{eq:sep} and \eqref{eq:rad_ang}.
        \item A direct calculation reveals that
        \begin{align*}
            \partial_{k_1} \tilde{h}^{A,B} (K) = -i \sum_{R \in \cR} R_1 e^{-i R \cdot K} h (R + \tau^{A,B}) = -i \sum_{R \in \cR} R_1 e^{-i 4\pi R_1/3a} h (R_1, R_2 - a/\sqrt{3}),
        \end{align*}
        where $R = (R_1, R_2)$ above. Since $h$ is assumed to be even in its first argument, the real part of the sum vanishes and thus $\alpha = \partial_{k_1} \tilde{h}^{A,B} (K) \in \mathbb{R}$. Similarly, we find that
        \begin{align*}
            \alpha_o = \partial^2_{k_1}\tilde{h}^{A,B} (K) = - \sum_{R \in \cR} R_1^2 e^{-i 4\pi R_1/3a} h (R_1, R_2 - a/\sqrt{3}) \in \mathbb{R}.
        \end{align*}
        \item The assumption on $\hperp$ implies that $\hathperp (k_1, -k_2; \eps) = \hathperp (k_1, k_2; \eps)$ for all $(k_1, k_2; \eps) \in \mathbb{R}^2 \times (0,1)$. 
        It follows that $\hathperpang (k_1, -k_2) = \hathperpang (k_1, k_2)$ and thus $\hathperpang ' (k_1, -k_2) = -\hathperpang ' (k_1, k_2)$ for all $(k_1, k_2) \in \mathbb{S}$. The result then follows from the definitions of the $\lambda_i$ and $\mu_j$ in \eqref{eq:rad_ang}.
        \item We write
        \begin{align*}
            \overline{\hathperp (k_1, -k_2; \eps)} = \int_{\mathbb{R}^2} e^{-i (-k_1 r_1 + k_2 r_2)} \overline{\hperp (r; \eps)} {\rm d} r = \int_{\mathbb{R}^2} e^{-i (-k_1 r_1 + k_2 r_2)} \hperp (-r_1, r_2; \eps) {\rm d} r = \hathperp (k; \eps),
        \end{align*}
        with the second equality above justified by our assumption on $\hperp$. Therefore, $$\overline{\hathperpang (k_1, -k_2)} = \hathperpang (k_1, k_2), \qquad \overline{\hathperpang ' (k_1, -k_2)} = -\hathperpang ' (k_1, k_2)$$ for all $(k_1, k_2) \in \mathbb{S}$, and the result follows from \eqref{eq:rad_ang}.
    \end{enumerate}
\end{proof}

\begin{proof}[Proof of Proposition \ref{prop:ph}]
    In the proof of Proposition \ref{prop:unitary}, we showed that $\cU \HBM \cU^\dagger = \HBM_0$, where $\HBM_0$ is the operator $\HBM$ when $\ls = 0$.
    Thus we can set $\ls = 0$ without loss of generality.
    We write $$\cM'_x = \begin{pmatrix}
        \sigma_1 & 0\\
        0 & -\sigma_1
    \end{pmatrix} \sm_x,$$
    with $\sm_x$ defined in the proof of Proposition \ref{prop:continuum_symmetries} (\ref{symm:mirror_x}). Our assumption that $\alpha \in \mathbb{R}$ implies that
    \begin{align*}
        &(L \sigma_1 \sm_x + \sigma_1 \sm_x L) f(r) \\
        &\hspace{1cm} =
        \begin{pmatrix}
            \alpha (-D_{r_1} - i D_{r_2}) + \bar{\alpha} (D_{r_1} + i D_{r_2}) & 0\\
            0 & \bar{\alpha} (-D_{r_1} + i D_{r_2}) + \alpha (D_{r_1} - i D_{r_2})
        \end{pmatrix} f (-r_1, r_2) = 0,
    \end{align*}
    hence $\diag (L,L) \cM'_x + \cM'_x \diag (L,L) = 0$. The off-diagonal blocks of $\Hfull$ satisfy
    \begin{align*}
        \begin{pmatrix}
            0 & \hoppingT(r)\\
            \hoppingT^\dagger (r) & 0
        \end{pmatrix}
        \begin{pmatrix}
            \sigma_1 & 0\\
            0 & -\sigma_1
        \end{pmatrix} \sm_x &= \begin{pmatrix}
            0 & -\hoppingT(r) \sigma_1\\
            \hoppingT^\dagger (r) \sigma_1 & 0
        \end{pmatrix}\sm_x,\\
        \begin{pmatrix}
            \sigma_1 & 0\\
            0 & -\sigma_1
        \end{pmatrix} \sm_x \begin{pmatrix}
            0 & \hoppingT(r)\\
            \hoppingT^\dagger (r) & 0
        \end{pmatrix} &= \begin{pmatrix}
            0 & \sigma_1 \hoppingT(-r_1, r_2)\\
            -\sigma_1 \hoppingT^\dagger (-r_1, r_2) & 0
        \end{pmatrix} \sm_x,
    \end{align*}
    where $\hoppingT(-r_1, r_2) = \sigma_1 \hoppingT(r) \sigma_1$ by our assumption that $\lambda_2 = \lambda_4$. We conclude that
    \begin{align*}
        \begin{pmatrix}
            0 & \hoppingT(r)\\
            \hoppingT^\dagger (r) & 0
        \end{pmatrix} \cM'_x + \cM'_x \begin{pmatrix}
            0 & \hoppingT(r)\\
            \hoppingT^\dagger (r) & 0
        \end{pmatrix} = 0,
    \end{align*}
    and the proof is complete.
\end{proof}

\section{Examples of interlayer hopping functions}\label{sec:examples}
In this section, we present two more examples of interlayer hopping functions satisfying Assumpiton \ref{assumption:hperp2}.
\srq{
\begin{example}\label{ex:1}
Fix $0 < \hoppingparam < \frac{\sqrt{3}}{2}|K|$. We first define the auxiliary function $\afn : [0,\infty) \to (0,\infty)$ by
\begin{align*}
    \afn (x) := \frac{2\pi}{(|K|^2 + \hoppingparam^2)^{3/2}} \hoppingparam e^{-x \sqrt{|K|^2 + \hoppingparam^2}}\left(1+x \sqrt{|K|^2 + \hoppingparam^2}\right).
\end{align*}
We verify that $\afn$ is monotonically decreasing and 
converges to $0$ at infinity, hence the inverse map $$\afn^{-1} : \Big(0, \frac{2\pi \hoppingparam}{(|K|^2 + \hoppingparam^2)^{3/2}}\Big] \to [0, \infty)$$ is well defined and monotonicallly decreasing. 
Define
\begin{align*}
    \afn_1 (x) := \frac{2\pi}{(|K|^2 + \hoppingparam^2)^{3/2}} \hoppingparam e^{-x \sqrt{|K|^2 + \hoppingparam^2}}, \qquad 0 \le x < \infty,
\end{align*}
so that $\afn_1 \le \afn$ and thus 
\begin{align}\label{eq:ex_lower_bd}
    \frac{1}{\sqrt{|K|^2 + \hoppingparam^2}} \log \left( \frac{2\pi \hoppingparam}{(|K|^2 + \hoppingparam^2)^{3/2}y}\right)=\afn^{-1}_1 (y) \le \afn^{-1} (y), \qquad 0 < y \le \frac{2\pi \hoppingparam}{(|K|^2 + \hoppingparam^2)^{3/2}}.
\end{align}
Fix $$0 < \lambda_0 \le \frac{2\pi \hoppingparam}{(|K|^2 + \hoppingparam^2)^{3/2}}$$ and define the function $\ell : (0,1) \to (0,\infty)$ by
\begin{align*}
    \ell (\eps) = \afn^{-1} (\eps \lambda_0).
\end{align*}
Now, define the interlayer hopping function by
$$\hperp (r;\eps) := e^{-\hoppingparam \sqrt{|r|^2 + \ell^2}},$$
where we henceforth use the shorthand $\ell := \ell (\eps)$.
The Fourier transform of $\hperp (\cdot \; ; \eps)$ is
\begin{align}\label{eq:hathperp_ex1}
    \hathperp (k; \eps) = 2\pi \frac{\hoppingparam e^{-\ell \sqrt{|k|^2 + \hoppingparam^2}} (1+\ell \sqrt{|k|^2 + \hoppingparam^2})}{(|k|^2 + \hoppingparam^2)^{3/2}},
\end{align}
from which we conclude that $\hathperp (K; \eps) = \afn(\ell) = \eps \lambda_0$ and verify the equality in \eqref{eq:neighbor_bds}, with $$\hathperprad (|k|; \eps) := \hathperp (|k|,0; \eps)/\lambda_0.$$ 
It remains to verify the decay estimates in \eqref{eq:neighbor_bds} and \eqref{eq:4bounds}.
The lower bound on $\afn^{-1}$ in \eqref{eq:ex_lower_bd} implies that
\begin{align*}
    \ell (\eps) \ge \frac{1}{\sqrt{|K|^2 + \hoppingparam^2}} \log \left( \frac{2\pi \hoppingparam}{(|K|^2 + \hoppingparam^2)^{3/2}\eps \lambda_0}\right), \qquad 0 < \eps < 1.
\end{align*}
Using that
\begin{align*}
    0 \le \hathperprad (|k|; \eps) \le \frac{2\pi}{\hoppingparam^2 \lambda_0}e^{-\ell \sqrt{|k|^2 + \hoppingparam^2}} (1+\ell \sqrt{|k|^2 + \hoppingparam^2}),
\end{align*}
we conclude that for any $\delta > 0$, there exists a positive constant $C_\delta$ such that
\begin{align*}
    0 \le \hathperprad(|k|; \eps) \le C_\delta \left( \frac{(|K|^2 + \hoppingparam^2)^{3/2}\eps \lambda_0}{2\pi \gamma}\right)^{(1-\delta)\frac{\sqrt{|k|^2 + \hoppingparam^2}}{\sqrt{|K|^2 + \hoppingparam^2}}}, \qquad |k| \ge 0, \quad 0 < \eps < 1.
\end{align*}
Derivatives of $\hathperprad (\cdot \; ; \eps)$ satisfy a similar estimate. Therefore, as in Example \ref{ex:2} we conclude that provided
\begin{align*}
    0 < \sdpt < \frac{\sqrt{7} |K|}{\sqrt{|K|^2+\hoppingparam^2}}-2,
\end{align*}
the estimate
\begin{align*}
    |\hathperprad (|k|; \eps)| + |\hathperprad '(|k|; \eps)| + |\hathperprad ''(|k|; \eps)| \le C (C\eps)^{\frac{2+\sdpt}{\sqrt{7}}|k|/|K|}, \qquad \qquad k \in \mathbb{R}^2, \quad 0 < \eps < 1
\end{align*}
holds. Thus we have verified \eqref{eq:neighbor_bds} and \eqref{eq:4bounds}, as desired.
\end{example}}
\begin{example}
    Examples \ref{ex:2} and \ref{ex:1} can be modified to allow for angular dependence. That is, instead of \eqref{eq:example_hopping}, 
    we could define the interlayer hopping function by its Fourier transform as
    \begin{align*}
        \hathperp (k; \eps) = \frac{2\pi}{\sqrt{|k|^2+\alpha^2}} \left( \frac{\lambda_0 \eps \sqrt{|K|^2 + \alpha^2}}{2\pi}\right)^{\sqrt{\frac{|k|^2+\alpha^2}{|K|^2+\alpha^2}}} \hathperpang (\angvar{k}) \chi (|k|), \qquad k \in \mathbb{R}^2,
    \end{align*}
    where $\hathperpang \in C^2 (\mathbb{S}; \mathbb{C})$ and the function $\chi \in C^\infty [0,\infty)$ satisfies $\chi (0) = 0$ and $\chi \equiv \lambda_0^{-1}$ in $(|K|-\delta,\infty)$ for some $\delta > 0$. Recall that $\angvar{k} := k/|k|$.
    For any non-constant $\hathperpang$, the vanishing of $\chi$ at zero 
    is necessary to maintain the smoothness of $\hathperp$ at the origin. A similar construction can be applied to the Fourier transform \eqref{eq:hathperp_ex1} of the hopping function from Example \ref{ex:1}.
\end{example}

\begin{remark}
    One might wish to 
    consider hopping functions with an arbitrary sublattice dependence.
    That is, instead of \eqref{eq:intra}-\eqref{eq:Hperp}, one could define the blocks of the tight-binding Hamiltonian \eqref{eq:H} by
\begin{align}\label{eq:intra_general}
    (H_{jj}\varphi_j)^{\sigma}_{R_j} := \sum_{R_j' \in \cR_j} \sum_{\sigma'} h^{\sigma \sigma'} (\rot_{-\theta_j/2} (R_j - R_j')) (\varphi_j^{\sigma'})_{R'_j}
\end{align}
and
\begin{align}\label{eq:Hperp_general}
    (\Hperp \varphi_2)^\sigma_{R_1} &:= \sum_{R_2 \in \cR_2} \sum_{\sigma'} \hperp^{\sigma \sigma'} (R_1 - R_2;\eps) (\varphi^{\sigma'}_{2})_{R_2}, 
\end{align}
for some functions $h^{\sigma \sigma'}$ and $\hperp^{\sigma \sigma'}$ indexed by $\sigma, \sigma' \in \{A,B\}$. An example is the Fang-Kaxiras tight-binding model \cite{Fang_Kaxiras_2016}, 
which is generalized in \cite[Section II C]{vafek2023continuum}.

As we show in Section \ref{subsec:symmetries_mono}, any $H_{jj}$ of the form \eqref{eq:intra_general} that satisfies natural symmetry assumptions must also be of the form \eqref{eq:intra} for some (sublattice-independent) hopping function $h$. 

However, 
the sublattice dependence of the hopping function in \eqref{eq:Hperp_general} would in general define an interlayer Hamiltonian $\Hperp$ which is outside of our framework \eqref{eq:Hperp}.
Indeed, we see that each off-diagonal block of the effective Hamiltonian $\Hfull$ (defined by \eqref{eq:Hfull} below) is a linear combination of the matrices
\begin{align*}
    \begin{pmatrix}
        1 & 1\\
        1 & 1
    \end{pmatrix}, \qquad
    \begin{pmatrix}
        1 & \sq{e^{-i2\pi/3}}\\
        \sq{e^{i2\pi/3}} & 1
    \end{pmatrix}, \qquad
    \begin{pmatrix}
        1 & \sq{e^{i2\pi/3}}\\
        \sq{e^{-i2\pi/3}} & 1
    \end{pmatrix},
\end{align*}
which would no longer be the case if we generalized \eqref{eq:Hperp} to \eqref{eq:Hperp_general}. One could still follow the proofs in this paper to derive a (more complicated) effective Hamiltonian in the more general case \eqref{eq:Hperp_general}, provided that each component $\hperp^{\sigma \sigma'}$ of the interlayer hopping function satisfies Assumption \ref{assumption:hperp2}. In particular, the factor $$\hathperp (k_1 + \cG_1 (G_2); \eps)$$ on the right-hand side of \eqref{eq:tildeHperp} would get replaced by $$\hat{\mathfrak{h}}^{\sigma \sigma'}_{12} (k_1+ \cG_1 (G_2); \eps) := e^{-i (k_1 + \cG_1 (G_2))\cdot (\tau^\sigma_1 - \tau^{\sigma'}_2)} \hathperp^{\sigma \sigma'}(k_1 + \cG_1 (G_2); \eps),$$
meaning that $\hat{\mathscr{h}}_{12} (k,q; \eps)$ in \eqref{eq:filtered_Taylor}-\eqref{eq:tildeF12} would get replaced by the filtered first-order Taylor expansion of $\hat{\mathfrak{h}}^{\sigma \sigma'}_{12} (k;\eps)$ about $k=q$.
The rest of the derivation would be the same.
\end{remark}

\end{document}